\documentclass[twoside, 12pt, reqno]{amsart}

\usepackage{preamble}
\usepackage{lipsum}

\makeatletter
\def\@captionheadfont{\normalfont}
\def\@captionfont{\normalfont}
\makeatother

\newif\ifshowcol
\showcolfalse 
\newcommand*{\bypierre}[2][red]{{\ifshowcol\color{#1}\fi#2}}

\allowdisplaybreaks

\begin{document}
    \title[Honeycomb Schrödinger operators with incommensurate line defects]{%
        Continuum honeycomb Schrödinger operators with incommensurate line defects
    }

    \author{P. Amenoagbadji}
    \address{Université Paris-Saclay, CNRS, Laboratoire de mathématiques d'Orsay, 91405, Orsay, France}
    \email{pierre.amenoagbadji.math@gmail.com}
    %
    %
    \author{M. I. Weinstein}
    \address{Department of Applied Physics and Applied Mathematics and Department of Mathematics, Columbia University, New York, NY, USA}
    \email{miw2103@columbia.edu}

    \begin{abstract}
        We study wave propagation in 2D honeycomb structures with a non-commensu\-rate or ``irrational'' line defect or edge. Our model is a Schrödinger operator which interpolates, across the edge, between two distinct bulk (asymptotic) Hamiltonians with a common spectral gap about the ``Dirac point'' of an unperturbed honeycomb operator. 
        We seek \emph{edge states}, eigenstates that are bounded and oscillatory parallel to the edge, and decaying in the transverse direction. For non-commensurate edges, the rigorous definition of these states is nontrivial due to the lack of translation invariance along the edge. To address this, we exploit quasiperiodicity along the edge by expressing the Hamiltonian as the restriction of a 3D (degenerate elliptic) Hamiltonian describing a 3D medium with a 2D interface within which there is periodicity. 
        Via multiscale analysis, we construct approximate edge states in this 3D setting and obtain by restriction 2D edge states which are quasiperiodic along the irrational edge. These edge states are seeded by eigenfunctions of an effective Dirac operator, which has an \emph{infinite} block-diagonal structure due to the non-commensurate geometry. A consequence is that infinitely many edge state eigenpairs arise, whose energies are dense in the perturbed bulk spectral gap.
        In a forthcoming paper, we rigorously construct these gap-filling edge states under a Diophantine condition. The main result here is a key tool in this construction: a resolvent expansion for the 3D Hamiltonian, whose leading term is the resolvent of the block-diagonal Dirac operator. The validity of this expansion requires an \emph{omnidirectional} non-resonance (\emph{no-fold}) condition on the dispersion functions of the unperturbed honeycomb Hamiltonian. This condition is satisfied in the strong binding regime. In contrast with earlier works on commensurate edges, the omnidirectional condition is independent of the edge.
    \end{abstract}

    \maketitle
    
    \setcounter{tocdepth}{1}
    \tableofcontents
    \setcounter{tocdepth}{2}

    \section{Introduction}
    \noindent
    Consider a continuous periodic Schrödinger operator $\calH_0$ in $\R^2$, with two band dispersion surfaces intersecting conically at a \emph{Dirac point} of energy $E_D$. In this article, we study a weak, slowly varying domain-wall perturbation $\calH^\delta_\EDGE$ of $\calH^0$, which smoothly interpolates between two distinct asymptotic (bulk) Hamiltonians across a \emph{line defect} or \emph{edge}, $\R\, \vvh_1$; see Figure \ref{fig:soft_junction}. 

    \begin{figure}[ht!]
        \makebox[\textwidth][c]{
            \includegraphics[page=7]{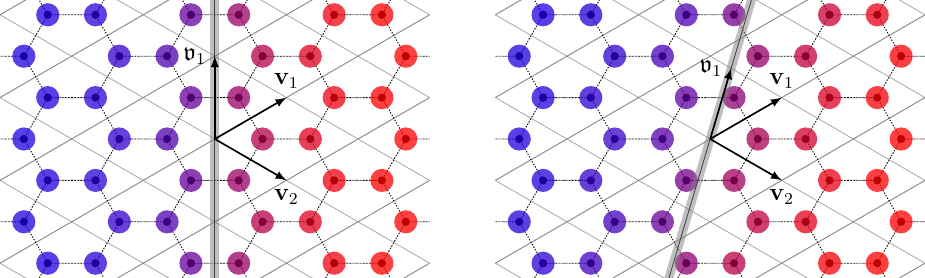}
        }
        \caption{The Hamiltonian studied in this paper models a smooth domain wall transition between two distinct bulk honeycomb Hamiltonians, across an \emph{edge}. If the edge is aligned with an element of the lattice $\Lambda = \Z \vvv_1 + \Z \vvv_2$, then it is said to be \emph{commensurate} (left panel); otherwise, it is \emph{non-commensurate} (right panel). \label{fig:soft_junction}}
    \end{figure}

    We call the edge \emph{rational} or \emph{commensurate} when it is oriented in the direction of a period lattice vector. The edge Hamiltonian $\calH^\delta_\EDGE$ is then translation-invariant along the edge, and thus can be studied using Floquet-Bloch theory. It is shown in \cite{fefferman2016edge,lee2019elliptic,drouot2019characterization,drouot2020edge} that $\calH^\delta_\EDGE$ has \emph{edge states} with energies $E$ in an $\calO (\delta)$ neighborhood of $E_D$: these are nontrivial solutions of the equation $\calH^\delta_\EDGE\, \Psi = E\, \Psi$ in $\R^2$ that are pseudo-periodic along the edge and decay in the transverse direction.

    In this article, we initiate the study of edge states when the edge is \underline{not} parallel to any period lattice vector. In this case, the edge is said to be \emph{irrational} or \emph{non-commensurate}. The main difficulty is a complete lack of translation invariance of the two-dimensional medium, which prevents a direct use of Floquet-Bloch theory.

    We propose a framework for the rigorous definition of edge states in the case of irrational edges. Our approach relies on the observation that $\calH^\delta_\EDGE$ is \emph{quasiperiodic} along the edge (see Remark \ref{rmk:quasiperiodic_nature_domain_wall_potential}). More precisely, it is realized as the restriction to a two-dimensional plane of an \emph{augmented} Hamiltonian $\calH^\delta_\AUG$ on $\R^3$. This augmented operator is translation-invariant within a planar interface, but it is degenerate elliptic. This observation has been used in the homogenization of periodic elliptic boundary value problems in presence of a sharp termination \cite{gerard2011homogenization,gerard2012homogenization} or an interface \cite{blanc2015local}. It has also been exploited in \cite{amenoagbadji2025time} to solve the Helmholtz equation numerically in a similar setting.

    Our main result, Theorem \ref{thm:resolvent_expansion}, provides, as $\delta \to 0^+$, an asymptotic resolvent expansion near the Dirac energy $E_D$ for the augmented Hamiltonian $\smash{\calH^\delta_\AUG}$ acting on an appropriate space of pseudo-periodic functions. The leading-order term in this expansion involves a countably \emph{infinite} family of effective one-dimensional Dirac operators. This contrasts with analogous results in \cite{drouot2019characterization,drouot2020edge} for rational edges, where only one or two effective Dirac operators emerge from the expansion.

    The resolvent expansion stated in Theorem \ref{thm:resolvent_expansion} plays a key role in forthcoming work \cite{amenoagbadji2026dense}, where we complete the program discussed above to construct genuine quasiperiodic edge states for irrational edges. Indeed, in \cite{amenoagbadji2026dense}, assuming a Diophantine condition on the irrational parameter which sets the orientation of the edge, we show that $\calH^\delta_\AUG$, acting on an appropriate space, admits a dense set of generalized eigenvalues that are dense in a neighborhood of the Dirac energy $E_D$. Moreover, the associated eigenfunctions can be restricted to the original two-dimensional hyperplane, hence giving rise to two-dimensional edge states with quasiperiodic behavior along the edge.

    \begin{figure}[ht!]
        \makebox[\textwidth][c]{
            \includegraphics[page=1]{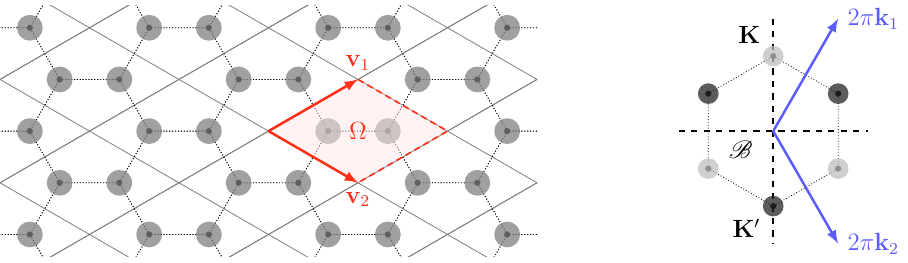}
        }
        \caption{Left: The equilateral triangular lattice $\Lambda = \Z \vvv_1 + \Z \vvv_2$ and the periodicity cell $\Omega$. $\mathbb R^2$ is the union of lattice translates of $\Omega$ and within this cell, are two lattice points. The honeycomb lattice is the union of the two triangular lattices, based at each of these points. Right: The dual lattice $\Lambda^* = \Z\vts (2\pi \kv_1) + \Z\vts (2\pi \kv_2)$ and the Brillouin zone $\scrB$, whose vertices are the high-symmetry quasi-momenta of $\Kv$ and $\Kv'$--type. \label{fig:lattice_brillouin}}
    \end{figure}
    
    \subsection{Mathematical framework}
    Before stating the main result, we briefly describe the setup. Let $\Lambda = \Z \vvv_1 + \Z \vvv_2$ and $\Lambda^* := \Z\vts (2\pi \kv_1) + \Z\vts (2\pi \kv_2)$ denote, respectively, the equilateral triangular lattice and its dual, shown in Figure \ref{fig:lattice_brillouin}, with $\kv_n \cdot \vvv_l = \delta_{nl}$. The Brillouin zone $\scrB$, a particular choice of a fundamental cell for $\Lambda^*$, is a regular hexagon shown in Figure \ref{fig:lattice_brillouin}. Its vertices consist of the points $\Kv$, $\Kv' = -\Kv$, and their images under rotations by integer multiples of $2\pi / 3$.

    Let $\calH^0 = -\Delta + V(\xv)$, where $V (\xv)$ is real-valued, bounded and $\Lambda$--periodic. For every quasi-momentum $\kv \in \scrB$, the eigenvalue problem $(\calH^0 - E)\, u = 0$ with $u(\xv + \vvv) = \euler^{\icplx\vts \kv \cdot \vvv}\, u(\xv)$ ($\vvv \in \Lambda$) has countably many eigenvalues $E_1 (\kv) \leq E_2 (\kv) \leq \dots$. The maps $\kv \in \scrB \mapsto E_b (\kv)$ are \emph{dispersion functions}; their graphs are \emph{dispersion surfaces}, and their ranges, $E_b (\scrB)$, are \emph{bands}. The union of these bands over $b \geq 1$ is the $L^2 (\R^2)$--spectrum of $\calH^0$; see Section \ref{sec:FB_states}.

    A \emph{honeycomb lattice potential} $V(\xv)$ is a smooth, real-valued function which is $\Lambda$--periodic, even, and invariant under clockwise rotation by $2 \pi/3$. For generic honeycomb lattice potentials, the Schrödinger operator $\calH_0 = -\Delta + V(\xv)$ has \emph{Dirac points} in its band structure. These are quasi-momentum/energy pairs, $(\Kv_\star, E_D)$, at which two dispersion surfaces touch conically over  the vertices $\Kv_\star$ of the Brillouin zone $\scrB$ \cite{colin1991singularites,grushin2009multiparameter,fefferman2012honeycomb,fefferman2017topologically,berkolaiko2018symmetry,lee2019elliptic,drouot2021ubiquity,cassier2021high,drouot2024band}; see Section \ref{sec:dirac_points}. The existence of Dirac points is rooted in the inversion, complex conjugation, and rotational symmetries of $\calH^0$. Concerning their stability, roughly speaking, Dirac points persist under small perturbations of $\calH^0$ that break rotational invariance (the dispersion cones merely translate and distort), but breaking either inversion or complex conjugation symmetry generically opens a local energy gap around the Dirac point. 

    One such perturbed operator is $\calH^\delta = \calH^0 + \delta\, \nabla \cdot a(\xv)\, \sigma_2\, \nabla$, where $\sigma_2$ is the second Pauli matrix, and $a(\xv)$ is smooth, $\Lambda$--periodic, even, and real-valued; see Section \ref{sec:motivation} for the physical interpretation of this perturbation. Because it breaks complex conjugation symmetry, for small enough nonzero $\delta$, the perturbation induced by $a(\xv)\, \sigma_2$ opens a local gap around the Dirac energy, $E_D$, provided that a generic nondegeneracy condition is satisfied (Section \ref{sec:conjugation_breaking_perturbations}).

    We then introduce a slow interpolation between the bulk operators $\calH^{-\delta}$ and $\calH^{+\delta}$, via a smooth \emph{domain wall function}, $\kappa (\zeta)$, which tends to $\pm 1$ as $\zeta \to \pm \infty$. The interpolation takes place in a direction transverse to an edge $\R\, \vvh_1$, where $\vvh_1 = \vvv_1 + r\vts \vvv_2$ for $r \in \R$. The corresponding operator is given by
    \begin{equation*}
        \calH^\delta_\EDGE = -\nabla \cdot A^\delta (\xv)\, \nabla + V (\xv) \ \ \textnormal{on $L^2 (\R^2)$}, \qquad A^\delta (\xv) = I - \delta\, \kappa (\delta\, \kvh_2 \cdot \xv)\, a(\xv)\, \sigma_2,
    \end{equation*}
    where $\kvh_2 := -r\vts \kv_1 + \kv_2$ is such that $\kvh_2 \cdot \vvh_1 = 0$. We call the edge $\R\, \vvh_1$ rational (commensurate) if $r \in \Q$; see Figure \ref{fig:soft_junction}, left panel. Otherwise, we say that the edge $\R\, \vvh_1$ is irrational (non-commensurate); see Figure \ref{fig:soft_junction}, right panel.

    Formally, an \emph{edge state} with energy $E \in \R$ is a nontrivial solution of the eigenvalue problem $\calH^\delta_\EDGE\, \Psi = E\, \Psi$ that decays as $|\kvh_2 \cdot \xv| \to \infty$ and exhibits ``\emph{plane wave-like}'' behavior along the edge $\R\, \vvh_1$. The precise nature of this plane wave-like behavior is determined by the invariance properties of $\calH^\delta_\EDGE$ in the $\vvh_1$--direction. In fact, if $r = b_1 / a_1$, where $a_1, b_1 \in \Z$ are coprime, then $\calH^\delta_\EDGE$ is invariant under the translation $\xv \mapsto \xv + a_1\vts \vvh_1$. Associated to this translation invariance is a parallel quasi-momentum $\kpar \in [-\pi, \pi)$, and the plane wave-like behavior corresponds to \emph{pseudo-periodicity}: $\Psi (\xv + a_1\vts \vvh_1) = \euler^{\vts \icplx\, \kpar}\, \Psi (\xv)$. Such solutions are called \emph{edge states with parallel quasi-momentum $\kpar$}.

    In contrast, little seems to be known for irrational edges. Since $\calH^\delta_\EDGE$ is not translation-invariant along the edge, the appropriate notion of the plane wave-like behavior of edge states requires a careful formulation. An important unanswered question is the following:
    \begin{question}
        For irrational edges, what are the detailed properties of the generalized states associated with the $L^2 (\R^2)$--spectrum of $\calH^\delta_\EDGE$ near the Dirac energy $E_D$, in particular their character along the edge and their decay behavior in the transverse direction?
    \end{question}
    We address this question in the present article and in the follow-up \cite{amenoagbadji2026dense}. Here, we introduce a formal definition of edge states in terms of an \emph{augmented}, three-dimensional Hamiltonian. Further, we prove in Theorem \ref{thm:resolvent_expansion} that the resolvent of the augmented Hamiltonian, near the Dirac energy, is governed by a countable family of one-dimensional effective Dirac operators. In \cite{amenoagbadji2026dense}, under a Diophantine condition on $r$, we use this resolvent expansion to answer the above question.

    \subsection{Main contributions of the present article}
    The goal of this article is to introduce and study the augmented operator, $\calH^\delta_\AUG$, and the associated eigenvalue problem used to define edge states for irrational edges. Our main result, Theorem \ref{thm:resolvent_expansion}, gives an expansion, near $E_D$, of the resolvent of $\smash{\calH^\delta_\AUG}$ on an appropriate space. By itself, Theorem \ref{thm:resolvent_expansion} does not prove the existence of edge states for $\calH^\delta_\EDGE$. Its main role is to identify the effective Dirac operators governing the augmented eigenvalue problem and obtain bounds that will be used to construct edge states for $\calH^\delta_\EDGE$ in \cite{amenoagbadji2026dense}. Besides, Theorem \ref{thm:resolvent_expansion} generalizes the resolvent expansions obtained in \cite{drouot2019characterization,drouot2020edge} for rational edges, and has implications for functional calculus; see Appendices \ref{sec:rational_matters} and \ref{sec:functional_calculus}.

    \medskip\noindent
    \textbf{Edge states for irrational edges.} %
    Although $\calH^\delta_\EDGE$ is not translation-invariant along irrational edges, we show in Section \ref{sec:quasiperiodicity_irrational_edge} that it extends to an \emph{augmented} Schrödinger operator, given for $(\xv, s) \in \R^2 \times \R$ by
    \begin{equation*}
        \calH^\delta_\AUG := -\nabla_\xv \cdot A^\delta_\AUG (\xv, s)\, \nabla_\xv + V (\xv)\ \ \textnormal{on $L^2 (\R^3)$},
    \end{equation*}
    where the coefficient $A^\delta_\AUG$, defined by \eqref{eq:def_augmented_potential}, satisfies $A^\delta_\AUG (\xv, s)|_{s = 0} = A^\delta (\xv)$. The augmented operator $\calH^\delta_\AUG$ also interpolates, via the domain wall function $\kappa (\cdot)$, between bulk augmented operators across a two-dimensional interface $\Gamma$ that contains the edge $\R\, \vvh_1 \times \{0\}$. Furthermore, \emph{this interface is spanned by two linearly independent vectors $\av_1, \av_2 \in \R^3$, along which $\calH^\delta_\AUG$ is translation-invariant}; see \eqref{eq:augmented_potential_property}.

    Accordingly, in Section \ref{sec:formulation_eigenvalue_problem_irrational_edges}, we formally define edge states as functions that can be written in the form $\Psi (\xv) = \Psi_\AUG (\xv, s)|_{s=0}$,\ where $\Psi_\AUG$ is a nontrivial solution of
    \begin{equation}\label{eq:augmented_evp}
        \left\{
            \begin{array}{l}
            \displaystyle
            \calH^\delta_\AUG\, \Psi_\AUG(\xv, s) = E\, \Psi_\AUG(\xv, s),
            \rets
            \Psi_\AUG(\cdot + \av_1) = \euler^{\vts \icplx\vts \kpar}\, \Psi_\AUG (\cdot), \quad \Psi_\AUG(\cdot + \av_2) = \euler^{\vts \icplx\vts \kperp}\, \Psi_\AUG (\cdot),
            \rets
            \Psi_\AUG(\xv, s) \to 0, \quad |s| \to \infty.
            \end{array}
        \right.
    \end{equation}
    In other words, $\Psi_\AUG$ is pseudo-periodic within the interface $\Gamma$ and decays in the direction transverse to $\Gamma$. Since $\calH^\delta_\AUG$ has no derivatives in $s$, multiplication by $\euler^{-\icplx\vts \tau\, s}$, for any $\tau \in \R$, changes $(\kpar, \kperp)$ into $(\kpar - \tau\, a_{1, s}, \kperp - \tau\, a_{2, s})$ where $a_{1, s}$ and $a_{2, s}$ are the respective $s$--components of $\av_1$ and $\av_2$; see the discussion before \eqref{eq:evp_irrational_edges}. Since $a_{2, s} \neq 0$, we may take $\kperp$ to be equal to zero without loss of generality. A solution of \eqref{eq:augmented_evp} is called an \emph{augmented edge state with parallel quasi-momentum $\kpar \in \R$}.
    
    For the rigorous formulation of \eqref{eq:augmented_evp} in Section \ref{sec:formulation_augmented_EVP_Sobolev_space}, we work in the space $L^2_{\kpar} (\cylaug)$ defined by \eqref{eq:def_L2aug} and equipped with the $L^2 (\cylaug)$--norm, where $\cylaug$ is the cylinder $\R^3 / (\Z\vts \av_1 + \Z\vts \av_2)$. Thus, \emph{augmented edge states with parallel quasi-momentum $\kpar \in \R$ are eigenfunctions of the operator $\calH^\delta_{\AUG, \kpar} = -\nabla_\xv \cdot A^\delta_\AUG (\xv, s)\, \nabla_\xv + V(\xv)$ acting on $L^2_{\kpar} (\cylaug)$}.
    
    Formally, if $\calH^\delta_{\AUG, \kpar}$ admits an eigenfunction that is continuous in $s$, and if its restriction to $s = 0$ is not identically zero, then this restriction is a two-dimensional edge state which is plane wave-like along the edge, in the sense that its restriction to $\R\, \vvh_1$ is a \emph{quasiperiodic} function; see also Remark \ref{rmk:quasiperiodic_nature_domain_wall_potential}. \bypierre{The rigorous construction of eigenfunctions continuous with respect to $s$ with nonzero restrictions to $s = 0$ is implemented in \cite{amenoagbadji2026dense}.}

    \medskip\noindent
    \textbf{Multiscale construction of approximate augmented edge states.} %
    Since we are interested in edge states that bifurcate from the Dirac point, we focus for simplicity on $\kpar = \Kv \cdot \vvh_1$. In the definition \eqref{eq:def_augmented_potential} of $A^\delta_\AUG$, for small $\delta > 0$, the length scale of the domain-wall transition across the edge is large compared with the scale of periodicity. Section \ref{sec:ansatz_and_multi_scale_analysis} takes advantage of this scale separation: using multiscale analysis, we construct a family of approximate eigenpairs of $\calH^\delta_{\AUG, \kpar = \Kv \cdot \vvh_1}$, with index set $\bbL (\delta) \subset \bbL := \{\Kv, \Kv'\} \times \Z$. For every $\iiv \in \bbL (\delta)$, the approximate eigenvalue is located in an $\calO (\delta)$ neighborhood of the Dirac energy $E_D$, and the approximate eigenfunction has the structure of a wave packet, whose slowly varying envelope is an eigenfunction of an effective Dirac operator $\calD^{\vts\delta}_\iiv$ defined on $H^1 (\R; \C^2)$.

    While this multiscale procedure was used in previous works for rational edges, a novel aspect of this work, caused by the non-ellipticity of $\calH^\delta_{\AUG, \kpar = \Kv \cdot \vvh_1}$, is the emergence of a countably \emph{infinite family of pairwise distinct effective Dirac operators}, and the implied \emph{countably infinite} family of approximate augmented edge states (indexed by $\bbL(\delta)$). For irrational edges, each of these states, restricted to $s = 0$, leads to a distinct, $\iiv$--dependent, two-dimensional approximate edge state for $\calH^\delta_\EDGE$; see Section \ref{sec:recap_multiscale}.

    \medskip\noindent
    \textbf{Resolvent expansion.} %
    \bypierre{Fix $\kpar = \Kv \cdot \vvh_1$ and $\iiv \in \bbL (\delta)$, and consider an approximate, $\iiv$--dependent eigenpair $(E_D + \delta\vts z_0, \Psi_{\AUG, 0} (\xv, s))$ of $\calH^\delta_{\AUG, \kpar}$, constructed via multiscale analysis. For this approximate eigenpair to yield a \emph{genuine} eigenpair $(E_D + \delta\vts z_0 + \delta^2\vts \tau, \Psi_{\AUG, 0} + \delta\, U)$, the corrector pair $(\tau, U (\xv, s))$ is required to satisfy an equation of the form
    \begin{equation}\label{eq:corrector_equation}
        \bigl(\vts \calH^\delta_{\AUG, \kpar} - E_D - \delta\vts z_0 - \delta^2\vts \tau \vts\bigr)\, U (\xv, s) = F (\xv, s),
    \end{equation}
    where the right-hand side depends on $E_D$, $\delta$, $z_0$, $\tau$, and $\Psi_{\AUG, 0}$. A critical step for obtaining the corrector pair is therefore to understand the behavior of the resolvent of $\calH^\delta_{\AUG, \kpar = \Kv \cdot \vvh_1}$ near the Dirac energy $E_D$. This motivates the main result of this article, Theorem \ref{thm:resolvent_expansion}, summarized as follows. For $\delta > 0$ small enough, there exist
    \begin{itemize}
        \item a countable subset $\bbL (\delta^{3/4}) \subset \bbL := \{\Kv, \Kv'\} \times \Z$,
        \item an operator, $\calJ_\delta$, bounded from $L^2_{\kpar = \Kv \cdot \vvh_1} (\cylaug)$ to $\ell^2 (\bbL; L^2(\R; \C^2))$, the space of square-summable sequences indexed by $\bbL$ with values in $L^2(\R; \C^2)$,
        \item a family of Dirac operators $\calD^{\vts \delta}_\iiv: H^1 (\R; \C^2) \to L^2 (\R; \C^2)$, $\iiv \in \bbL$,
    \end{itemize}
    such that the following expansion holds in operator norm as $\delta \to 0^+$, uniformly in $z$ varying in a compact subset of $\C \setminus \R$:
    \begin{equation}\label{eq:resolvent_expansion_intro}
        \bigg(\vts \frac{\calH^\delta_{\AUG, \kpar = \Kv \cdot \vvh_1} - E_D}{\delta} - z \vts\bigg)^{-1}\ =\  \calJ^*_\delta\; \mathds{1}_{\bbL (\delta^{3/4})}\; \big(\vts \calD^{\vts \delta} - z \vts\big)^{-1}\; \mathds{1}_{\bbL (\delta^{3/4})}\; \calJ^{}_\delta\ +\ \calO (\delta^{1/4}).
    \end{equation}
    Here, for $\bbL' \subseteq \bbL$, $\mathds{1}_{\bbL'}$ is the projection operator $\operatorname{diag} (\mathds{1}_{\iiv \in \bbL'})_{\iiv \in \bbL}$ bounded in $\ell^2 (\bbL; L^2(\R; \C^2))$, and $\calD^{\vts \delta}$ is the block-diagonal Dirac operator:
    \begin{equation*}
        \calD^{\vts \delta} = %
        \begin{pmatrix}
            \ddots & & (0) 
            \\
            & \calD^{\vts \delta}_\iiv & 
            \\
            (0) & & \ddots
        \end{pmatrix}_{\iiv \in \bbL}: \ell^2 (\bbL; H^1(\R; \C^2)) \to \ell^2 (\bbL; L^2(\R; \C^2)).
    \end{equation*}%
    The choice of the $3/4$-exponent in $\bbL (\delta^{3/4})$ is mainly technical; see Remark \ref{rmk:technical_exponent}.} 
    
    In the case of rational edges, analogous resolvent expansions were established in \cite{drouot2019characterization,drouot2020edge} for the two-dimensional Schrödinger operator $\calH^\delta_\EDGE$, acting on a suitable space. Equation \eqref{eq:resolvent_expansion_intro}, though formulated in terms of a three-dimensional operator, provides a sense in which these resolvent expansions  generalize to the case of irrational edges. In fact, Theorem \ref{thm:resolvent_expansion_rational} shows that \eqref{eq:resolvent_expansion_intro} is consistent with the results in \cite{drouot2019characterization,drouot2020edge} when $r$ is rational. 

    We discuss applications of the resolvent expansion in Theorem \ref{thm:resolvent_expansion}. It can be combined with the Helffer-Sjöstrand formula to estimate operators of the form $w(\vts\delta^{-1}\, (\calH^\delta_{\AUG, \Kv \cdot \vvh_1} - E_D) \vts)$, with $w \in \scrC^\infty_0 (\R)$; see Corollary \ref{cor:functional_calculus_estimates}. Thus, with appropriate choices of $w$, Corollary \ref{cor:functional_calculus_estimates} enables us to estimate smoothed spectral projections of $\calH^\delta_{\AUG, \Kv \cdot \vvh_1}$ near $E_D$, or to derive finite but large time effective dynamics of $\exp(-\icplx\, t\, \calH^\delta_{\AUG, \Kv \cdot \vvh_1})$, for wave packet initial data which are spectrally localized near the Dirac point.

    A key assumption under which we derive \eqref{eq:resolvent_expansion_intro} is an {\it omnidirectional no-fold} condition, which we introduce and motivate in Section \ref{sec:no_fold_condition}. The validity of this condition in the strong binding regime has been established in \cite{fefferman2018honeycomb}.

    \subsection{Further discussion and future directions}
    The spectrum of the block-diagonal effective Hamiltonian, $\calD^{\vts \delta}$, consists of an absolutely continuous part of the form $\R \setminus (-\theta_\GAP, \theta_\GAP)$, for $\theta_\GAP > 0$, and a pure point part. \bypierre{These spectral parts overlap for irrational edges: we prove that} the point spectrum of $\calD^{\vts \delta}$ is dense, and that a subset of its eigenvalues densely fill $(-\theta_\GAP, \theta_\GAP)$; its other eigenvalues are embedded in the absolutely continuous spectrum; see Figure \ref{fig:spectrum_block_diagonal_operator}. %
    Two questions arise:
    \begin{itemize}
        \item Do the eigenvalues of $\calD^{\vts \delta}$ in $(-\theta_\GAP, \theta_\GAP)$ seed genuine eigenvalues of $\calH^\delta_{\AUG, \Kv \cdot \vvh_1}$?
        \item How do the eigenvalues of $\calD^{\vts \delta}$ embedded in its absolutely continuous spectrum manifest in the spectral properties of $\calH^\delta_{\AUG, \Kv \cdot \vvh_1}$\bypierre{, and subsequently of $\calH^\delta_\EDGE$}?
    \end{itemize}
    What makes these questions challenging is that the eigenvalues of $\calD^{\vts \delta}$, whether embedded \bypierre{in the absolutely continuous spectrum or not, are not ``\emph{well-separated from one another}'', in the sense that every neighborhood of an eigenvalue contains (countably many) other eigenvalues of $\calD^{\vts \delta}$}. Therefore, general arguments (based on quasi-modes) for gapped self-adjoint operators \cite{drouot2019characterization,drouot2020edge,drouot2020defect} do not apply. What makes the first question nontrivial is that dense point spectrum is unstable under general perturbations; see \cite{donoghue1965perturbation,carey1976unitary,simon1986singular,howland1987perturbation}. Proving that the dense eigenvalues of $\calD^\delta$ seed eigenvalues of $\calH^\delta_{\AUG, \Kv \cdot \vvh_1}$ is therefore a delicate task.
    
    For the second question, while ``well-separated'' eigenvalues embedded in the absolutely continuous spectrum are known under certain conditions to perturb to scattering resonances \cite{friedrichs1948perturbation,howland1968perturbation,gohberg1971operator,soffer1998time,costin2001resonance} (see also \cite[Section XII.6]{simon1978methodsvol4} and \cite[Chapter 16]{hislop2012introduction}), \bypierre{the behavior of a dense set of embedded eigenvalues under perturbation remains an open question.}

    On the other hand, we provide an affirmative answer to the first question in \cite{amenoagbadji2026dense}. Specifically, we prove that if the irrational edge parameter, $r$, satisfies a \emph{Diophantine} condition, then the spectrum of $\calH^\delta_{\AUG, \Kv \cdot \vvh_1}$, near $E_D$, consists of a \emph{dense} set of eigenvalues seeded by the eigenvalues of $\calD^{\vts \delta}$. In \cite{amenoagbadji2026dense}, these eigenvalues and the associated eigenfunctions are constructed explicitly by solving the corrector equation \eqref{eq:corrector_equation} similarly to \cite{fefferman2016edge,fefferman2017topologically}, rather than using quasi-mode arguments \cite{drouot2019characterization,drouot2020edge}. 
    
    \bypierre{Further, we prove in \cite{amenoagbadji2026dense} that these eigenfunctions are smooth in $s$. This does not follow from any regularity property of the degenerate operator $\calH^\delta_{\AUG, \Kv \cdot \vvh_1}$. Instead, our approach to smoothness in $s$ is motivated by differentiating \eqref{eq:augmented_evp} with respect to $s$. This procedure is rigorously implemented in \cite{amenoagbadji2026dense} by bounding the weighted resolvent $\langle D_s \rangle^m\, \mathfrak{R} (E_D)\, \langle D_s \rangle^{-m}$, where $\mathfrak{R} (E_D)$ denotes the resolvent of $\calH^\delta_{\AUG, \Kv \cdot \vvh_1}$ at $E_D$ reduced on an appropriate space.} As a consequence, the eigenfunctions, restricted to $s = 0$, yield genuine two-dimensional edge states. This result indicates that the spectrum of $\calH^\delta_\EDGE$ near $E_D$ corresponds to states that are quasiperiodic along the edge and decay in the direction transverse to the edge.

    \subsection{Motivation and previous work}\label{sec:motivation}~
    The study of honeycomb materials in condensed matter physics was sparked by the experimental isolation of graphene, a crystalline two-dimensional hexagonal arrangement of carbon atoms extending to a macroscale \cite{geim2007rise,novoselov2011nobel}. The remarkable properties of quantum electronic waves in graphene have catalyzed investigations of wave propagation in diverse engineered two-dimensional structures with  hexagonal and other novel symmetries. These so-called \emph{artificial graphene} systems have been studied and realized experimentally, for instance, in electronics \cite{singha2011two}, photonics \cite{peleg2007conical,khanikaev2013photonic,rechtsman2013topological,rechtsman2013photonic,bellec2013tight,wu2015scheme,ozawa2019topological,guglielmon2021landau,barsukova2024direct}, acoustics \cite{torrent2012acoustic, khanikaev2015topologically, mousavi2015topologically, yang2015topological}, and mechanics \cite{nash2015topological}; see also the survey \cite{shah2024colloquium} for further references.

    Many of the interesting properties of graphene-like materials are related to the presence of \emph{Dirac points} in their Floquet-Bloch band structure (Section \ref{sec:FB_theory}), which encodes the spectral properties of an underlying periodic Hamiltonian, and hence wave propagation properties of the system. \bypierre{The existence of Dirac points for the Hamiltonian $\calH^0$ was established in} \cite{fefferman2012honeycomb,fefferman2017topologically}; %
    see also \cite{berkolaiko2018symmetry} for a group representation perspective, and \cite{drouot2021ubiquity} for a discussion of why conical degeneracies are the generic type of degeneracy in Hermitian Hamiltonians.
    A manifestation of this phenomenon is the emergence of conical degeneracies when a $\Z^2$--lattice potential with quadratic degeneracies is slightly deformed; see \cite{keller2018spectral,chaban2025instability}.
    The existence of Dirac points was studied as well for $\|V\|_\infty$ small \cite{grushin2009multiparameter, colin1991singularites}, when $V$ is a superposition of Dirac delta distributions \cite{lee2016dirac}, and in the strong binding regime \cite{fefferman2018honeycomb}. Similar results hold for divergence-form elliptic operators \cite{ammari2020honeycomb,cassier2021high,lee2019elliptic,li2024interface}, quantum graph Hamiltonians \cite{kuchment2006spectra}, and continuum Laplacians on graph-like domains \cite{delourme2024guided}, all exhibiting honeycomb symmetries. A general framework for the study of band dispersion degeneracies in $n$-dimensional lattices can be found in \cite{drouot2024band}.

    The continuum Schrödinger Hamiltonians $\calH^\delta_\EDGE$ studied in this article were introduced in \cite{fefferman2016edge}, building on ideas from \cite{haldane2008possible,raghu2008analogs}, about the realization of quantum Hall-like edge states in photonic crystals. Further mathematical investigations appear in \cite{drouot2019characterization,lee2019elliptic,drouot2020edge}, and related work on square and deformed square lattices can be found in \cite{chaban2025instability,chaban2025edge}. 

    For rational edges, edge states can be characterized as solutions of a family of Floquet-Bloch eigenvalue problems depending on a parallel quasi-momentum, $\kpar$. The corresponding eigenvalues, represented in terms of $\kpar$, form a collection of \emph{edge state curves}. An algebraic count of the number of gap traversing edge state curves (spectral flow) is equal to the difference of bulk Chern numbers, integer-valued topological indices determined by the asymptotic \underline{bulk} band structures. In cases where complex conjugation symmetry is broken, this difference may be non-zero, ensuring the existence of gap traversing edge state curves. Spectral flow is also related to a notion of quantum conductivity. %
    The corresponding edge states are said to be {\it topologically protected} since a relatively compact perturbation of the Hamiltonian (even large) does  not change topological indices which depend only on the bulk. This is a manifestation of  the bulk-edge correspondence principle; see, for example, \cite{hatsugai1993chern,elbau2002equality,elgart2005equality,hasan2010colloquium,graf2013bulk,drouot2024topological,drouot2024bulk} for discrete models and \cite{bal2019continuous,drouot2020edge,gontier2020edge,drouot2021microlocal,drouot2021bulk,shapiro2022tight,quinn2024approximations,bal2022topological,bal2023topological} for continuous Hamiltonians.

    For $\delta$ small, the bifurcation of edge states from the Dirac point has been established in \cite{fefferman2016edge,lee2019elliptic,drouot2019characterization,drouot2020edge} using a Lyapunov-Schmidt/Schur complement reduction procedure. %
    Sharp junctions, that is, discontinuous transitions between bulk media (realized, for instance, by taking $\kappa (\zeta) = \sign (\zeta)$), were investigated in \cite{li2024interface} using an integral equation formulation of the edge state problem and a generalized Rouché theorem. Similar bifurcation results were derived for one-dimensional systems \cite{fefferman2014topologically,fefferman2017topologically,drouot2020defect}, periodic waveguides \cite{qiu2023mathematical}, and square lattice potentials \cite{chaban2025edge,qiu2024mathematical}; we also refer to \cite{fefferman2020continuum,delourme2024guided} for related results in different perturbation settings (strong binding regime and thin graph-like domains), and to \cite{fefferman2022discrete,fefferman2024discrete,graf2013bulk,gomez2025edge} for non-perturbative results in discrete models. 

    The case of irrational edges, on the other hand, seems to be far less well understood, due to the lack of translation invariance of $\calH^\delta_\EDGE$ along the edge. The results in \cite{hempel2011spectral,hempel2012dislocation,gontier2021spectral} for models of periodic media with sharp terminations or junctions between two periodic media suggest that \emph{all} gaps in the $L^2(\R^2)$--spectrum of the bulk operators $\calH^{\pm \delta}$ are filled with ``edge spectrum''. These papers do not provide information about the type of spectrum filling bulk gaps or the nature of the corresponding modes. Results in \cite{bols2022absolutely,drouot2025edge} for a discrete variant of $\calH^\delta_\EDGE$ suggest that the gaps are filled with absolutely continuous spectrum, although there could be other spectral types. \bypierre{It is worth noting that \cite{bols2022absolutely,drouot2025edge} do not rely explicitly on the irrational nature of the edge.}

    The interpretation of $\calH^\delta_\EDGE$ as the restriction of a higher-dimensional operator, $\calH^\delta_\AUG$, is the continuum analogue of the \emph{cut-and-project} method \cite{duneau1985quasiperiodic,kalugin19850,elser1986diffraction}, originally introduced to model quasicrystals \cite{shechtman1984metallic,de1981algebraic,kramer1984periodic}. \bypierre{The idea of introducing and studying an augmented operator such as $\calH^\delta_\AUG$ already appears in Kozlov's work \cite{kozlov1979averaging} and has been used for the homogenization of quasicrystals and continuum quasiperiodic structures \cite{bouchitte2010homogenization,wellander2019homogenization,ferreira2021homogenization}}. Numerical studies of electronic, electromagnetic, or water wave propagation in non-commensurate systems based on this approach appear in \cite{rodriguez2008computation,wilkening2021quasi,amenoagbadji2023wave,nicholls2025method,wang2025convergence,nicholls2025analyticity,jiang2025projection}. \bypierre{To the best of our knowledge, however, the spectral analysis of degenerate elliptic operators such as $\calH^\delta_\AUG$ has received little attention}.

    We finish by noting that the spectral properties of quasiperiodic systems are known to be very subtle. A well-known example which reflects this complexity is given by the almost Mathieu operator $-\Delta_\Z + 2 \lambda \cos (2\pi\vts \alpha\vts n + x)$, which models an electron in a magnetic field with constant flux $\alpha \in \R$ on a lattice. The spectrum of this operator, for $\lambda \neq 0$ and $\alpha$ irrational, was conjectured to be a Cantor set (that is, a compact set with no isolated points and with empty interior) \cite{azbel1964energy}. This conjecture was proven in \cite{avila2009ten}. For a similar result for a graphene model, see \cite{becker2019cantor}. Related contributions can be found, for instance, in \cite{grempel1982localization,bellissard1983localization,simon1985almost,casdagli1986symbolic,suto1987spectrum,eliasson1997discrete,jitomirskaya1999metal,bourgain2000nonperturbative,bourgain2002absolutely,klein2005anderson,avila2006reducibility,damanik2016fibonacci,avila2017spectral,jitomirskaya2018all} for discrete models, and in \cite{dinaburg1975one,sarnak1982spectral,moser1984extension,sinai1987anderson,frohlich1990localization,eliasson1992floquet} for continuous models.

    \subsection{Outline}
    In Section \ref{sec:FB_theory}, we review Floquet-Bloch theory for two-dimensional Schrödinger operators with periodic potentials, introduce the notion of honeycomb lattice potentials, and recall existing results on the existence of Dirac points of $\calH^0$ \bypierre{and their stability under symmetry-breaking perturbations. We also introduce the key assumptions of the article, including the omnidirectional no-fold condition; see Section \ref{sec:no_fold_condition}}. %
    \bypierre{In Sections \ref{sec:dw_modulated_operator} and \ref{sec:irrational_edges}, we present the construction of the edge operator, $\calH^\delta_\EDGE$, for an irrational edge orientation, $r \in \R$. In Section \ref{sec:quasiperiodicity_irrational_edge}, we highlight the quasiperiodic nature of $\calH^\delta_\EDGE$ along the edge. This allows us to propose, in Section \ref{sec:formulation_eigenvalue_problem_irrational_edges},} a notion of edge states for irrational edges: these states are defined as restrictions of three-dimensional augmented edge states. %
    In Section \ref{sec:ansatz_and_multi_scale_analysis}, we describe the multiscale asymptotic expansion procedure used to formally construct augmented edge states. These approximate edge states arise from the point eigenvalues of the block-diagonal effective Dirac operator $\calD^{\vts \delta}$. The spectral properties of $\calD^{\vts \delta}$ are also discussed; see Section \ref{sec:pties_effective_Dirac_operator}.
    In Section \ref{sec:formulation_augmented_EVP_Sobolev_space}, we introduce an appropriate Sobolev space framework for the rigorous study of the augmented edge state eigenvalue problem. We then state and comment on the main result, Theorem \ref{thm:resolvent_expansion}, in Section \ref{sec:main_results}, before providing the strategy for its proof in Section \ref{sec:strategy_proof}. 
    \bypierre{In Section \ref{sec:spectrum_unperturbed_augmented_operator}, in preparation for the proof of Theorem \ref{thm:resolvent_expansion}, we study the spectral properties of the unperturbed augmented operator, $\calH^0_{\AUG, \Kv \cdot \vvh_1}$, and} we introduce the notion of edge quasi-momentum slices, parameterizing those Floquet-Bloch modes that participate in the construction of augmented edge states. Two useful tools for the analysis in this section are the Fourier transform with respect to the augmented variable $s$, and a Wrapping Lemma \ref{lem:wrapping_lemma} to characterize neighborhoods of high-symmetry quasi-momenta.
    In Sections \ref{sec:proof_Theorems_A_B} and \ref{sec:asymptotic_expansions}, we present the proof of Theorem \ref{thm:resolvent_expansion} using a Schur complement reduction strategy. %
    Appendix \ref{sec:rational_matters} shows how the resolvent expansion in Theorem \ref{thm:resolvent_expansion} reduces to known results for rational edges; see Theorem \ref{thm:resolvent_expansion_rational}.
    Appendix \ref{sec:functional_calculus} presents an application of the resolvent expansion in Theorem \ref{thm:resolvent_expansion} by expanding smooth functions of the centered and rescaled operator $\delta^{-1}\, (\calH^\delta_{\AUG, \Kv \cdot \vvh_1} - E_D)$; see Corollary \ref{cor:functional_calculus_estimates}.
    \bypierre{Finally, in Appendix \ref{sec:directional_Fourier}, we prove the properties of the partial Fourier transform in $s$ that were stated in Section \ref{sec:Fourier_analysis_L2aug}.} 

    \subsection*{Notation and conventions}
    \begingroup\small
    \begin{itemize}[leftmargin=*]
        \item $\N_0 := \{0,1,2,...\}$ and $\N = \{1,2,...\}$ denote the sets of non-negative and positive integers, respectively.
        \item Given $d \geq 1$ and $\xv = (x_1, \dots, x_d), \yv = (y_1, \dots, y_d) \in \R^d$, let $\xv \cdot \yv := x_1 y_1 + \dots + x_d y_d$ and $|\xv|^2 := \xv \cdot \xv$. For $d = 2$, we define $\xv \wedge \yv = x_1 y_2 - x_2 y_1$.
        \item Given two Hilbert spaces $(\scrX, \|\cdot\|_\scrX)$, $(\scrY, \|\cdot\|_\scrY)$, we call $\scrL (\scrX, \scrY)$ the space of \emph{bounded} linear operators from $\scrX$ to $\scrY$, with the norm $\smash{\displaystyle \| T \|_{\scrL (\scrX, \scrY)} :=\!\!\!\!\! \sup_{x \in \scrX, \|x\|_\scrX = 1}\!\!\! \|T\, x\|_\scrY}$ for any $T \in \scrL (\scrX, \scrY)$. We set $\scrL(\scrX) := \scrL(\scrX, \scrX)$. 
        \item Given a measurable set $S \subset \R^d$ and a Hilbert space $(\scrX, \|\cdot\|_\scrX)$, we denote by $L^2(S; \scrX)$ the Hilbert space of strongly measurable functions $u: S \mapsto \scrX$ such that $\int_S \|u(t)\|^2_\scrX\, dt < \infty$. Similarly, for any discrete set $\bbI$, we define the space
        \begin{equation*}
           \ell^2(\bbI; \scrX) := \bigg\{ (g_\iiv)_{\iiv \in \bbI} \in \scrX^{\bbI}\ \ \bigg|\ \ \sum_{\iiv \in \bbI} \|g_\iiv\|^2_\scrX < \infty  \bigg\},
        \end{equation*}
        equipped with its natural inner product.
        \item\label{item:notation_lesssim_simeq} Given $a, b \geq 0$, we write $a \lesssim b$ if there exists a constant $C > 0$ such that $a \leq C\vts b$. We write $a \simeq b$ if $a \lesssim b$ and $b \lesssim a$.
        \item If $u_\veps \in \scrX$ and $f: \R_+ \to \R_+$, we write $u_\veps = \calO_\scrX (f(\veps))$ if there exist constants $\veps_0, C > 0$ such that $\|u_\veps\|_\scrX \leq C\, f(\veps)$ for any $\veps \in (0, \veps_0)$.
        \item\label{item:notation_otimes} Given a Hilbert space $(\scrX, \|\cdot\|_{_\scrX})$ and $\varphi \in \scrX$, we denote by $\varphi \otimes \varphi$ the rank-one operator defined by $(\varphi \otimes \varphi)\, u := \langle \varphi,\, u \rangle_{_\scrX}\, \varphi$ for any $u \in \scrX$. Note that $\varphi \otimes \varphi$ is an orthogonal projector if $\|\varphi\|_\scrX = 1$.
        \item The domain of an unbounded operator $\calH$ is $\DOM (\calH)$, and its spectrum is $\spec(\calH)$.
        \item Following the convention in \cite{graham1994concrete}, $\modulo{x} = x - \lfloor x \rfloor \in [0, 1)$ is the fractional part of $x \in \R$. Here, $\lfloor x\rfloor \in (x-1, x]$ denotes the greatest integer less than or equal to $x$.
        \item $\Lambda := \Z \vvv_1 + \Z \vvv_2$ and $\Lambda^* := \Z\vts (2\pi \kv_1) + \Z\vts (2\pi \kv_2)$ denote the equilateral triangular lattice and its dual respectively, with
        \begin{equation}
            \label{eq:honeycomb_vector_quasimomenta}
            \vvv_1 := \varsigma \begin{pmatrix}\displaystyle \sqrt{3}/2 \rets \displaystyle 1/2 \end{pmatrix}, \ %
            \vvv_2 := \varsigma \begin{pmatrix}\displaystyle \sqrt{3}/2 \rets \displaystyle-1/2 \end{pmatrix}, \ \ %
            \kv_1 := \varsigma \begin{pmatrix} \displaystyle 1/2 \rets \displaystyle \sqrt{3}/2 \end{pmatrix}, \ %
            \kv_2 := \varsigma \begin{pmatrix} \displaystyle 1/2 \rets \displaystyle -\sqrt{3}/2 \end{pmatrix},%
        \end{equation}
        and where $\smash{\varsigma = \sqrt{2 / \sqrt{3}}}$ is chosen so that $|\vvv_1 \wedge \vvv_2| = 1$ and $\kv_n \cdot \vvv_l = \delta_{n l}$ for any $n, l \in \{1, 2\}$.
        \item $\Kv := 2\pi\, (\kv_1 - \kv_2) / 3$ and $\Kv' := -\Kv = 2\pi\, (- \kv_1 + \kv_2) / 3$ are high-symmetry quasi-momenta, which generate the union of dual sublattices $\bbK := (\Kv + \Lambda^*) \cup (\Kv' + \Lambda^*) \subset \R^2$. 
        \item A generic element of $\bbL := \{\Kv, \Kv'\} \times \Z$ is denoted by $\iiv = (\Kv_\iiv, m_\iiv)$. We refer to generic subsets of $\bbL$ as $\bbL' \subseteq \bbL$.
        \item The Pauli matrices are
        \begin{equation}\label{eq:def_pauli_matrices}
            \displaystyle
            \sigma_1 := \begin{pmatrix} 0 & 1 \\ 1 & 0 \end{pmatrix}, \quad %
            \sigma_2 := \begin{pmatrix} 0 & -\icplx \\ \icplx & 0 \end{pmatrix}, \quad %
            \sigma_3 := \begin{pmatrix} 1 & 0 \\ 0 & -1 \end{pmatrix}. %
        \end{equation}
        \item We consider the unitary Fourier transform and its inverse given for $f, g \in L^2(\R)$ by
        \begin{equation}\label{eq:def_unitary_Fourier_transform}
            \calF f (\lvar) := \frac{1}{\sqrt{2\pi}} \int_\R f(\zeta)\, \euler^{-\icplx\vts \lvar\vts \zeta}\, d\zeta, \quad (\calF^* g)(\zeta) = \frac{1}{\sqrt{2\pi}} \int_\R g(\lvar)\, \euler^{\vts \icplx\vts \lvar\vts \zeta }\, d\lvar.
        \end{equation}
        We recall the Plancherel identity $\|\calF f\|_{L^2(\R)} = \|f\|_{L^2(\R)}$.
        \item 
        Fourier multiplier, $h(D_\zeta)$: Let $h \in L^\infty (\R)$. Define 
        \begin{equation}\label{eq:def_Fourier_multiplier}
            D_\zeta := - \icplx\, \partial_\zeta,
        \end{equation}
        and, using  \eqref{eq:def_unitary_Fourier_transform}, the bounded linear operator on $L^2(\R)$
        \begin{equation}\label{eq:f-multiplier}
            h (D_\zeta)f(\zeta) = (\calF^*\, h\, \calF) f(\zeta) = \frac{1}{\sqrt{2\pi}} \int_\R h(\lvar)\, \widehat{f}(\lvar)\, \euler^{\vts \icplx\vts \lvar\vts \zeta }\, d\lvar, \textWHERE \widehat{f} (\lvar) := \calF\, f (\lvar).
        \end{equation}
        \item Let $\lvar\mapsto\chi(\lvar)$ denote the indicator function of the interval $[-1, 1]$.  We  shall frequently spectrally localize on a set of size $\delta^\nu$ using the operator $\chi(\delta^{-\nu}D_\zeta)$ on $L^2(\R)$.
    \end{itemize}
    \endgroup

    \subsection*{Acknowledgements}
    The authors were supported in part by NSF grants DMS-1908657, DMS-1937254, DMS-2510769  and Simons Foundation Math + X Investigator Award {\# 376319} (MIW). Part of this research was carried out during the 2023-24 academic year, when MIW was a Visiting Member in the School of Mathematics - Institute of Advanced Study, Princeton, supported by the Charles Simonyi Endowment, and a Visiting Fellow in the Department of Mathematics at Princeton University. The authors thank G. Bal and A. Sagiv for stimulating discussions, and A. Drouot for suggesting the Helffer-Sjöstrand formula used in Appendix \ref{sec:functional_calculus}. 

    \section{Floquet-Bloch theory, honeycomb lattice potentials, and Dirac points}\label{sec:FB_theory}
    \noindent
    We first review Floquet-Bloch theory; more detailed presentations can be found in \cite{kuchment2012floquet,kuchment2016overview,kuchment2001mathematics,lechleiter2017floquet,gerard1998mourre} or \cite[Chapter XIII]{simon1978methodsvol4}. 

    \subsection{Fourier analysis in \texorpdfstring{$L^2(\R^2 / \Lambda)$}{L2(R/Λ)}}
    Let $\{\vvv_1, \vvv_2\}$ be a linearly independent set in $\R^2$. We introduce:
    \begin{itemize}
    \item the \emph{lattice} $\Lambda = \Z \vvv_1 + \Z \vvv_2$;
    \item the \emph{fundamental unit cell} $\Omega := \left\{ \yrm_1 \vvv_1 + \yrm_2 \vvv_2 \ /\ \yrm_1, \yrm_2 \in (0, 1) \right\}$;
    \item the \emph{dual lattice} $\Lambda^* := \Z\vts (2\pi \kv_1) + \Z\vts (2\pi \kv_2)$, where $\kv_n \cdot \vvv_l = \delta_{n l}$ for $n, l \in \{1, 2\}$;
    \item the \emph{first Brillouin zone} $\scrB$, a fundamental unit cell defined as the set of points in $\R^2$ that are closer to the origin than to any other point in the dual lattice $\Lambda^*$.
    \end{itemize}

    \vspace{1\baselineskip} \noindent  
    Given $\kv \in \R^2$, let $L^2_\kv$ be the space of \emph{$\kv$--pseudo-periodic functions with respect to $\Lambda$}:
    \begin{equation}\label{eq:k_pseudo_periodic_space}
    \displaystyle
    L^2_\kv := \big\{\vts f \in L^2_{\LOC}(\R^2)\ \ \big|\ \ f(\xv + \vvv) = \euler^{\vts \icplx\vts \kv \cdot \vvv}\, f(\xv) \textFOR \xv \in \R^2, \vvv \in \Lambda \vts\big\}.
    \end{equation}
    Note that $f$ is $\kv$--pseudo-periodic with respect to $\Lambda$ if and only if $\euler^{\vts- \icplx \vts \kv \cdot \xv}\, f$ is $\Lambda$--periodic. In particular, $L^2_\kv = L^2(\R^2 / \Lambda)$ for $\kv = 0$.  We equip $L^2_\kv$ with the inner product and the norm
    \begin{equation*}
    \langle g, f \rangle_{L^2_\kv} := \int_\Omega \overline{g(\xv)}\, f(\xv)\, d\xv, \quad \|f\|^2_{L^2_\kv} := \langle f, f \rangle_{L^2_\kv}.
    \end{equation*}
    Similarly, one defines the Sobolev space
    \begin{equation*}
    \spforall q \in \N_0, \quad H^q_\kv := \big\{\vts f \in L^2_\kv\ \ \big|\ \ \euler^{\vts- \icplx \vts \kv \cdot \xv} f \in H^q (\R^2 / \Lambda) \vts\big\},
    \end{equation*}
    equipped with the $H^q(\Omega)$--norm. Note that $H^{q = 0}_\kv = L^2_\kv$.

    Any $f \in L^2(\R^2 / \Lambda)$ can be expanded in Fourier series:
    \begin{equation*}
    f(\xv) = \sum_{\kv \in \Lambda^*} \widehat{f}_\kv\, \euler^{\vts \icplx\vts \kv \cdot \xv}, \textWITH \widehat{f}_\kv := \frac{1}{|\Omega|} \int_\Omega \euler^{- \icplx\vts \kv \cdot \yv}\, f(\yv)\, d\yv.
    \end{equation*}
    Moreover, the Sobolev regularity of $f \in L^2(\R^2 / \Lambda)$ is related to the decay rate of its Fourier coefficients \cite[Chapter 4.3]{taylor2023partial}:
    \begin{equation}\label{eq:link_Sobolev_regularity_decay_Fourier_coeffs}
    \spforall q \in \N_0, \quad \spforall f \in H^q (\R^2 / \Lambda), \quad \|f\|^2_{H^q (\R^2 / \Lambda)} \hyperref[item:notation_lesssim_simeq]{\simeq} \sum_{\kv \in \Lambda^*} (1 + |\kv|^2)^q\, |\widehat{f}_\kv|^2.
    \end{equation}

    \subsection{Floquet-Bloch states}\label{sec:FB_states}
    Given $V \in \scrC^\infty(\R^2 / \Lambda; \R)$, a real-valued, smooth, and $\Lambda$--periodic potential, we consider the self-adjoint Schrödinger operator:
    \begin{align*}
    \displaystyle
    \calH^0 &:= -\Delta + V(\xv), & \DOM(\calH^0) &:= H^2 (\R^2).
    \\
    \intertext{In addition, given $\kv \in \R^2$, define:}
    \displaystyle
    \calH^0_\kv &:= -\Delta + V(\xv), & \DOM(\calH^0_\kv) &:= H^2_\kv.
    \end{align*}
    The spectral properties of $\calH^0$ can be studied using the eigenpairs of $\calH^0_\kv$, for $\kv \in \R^2$, which correspond to the nontrivial solutions of the eigenvalue problem
    \begin{equation}\label{eq:evp_cell}
    \displaystyle
    \left\{
    \begin{array}{r@{\ =\ }l@{\quad}l}
        \calH^0_\kv\, \Phi(\xv; \kv) & E(\kv)\, \Phi(\xv; \kv), & \xv \in \R^2,
        \rets
        \Phi(\xv + \vvv; \kv) & \euler^{\vts \icplx\vts \kv \cdot \vvv}\, \Phi(\xv; \kv), & \xv \in \R^2,\ \ \spforall \vvv \in \Lambda.
    \end{array}
    \right.
    \end{equation}
    The $\kv$--pseudo-periodicity condition in \eqref{eq:evp_cell} is invariant under the translation $\kv \mapsto \kv + \kv'$ for any $\kv' \in \Lambda^*$, so that the study of \eqref{eq:evp_cell} can be restricted to $\kv$ varying in the Brillouin zone $\scrB$. A solution of \eqref{eq:evp_cell} is called a \emph{Floquet-Bloch eigenpair} with quasi-momentum $\kv$. 

    For each fixed $\kv$, \bypierre{the operator $\calH^0_\kv$ is self-adjoint, bounded from below, and with compact resolvent; hence} its spectrum is a countable set of real eigenvalues, listed with multiplicity:
    \begin{equation*}
    \displaystyle
    E_1 (\kv) \leq E_2 (\kv) \leq E_3 (\kv) \leq \dots \leq E_b(\kv) \leq \dots; \quad \lim_{b \to +\infty} E_b(\kv) = +\infty,
    \end{equation*}
    with eigenfunctions $\Phi_b(\xv; \kv)$, $b \geq 1$. For a fixed $\kv \in \scrB$, $\{\xv \mapsto \Phi_b(\xv; \kv)\}_{b \geq 1}$ can be taken to be an orthonormal basis for $L^2_\kv$. By elliptic regularity, the smoothness of $V(\xv)$ implies that the functions $\xv \mapsto \Phi_b (\xv; \kv)$ are smooth, uniformly in $\kv \in \scrB$. Moreover, similarly to \eqref{eq:link_Sobolev_regularity_decay_Fourier_coeffs}, 
    \begin{equation}\label{eq:link_Sobolev_regularity_decay_FB_coeffs}
    \spforall (\kv, q) \in \scrB \times \N_0, \quad \spforall f \in H^q_\kv, \quad  \|f\|^2_{H^q_\kv} \hyperref[item:notation_lesssim_simeq]{\simeq} \sum_{b \geq 1} (1 + b^2)^{q/2}\, \big| \lla\vts \Phi_b (\cdot\,; \kv), f \vts\rra_{L^2_\kv} \big|^2,
    \end{equation}
    \emph{where the implied constants can be chosen to be independent in $\kv$}.

    The sets $\left\{(\kv, E_b(\kv)),\ \kv \in \scrB\right\} \subset \scrB \times \R$, for $b \geq 1$, are called \emph{dispersion surfaces}, and the mappings $\kv \in \scrB \mapsto E_b(\kv)$, \emph{dispersion functions}. These functions are Lipschitz continuous; see \cite{kuchment2012floquet,avron1978analytic,kuchment2016overview} and \cite[Appendix A]{fefferman2014wave}. For $b \geq 1$, as $\kv$ varies over $\scrB$, each function $E_b(\kv)$ ranges over a closed interval called the \emph{$b$--th band}. The union of these bands over $b \geq 1$ forms the $L^2(\R^2)$--spectrum of $\calH^0$. Moreover, $\{\xv \mapsto \Phi_b(\xv; \kv)\}_{b \geq 1, \kv \in \scrB}$ is complete in $L^2(\R^2)$:
    \begin{equation*}
        \displaystyle
        \spforall f \in L^2(\R^2), \quad f(\xv) = \sum_{b \geq 1} \int_{\scrB} \lla\vts \Phi_b (\cdot\,; \kv), f \vts\rra_{L^2(\R^2)}\; \Phi_b (\xv; \kv)\, d\vts\kv, \quad \xv \in \R^2,
    \end{equation*}
    \bypierre{where $\kv \in \scrB \mapsto \lla\vts \Phi_b (\cdot\,; \kv), f \vts\rra_{L^2(\R^2)}$ is understood in the $L^2$-sense rather than pointwise.}

    \subsection{Honeycomb lattice potentials}
    We now apply the results of  Section \ref{sec:FB_states}  to the equilateral triangular lattice $\Lambda := \Z \vvv_1 + \Z \vvv_2$ and its dual $\Lambda^* := \Z\vts (2\pi \kv_1) + \Z\vts (2\pi \kv_2)$, where $(\vvv_1, \vvv_2, \kv_1, \kv_2)$ are defined by \eqref{eq:honeycomb_vector_quasimomenta} and shown in Figure \ref{fig:lattice_brillouin}. The Brillouin zone $\scrB$ is a closed regular hexagon shown in Figure \ref{fig:lattice_brillouin} (right panel). Consider its vertices $\Kv$ and $\Kv'$ defined by
    \begin{equation*}
    \displaystyle
    \Kv := \frac{2\pi}{3} (\kv_1 - \kv_2), \quad \Kv' := -\Kv = \frac{2\pi}{3} (-\kv_1 + \kv_2).
    \end{equation*}
    The vertices of $\scrB$ can be generated from $\Kv$, $\Kv'$ using the clockwise rotation matrix by $2\pi/3$:
    \begin{equation}\label{eq:rotation_matrix_2pi_3}
    \displaystyle
    R = \begin{pmatrix}
        \displaystyle-1/2 & \displaystyle\sqrt{3}/2
        \rets
        \displaystyle-\sqrt{3}/2 & \displaystyle-1/2
    \end{pmatrix}.
    \end{equation}
    These  vertices are classified into two categories: 
    \begin{itemize}
    \item $\Kv$ type-points: $\Kv$, $R\, \Kv = \Kv + 2\pi\vts \kv_2$, and $R^2\, \Kv = \Kv - 2\pi\vts \kv_1$;
    \item $\Kv'$ type-points: $\Kv'$, $R\, \Kv' = \Kv' - 2\pi\vts \kv_2$, and $R^2\, \Kv' = \Kv' + 2\pi\vts \kv_1$.
    \end{itemize}
    
    Consider the following symmetry operators: complex-conjugation, inversion  and rotation by $2\pi/3$:
    \[
    \rets
    (\vts\calC f\vts)\vts (\xv) := \overline{f(\xv)}, \quad (\vts\calP f\vts)\vts (\xv) := f(-\xv),  \textAND \calR f(\xv) := f(R^* \xv),
    \]
    where $R$ is given by \eqref{eq:rotation_matrix_2pi_3}, and $R^* = R^\top= R^{-1}$.

    \begin{definition}[Definition 2.1, \cite{fefferman2012honeycomb}]\label{defi:honeycomb_lattice_potential}
        A function $V \in \scrC^\infty (\R^2)$ is said to be a \emph{honeycomb lattice potential} if $V$ is real-valued, $\Lambda$--periodic and, with respect to some origin of coordinates which we may take to be $\xv=0$, $V(\xv)$ is also  even, and $2\pi/3$--rotationally invariant:
        \begin{gather*}
            [\calC, V(\xv)] = 0, \quad [\calP, V(\xv)] = 0, \textAND [\calR, V(\xv)] = 0.
        \end{gather*}
    \end{definition}

    \noindent
    Due to the $\calR$--invariance of a honeycomb lattice potential $V$, the operator $\calH^0_{\Kv_\star} = -\Delta + V(\xv)$, acting on $L^2_{\Kv_\star}$, commutes with $\calR$ \bypierre{when} $\Kv_\star$ is a vertex of the Brillouin zone $\scrB$ modulo $\Lambda^*$--translations; see \cite[Proposition 2.2]{fefferman2012honeycomb}. For this reason, the vertices of $\scrB$ are called \emph{high-symmetry quasi-momenta}. \bypierre{The set of all high-symmetry quasi-momenta is the union of sublattices generated by the vertices of $\scrB$ and their translates along $\Lambda^*$:
    \begin{equation*}
        \bbK := (\Kv + \Lambda^*) \cup (\Kv' + \Lambda^*).
    \end{equation*}
    If $\Kv_\star \in \bbK$, then} the $L^2_{\Kv_\star}$--spectrum of $\calR$ consists of the cube roots of unity $1$, $\tau := \euler^{2\vts \icplx\vts \pi/3}$, and $\overline{\tau}$, and hence $L^2_{\Kv_\star}$ can be decomposed in terms of the eigenspaces of $\calR$:
    \begin{align*}
    \displaystyle
    L^2_{\Kv_\star} &\phantom{:}= L^2_{\Kv_\star, 1} \oplus L^2_{\Kv_\star, \tau} \oplus L^2_{\Kv_\star, \overline{\tau}},
    \intertext{where} 
    L^2_{\Kv_\star, z} &:= \big\{\vts f \in L^2_{\Kv_\star}\ /\ \ \calR f = z f \vts\big\}, \quad z \in \C.
    \end{align*}

    \subsection{Dirac points}\label{sec:dirac_points}
    A {\it Dirac point} is a quasi-momentum/energy pair $(\Kv_\star, E_D)$ for which the following holds: there exist an integer $b_\star \geq 1$ and Floquet-Bloch eigenpair mappings
    \begin{equation*}
    \displaystyle
    \kv \mapsto \big(E_{b_\star}(\kv), \Phi_{b_\star}(\xv; \kv)\big) \textAND \kv \mapsto \big(E_{b_\star + 1}(\kv), \Phi_{b_\star + 1}(\xv; \kv)\big),
    \end{equation*}    
    such that the dispersion functions $E_-: \kv \mapsto E_{b_\star}(\kv)$ and $E_+: \kv \mapsto E_{b_\star + 1}(\kv)$ touch conically at $(\Kv_\star, E_D)$.  More precisely, there exist $\upsilon_D > 0$, constants $C,\; \veps_0 > 0$, and functions $e_\pm (\kv)$, defined for $|\kv - \Kv_\star| < \veps_0$ with $|e_\pm (\kv)| \leq C |\kv - \Kv_\star|$, such that
    \begin{equation}\label{eq:conical_behavior_Dirac}
    \displaystyle
    \begin{array}{r@{\ =\ }l}
        E_-(\kv) - E_D & - \upsilon_D\, |\kv - \Kv_\star|\,(1 + e_-(\kv)),
        \rets
        E_+(\kv) - E_D & + \upsilon_D\, |\kv - \Kv_\star|\,(1 + e_+(\kv)).
    \end{array}
    \end{equation}
    The number $\upsilon_D$ is called the \emph{Dirac velocity}.

    The following result provides sufficient conditions for the existence of Dirac points at the vertices of the Brillouin zone $\scrB$; see \cite{fefferman2012honeycomb}.
    \begin{proposition} \label{prop:sufficient_conditions_dirac_point}
    Let $\Kv_\star$ be a vertex of the Brillouin zone $\scrB$, and let $E_D \in \R$. 

    A sufficient condition for the  quasi-momentum/energy pair $(\Kv_\star, E_D)$ to be a \emph{Dirac point} of $\calH^0 = -\Delta + V(\xv)$  is the existence of an integer $b_\star \geq 1$ and Floquet-Bloch eigenpairs mappings
    \begin{equation*}
        \displaystyle
        \kv \mapsto \big(E_{b_\star}(\kv), \Phi_{b_\star}(\xv; \kv)\big) \textAND \kv \mapsto \big(E_{b_\star + 1}(\kv), \Phi_{b_\star + 1}(\xv; \kv)\big),
    \end{equation*}    
    such that the following conditions are satisfied:
    \begin{enumerate}[label={$(\roman*).$}, ref={$(\roman*)$}]
        \item $E_D = E_{b_\star}(\Kv_\star) = E_{b_\star + 1}(\Kv_\star)$ is a double eigenvalue of $\calH^0_{\Kv_\star}$;
        \item\label{item:dirac_point_item_2} There is a
        \emph{Dirac eigenbasis}, consisting of functions
        \begin{equation*}
        \displaystyle
        \Phi^{\Kv_\star}_1 \in L^2_{\Kv_\star, \tau} \textAND \Phi^{\Kv_\star}_2 = \calP\, \calC\, \Phi^{\Kv_\star}_1 \in L^2_{\Kv_\star, \overline{\tau}},
        \end{equation*}
        with $\langle \Phi^{\Kv_\star}_j, \Phi^{\Kv_\star}_l \rangle_{L^2_{\Kv_\star}} = \delta_{jl}$ for $j, l \in \{1, 2\}$, and
        \begin{equation*}
        \Ker\, (\calH^0_{\Kv_\star} -  E_D) = \vect \big\{\Phi^{\Kv_\star}_1, \Phi^{\Kv_\star}_2 \big\};
        \end{equation*}
        \item\label{item:dirac_point_item_3} One has
        \begin{equation}\label{eq:component_phi_1_phi_2}
        \displaystyle
        \upsilon^{\Kv_\star} := \frac{1}{2}\, \overline{\big\langle \Phi^{\Kv_\star}_1, - 2 \icplx\vts \nabla \Phi^{\Kv_\star}_2 \big\rangle_{L^2_{\Kv_\star}}} \cdot \begin{pmatrix} 1 \\ \icplx \end{pmatrix} \neq 0.
        \end{equation} 
    \end{enumerate}
    If the above conditions are satisfied, then \eqref{eq:conical_behavior_Dirac} holds with $\upsilon_D := |\upsilon^{\Kv_\star}|$. Moreover, if \eqref{eq:component_phi_1_phi_2} holds, then there exists $\theta \in \R$ such that by replacing $\Phi^{\Kv_\star}_1$ with $\euler^{\vts \icplx\, \theta}\, \Phi^{\Kv_\star}_1$, and $\Phi^{\Kv_\star}_2$ with $\euler^{- \icplx\, \theta}\, \Phi^{\Kv_\star}_2$, we obtain an equivalent Dirac basis that satisfies Condition \ref{item:dirac_point_item_2}, and such that $\upsilon^{\Kv_\star} \equiv \upsilon^{\Kv_\star} (\theta)$ is a real number, with $\upsilon^{\Kv} = \upsilon_D > 0$ and $\upsilon^{\Kv'} = -\upsilon_D < 0$.
    \end{proposition}

    \begin{remark}\label{rmk:Dirac_point_K_sufficient}
        \bypierre{Suppose that $\calH^0 = -\Delta + V(\xv)$ has a Dirac point $(\Kv, E_D)$. Since $\calH^0_\Kv = \calH^0_{R \Kv} = \calH^0_{R^2 \Kv}$, it follows that $(R \Kv, E_D)$ and $(R^2 \Kv, E_D)$ are also Dirac points. Moreover, let $\{\Phi^{\Kv}_1, \Phi^{\Kv}_2\}$ denote the associated eigenbasis, in the sense of Proposition \ref{prop:sufficient_conditions_dirac_point}. Since $V$ is even, $(-\Kv, E_D) = (\Kv', E_D)$ is also a Dirac point with associated eigenbasis $\{\calP\, \Phi^{\Kv}_1, \calP\, \Phi^{\Kv}_2\}$.}
    \end{remark}

    \vspace{1\baselineskip} \noindent
    The existence of Dirac points is shown in \cite{fefferman2012honeycomb} for generic honeycomb lattice potentials. 
    \begin{theorem}\label{thm:existence_Dirac_points}
        Let $V$ be a honeycomb lattice potential (see Definition \ref{defi:honeycomb_lattice_potential}) such that
        \begin{equation*}
            \displaystyle
            \frac{1}{|\Omega|}\int_{\Omega} V(\xv)\, \euler^{\vts -\icplx\vts 2\pi\vts (\kv_1 + \kv_2) \cdot \xv}\; d\xv \neq 0.
        \end{equation*} 
        Then there exists a \bypierre{countable and closed set $S \subset \R$} such that for any $\mu \in \R \setminus S$, the operator $-\Delta + \mu\vts V(\xv)$ satisfies the conditions in Proposition \ref{prop:sufficient_conditions_dirac_point}. As a consequence, $-\Delta + \mu\vts V(\xv)$ admits a Dirac point $(\Kv_\star, E_D)$ for $\Kv_\star \in \{\Kv, \Kv'\}$.
    \end{theorem}

    \noindent
    Following Theorem \ref{thm:existence_Dirac_points} and Remark \ref{rmk:Dirac_point_K_sufficient}, we make the following assumption:
    
    \begin{assumption}[Dirac point]\label{hyp:Dirac_point}
        $V$ is a honeycomb lattice potential in the sense of Definition \ref{defi:honeycomb_lattice_potential}, such that $\calH^0 = -\Delta + V(\xv)$ has a Dirac point $(\Kv_\star, E_D) \in \{\Kv, \Kv'\} \times \R$, and satisfies the conditions in Proposition \ref{prop:sufficient_conditions_dirac_point}. Moreover, following the last part of Proposition \ref{prop:sufficient_conditions_dirac_point}, we always choose the Dirac eigenbasis so that $\upsilon^{\Kv} = -\upsilon^{\Kv'} = \upsilon_D > 0$.
    \end{assumption}

    \subsection{The omnidirectional no-fold condition}\label{sec:no_fold_condition}
    In addition to Assumption \ref{hyp:Dirac_point}, we require the following, so-called \emph{no-fold condition}:
    \begin{assumption}\label{hyp:omnidirectional_no_fold}
        $\calH^0$ satisfies the omnidirectional no-fold condition, that is,
        \begin{equation}\label{eq:omnidirectional_no_fold}
            \displaystyle
            \textnormal{if $\kv \in \R^2$ is such that $E_b (\kv) = E_D$ for $b \geq 1$, then $\kv \in \bbK := (\Kv + \Lambda^*) \cup (\Kv' + \Lambda^*)$.}
        \end{equation}
    \end{assumption}

    \noindent
    \bypierre{Assumption \ref{hyp:omnidirectional_no_fold} ensures that the dispersion functions that attain the Dirac energy, $E_D$, do so at high-symmetry quasi-momenta only. We shall prove in Lemma \ref{lem:dispersion_curves_bounded_away_from_Dirac} that under this assumption, the dispersion functions $E_b (\kv)$ are bounded away from $E_D$ whenever $\kv$ is located outside a neighborhood of $\bbK$.}

    \bypierre{A similar, but weaker assumption was introduced in previous work \cite{fefferman2016edge,lee2019elliptic,drouot2019characterization,drouot2020edge} for rational edges.} The above omnidirectional no-fold condition allows to study irrational edges. It can be verified explicitly that Assumption \ref{hyp:omnidirectional_no_fold} holds for the two-band tight-binding model of graphene, as well as for the honeycomb Schrödinger operator $-\Delta + g^2\, V(\xv)$ in the {\it strong binding regime} ($g \gg 1$). Indeed, in the latter setting it was shown in \cite[Theorem 6.1]{fefferman2018honeycomb} that the lowest two bands, after appropriate scaling, converge uniformly to those of the two-band tight-binding model.

    \subsection{\texorpdfstring{$\calC$}{C}-breaking perturbations and instability of Dirac points}
    \label{sec:conjugation_breaking_perturbations}
    Let $a \in \scrC^\infty (\R^2 / \Lambda; \R)$ be smooth, $\Lambda$--periodic, even, and real-valued, so that $[\calC, a(\xv)] = 0$ and $[\calP, a(\xv)] = 0$. We define
    \begin{equation*}
        \displaystyle
        \calH^\delta := -\nabla \cdot \big(\vts \Id - \delta\, a(\xv)\, \hyperref[eq:def_pauli_matrices]{\sigma_2} \vts\big)\, \nabla + V(\xv), \qquad \DOM(\calH^\delta) := H^2(\R^2).
    \end{equation*}
    Since $\calP\, a\, \sigma_2 = a\, \sigma_2\, \calP$ but $\calC\, a\, \sigma_2 = -a\, \sigma_2\, \calC$, the perturbation induced by $a (\xv)\, \sigma_2$ preserves parity symmetry, while conjugation symmetry is broken: $[\calH^\delta, \calP] = 0$ and $[\calH^\delta, \calC] \neq 0$. Moreover, because $\calC\, \calH^\delta = \calH^{-\delta}\, \calC$, the operators $\calH^\delta$ and $\calH^{-\delta}$ share the same spectrum. 

    In what follows, we shall also assume that
    \begin{assumption}[Nondegeneracy]\label{hyp:non_degeneracy}
    The function $a$ is smooth, real-valued, $\Lambda$--periodic, and even. Moreover,
    \begin{equation}\label{eq:non_degeneracy}
        \vartheta^{\Kv_\star} := \lla\vts \Phi^{\Kv_\star}_1, \nabla \cdot \bigl(\vts a(\xv)\, \sigma_2\, \nabla \Phi^{\Kv_\star}_1 \vts\bigr) \vts\rra_{L^2_{\Kv_\star}} \neq 0,
    \end{equation}
    where $\{\Phi^{\Kv_\star}_1, \Phi^{\Kv_\star}_2\}$ is the Dirac eigenbasis in Part \ref{item:dirac_point_item_2} of Proposition \ref{prop:sufficient_conditions_dirac_point}.
    \end{assumption}

    \noindent
    \bypierre{Note that, since $a(-\xv) = a (\xv)$ and $\Phi^{\Kv'}_j (\xv) = \Phi^{\Kv}_j (-\xv)$ (Remark \ref{rmk:Dirac_point_K_sufficient}), we have $\vartheta^{\Kv} = \vartheta^{\Kv'}$.}

    Under the Dirac point Assumption \ref{hyp:Dirac_point} and the nondegeneracy Assumption \ref{hyp:non_degeneracy}, for a small enough $\delta \in \R \setminus \{0\}$, the $\calC$-breaking perturbation opens a \emph{local} gap in the neighborhood of the Dirac point $(\Kv_\star, E_D)$ (\cite[Section 5]{lee2019elliptic}). More precisely, for all $\kv \in \scrB$ define:
    \begin{equation*}
        \calH^\delta_\kv := -\nabla \cdot \big(\vts \Id - \delta\, a(\xv)\, \sigma_2 \vts\big)\, \nabla + V(\xv), \qquad \DOM(\calH^\delta_\kv) := H^2_\kv\ .
    \end{equation*}
    \bypierre{Then, there exist constants $\veps_0, \delta_0, C > 0$ and coefficients $e^{\delta, \veps}_\pm \in \R$ defined for $0 < \delta < \delta_0$ and $0 < \veps < \veps_0$ with $|e^{\delta, \veps}_\pm| \leq C\, \delta^2$ such that, for any $|\kv - \Kv_\star| \leq \veps$,
    \begin{equation*}
        \spec ( \calH^\delta_{\kv} ) \cap \bigl(E_D - |\delta\, \vartheta^{\Kv_\star}| + e^{\delta, \veps}_-, E_D + |\delta\, \vartheta^{\Kv_\star}| + e^{\delta, \veps}_+\bigr) = \emptyset,
    \end{equation*}
    Furthermore, if the omnidirectional no-fold condition \eqref{eq:omnidirectional_no_fold} is satisfied, then the local gap opened by the $\calC$-breaking perturbation is actually a \emph{global}, spectral gap around $E_D$: In fact, there exist $\delta_0, C > 0$ and coefficients $e^\delta_\pm$ defined for $0 < \delta < \delta_0$ with $|e^\delta_\pm| \leq C\, \delta^2$, such that 
    \begin{equation*}
        \spec ( \calH^\delta ) \cap \bigl(E_D - |\delta\, \vartheta^{\Kv_\star}| + e^\delta_-, E_D + |\delta\, \vartheta^{\Kv_\star}| + e^\delta_+\bigr) = \emptyset,
    \end{equation*}
    The same holds for $\calH^{-\delta}_\kv$ and $\calH^{-\delta}$ respectively.}

    \section{Honeycomb structures with a non-commensurate line defect}\label{sec:honeycomb_structures_line_defect}
    \subsection{Domain wall-modulated Schrödinger operator}\label{sec:dw_modulated_operator}
    We construct a class of operators that spatially interpolate smoothly and slowly between $\calH^\delta$ and $\calH^{-\delta}$ in the direction transverse to an edge $\R\, \vvh_1$, with $\vvh_1 \in \R^2 \setminus \{0\}$. The interpolation is realized via a \emph{domain wall function}.
    \begin{definition}\label{defi:dw}
        A \emph{domain wall function} is a function $\kappa \in \scrC^\infty (\R)$ such that $\kappa' \in L^\infty(\R)$ and $\displaystyle \lim_{\zeta \to \pm \infty} \kappa(\zeta) = \pm 1$.
    \end{definition}
    \noindent
    Let $\kvh_2 \in \R^2 \setminus \{0\}$ be such that $\kvh_2 \cdot \vvh_1 = 0$. Given a domain wall function $\kappa$, in the sense of Definition \ref{defi:dw}, we consider the Schrödinger operator: 
    \begin{equation}\label{eq:def_dw_coefficient}
        \displaystyle
        \begin{array}{r@{\ :=\ }l}
            \calH^\delta_\EDGE & -\nabla \cdot A^\delta (\xv)\, \nabla + V(\xv), \qquad \DOM (\calH^\delta_\EDGE) = H^2 (\R^2),
            \ret
            \textWITH A^\delta (\xv) & \Id - \delta\, \kappa (\delta \kvh_2 \cdot \xv)\, a(\xv)\, \sigma_2.
        \end{array}
    \end{equation}
    Since $\kappa (\zeta) \to \pm 1$ as $\zeta \to \pm \infty$, $\calH^\delta_\EDGE$ asymptotically behaves as $\calH^{\pm \delta}$ far from the edge $\R\, \vvh_1$. We refer to Figure \ref{fig:soft_junction} for a schematic representation of $\calH^\delta_\EDGE$.

    \subsection{Irrational edges}\label{sec:irrational_edges}
    The edge vector $\vvh_1$ can be written as $\vvh_1 = a_1\vts \vvv_1 + b_1\vts \vvv_2$ where $a_1, b_1 \in \R$. Since $\vvh_1 \neq 0$, $a_1$ and $b_1$ cannot both vanish, and  because the roles of $\vvv_1$ and $\vvv_2$ are interchangeable, we can assume, without any loss of generality, that $a_1 \neq 0$. \bypierre{Further, we may normalize $\vvh_1$ by choosing $a_1 = 1$, and set}
    \begin{equation} 
        \vvh_1 := \vvv_1 + r \vvv_2,\quad  \textrm{with}\quad r \in \R.\label{eq:v1-def}
    \end{equation}
    If $r$ is rational, then we call $\R\, \vvh_1$ a \emph{commensurate} or \emph{rational} edge. Two classical examples include the \emph{zigzag edge} ($r = 0$) and the \emph{armchair} edge ($r = 1$). If $r \in \R \setminus \Q$, then we call $\R\, \vvh_1$ a \emph{non-commensurate} or \emph{irrational} edge.

    We may complete $\vvh_1$ with the vectors $\vvh_2$, $\kvh_1$, and $\kvh_2$, to form a Bravais lattice $\Z\vts \vvh_1 + \Z\vts \vvh_2$, associated to the edge, and its dual, $\Z\vts (2\pi \kvh_1) + \Z\vts (2\pi \kvh_2)$:
    \begin{equation}\label{eq:def_edge_vectors}
        \displaystyle
        \vvh_1 := \vvv_1 + r \vvv_2, \quad \vvh_2 := \vvv_2 \textAND \kvh_1 := \kv_1, \quad \kvh_2 := -r \kv_1 + \kv_2. 
    \end{equation}
    These vectors satisfy
    \[ \kvh_n \cdot \vvh_l = \delta_{n l},\quad \textrm{for}\quad  n, l = 1, 2.\]
    However, note that $\Z\vts \vvh_1 + \Z\vts \vvh_2 \neq \Lambda$ and $\Z\vts (2\pi \kvh_1) + \Z\vts (2\pi \kvh_2) \neq \Lambda^*$ for irrational $r$. 

    \subsection{Quasiperiodicity along irrational edges}\label{sec:quasiperiodicity_irrational_edge}
    We assume that $(\vvh_1, \vvh_2, \kvh_1, \kvh_2)$ are given by \eqref{eq:def_edge_vectors} with a slope $r \in \R$. If $r$ is irrational, then the coefficients $(A^\delta, V)$ of the domain wall-modulated potential $\calH^\delta_\EDGE$ introduced in \eqref{eq:def_dw_coefficient} are not periodic in the $\vvh_1$--direction. However, in the spirit of \cite{gerard2011homogenization,blanc2015local}, we show that periodicity can be recovered by extending $(A^\delta, V)$ into higher dimensions. 

    More precisely, $A^\delta$ can be expressed as the restriction of an  extension  $A^\delta_\AUG$, defined in a $3$--dimensional space, which is periodic within a $2$--dimensional plane containing the edge $\R\, \vvh_1$, and similarly for the potential $V$. This is closely related to the \emph{quasiperiodic} behavior of $(A^\delta, V)$ in the $\vvh_1$--direction, as explained in Remark \ref{rmk:quasiperiodic_nature_domain_wall_potential} below.

    We introduce the following extensions of $V$ and $A^\delta$ to functions on a three-dimensional domain:
    \begin{align}
    \spforall (\xv, s) \in \R^2 \times \R, \quad V_\AUG (\xv, s) &:= V(\xv), \nonumber
    \retss 
    A^\delta_\AUG (\xv, s) &:= \Id - \delta\, \kappa\big(\delta \kvh_2 \cdot (\xv + s \vvv_2)\big)\, a(\xv)\, \sigma_2.\label{eq:def_augmented_potential}
    \end{align}
    Then, $A^\delta$ is the restriction of $A^\delta_\AUG$ to the plane $\R^2 \times \{0\}$, {\it i.e.} $A^\delta(\xv) = A^\delta_\AUG (\xv, 0)$. We refer to $A^\delta_\AUG$ as an \emph{augmented} coefficient or an \emph{extension\;/\;lift} of $A^\delta$, and similarly for $V_\AUG$.

    From its expression, $A^\delta_\AUG$ interpolates smoothly and (for $\delta$ small) slowly, via the domain wall $\kappa$, between the functions $(\xv, s) \mapsto \Id + \delta\, a(\xv)\, \sigma_2$ and $(\xv, s) \mapsto \Id - \delta\, a(\xv)\, \sigma_2$, in the direction transverse to the two-dimensional interface
    \begin{equation}
    \Gamma := \left\{(\xv, s) \in \R^2 \times \R \ /\ \kvh_2 \cdot (\xv + s \vvv_2) = 0 \right\}.\label{eq:interface_Gamma}
    \end{equation}
    Introduce the three-dimensional vectors:
    \begin{equation}\label{eq:3D_lattice_vectors}
    \av_1 := \begin{pmatrix}\vvv_1 \\ r\end{pmatrix}, \quad \av_2 := \begin{pmatrix}\vvv_2 \\ -1\end{pmatrix}, \quad \av_\perp := \begin{pmatrix}0_2 \\ 1\end{pmatrix} \ .
    \end{equation}
    The set $\{\av_1, \av_2\}$ spans $\Gamma$, while the vector $\av_\perp$ is transverse to $\Gamma$.

    Our interest in $A^\delta_\AUG$ lies in its \emph{translation invariance} within the interface $\Gamma$. In contrast to $A^\delta$ which is not necessarily translation-invariant along the edge $\R\, \vvh_1$, the relations \eqref{eq:def_edge_vectors} and the $\vvv_1$--periodicity of $a(\cdot)$ imply
    \begin{subequations}\label{eq:augmented_potential_property}
        \begin{equation}\label{eq:augmented_potential_property_1}
            \displaystyle 
            \spforall (\xv, s) \in \R^2 \times \R, \quad A^\delta_\AUG (\xv + \vvv_1, s + r) = A^\delta_\AUG (\xv, s).
        \end{equation}
        On the other hand, the $\vvv_2$--periodicity of $a(\cdot)$ implies that
        \begin{equation}\label{eq:augmented_potential_property_2}
            \displaystyle 
            \spforall (\xv, s) \in \R^2 \times \R, \quad A^\delta_\AUG (\xv + \vvv_2, s - 1) = A^\delta_\AUG (\xv, s).
        \end{equation}
    \end{subequations}
    In other words, $A^\delta_\AUG$ is invariant under translations by the vectors $\av_1, \av_2 \in \R^3$, which span the interface $\Gamma$. The relations (\ref{eq:augmented_potential_property_1}, \ref{eq:augmented_potential_property_2}) are also satisfied by $V_\AUG$.

    If $r = b_1 / a_1$ is rational, then (\ref{eq:augmented_potential_property_1}, \ref{eq:augmented_potential_property_2}) imply that $\xv \mapsto A^\delta_\AUG (\xv, s)$ is $\Z(a_1\vts \vvh_1)$--periodic for any $s \in \R$, in particular for $s = 0$. This is expected for rational edges. For an irrational $r$, (\ref{eq:augmented_potential_property_1}, \ref{eq:augmented_potential_property_2}) encode the more complex, \emph{quasiperiodic}, behavior of $A^\delta$ along $\R\, \vvh_1$.

    \begin{remark}\label{rmk:quasiperiodic_nature_domain_wall_potential}
    \begin{enumerate}[label={$(\alph*).$}, ref={\theremark.$(\alph*)$}, wide=0pt]
        \item\label{rmk:quasiperiodic_nature_domain_wall_potential_a} One can interpret $\{\xv \mapsto\!\vts A^\delta_\AUG (\xv, s), s \in \R\}$ as a family of two-dimensional dislocations of $A^\delta$ in the $\vvh_2$--direction. In fact, for $s \in \R$ fixed, assuming for simplicity that $\kappa (\zeta)={\rm sgn}(\zeta)$, $\xv \mapsto A^\delta_\AUG (\xv, s)$ models a sharp junction, in the direction transverse to the interface $\R\, \vvh_1 - s\vts \vvh_2$, between media modeled by $\xv \mapsto \Id + \delta\, a(\xv)\, \sigma_2$ and $\xv \mapsto \Id - \delta\, a(\xv)\, \sigma_2$. Dislocations may be studied independently from one another if $r$ is rational, but we see from (\ref{eq:augmented_potential_property_1}, \ref{eq:augmented_potential_property_2}) that these dislocations are all coupled if $r$ is irrational.

        \item The existence of a function $A^\delta_\AUG$ with the translation-invariance properties (\ref{eq:augmented_potential_property_1}, \ref{eq:augmented_potential_property_2}) and such that $A^\delta(\xv) = A^\delta_\AUG(\xv, 0)$ is closely related to the quasiperiodicity of $A^\delta$ along $\R\, \vvh_1$. More precisely, a function $q \in \scrC^0(\R)$ is said to be \emph{quasiperiodic} of order $n \geq 1$ if there exist a $\Z^n$--periodic function $Q  \in \scrC^0(\R^n)$ and a vector $\bs{\theta} = (\theta_1, \dots, \theta_n) \in \R^n$ with $\theta_j \neq 0$, such that
        \begin{equation}
        \label{eq:def_quasiperiodic_function}
        \spforall y \in \R, \quad q(y) = Q (y\, \bs{\theta}) := Q (y\vts \theta_1, \dots, y\vts \theta_n).
        \end{equation}
        Examples of quasiperiodic functions include finite sums and products of periodic functions. For a detailed study of quasiperiodic (and more generally of almost periodic) functions, we refer to the monographs \cite{besicovitch1954almost,bohr2018almost,levitan1982almost,fink2006almost}; see also \cite{moser1966theory}.

        \hspace{15pt}
        We claim that for a fixed $\zeta \in \R$, the function $q_\zeta: y \in \R \mapsto A^\delta (y\, \vvh_1 + \zeta\, \vvh_2)$, which represents $A^\delta$ in the $\vvh_1$--direction, is quasiperiodic of order $2$ in the sense of \eqref{eq:def_quasiperiodic_function}. In fact, from $\vvh_1 = \vvv_1 + r \vvv_2$ and $\vvh_2 = \vvv_2$, it follows that
        \begin{equation*}
        q_\zeta (y) = \Id - \delta\, \kappa (\delta\vts \zeta)\, a \big(y\, \vvv_1 + (y\vts r + \zeta)\, \vvv_2\big)\, \sigma_2, \quad \spforall y \in \R.
        \end{equation*}
        In particular, with $\bs{\theta} := (1, r)$, we have $q_\zeta (y) = Q_\zeta (y\, \bs{\theta}) = Q_\zeta (y, y\vts r)$, where
        \begin{equation*}
        Q_\zeta(\tau_1, \tau_2) := \Id - \delta\, \kappa (\delta\vts \zeta)\, a \big(\tau_1\, \vvv_1 + (\tau_2 + \zeta)\, \vvv_2\big)\, \sigma_2, \quad \spforall (\tau_1, \tau_2) \in \R^2.
        \end{equation*}
        Since $a(\cdot)$ is $\Lambda$--periodic, it follows that $Q_\zeta$ is $\Z^2$--periodic.  Hence, $q_\zeta$ is  quasiperiodic in the sense of \eqref{eq:def_quasiperiodic_function}. The link between the augmented functions $(\tau_1, \tau_2, \zeta) \mapsto Q_\zeta(\tau_1, \tau_2)$ and $(\xv, s) \mapsto A^\delta_\AUG (\xv, s)$ follows from the change of variables 
        \begin{equation*}
        (\tau_1, \tau_2, \zeta) \mapsto (\xv, s) := (\tau_1\, \vvv_1 + (\tau_2 + \zeta)\, \vvv_2,\, r\vts \tau_1 - \tau_2).
        \end{equation*}
    \end{enumerate}   
    \end{remark}

    \subsection{Formulation of the eigenvalue problem for irrational edges}\label{sec:formulation_eigenvalue_problem_irrational_edges}
    Edge states are nontrivial solutions of the eigenvalue equation $\calH^\delta_\EDGE\, \Psi = E\, \Psi$ that are oscillatory (plane wave-like) along the edge $\R\, \vvh_1$ and decaying in the transverse direction. We now specify the behavior along $\R\, \vvh_1$ of edge states in the irrational case, by exploiting the quasiperiodic nature of $(A^\delta, V)$. The relations $A^\delta (\xv) = A^\delta_\AUG (\xv, 0)$ and $V (\xv) = V_\AUG (\xv, 0)$ suggest looking for edge states in the form 
    \begin{equation}\label{eq:edge_states_ansatz}
    \displaystyle
    \Psi(\xv) = \Psi_\AUG (\xv, s)\big|_{s=0}, \textFOR \xv \in \R^2,
    \end{equation}
    where $\Psi_\AUG (\xv, s)$ satisfies the degenerate elliptic equation
    \begin{equation*}
    \big( -\nabla_\xv \cdot A^\delta_\AUG (\xv, s)\, \nabla_\xv + V_\AUG (\xv, s) \vts\big)\, \Psi_\AUG(\xv, s) = E\, \Psi_\AUG(\xv, s), \quad (\xv, s) \in \R^2 \times \R.
    \end{equation*}
    Moreover, we require decay in any direction transverse to the interface $\Gamma$ defined by \eqref{eq:interface_Gamma}. For instance, choosing the vector $\av_\perp$ defined by \eqref{eq:3D_lattice_vectors}, we may impose:
    \begin{equation*}
        \Psi_\AUG(\xv, s) \to 0, \quad |s| \to +\infty.
    \end{equation*}
    Finally, associated with the periodicity properties (\ref{eq:augmented_potential_property_1}, \ref{eq:augmented_potential_property_2}) of $(A^\delta_\AUG, V_\AUG)$ along the two-dimensional plane $\Gamma$, we introduce,  for parallel quasi-momentum pairs,  $(\kpar, \kperp) \in [-\pi, \pi)^2$ the pseudo-periodic boundary conditions within $\Gamma$:
    \begin{equation*}
    \displaystyle
    \begin{array}{r@{\ =\ }l@{\quad}l}
        \Psi_\AUG(\xv + \vvv_1, s + r) & \euler^{\vts \icplx\vts \kpar}\, \Psi_\AUG(\xv, s), & (\xv, s) \in \R^2 \times \R,
        \rets
        \Psi_\AUG(\xv + \vvv_2, s - 1) & \euler^{\vts \icplx\vts \kperp}\, \Psi_\AUG(\xv, s), & (\xv, s) \in \R^2 \times \R,
    \end{array}
    \end{equation*}
    which describe the oscillatory, or more precisely, the quasiperiodic behavior of $\Psi$ along $\R\, \vvh_1$. 

    {\it Note that for $E \in \R$ and for a fixed pair $(\kpar, \kperp) \in [-\pi, \pi)^2$, the above equations are satisfied by $\Psi_\AUG$ if and only if ${(\xv, s) \mapsto \euler^{\vts \icplx\vts \kperp\vts s}\, \Psi_\AUG (\xv, s)}$ satisfies the same equations for the quasi-momenta $(\kpar + r\vts \kperp, 0)$. Therefore, we may take $\kperp = 0$ with no loss of generality.} %

    This leads to the \emph{augmented edge state eigenvalue problem with parallel quasi-momentum $\kpar \in [-\pi, \pi)$}:
    \begin{equation}\label{eq:evp_irrational_edges}
        \left\{
            \begin{array}{r@{\ }l}
            \displaystyle
            \big( -\nabla_\xv \cdot A^\delta_\AUG (\xv, s)\, \nabla_\xv + V_\AUG (\xv, s) \vts\big)\, &\Psi_\AUG(\xv, s) = E\, \Psi_\AUG(\xv, s),
            \rets
            \Psi_\AUG(\xv + \vvv_1, s + r) &= \euler^{\vts \icplx\vts \kpar}\, \Psi_\AUG(\xv, s),
            \rets
            \Psi_\AUG(\xv + \vvv_2, s - 1) &= \Psi_\AUG(\xv, s),
            \rets
            \Psi_\AUG(\xv, s) \to 0, \quad{} & |s| \to +\infty. 
            \end{array}
        \right.
    \end{equation}
    Formally, if an augmented edge state, $\Psi_\AUG(\xv, s)$, exists and is continuous with respect to the augmented variable $s$, then its restriction $\Psi(\xv) := \Psi_\AUG (\xv, 0)$ is formally an edge state. \bypierre{First, such a restriction cannot vanish: in fact, note that
    \begin{equation*}
        \spforall n_1, n_2 \in \Z, \quad \Psi (\xv + n_1\vts \vvv_1 + n_2\vts \vvv_2) = \Psi_\AUG (\xv + n_1\vts \vvv_1 + n_2\vts \vvv_2, 0) = \euler^{\vts \icplx\vts n_1\vts \kpar}\, \Psi_\AUG (\xv, n_2 - n_1\vts r),
    \end{equation*}
    where we used the boundary conditions in \eqref{eq:evp_irrational_edges}. Therefore, if $\Psi$ was zero, then we would obtain $\Psi_\AUG (\xv, n_2 - n_1\vts r) = 0$ for $n_1, n_2 \in \Z$, which implies $\Psi_\AUG (\xv, s) = 0$ for $(\xv, s) \in \R^2 \times \R$, by density of $\Z + r\vts \Z$ in $\R$ for irrational $r$ and by continuity of $\Psi_\AUG$ in $s$. This contradicts the assumption that $\Psi_\AUG \neq 0$. Moreover, $(\calH^\delta_\EDGE - E)\, \Psi = 0$ and $\xv \mapsto \euler^{-\icplx\vts \kpar\vts (\kvh_1 \cdot \xv)}\, \Psi (\xv)$ is quasiperiodic along the edge $\R\, \vvh_1$ (Remark \ref{rmk:quasiperiodic_nature_domain_wall_potential}). Finally, $\Psi$ decays in the transverse ($\vvh_2 = \vvv_2$) direction: In fact, writing $\zeta = t + n$ where $t \in [0, 1)$ and $n = \lfloor \zeta \rfloor \in \Z$, the second of the boundary conditions in \eqref{eq:evp_irrational_edges} implies $\Psi(\xv + \zeta\vts \vvh_2) = \Psi_\AUG(\xv + (t + n)\vts \vvv_2, 0) = \Psi_\AUG(\xv + t\vts \vvv_2, n) $. Therefore, if the decay property in \eqref{eq:evp_irrational_edges} is uniform in $\xv$, then for $\yrm_1 \in \R$, 
    \begin{equation*}
        |\Psi (\yrm_1\vts \vvh_1 + \zeta\vts \vvh_2)| \leq \sup_{t \in [0, 1)} |\Psi_\AUG (\yrm_1\vts \vvh_1 + t\vts \vvh_2, n)| \to 0 \quad \textnormal{as $|\zeta| \to +\infty$}.
    \end{equation*}}

    \begin{remark}
    \label{rmk:elliptically_degenerate_operator}
    An important feature of the augmented eigenvalue problem \eqref{eq:evp_irrational_edges} is that the principal part of the operator $-\nabla_\xv \cdot A^\delta_\AUG\, \nabla_\xv + V_\AUG$ is degenerate elliptic; it does not involve any derivative with respect to $s$, in contrast to a classical Schrödinger operator, such as: $-\nabla_\xv \cdot A^\delta_\AUG\, \nabla_\xv - \partial^2_s + V_\AUG$.%
    \end{remark}

    \section{Multiscale expansion and effective Dirac operators}\label{sec:ansatz_and_multi_scale_analysis}
    \noindent
    In this section, we extend the procedure in \cite{fefferman2016edge,fefferman2017topologically,lee2019elliptic,drouot2020edge} and present a formal multiscale expansion of solutions of the augmented eigenvalue problem \eqref{eq:evp_irrational_edges}. The approximate edge states emerging from our asymptotic analysis suggest that, for irrational edges, there are infinitely many edge state curves bifurcating from the Dirac point  for $0<\delta\ll 1$. In contrast to earlier studies of edge states for rational edges, we find that these bifurcations are controlled by an \emph{infinite family} of Dirac effective operators. The multiple scale expansion is summarized in Section \ref{sec:recap_multiscale}.

    \subsection{The ansatz}
    We are interested in constructing approximate solutions $\Psi_\AUG$ to the augmented eigenvalue problem \eqref{eq:evp_irrational_edges} which are of wave packet type, and in a sense, spectrally concentrated near $\Kv$. Thus, we take
    \begin{equation*}
    \kpar = \Kv \cdot \vvh_1 + \delta \mu = (\Kv + \delta \mu \kvh_1) \cdot \vvh_1, \textWITH (\mu, \delta) \in \R \times \R_+, \quad 0 < \delta \ll 1.
    \end{equation*}
    The multiple scale structure of $A^\delta_\AUG$ (see \eqref{eq:def_augmented_potential}) suggests that we seek $\Psi_\AUG$ as a function of decoupled variables: a pair of \emph{fast} variables $(\xv, s)$ associated to the augmented bulk honeycomb medium, and a \emph{slow} variable $\zeta = \delta \kvh_2 \cdot (\xv + s \vvv_2)$ associated with the slow variation of the domain wall. Thus we write
    \[
    \Psi_\AUG (\xv, s) = \psi_\AUG (\xv, s, \zeta),\quad \zeta = \delta \kvh_2 \cdot (\xv + s \vvv_2).
    \] 
    Substituting this ansatz into the eigenvalue problem  \eqref{eq:evp_irrational_edges}, and applying the derivation rule
    \begin{equation*}
    \nabla_\xv \Psi_\AUG (\xv, s) = \big(\vts \nabla_\xv + \icplx\vts \delta\vts \kvh_2\, D_\zeta\vts \big)\, \psi_\AUG (\xv, s, \zeta ) \Big|_{\zeta = \delta \kvh_2 \cdot (\xv + s \vvv_2)}, \textWHERE\  D_\zeta := -\icplx\, \partial_\zeta\ ,
    \end{equation*}
    we find that $\psi_\AUG:(\xv, s, \zeta) \in \R^2 \times \R \times \R\mapsto\mathbb C$ satisfies an eigenvalue problem  in the extended set of variables:
    \begin{subequations}\label{eq:evp_decoupled_scales}
        \begin{empheq}[left={\!\!\empheqlbrace }]{align}
            \displaystyle\Big[\! -\big(\vts \nabla_\xv + \icplx\vts \delta\vts \kvh_2\vts D_\zeta\vts \big) \cdot \big(\vts \Id - \delta\vts \kappa(\zeta)\vts a(\xv)\vts \sigma_2 \vts\big)& \big(\vts \nabla_\xv + \icplx\vts \delta\vts \kvh_2\vts D_\zeta\vts \big) + V(\xv) - E \Big]\, \psi_\AUG = 0, \label{eq:evp_decoupled_scales_1}
            \\
            \psi_\AUG(\xv + \vvv_1, s + r, \zeta) &= \euler^{\vts \icplx\vts (\Kv + \delta \mu \kvh_1) \cdot \vvh_1}\, \psi_\AUG (\xv, s, \zeta),\label{eq:evp_decoupled_scales_2}
            \retss
            \psi_\AUG(\xv + \vvv_2, s - 1, \zeta) &= \psi_\AUG (\xv, s, \zeta),\label{eq:evp_decoupled_scales_3}
            \retss
            \psi_\AUG(\xv, s, \zeta) & \to 0, \quad |\zeta| \to +\infty.
        \end{empheq}
    \end{subequations}
    If $\psi_\AUG(\xv,s,\zeta)$ is a solution of \eqref{eq:evp_decoupled_scales}, then $\smash{\psi_\AUG(\xv,s,\zeta) \big|_{\zeta = \delta \kvh_2 \cdot (\xv + s \vvv_2)}}$ is a solution of the eigenvalue problem \eqref{eq:evp_irrational_edges}.

    Let $\Kv_\star \in \{\Kv, \Kv'\}$. Since our goal is to construct $\psi_\AUG$ perturbatively around the Dirac point $(\Kv_\star, E_D)$, we look for a solution of the form
    \begin{equation}\label{eq:ansatz_solution_decoupled_variables}
    \psi_\AUG (\xv, s, \zeta) = \exp \Big[\vts\icplx\vts \big(\vts \delta\, (\mu + \gammatilde)\, \kvh_1 \cdot \xv + (\lvar + \Kv \cdot \vvv_2)\vts s \vts\big)\vts\Big]\, \psi (\xv, \zeta),
    \end{equation} 
    where
    \begin{equation}
    \xv \mapsto \psi(\xv, \zeta) \in L^2_{\Kv_\star}, \quad \zeta \mapsto \psi (\xv, \zeta) \in L^2(\R), \quad |\gammatilde| \leq 1, \quad \lvar \in \R.\label{eq:psi_x_zeta}
    \end{equation}

    \begin{remark}[On the Ansatz \eqref{eq:ansatz_solution_decoupled_variables}
    for $\psi_\AUG (\xv, s, \zeta)$]
    Note  that $\xv\mapsto \psi_\AUG(\xv,s,\zeta)$ is formally a wave packet of bandwidth $\delta$; it is spectrally localized to an $\calO(\delta)$--neighborhood of the Dirac point ($\Kv_\star,E_D)$. Here, the exponential dependence in $s$ is motivated by the fact that the eigenvalue problem \eqref{eq:evp_decoupled_scales_1} is invariant under translations in the $s$ variable. Finally the parameters $\tilde\gamma$ and $\lambda$ will be chosen so that the conditions \eqref{eq:evp_decoupled_scales_2}--\eqref{eq:evp_decoupled_scales_3} and 
    \eqref{eq:psi_x_zeta} are consistent. Just below, we'll impose this consistency and find that  $\delta\tilde\gamma$ and $\lambda$ are ``quantized''; see the relations \eqref{eq:gamma_Iandlambda_I} below.
    We'll then, in Section \ref{sec:multi_scale_procedure}, formulate the PDE problem for $\psi(\xv,\zeta)$ and, for admissible $\tilde\gamma$ and $\lambda$, derive an asymptotic expansion of its solution in powers of $\delta$.
    \end{remark}

    \noindent
    Let us now determine the pairs $(\gammatilde, \lvar) \in \R \times \R$ for which the pseudo-periodic boundary conditions in \eqref{eq:evp_decoupled_scales} and \eqref{eq:psi_x_zeta} satisfied by $\psi$ and $\psi_\AUG$ are compatible. Combining \eqref{eq:evp_decoupled_scales_2} with the relation $\psi (\xv + \vvv_1, \zeta) = \euler^{\vts \icplx\vts \Kv_\star \cdot \vvv_1}\, \psi(\xv, \zeta)$, and \eqref{eq:evp_decoupled_scales_3} with the relation $\psi (\xv + \vvv_2, \zeta) = \euler^{\vts \icplx\vts \Kv_\star \cdot \vvv_2}\, \psi(\xv, \zeta)$, we obtain:
    \begin{equation}\label{eq:corrector_compatibility_relations}
    \left\{
        \begin{array}{r@{\quad}l}
        \Kv \cdot \vvh_1 = \delta\, \gammatilde + (\lvar + \Kv \cdot \vvv_2)\vts r + \Kv_\star \cdot \vvv_1 + 2\pi n & n \in \Z,
        \retss
        0 = - \lvar - \Kv \cdot \vvv_2 + \Kv_\star \cdot \vvv_2 + 2\pi m & m \in \Z.
        \end{array}
    \right.
    \end{equation}
    Here, we have also used that $\kvh_1 \cdot \vvv_1 = \kv_1 \cdot \vvv_1 = 1$ and $\kvh_1 \cdot \vvv_2 = \kv_1 \cdot \vvv_2 = 0$. %
    The second equation in \eqref{eq:corrector_compatibility_relations} determines $\lvar$. Substituting this expression into the first equation, and using the relation $\vvv_1 + r \vvv_2 = \vvh_1$, we get
    \begin{equation}\label{eq:corrector_second_relation}
    \delta\, \gammatilde = (\Kv - \Kv_\star) \cdot \vvh_1  - 2\pi (n + m r), \quad (m, n) \in \Z^2.
    \end{equation}
    We seek integers $m,n$ and $\delta\, \gammatilde \in [-\pi,\pi)$ satisfying  \eqref{eq:corrector_second_relation}. The latter constraint implies
    \begin{equation*}
    \begin{alignedat}{3}
        \qquad&   \delta\, \gammatilde\ +\pi = (\Kv - \Kv_\star) \cdot \vvh_1  - 2\pi (n + m r) + \pi &&\in [0,2\pi)
        \\
        \iff\quad&   
        \frac{\delta\, \gammatilde}{2\pi}\ +\frac{1}{2} = \frac{(\Kv - \Kv_\star)\cdot \vvh_1}{2\pi}   - (n + m r) + \frac{1}{2} &&\in[0,1)
        \\
        \iff\quad&   \frac{\delta\, \gammatilde}{2\pi}\ + \frac{1}{2}= \modulo{ 
        \frac{(\Kv - \Kv_\star) \cdot \vvh_1}{2\pi}  -  m r + \frac{1}{2}} &&\in [0,1)
        \retss
        \iff\quad&  \delta\, \gammatilde  = 2\pi\ \modulo{ 
        \frac{(\Kv - \Kv_\star)\cdot \vvh_1}{2\pi}   -  m r + \frac{1}{2}}   - \pi  &&\in [-\pi,\pi),
    \end{alignedat}
    \end{equation*}
    which is independent of $n$.
    Here, $\modulo{x} := x - \lfloor x \rfloor \in [0, 1)$ is the fractional part of $x \in \R$.

    Let 
    \[ \bbL := \{\Kv, \Kv'\} \times \Z\]
    and choose any $\iiv = (\Kv_\star, m) \in \bbL$. Then, the pair $(\gammatilde_\iiv, \lvar_\iiv)$ given by 
    \begin{subequations}\label{eq:gamma_Iandlambda_I}
        \begin{align}
            \delta\, \gammatilde^\delta_\iiv = \gamma_\iiv &:=  2\pi\ \modulo{\frac{(\Kv - \Kv_\star)\cdot \vvh_1}{2\pi}   -  m r + \frac{1}{2}}   - \pi,
            \retss
            \lvar_\iiv &:= (\Kv_\star - \Kv) \cdot \vvv_2 + 2\pi m,
        \end{align}
    \end{subequations}
    satisfies the two compatibility relations in \eqref{eq:corrector_compatibility_relations}.  Finally, to ensure that $|\gammatilde^\delta_\iiv| = \delta^{-1}\, |\gamma_\iiv| \leq 1$, we consider indices $\iiv$ that belong to $\bbL(\delta)$, where 
    \begin{equation}\label{eq:def_L_veps}
        \spforall \veps > 0, \quad \bbL(\veps) := \big\{\vts \iiv \in \bbL,\ \ |\gamma_\iiv| \leq \veps \vts\big\}.
    \end{equation}
    The objects $\lvar_\iiv$, $\gamma_\iiv$, $\bbL$ and $\bbL(\delta)$ naturally arise and play an important role  in Section \ref{sec:veps_neighborhood}. 

    \begin{remark}[The set $\bbL(\delta)$ for rational and irrational edges]
    \label{remark:bbL-rational_v_irrational}
    As shown in Proposition \ref{prop:veps_neighborhood_K_rational}, the set $\bbL(\delta)$ is quite explicit for $r$ rational and $\delta$ small enough. For instance, $\bbL(\delta) = \{\Kv\} \times \Z$ for $r = 0$, and $\bbL(\delta) = \{\Kv, \Kv'\} \times \Z$ for $r = 1$; see Figure \ref{fig:broken_quasi_momentum_line}. Moreover, for rational edges, the sequence $(\gamma_\iiv)_{\iiv \in \bbL(\delta)}$ takes a finite number of values including $0$, and so $\gamma_\iiv = 0$ for $\iiv \in \bbL(\delta)$, provided that $\delta$ is small enough (Proposition \ref{prop:veps_neighborhood_K_rational}). On the other hand, if $r$ is irrational, $(\gamma_\iiv)_{\iiv \in \bbL}$ is dense in $[-\pi, \pi]$, so one can extract a subsequence which tends to $0$. In particular, $\bbL(\delta)$ is a non-empty, nontrivial, and countable set.
    \end{remark} 

    \medskip\noindent
    In the following, for any fixed index $\iiv \in \bbL(\delta)$, we assume $\psi_\AUG (\xv, s, \zeta)$ to be of the form \eqref{eq:ansatz_solution_decoupled_variables}, where $\gammatilde \equiv \gammatilde^\delta_\iiv$ and $\lvar=\lvar_\iiv$ are given by \eqref{eq:gamma_Iandlambda_I}.

    \subsection{Multiscale expansion}\label{sec:multi_scale_procedure}~
    Fix $\iiv = (\Kv_\star, m) \in \bbL (\delta)$, so that $|\gammatilde^\delta_\iiv| \leq 1$. Define 
    \begin{equation}\label{eq:diff_op_slow_variable}
    \Upsilon^\delta_\iiv (\mu; D_\zeta) := \icplx\, \big(\vts \kvh_2\, D_\zeta + (\vts \mu + \gammatilde^\delta_\iiv \vts)\, \kvh_1 \vts\big) \qquad (\textnormal{with} \quad D_\zeta := -\icplx\, \partial_\zeta).
    \end{equation}
    Substituting the ansatz (\ref{eq:ansatz_solution_decoupled_variables}, \ref{eq:psi_x_zeta}, \ref{eq:gamma_Iandlambda_I}) into \eqref{eq:evp_decoupled_scales}, we find that $(E, \psi) \equiv (E_\iiv, \psi_\iiv)$ satisfies the problem
    \begin{equation*}
    \left\{
        \begin{array}{r@{\ }l}
        \big( -\Delta_\xv + V(\xv) + \delta\, \scrG^{(1)} + \delta^2\, \scrG^{(2)} + \delta^3\, \scrG^{(3)} \big)\,\psi_\iiv(\xv, \zeta) &= E_\iiv\, \psi_\iiv(\xv, \zeta), \quad (\xv, \zeta) \in \R^2 \times \R,
        \ret
        \xv \mapsto \psi_\iiv(\xv, \zeta) \in L^2_{\Kv_\star}, \quad \zeta \mapsto \psi_\iiv(\xv, \zeta) &\in L^2(\R),
        \end{array}
    \right.
    \end{equation*}
    where
    \begin{align*}
    \scrG^{(1)} &:= -\nabla_\xv \cdot \big(\vts 2\, \Upsilon^\delta_\iiv (\mu; D_\zeta) - \kappa (\zeta)\, a(\xv)\, \sigma_2\, \nabla_\xv \vts\big),
        \ret
        \scrG^{(2)} &:= \nabla_\xv \cdot \kappa (\zeta)\, a(\xv)\, \sigma_2\, \Upsilon^\delta_\iiv (\mu; D_\zeta) + \Upsilon^\delta_\iiv (\mu; D_\zeta) \cdot \kappa (\zeta)\, a(\xv)\, \sigma_2\, \nabla_\xv - \Upsilon^\delta_\iiv (\mu; D_\zeta) \cdot \Upsilon^\delta_\iiv (\mu; D_\zeta),
        \ret
        \scrG^{(3)} &:= \Upsilon^\delta_\iiv (\mu; D_\zeta) \cdot \kappa (\zeta)\, a(\xv)\, \sigma_2\, \Upsilon^\delta_\iiv (\mu; D_\zeta).
    \end{align*}
    We seek a multiscale expansion of the solution, as in \cite{fefferman2016edge,fefferman2017topologically,lee2019elliptic,drouot2020edge}. In particular,  we expand $(E_\iiv, \psi_\iiv)$ in powers of $\delta$:
    \begin{equation*}
    \left\{
        \begin{array}{r@{\ =\ }l@{\qquad}r@{\ \equiv\ }l}
        E_\iiv & E_D + \delta\vts E^{(1)} + \delta^2\vts E^{(2)} + \dots; & E^{(n)} & E^{(n)}_\iiv (\mu),
        \rets
        \psi_\iiv & \psi^{(0)} + \delta\, \psi^{(1)} + \delta^2\, \psi^{(2)} + \dots; & \psi^{(n)}(\xv, \zeta) & \psi^{(n)}_\iiv (\xv, \zeta; \mu).
        \end{array}
    \right.
    \end{equation*}
    By substituting this expansion in the PDE satisfied by $\psi_\iiv$, and by equating terms of same order in powers of $\delta$, we obtain a hierarchy of equations satisfied by $\psi^{(n)}$, $n \geq 0$\footnote{In the multiscale expansion, $\gammatilde^\delta_\iiv$ is treated as an $\calO(1)$ term, even though its actual order may be higher.}. Each of these equations is completed with the boundary conditions
    \begin{equation*}
    \displaystyle
    \spforall n \geq 0, \quad \xv \mapsto \psi^{(n)}(\xv, \zeta) \in L^2_{\Kv_\star}, \quad \zeta \mapsto \psi^{(n)}(\xv, \zeta) \in L^2(\R),
    \end{equation*}
    which express $\Kv_\star$--pseudo-periodicity in $\xv$ and decay in $\zeta$ respectively.

    At order $\delta^0$ and for a fixed $\zeta \in \R$, one finds that $\psi^{(0)}(\cdot\,, \zeta)$ satisfies
    \begin{equation*}
    (\calH^0_{\Kv_\star} - E_D)\, \psi^{(0)} (\cdot\,, \zeta) = 0.
    \end{equation*}
    Since the kernel of $\calH^0_{\Kv_\star} - E_D$ is a two-dimensional space spanned by the pair $\{\Phi^{\Kv_\star}_1, \Phi^{\Kv_\star}_2\}$ of eigenfunctions associated to $(\Kv_\star, E_D)$ (Condition \ref{item:dirac_point_item_2} in Proposition \ref{prop:sufficient_conditions_dirac_point}), we obtain
    \begin{align}
    \psi^{(0)}(\xv, \zeta) &= \alpha_1 (\zeta)\, \Phi^{\Kv_\star}_1(\xv) + \alpha_2 (\zeta)\, \Phi^{\Kv_\star}_2(\xv)\nonumber
    \retss
    &= \Phi^{\Kv_\star} (\xv)^\transp \alpha(\zeta), \textWITH \Phi^{\Kv_\star} = (\Phi^{\Kv_\star}_1, \Phi^{\Kv_\star}_2)^\transp, \label{eq:epxression_Psi_0}
    \end{align}
    and where $\alpha = (\alpha_1, \alpha_2)^\transp \in L^2(\R; \C^2)$ remains to be determined.

    At order $\delta^1$, one finds that $\psi^{(1)} (\cdot\,, \zeta)$ satisfies
    \begin{equation}\label{eq:bvp_psi1}
    ( \calH^0_{\Kv_\star} - E_D )\, \psi^{(1)} (\cdot\,, \zeta) = (- \scrG^{(1)} + E^{(1)} )\, \psi^{(0)} (\cdot\,, \zeta).
    \end{equation}
    For $\zeta \in \R$ fixed, \eqref{eq:bvp_psi1} has a solution in $L^2_{\Kv_\star}$ only if its right-hand side is $L^2_{\Kv_\star}$--orthogonal to the Dirac basis functions $\Phi^{\Kv_\star}_l$, that is,
    \begin{equation*}
    \big\langle \Phi^{\Kv_\star}_{l},\, (- \scrG^{(1)} + E^{(1)} )\, \psi^{(0)} (\cdot\,, \zeta) \big\rangle_{L^2_{\Kv_\star}} = 0, \quad l = 1, 2.
    \end{equation*}
    These solvability conditions, combined with the form \eqref{eq:epxression_Psi_0} of $\psi^{(0)}$ and the definition of $\scrG^{(1)}$, lead to the eigenvalue problem: \emph{find $(E^{(1)}, \alpha) \in \C \times L^2(\R; \C^2)$ such that for $l = 1, 2$,}
    \begin{multline}\label{eq:effective_Dirac_1}
    \displaystyle
    \sum_{j = 1,2} \big\langle \Phi^{\Kv_\star}_{l}, -2\vts \icplx\vts \kvh_2 \cdot \nabla \Phi^{\Kv_\star}_j \big\rangle_{L^2_{\Kv_\star}} (D_\zeta\vts \alpha_j)\vts (\zeta) + %
    (\mu + \gammatilde^\delta_\iiv)\vts \sum_{j = 1,2} \big\langle \Phi^{\Kv_\star}_{l}, -2 \icplx \kvh_1 \cdot \nabla \Phi^{\Kv_\star}_j \big\rangle_{L^2_{\Kv_\star}} \alpha_j (\zeta)
    \\
    + \kappa(\zeta) \sum_{j = 1,2} \big\langle \Phi^{\Kv_\star}_l, \nabla \cdot \big( a(\xv)\, \sigma_2\, \nabla \Phi^{\Kv_\star}_j \big) \big\rangle_{L^2_{\Kv_\star}} \alpha_j (\zeta) -%
    E^{(1)}\, \alpha_{l} (\zeta) = 0.
    \end{multline}
    We achieve a simplification of  \eqref{eq:effective_Dirac_1} using the following propositions.

    \begin{proposition}[\texorpdfstring{\cite[Eq. (7.28)--(7.29)]{fefferman2014wave}}{\cite{fefferman2014wave}}]\label{prop:expression_Phi_l_nabla_Phi_j}
    For $\Kv_\star \in \{\Kv, \Kv'\}$ and $\modK = (\modK^{(1)}, \modK^{(2)}) \in \R^2$, let 
    \begin{equation}\label{eq:def_twisted_Pauli_matrix_1}
        \sigma^{\Kv_\star} (\modK)  := \displaystyle
        \begin{pmatrix}
        \big\langle \Phi^{\Kv_\star}_1, -2\icplx\vts \modK \cdot \nabla \Phi^{\Kv_\star}_1 \big\rangle_{L^2_{\Kv_\star}} & %
        \big\langle \Phi^{\Kv_\star}_1, -2\icplx\vts \modK \cdot \nabla \Phi^{\Kv_\star}_2 \big\rangle_{L^2_{\Kv_\star}} \rets %
        \big\langle \Phi^{\Kv_\star}_2, -2\icplx\vts \modK \cdot \nabla \Phi^{\Kv_\star}_1 \big\rangle_{L^2_{\Kv_\star}} & %
        \big\langle \Phi^{\Kv_\star}_2, -2\icplx\vts \modK \cdot \nabla \Phi^{\Kv_\star}_2 \big\rangle_{L^2_{\Kv_\star}}
        \end{pmatrix}.%
    \end{equation}
    Then,  
    \begin{equation}
        \sigma^{\Kv_\star} (\modK) = \upsilon^{\Kv_\star}\, 
        \begin{pmatrix}
        0 & \modK^{(1)} + \icplx\vts \modK^{(2)}
        \\
        \modK^{(1)} - \icplx\vts \modK^{(2)} & 0
        \end{pmatrix} =
        \upsilon^{\Kv_\star}\, \big(\vts \modK^{(1)}\sigma_1 - \modK^{(2)}\sigma_2 \vts\big). \label{eq:def_twisted_Pauli_matrix_2}
    \end{equation}
    Here, $\upsilon^{\Kv} = \upsilon_D$ and $\upsilon^{\Kv'} = -\upsilon_D$, where $\upsilon_D > 0$ is the Dirac velocity \eqref{eq:conical_behavior_Dirac}; see Condition \ref{item:dirac_point_item_3} in Proposition \ref{prop:sufficient_conditions_dirac_point}.
    \end{proposition}

    \begin{proposition}[\texorpdfstring{\cite[Proposition 4.2]{drouot2020edge}}{}]\label{prop:expression_Phi_l_W_Phi_j}
    For $\Kv_\star \in \{\Kv, \Kv'\}$,
    \begin{equation*}
        \displaystyle
        \begin{pmatrix}
        \big\langle \Phi^{\Kv_\star}_1, \nabla \cdot \big( a(\xv)\, \sigma_2\, \nabla \Phi^{\Kv_\star}_1 \big) \big\rangle_{L^2_{\Kv_\star}} & %
        \big\langle \Phi^{\Kv_\star}_1, \nabla \cdot \big( a(\xv)\, \sigma_2\, \nabla \Phi^{\Kv_\star}_2 \big) \big\rangle_{L^2_{\Kv_\star}} \rets %
        \big\langle \Phi^{\Kv_\star}_2, \nabla \cdot \big( a(\xv)\, \sigma_2\, \nabla \Phi^{\Kv_\star}_1 \big) \big\rangle_{L^2_{\Kv_\star}} & %
        \big\langle \Phi^{\Kv_\star}_2, \nabla \cdot \big( a(\xv)\, \sigma_2\, \nabla \Phi^{\Kv_\star}_2 \big) \big\rangle_{L^2_{\Kv_\star}} %
        \end{pmatrix} %
        = \vartheta^{\Kv_\star}\, \sigma_3,
    \end{equation*}
    with $\vartheta^{\Kv'} = \vartheta^{\Kv} := \lla \Phi^{\Kv}_1, \nabla \cdot \big( a(\xv)\, \sigma_2\, \nabla \Phi^{\Kv}_1 \big) \rra_{L^2_{\Kv}} \neq 0$ under Assumption \ref{hyp:non_degeneracy} (Section \ref{sec:conjugation_breaking_perturbations}).
    \end{proposition}

    \vspace{1\baselineskip} \noindent
    By Proposition \ref{prop:expression_Phi_l_nabla_Phi_j} (with $\modK = \kvh_1, \kvh_2$) and Proposition \ref{prop:expression_Phi_l_W_Phi_j}, Equation \eqref{eq:effective_Dirac_1} becomes
    \begin{equation}\label{eq:effective_Dirac_2}
    \big(\vts \calD^{\Kv_\star} (\mu + \gammatilde^\delta_\iiv) - E^{(1)} \vts\big)\, \alpha(\zeta) = 0,
    \end{equation}
    where $\calD^{\Kv_\star} (\muhat): H^1(\R; \C^2) \to L^2 (\R; \C^2)$ is the \emph{effective Dirac operator}:
    \begin{equation} \label{eq:effective_Dirac_operator}
    \spforall \muhat \in \R, \quad \calD^{\Kv_\star} (\muhat) := \displaystyle \sigma^{\Kv_\star} (\kvh_2)\, D_\zeta %
    + \muhat\, \sigma^{\Kv_\star} (\kvh_1) + \vartheta^{\Kv_\star}\, \kappa(\zeta)\, \sigma_3.
    \end{equation}
    In the next section, Section \ref{sec:pties_effective_Dirac_operator}, we show (using the property $\kappa (\zeta) \to \pm 1$ as $\zeta \to \pm \infty$) that $\calD^{\Kv_\star} (\muhat)$ has a gap in its continuous spectrum centered around $0$, containing at least one isolated eigenvalue, and that all such isolated eigenvalues are simple. As explained in Section \ref{sec:recap_multiscale} below, the eigenpairs $(E^{(1)},\alpha(\zeta))$ of the effective operators $\calD^{\Kv_\star}(\mu + \gammatilde^\delta_\iiv)$ give rise to approximate edge states which, at leading order in $\delta$, have the form $\psi^{(0)}(\xv, \zeta) = \Phi^{\Kv_\star} (\xv)^\transp \alpha(\zeta)$, with $\zeta = \delta \kvh_2 \cdot (\xv + s \vvv_2)$.

    The order $\delta$ correction to $\psi^{(0)}(\xv,\zeta)$ is derived  as follows. Let $(E^{(1)}, \alpha (\zeta)) \in \R \times L^2 (\R; \C^2)$ be an eigenpair \eqref{eq:effective_Dirac_2}. Then the right-hand side of \eqref{eq:bvp_psi1} belongs to $P_\perp L^2_{\Kv_\star}$, where $P_\perp$ is the $L^2_{\Kv_\star}$--projector onto the orthogonal complement of $\Ker\, (\calH^0_{\Kv_\star} - E_D) = \vect \{\Phi^{\Kv_\star}_1, \Phi^{\Kv_\star}_2\}$. By Fredholm's alternative (see e.g. \cite[Theorem 6.6]{brezis_functional_2011}), it follows that $\smash{(\calH^0_{\Kv_\star} - E_D)|_{P_\perp L^2_{\Kv_\star}}}$ has a bounded inverse $R_{\Kv_\star} \in \scrL(P_\perp L^2_{\Kv_\star})$, and we find that
    \begin{align}
    \displaystyle
    \psi^{(1)} (\xv, \zeta) &= \psi^{(1)}_p (\xv, \zeta) + \psi^{(1)}_h (\xv, \zeta), \textWITH \nonumber
    \retss
    \psi^{(1)}_p (\xv, \zeta) &:= R_{\Kv_\star}\, \big(- \scrG^{(1)} + E^{(1)} \big)\, \psi^{(0)} (\xv, \zeta) \textAND \psi^{(1)}_h (\xv, \zeta) := \Phi^{\Kv_\star} (\xv)^\transp \beta (\zeta),\label{eq:psi1_particular_plus_homogeneous}
    \end{align}
    and where $\beta \in L^2(\R; \C^2)$ is obtained using the solvability conditions of the equation satisfied at order $\delta^2$. This multiple scale expansion procedure can be continued to construct approximations of edge states at any arbitrary order $\delta^n$, $n \geq 2$.

    \subsection{Spectral properties of the effective Dirac operator}\label{sec:pties_effective_Dirac_operator}
    The  Dirac operator $\calD^{\Kv_\star} (\muhat)$  defined by \eqref{eq:effective_Dirac_operator}, and which emerges from the multiscale expansion in Section \ref{sec:multi_scale_procedure}, is unitarily equivalent to another Dirac operator studied for instance in \cite{drouot2020edge}.
    \begin{lemma}\label{lem:equivalent_Dirac_operators}
    Let $\Kv_\star \in \{\Kv, \Kv'\}$, $\muhat \in \R$, and introduce the Dirac operator
    \begin{equation*}
        \displaystyle
        \calD^{\Kv_\star}_0 (\muhat) := \upsilon^{\Kv_\star}\, |\kvh_2|\, \sigma_1\, D_\zeta - \frac{\upsilon^{\Kv_\star}}{|\kvh_2|}\, \muhat\, \sigma_2 + \vartheta^{\Kv_\star}\, \kappa(\zeta)\, \sigma_3,
    \end{equation*}
    with  $\vartheta^{\Kv'} = \vartheta^{\Kv}$. There exist constants $(K, \omega) \in \R \times \C$ depending on $\kvh_1$, $\kvh_2$, and $\upsilon_D$ only, and with $|\omega| = 1$, such that the following unitary equivalence relation holds:
    \begin{equation}\label{eq:unitary-equiv}
        \calD^{\Kv_\star} (\muhat) = \calN (\muhat)^*\; \calD^{\Kv_\star}_0 (\muhat)\; \calN (\muhat).\quad  \textrm{Here,}\quad (\calN (\muhat)\, \alpha)\vts (\zeta) := %
        \begin{pmatrix}
        \omega & 0
        \\
        0 & \overline{\omega}
        \end{pmatrix}\, \alpha(\zeta)\; \euler^{\vts\icplx\vts \muhat\vts K \vts \zeta}.
    \end{equation}
    \end{lemma}

    \begin{dem}
    It suffices to find $(K, \omega) \in \R \times \C$ for which \eqref{eq:unitary-equiv} holds. Let $\zvh_j := \kvh^{(1)}_j + \icplx\vts \kvh^{(2)}_j$, with $|\kvh_j| = |\zvh_j|$, $j = 1, 2$. Using the definition \eqref{eq:effective_Dirac_operator} of $\calD^{\Kv_\star} (\muhat)$, the expression for $\calN (\muhat)$, and the constraint $|\omega|^2 = 1$, we obtain
    \begin{align*}
        \calN (\muhat)\, \calD^{\Kv_\star} (\muhat)\, \calN (\muhat)^* \!\!&\phantom{:}= %
        \begin{pmatrix}
        0 & Z_1
        \\
        \overline{Z_1} & 0
        \end{pmatrix}\, D_\zeta + \muhat\, %
        \begin{pmatrix}
        0 & Z_2
        \\
        \overline{Z_2} & 0
        \end{pmatrix} + \vartheta^{\Kv_\star}\, \kappa(\zeta)\, \sigma_3,
        \retss
        \textWITH Z_1 &:= \upsilon^{\Kv_\star}\, \omega^2\, \zvh_2 \textAND Z_2 := \upsilon^{\Kv_\star}\, \omega^2\, (\zvh_1 - K\vts \zvh_2).
    \end{align*}
    The right-hand side of the above equation coincides with $\calD^{\Kv_\star}_0 (\muhat)$ if and only if
    \begin{equation}\label{eq:equivalent_Dirac_operators_dem_2}
        \displaystyle
        Z_1 = \omega^2\, \upsilon^{\Kv_\star}\, \zvh_2 = \upsilon^{\Kv_\star}\, |\zvh_2| \textAND Z_2 = \omega^2\, \upsilon^{\Kv_\star} (\zvh_1 - K\vts \zvh_2) = \icplx\vts\, \frac{\upsilon^{\Kv_\star}}{|\zvh_2|}.
    \end{equation}
    The first equation in \eqref{eq:equivalent_Dirac_operators_dem_2} determines $\omega$ up to sign, with $|\omega| = 1$. We substitute this into the second equation in \eqref{eq:equivalent_Dirac_operators_dem_2}, to obtain
    \begin{equation*}
        K = \frac{\zvh_1\, \overline{\zvh_2} - \icplx}{|\zvh_2|^2}.
    \end{equation*}
    Finally, using the relation
    \begin{align*}
        \zvh_1\, \overline{\zvh_2} &= (\kvh_1^{(1)} + \icplx\vts \kvh_1^{(2)})\, (\kvh_2^{(1)} - \icplx\vts \kvh_2^{(2)}) = (\kvh_1 \cdot \kvh_2) + \icplx\vts (\kvh_2 \wedge \kvh_1)
        \\
        &= (\kvh_1 \cdot \kvh_2) + \icplx\vts (\kv_2 \wedge \kv_1) = (\kvh_1 \cdot \kvh_2) + \icplx \quad \textnormal{by (\ref{eq:def_edge_vectors}, \ref{eq:honeycomb_vector_quasimomenta})},
    \end{align*}
    we deduce that $K = (\kvh_1 \cdot \kvh_2) / |\kvh_2|^2$ is indeed a real coefficient. Therefore, the proof of Lemma \ref{lem:equivalent_Dirac_operators} is complete.
    \end{dem}

    \vspace{1\baselineskip}\noindent
    Thanks to Lemma \ref{lem:equivalent_Dirac_operators}, the spectral properties of $\calD^{\Kv_\star} (\muhat)$ can be derived from those of $\calD^{\Kv_\star}_0 (\muhat)$. We use the results \cite[Lemma 3.1]{drouot2019characterization} and \cite[Proposition 4.5]{drouot2020edge}, on the spectrum of $\calD^{\Kv_\star}_0 (\muhat)$.

    \begin{proposition}\label{prop:pties_Dirac_operator}
        Let $\Kv_\star \in \{\Kv, \Kv'\}$. 
        Assume $\kappa(\zeta)$ is a domain wall function (see Definition \ref{defi:dw}), and that $\vartheta^{\Kv_\star} \neq 0$ (see Assumption \ref{hyp:non_degeneracy}). Then the following holds:
        \begin{enumerate}[label={$(\alph*).$}, ref={\theproposition.$(\alph*)$}, wide=0pt]
            \item\label{prop:pties_Dirac_operator_item_1} The essential spectrum of $\calD^{\Kv_\star}(0)$ coincides with its absolutely continuous part, which is $\R \setminus (-\theta_\GAP, \theta_\GAP)$, with $\theta_\GAP := |\vartheta^{\Kv_\star}|$. Moreover, zero is always in the point spectrum, and there exists an integer $N \geq 0$ such that the discrete spectrum of $\calD^{\Kv_\star}(0)$ consists of $2 N + 1$ simple eigenvalues located in the gap $(-\theta_\GAP, \theta_\GAP)$, and symmetric about $0$:
            \begin{equation*}
            -\theta_\GAP < z_{-N} < \dots < z_{-1} < z_0 = 0 < z_1 < \dots < z_N < \theta_\GAP, \quad z_{-j} = - z_j.
            \end{equation*}
            \item For any $\muhat \in \R$, the essential spectrum of $\calD^{\Kv_\star}(\muhat)$ is purely absolutely continuous, and is given by 
            \begin{equation*}
            \R \setminus (-\theta_\GAP(\muhat), \theta_\GAP(\muhat)), \textWHERE \theta_\GAP(\muhat) := \sqrt{\theta_\GAP^2 + \frac{\upsilon^2_D}{|\kvh_2|^2}\, \muhat^2}.
            \end{equation*}
            The discrete spectrum of $\calD^{\Kv_\star}(\muhat)$ consists of $2 N + 1$ simple eigenvalues $z_j (\muhat)$, $-N \leq j \leq N$, in the gap $(-\theta_\GAP(\muhat), \theta_\GAP(\muhat))$:
            \begin{equation*}
            \displaystyle
            z_0(\muhat) := -\frac{\upsilon_D}{|\kvh_2|}\, \muhat\, \sign (\vartheta^{\Kv_\star}), \quad z_{\pm j} (\muhat) := \pm \sqrt{z^2_j + \frac{\upsilon^2_D}{|\kvh_2|^2}\, \muhat^2}, \quad \spforall j = 1, \dots, N.
            \end{equation*}
            where $z_j$, $-N \leq j \leq N$, are the eigenvalues of $\calD^{\Kv_\star}(0)$. 
        \end{enumerate}
    \end{proposition}

    \noindent
    The spectrum of $\calD^{\Kv_\star} (\muhat)$ as a function of $\muhat$ is plotted in Figure \ref{fig:spectrum_Dirac_operator} for $N = 1$.

    \begin{figure}[ht!]
        \makebox[\textwidth][c]{
            \includegraphics[page=5]{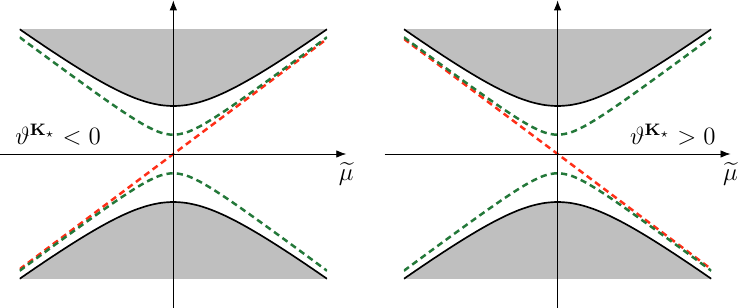}
        }
        \caption{Spectrum of $\calD^{\Kv_\star} (\muhat)$ with respect to $\muhat$ for $N = 1$. The gray area represents the essential spectrum, the red curve is the topologically protected zero eigenvalue $\muhat \mapsto z_0(\muhat)$, and the green curves represent $\muhat \mapsto (z_{-1}(\muhat), z_{1}(\muhat))$. \label{fig:spectrum_Dirac_operator}}
    \end{figure}

    \begin{remark}\label{rmk:eigenvalues_Dirac_operator}
    The eigenvalue curve, $\muhat\mapsto z_0(\muhat)$, is ``topologically protected'' in the sense that it persists against arbitrary (even large) spatially localized perturbations of the  domain wall function, $\kappa(\zeta)$. Further, %
    while $\calD^{\Kv_\star} (\muhat)$ has at least one eigenvalue $z_0 (\muhat)$, the total number of eigenvalues depends on the details of $\kappa$.  For a prescribed $N \geq 0$, the function $\kappa$ can be constructed to ensure that $\calD^{\Kv_\star} (\muhat)$ admits exactly $2N + 1$ eigenvalues; see \cite{lu2020dirac}.
    \end{remark}

    \subsection{Summary of multiscale procedure and implications for irrational edges}\label{sec:recap_multiscale}
    For $\iiv = (\Kv_\star, m) \in \bbL(\delta)$ and $-N \leq j \leq N$, let $(z_j (\muhat),\, \alpha^{\Kv_\star}_j (\cdot\,;\, \muhat)) \in \R \times H^1(\R; \C^2)$ denote an eigenpair of the effective Dirac operator $\calD^{\Kv_\star} (\muhat)$. Then setting $\kpar = \Kv \cdot \vvh_1 + \delta \mu$ with $\mu \in \R$, the formal analysis of Section \ref{sec:multi_scale_procedure} shows that the function:
    \begin{multline}\label{eq:approximate_augmented_eigenfunction_1}
        \Psi^{\delta, (0)}_{\AUG, \iiv, j}: (\xv, s) \mapsto \exp \Big[\vts\icplx\vts \big(\vts (\mu\, \delta + \gamma_\iiv)\, \kvh_1 \cdot \xv + (\lvar_\iiv + \Kv \cdot \vvv_2)\vts s \vts\big)\vts\Big]\, 
        \\
        \times \Phi^{\Kv_\star} (\xv)^\transp\, \alpha^{\Kv_\star}_j \big(\delta \kvh_2 \cdot (\xv + s \vvv_2)\,;\, \mu + \delta^{-1} \gamma_\iiv \big)
    \end{multline}
    is an $\calO(\delta)$--approximate solution to \eqref{eq:evp_irrational_edges} associated with the energy $E_D + \delta\vts z_j(\mu + \delta^{-1} \gamma_\iiv)$. 

    Evaluating this solution at $s = 0$ (note that $\alpha^{\Kv_\star}_j (\cdot\,;\, \muhat) \in H^1(\R) \subset \scrC^0(\R)$), leads to the two-dimensional $\calO(\delta)$--approximate edge state for the Hamiltonian  $\calH^\delta_\EDGE$ (see \eqref{eq:def_dw_coefficient}):
    \begin{equation}
        \Psi^{\delta, (0)}_{\iiv, j, \kpar} (\xv) = \exp \Big[\vts\icplx\vts \big(\vts (\mu\, \delta + \gamma_\iiv)\, \kvh_1 \cdot \xv \vts\big)\vts\Big]\, \Phi^{\Kv_\star} (\xv)^\transp\, \alpha^{\Kv_\star}_j \big(\delta \kvh_2 \cdot \xv\,;\, \mu + \delta^{-1} \gamma_\iiv \big).\label{eq:quasi-per-edge-states}
    \end{equation}
    For the case of irrational edges, the expression \eqref{eq:quasi-per-edge-states} is a quasiperiodic function as $\xv$ varies along the line $\R\, \vvh_1$, and it decays on a long length scale, of order $\delta^{-1}$, in directions transverse to this line.

    For rational edges, if $\delta$ is small enough, then $\gamma_\iiv = 0$ for $\iiv \in \bbL (\delta)$. Therefore, the expression for the $\calO(\delta)$--approximate edge state reduces to $\Psi^{\delta, (0)}_{\iiv, j, \kpar} (\xv) = \euler^{\vts \icplx\, \mu\vts \delta\, \kvh_1 \cdot \xv}\, \Phi^{\Kv_\star} (\xv)^\transp\, \alpha^{\Kv_\star}_j (\delta \kvh_2 \cdot \xv\,;\, \mu)$, which only involves $\calD^{\Kv_\star} (\mu)$. This is consistent with the results in \cite{fefferman2016edge,fefferman2017topologically,lee2019elliptic,drouot2020edge} for rational edges. Note that even though the augmented solutions are distinct, their restriction to $s = 0$ only depends on $\Kv_\star$. On the other hand, if $r$ is irrational, then $\{\delta^{-1}\, \gamma^\delta_\iiv\}_{\iiv \in \bbL(\delta)}$ forms a nontrivial and dense set in $[-1, 1]$, therefore giving rise to countably many  approximate $\iiv$--dependent edge states, seeded by the point eigenvalues of a corresponding countable family of effective Dirac operators.

    We conclude this section by noting that the procedure in Section \ref{sec:multi_scale_procedure} is not a standard multiple scale expansion in that there is a $\delta$-dependence of the individual terms through $\gammatilde^\delta_\iiv$. Yet it does respect ``scale separation''. Specifically, since $|\gamma_\iiv|\le\delta$ ($\iiv\in\bbL(\delta)$) the expression \eqref{eq:quasi-per-edge-states} is a slow modulation of $\Phi^{\Kv_\star} (\xv)$:
    \begin{equation}
        \Psi^{\delta, (0)}_{\iiv, j, \kpar} (\xv) = \Phi^{\Kv_\star} (\xv)^\transp\, \underbrace{\alpha^{\Kv_\star}_j \big(\delta \kvh_2 \cdot \xv\,;\, \mu + \gammatilde^\delta_\iiv \big)\times \exp \Big[\vts\icplx\vts \big(\vts (\mu\,  + \gammatilde^\delta_\iiv)\, \delta \kvh_1 \cdot \xv \vts\big)\vts\Big]}_{U \scalebox{0.82}{$\big($}\vts \delta \kvh_1 \cdot \xv,\, \delta \kvh_2 \cdot \xv\,;\, \gammatilde^\delta_\iiv \vts\scalebox{0.82}{$\big)$}} \, ,\label{eq:quasi-per-edge-states-alt}
    \end{equation}
    where $\gamma_\iiv := \delta\, \gammatilde^\delta_\iiv$ and $|\gammatilde^\delta_\iiv|\le 1$. 

    \section{Rigorous formulation and main theorem}
    \noindent
    \bypierre{In this section, we reformulate the augmented eigenvalue problem \eqref{eq:evp_irrational_edges} using an appropriate framework involving Sobolev spaces. We then state and comment on the main theorem of this article, Theorem \ref{thm:resolvent_expansion}.}

    \subsection{Formulation of the augmented eigenvalue problem in a Sobolev space}\label{sec:formulation_augmented_EVP_Sobolev_space}
    Due to their periodicity properties \eqref{eq:augmented_potential_property} along the lattice $\Z\vts \av_1 + \Z\vts \av_2$ defined by \eqref{eq:3D_lattice_vectors}, the coefficients $A^\delta_\AUG$, $V_\AUG$, defined on $\R^3$, are completely determined by their restriction to the cylinder
    \begin{equation*}
        \displaystyle
        \cylaug := \R^3 / (\Z\vts \av_1 + \Z\vts \av_2).
    \end{equation*}
    Moreover, since $\av_\perp$ is transverse to $\Gamma$, \eqref{eq:interface_Gamma}, the cylinder $\cylaug$ may be identified with $\Omega \times \R$ in the following trivial manner:
    \begin{align}
        \cylaug &\phantom{:}\equiv \big\{ \tau_1\, \av_1 + \tau_2\, \av_2 + s\, \av_\perp \ \ \big|\ \ \tau_1, \tau_2 \in (0, 1), \ s \in \R \big\} \nonumber
        \retss
        &:= \big\{ (\tau_1\, \vvv_1 + \tau_2\, \vvv_2,\, r\vts \tau_1 - \tau_2 + s) \ \ \big|\ \ \tau_1, \tau_2 \in (0, 1), \ s \in \R \big\} = \Omega \times \R. \label{eq:identification_cylaug}
    \end{align}
    Given $\kpar \in \R$, we shall work with the space of \emph{augmented $\kpar$--pseudo-periodic} functions:
    \begin{equation}\label{eq:def_L2aug}
        L^2_{\kpar} (\cylaug) = \Big\{ F \in L^2_{\LOC} (\R^3)\ \ \Big|\ \ \euler^{-\icplx\vts \kpar (\kvh_1 \cdot \xv)}\; F \in L^2(\cylaug)\Big\}.
    \end{equation}
    In particular, $L^2_{\kpar = 0} (\cylaug) = L^2(\cylaug)$. Note that $L^2_{\kpar} (\cylaug)$ encodes the boundary conditions in \eqref{eq:evp_irrational_edges}. The identification $\cylaug \equiv \Omega \times \R$ justified in \eqref{eq:identification_cylaug} allows us to equip $L^2_{\kpar} (\cylaug)$ with the $L^2(\Omega \times \R)$--inner product and the associated norm. 

    Similarly, for $q \in \N_0$, we set  $H^{q = 0}_{\kpar, \xv} (\cylaug) := L^2_{\kpar} (\cylaug)$, and inductively define:
    \begin{equation}\label{eq:def_Hqaug_edge} 
        \displaystyle
        \spforall q \in \N, \quad H^q_{\kpar, \xv} (\cylaug) := \Big\{ F \in H^{q-1}_{\kpar, \xv} (\cylaug)\ \ \Big|\ \ \nabla_\xv\, F \in [\vts H^{q-1}_{\kpar, \xv} (\cylaug)\vts]^2 \Big\}.
    \end{equation}
    We equip $H^q_{\kpar, \xv} (\cylaug)$ with the $L^2(\R_s; H^q (\Omega))$--inner product and the associated norm. The subscript ``$\xv$'' in $\smash{H^q_{\kpar, \xv} (\cylaug)}$ indicates regularity with respect to the variable $\xv$ alone, which makes it an anisotropic space. Elements of $\smash{H^q_{\kpar, \xv} (\cylaug)}$, $q \geq 1$, can also be characterized via trace (restriction) operators \cite{evans2022partial,lions2012non,adams2003sobolev}, although their definition can be very tricky due to the anisotropy \cite{joly1992some}.

    Finally, we introduce the operator
    \begin{equation}\label{eq:augmented_dw_operator}
        \calH^\delta_{\AUG, \kpar}  := -\nabla_\xv \cdot A^\delta_\AUG (\xv, s)\, \nabla_\xv + V_\AUG (\xv, s), \qquad \DOM(\calH^\delta_{\AUG, \kpar}) := H^2_{\kpar, \xv} (\cylaug).
    \end{equation}
    Then, the eigenvalue problem \eqref{eq:evp_irrational_edges} can be reformulated as:
    \begin{equation}\label{eq:evp_irrational_edges_bis}
        \displaystyle
        (\vts \calH^\delta_{\AUG, \kpar} - E \vts)\, \Psi_\AUG = 0, \qquad \Psi_\AUG \in \DOM(\calH^\delta_{\AUG, \kpar}).
    \end{equation}
    In other words, augmented edge states are eigenfunctions of $\calH^\delta_{\AUG, \kpar}$.

    While not essential for the remainder of this work, the next result is intended to offer more insight into the spectral properties of $\calH^\delta_{\AUG, \kpar}$. This is a direct adaptation of results on almost periodic \cite{simon1982almost} and random \cite{pastur1980spectral} Schrödinger operators.

    \begin{proposition}\label{prop:relation_augmented_operator}~
        \begin{enumerate}[label={$(\alph*).$}, ref={\theremark.$(\alph*)$}, wide=0pt]
            \item We have:
            \begin{equation}\label{eq:relation_augmented_operator}
            \spforall \kpar \in \R, \quad \calH^\delta_{\AUG, \kpar + 2\pi} = \calH^\delta_{\AUG, \kpar} \textAND \calH^\delta_{\AUG, \kpar + 2\pi r} = \euler^{\vts \icplx\vts 2\pi\vts s}\, \calH^\delta_{\AUG, \kpar}\, \euler^{- \icplx\vts 2\pi\vts s},
            \end{equation}
            where $\euler^{\pm \icplx\vts 2\pi\vts s}$ is to be understood as a multiplication operator. Therefore: 
            \item If $r = b_1 / a_1$ with coprime $(b_1, a_1) \in \Z \times \N$, then $\spec (\calH^\delta_{\AUG, \kpar + 2 \pi / a_1}) = \spec (\calH^\delta_{\AUG, \kpar})$.
            \item If $r$ is irrational, then the spectrum of $\calH^\delta_{\AUG, \kpar}$ is independent of $\kpar$.
        \end{enumerate}
    \end{proposition}

    \begin{dem}
        $(a)$. \textit{Equivalence relations}. The operator ${\calH^\delta_{\AUG, \kpar}}$ depends on $\kpar$ only through its domain. Therefore, \eqref{eq:relation_augmented_operator} follows from the relation $\smash{L^2_{\kpar + 2\pi} (\cylaug) = L^2_{\kpar} (\cylaug)}$ and from the fact that $\smash{F \in L^2_{\kpar} (\cylaug)}$ if and only if $\smash{\euler^{\vts \icplx\vts 2\pi\vts s}\, F \in L^2_{\kpar + 2\pi r} (\cylaug)}$. 

        \vspace{1\baselineskip} \noindent
        $(b)$. \textit{Spectrum for rational edges}. Assume that $r = b_1 / a_1$, with coprime $(a_1, b_1) \in \N \times \Z$. Since there exists a pair of coprime integers $(a_2, b_2) \in \Z^2$ such that $a_1\vts b_2 - a_2\vts b_1 = 1$, we get the unitary equivalence relation
        \begin{equation*}
            \calH^\delta_{\AUG, \kpar + 2\pi / a_1} = \calH^\delta_{\AUG, \kpar - 2\pi\vts a_2\vts r + 2\pi b_2} = \euler^{- \icplx\vts 2\pi\vts a_2\vts s}\, \calH^\delta_{\AUG, \kpar}\, \euler^{\vts \icplx\vts 2\pi\vts a_2\vts s} \quad \textnormal{by \eqref{eq:relation_augmented_operator},}
        \end{equation*}
        which implies that $\spec (\calH^\delta_{\AUG, \kpar + 2\pi / a_1}) = \spec (\calH^\delta_{\AUG, \kpar})$.

        \vspace{1\baselineskip} \noindent
        $(c)$. \textit{Spectrum for irrational edges}. In order to prove that $\spec (\calH^\delta_{\AUG, \kpar})$ is independent of $\kpar$ for irrational edges, it is convenient to work with the operator acting on $L^2(\cylaug)$:
        \begin{equation*}
            \begin{array}{r@{\ }l}
            \calH^\delta_{\AUG} (\kpar) &:= \displaystyle \euler^{-\icplx\vts \kpar\vts (\kvh_1 \cdot \xv)}\, \calH^\delta_{\AUG, \kpar}\, \euler^{\vts \icplx\vts \kpar\vts (\kvh_1 \cdot \xv)} 
            \rets
            &\phantom{:}= -(\vts \nabla_\xv + \icplx\, \kpar\vts \kvh_1 \vts) \cdot A^\delta_\AUG (\xv, s)\, (\vts \nabla_\xv + \icplx\, \kpar\vts \kvh_1 \vts) + V_\AUG (\xv, s).
            \end{array}
        \end{equation*}
        It follows from \eqref{eq:relation_augmented_operator} that for any $(m, n) \in \Z^2$,
        \begin{equation*}
            \calH^\delta_\AUG (\kpar + 2\pi (m + n r)) = \euler^{- \icplx\vts 2\pi (m + n r)\, \kvh_1 \cdot \xv}\, \euler^{\vts \icplx\vts 2\pi n\vts s}\, \calH^\delta_\AUG (\kpar)\, \euler^{- \icplx\vts 2\pi n\vts s}\, \euler^{\vts\icplx\vts 2\pi (m + n r)\, \kvh_1 \cdot \xv}.
        \end{equation*}
        As a consequence,
        \begin{equation}\label{eq:relation_augmented_operator_fixed_domain}
            \spforall (m, n) \in \Z^2, \quad \spec \big(\vts \calH^\delta_\AUG (\kpar + 2\pi (m + n r)) \vts\big) = \spec \big(\vts \calH^\delta_\AUG (\kpar) \vts\big).
        \end{equation}
        Since $\calH^\delta_{\AUG} (\kpar)$ varies polynomially in $\kpar$, the map $R: \kpar \mapsto (\vts \calH^\delta_{\AUG} (\kpar) - \icplx \vts)^{-1}$ is continuous from $\R$ to $\scrL(L^2(\cylaug))$. By classical perturbation arguments \cite[Theorem V.4.10]{kato1995perturbation}, it follows that $\kpar \in \R \mapsto \spec (R(\kpar))$ is continuous with respect to the Hausdorff metric. Moreover, by \eqref{eq:relation_augmented_operator_fixed_domain}, $\kpar \mapsto \spec (R(\kpar))$ is constant on the set of periods $2\pi r \Z + 2\pi \Z$. Since $2\pi r \Z + 2\pi \Z$ is dense in $\R$ by Kronecker's theorem \cite[Theorem 442]{hardy1979introduction}, it follows that $\kpar \mapsto \spec (R(\kpar))$ is constant on $\R$. Thus, $\spec (\calH^\delta_{\AUG} (\kpar))$, or equivalently $\spec (\calH^\delta_{\AUG, \kpar})$, is independent of $\kpar$.
    \end{dem}

    \subsection{Statement of the main result, Theorem \texorpdfstring{\ref{thm:resolvent_expansion}}{7.1} on the resolvent expansion}\label{sec:main_results}
    \noindent
    In this section, we state and discuss the main result of this paper, Theorem \ref{thm:resolvent_expansion}, on the asymptotic expansion of the resolvent of $\calH^\delta_{\AUG, \Kv \cdot \vvh_1 + \delta \mu}$ as $\delta \to 0$. Motivated by the multiscale study in Section \ref{sec:ansatz_and_multi_scale_analysis} we begin by introducing the operators involved in the formulation of Theorem \ref{thm:resolvent_expansion}. 

    For any $\iiv = (\Kv_\iiv, m_\iiv) \in \bbL = \{\Kv, \Kv'\} \times \Z$ and $\Phi^{\Kv_\iiv} = (\Phi^{\Kv_\iiv}_1, \Phi^{\Kv_\iiv}_2)^\transp$, using the pair $(\gamma_\iiv, \lvar_\iiv)$ defined by \eqref{eq:gamma_Iandlambda_I}, we introduce
    \begin{equation}\label{eq:mode_for_averaging_operator}
        \varphi_\iiv (\xv, s) := \exp \Big[\vts\icplx\vts \big(\vts \hyperref[eq:gamma_Iandlambda_I]{\gamma_\iiv}\, \kvh_1 \cdot \xv + (\hyperref[eq:gamma_Iandlambda_I]{\lvar_\iiv} + \Kv \cdot \vvv_2)\vts s \vts\big)\vts\Big]\, \Phi^{\Kv_\iiv} (\xv), \quad (\xv, s) \in \R^2 \times \R.
    \end{equation}
    Moreover, define the bounded linear operator $F(\xv, s) \mapsto  (\calT_\iiv\, F)(\zeta)$, which averages over $\Omega$:
    \begin{gather}
        \calT_\iiv : L^2_{\Kv \cdot \vvh_1} (\cylaug) \to L^2 (\R; \C^2), \nonumber
        \\
        (\calT_\iiv\, F) (\zeta) := 
        \int_\Omega \vts\overline{\varphi_\iiv(\xv,\zeta - \kvh_2 \cdot \xv)}\, F\vts\vts (\xv, \zeta - \kvh_2 \cdot \xv)\, d\xv, \quad \zeta \in \R, \label{eq:averaging_operator}
    \end{gather}
    and its bounded linear adjoint $g(\zeta)\mapsto (\calT^*_\iiv\, g) (\xv, s)$:
    \begin{gather}
    \calT^*_\iiv: L^2 (\R; \C^2)\to L^2_{\Kv \cdot \vvh_1} (\cylaug), %
    \\
    (\calT^*_\iiv\, g) (\xv, s) := \varphi_\iiv (\xv, s)^\transp\, g \big(\kvh_2 \cdot (\xv + s \vvv_2)\big), \quad (\xv, s) \in \R^2 \times \R.\label{eq:averaging_operator_adjoint}
    \end{gather}
    Further, define the unitary scaling operators on $L^2(\R)$ given by
    \begin{equation}
        \label{eq:dilation_operators}
        \displaystyle
        \calU_\delta\, g (\zeta) := \delta^{-1/2}\, g(\delta^{-1}\, \zeta) \textAND \calU^*_\delta\, g (\zeta) = \delta^{1/2}\, g(\delta\, \zeta), \quad \zeta \in \R.
    \end{equation}
    and introduce the composition:
    \begin{equation}\label{eq:def_calJ}
        \calJ_{\delta, \iiv} := \calU_\delta\, \calT_\iiv \in \scrL(L^2_{\Kv \cdot \vvh_1} (\cylaug), L^2 (\R; \C^2)).
    \end{equation}
    The results of Section \ref{sec:recap_multiscale} may be summarized as follows: for any given $\iiv \in \bbL (\delta)$, an eigenpair $(z_j (\muhat),\, \alpha^{\Kv_\iiv}_j (\cdot\,;\, \muhat)) \in \R \times L^2(\R; \C^2)$ of the effective Dirac operator $\calD^{\Kv_\iiv} (\muhat)$ gives rise, for any $\mu$ in compact subsets of $\R$, to an $\calO(\delta)$--approximate eigenvalue $E_D + \delta\vts z_j(\mu + \delta^{-1} \gamma_\iiv)$ of $\calH^\delta_{\AUG, \kpar = \Kv \cdot \vvh_1 + \delta \mu}$, with an associated $\calO(\delta)$--approximate eigenfunction, defined by \eqref{eq:approximate_augmented_eigenfunction_1}, expressible as
    \begin{equation}\label{eq:approximate_augmented_eigenfunction_2}
        \Psi^{\delta, (0)}_{\AUG, \iiv, j} = \delta^{-1/2}\, \euler^{\vts\icplx\vts \delta \mu \kvh_1 \cdot \xv}\, \calJ_{\delta, \iiv}^*\, \alpha^{\Kv_\iiv}_j(\cdot\, ; \mu + \delta^{-1} \gamma_\iiv).
    \end{equation}
    This expression underlies the role played by $\calJ_{\delta, \iiv}$ in our main Theorem \ref{thm:resolvent_expansion}, below.

    By arranging the families $\smash{(\calT_\iiv)_{\iiv \in \bbL}}$ and $\smash{(\calJ_{\delta, \iiv})_{\iiv \in \bbL}}$ into columns, we obtain
    \begin{equation}\label{eq:concatenated_averaging_operator}
        \calT := \begin{pmatrix} \vdots \\ \calT_\iiv \\ \vdots \end{pmatrix}, \quad \calJ_\delta := \calU_\delta\, \calT = \begin{pmatrix} \vdots \\ \calJ_{\delta, \iiv} \\ \vdots \end{pmatrix}: L^2_{\Kv \cdot \vvh_1} (\cylaug) \to [ L^2(\R; \C^2) ]^{\bbL}.
    \end{equation}
    Proposition \ref{prop:pties_averaging_operator} presents some useful properties for $\calT_\iiv$, $\calJ_{\delta, \iiv}$, $\calT$, $\calJ_\delta$. In particular, we prove that the column operators $\calT, \calJ_\delta$ are bounded:
    \begin{equation}\label{eq:T_J_delta_are_bounded}
        \spforall q \in \N_0, \quad \calT, \calJ_\delta \in \scrL(H^q_{\Kv \cdot \vvh_1, \xv} (\cylaug),\, \ell^2_q (\bbL; L^2(\R; \C^2))).
    \end{equation}
    Here, $H^q_{\Kv \cdot \vvh_1, \xv} (\cylaug)$ is defined by \eqref{eq:def_Hqaug_edge}, and for any Banach space $\scrX$, $\ell^2_q (\bbL; \scrX)$ is defined by
    \begin{equation}\label{eq:def_ell2_L_X}
        \ell^2_q (\bbL; \scrX) := \bigg\{ (g_\iiv)_{\iiv \in \bbL} \in \scrX^{\bbL}\ \ \Big|\ \ \sum_{\iiv \in \bbL} (1 + m^2_\iiv)^q\, \|g_\iiv\|^2_{\scrX} < \infty \bigg\},
    \end{equation}
    and is equipped with the norm arising from \eqref{eq:def_ell2_L_X}.
    We use the following abbreviated notations:  $\ell^2 (\bbL; \scrX) = \ell^2_{q = 0} (\bbL; \scrX)$, $\ell^2_q (\bbL) = \ell^2_q (\bbL; \C)$, and $\ell^2 (\bbL) = \ell^2_{q = 0} (\bbL; \C)$. 

    In preparation for our resolvent expansion, we introduce the block-diagonal Dirac operator defined for $(\mu, \delta) \in \R \times \R^*_+$ by
    \begin{multline}\label{eq:block_diagonal_dirac_operator}
        \calD^{\vts \delta} (\mu) = %
        \begin{pmatrix}
            \ddots & & (0)
            \\
            & \calD^{\Kv_\iiv} (\mu + \delta^{-1} \gamma_\iiv) & 
            \\
            (0) & & \ddots
        \end{pmatrix}_{\iiv\in\bbL}
        : \ell^2(\bbL; H^1(\R; \C^2)) \to \ell^2 (\bbL; L^2(\R; \C^2)), \quad \textnormal{that is},
        \rets
        \spforall \bm{g} = (g_\iiv)_{\iiv\in\bbL} \in \ell^2(\bbL; H^1(\R; \C^2)), \quad \big( \calD^{\vts \delta} (\mu)\, \bm{g} \big)_\iiv := \calD^{\Kv_\iiv} (\mu + \delta^{-1} \gamma_\iiv)\, g_\iiv, \quad \spforall \iiv \in \bbL.
    \end{multline}
    The operators $\calD^{\Kv_\iiv} (\mu + \delta^{-1} \gamma_\iiv)$ are the effective Dirac operators, arising in our formal multiple scale expansion; see \eqref{eq:effective_Dirac_operator}.

    For any $\bbL' \subseteq \bbL$, we introduce a block-diagonal projection operator, denoted $\mathds{1}_{\bbL'}$. For every $\bm{g} = (g_\iiv)_{\iiv\in\bbL} \in \ell^2(\bbL; L^2(\R; \C^2))$, we define $\mathds{1}_{\bbL'}\; \bm{g}\in \ell^2(\bbL; L^2(\R; \C^2))$ by 
    \begin{equation}
        \big( \mathds{1}_{\bbL'}\; \bm{g} \big)_\iiv := 
        \left\{
            \begin{array}{l@{\quad}l}
            g_\iiv, & \spforall \iiv \in \bbL',
            \rets
            0, & \spforall \iiv \not\in \bbL'.\label{eq:block_diagonal_restriction_operator}
            \end{array}
        \right.   
    \end{equation}
    Clearly, $\mathds{1}_{\bbL'}$  is bounded on $\ell^2(\bbL; L^2(\R; \C^2))$.

    Our main result states that after an energy-centering about $E=E_D$ and $\delta$--dependent rescaling of $\calH^\delta_{\AUG, \kpar}$, $\kpar = \Kv \cdot \vvh_1 + \delta \mu$, so that the bulk energy gap is of order one, the resolvent  is well-approximated, for $\delta \to 0$, by the resolvent of the  block-diagonal effective Dirac operator $\calD^{\vts \delta} (\mu)$.

    \begin{theorem}[Resolvent expansion for general edges]\label{thm:resolvent_expansion}
        Let $\vvh_1 = \vvv_1 + r \vvv_2$, where $r \in \R$. Suppose that the Assumptions \ref{hyp:Dirac_point}, \ref{hyp:omnidirectional_no_fold}, \ref{hyp:non_degeneracy} hold. Fix $\eta, \mu_0 > 0$, and a compact set $S \subset \C$. There exists $\delta_0 > 0$ such that if
        \begin{equation*}
            \delta \in (0, \delta_0), \quad \mu \in [-\mu_0, \mu_0], \quad \kpar = \Kv \cdot \vvh_1 + \delta \mu, \quad z \in S, \quad \dist \big(z,\; \spec \calD^{\vts \delta} (\mu)\big) > \eta,
        \end{equation*}
        then the operator $\calH^\delta_{\AUG, \kpar} - E_D - \delta\, z$ is invertible from $H^2_{\kpar, \xv} (\cylaug)$ to $L^2_{\kpar} (\cylaug)$, and we have the expansion:
        \begin{multline}\label{eq:resolvent_expansion_Haug}
            \bigg(\vts \frac{\calH^\delta_{\AUG, \kpar} - E_D}{\delta} - z \vts\bigg)^{-1} = \euler^{\vts\icplx\vts \delta \mu \kvh_1 \cdot \xv}\, \calJ^*_\delta\; \mathds{1}_{\bbL (\delta^{3/4})}\; \big(\vts \calD^{\vts \delta} (\mu) - z \vts\big)^{-1}\; \mathds{1}_{\bbL (\delta^{3/4})}\; \calJ^{}_\delta\, \euler^{-\icplx\vts \delta \mu \kvh_1 \cdot \xv} %
            \\
            + \frac{1}{\eta}\, \calO (\delta^{1/4}), \quad \textnormal{in}\ \ \scrL(L^2_{\kparsubscript} (\cylaug)),
        \end{multline}
        \bypierre{where $\bbL (\delta^{3/4}) \subset \bbL$ is the set introduced in \eqref{eq:def_L_veps}.}
    \end{theorem}

    \begin{remark}
        By simply replacing $\Kv$ with $\Kv'$ in the various objects involved in the statement of Theorem \ref{thm:resolvent_expansion}, a similar resolvent expansion can be established for parallel quasi-momenta located in a neighborhood of $\Kv' \cdot \vvh_1$.
    \end{remark}

    \vspace{1\baselineskip} \noindent
    Theorem \ref{thm:resolvent_expansion} is a result on the centered and $\delta$--scaled operator $\delta^{-1}\, (\calH^\delta_{\AUG, \Kv \cdot \vvh_1 + \delta \mu} - E_D)$. The resolvent of this operator is shown in \eqref{eq:resolvent_expansion_Haug} to be similar, at first-order, to the resolvent of the block-diagonal Dirac operator $\calD^{\vts \delta} (\mu)$ defined by \eqref{eq:block_diagonal_dirac_operator}. For rational edges, a similar expansion was derived in \cite{drouot2020edge} for the two-dimensional Schrödinger operator $-\nabla \cdot A^\delta (\xv)\, \nabla + V (\xv)$ acting on an appropriate space. Here, working with the augmented Schrödinger operator enables us to study irrational edges. In Section \ref{sec:resolvent_expansion_rational_edges}, we discuss this expansion for rational edges. 

    Given its definition \eqref{eq:resolvent_expansion_Haug}, the spectral properties of $\calD^{\vts \delta} (\mu)$ can be derived from those of the effective Dirac operators $\calD^{\Kv_\star} (\muhat)$ defined in \eqref{eq:effective_Dirac_operator} for $\Kv_\star \in \{\Kv, \Kv'\}$, $\muhat \in \R$. For instance,
    \begin{equation*}
        \spec \calD^{\vts \delta} (\mu) = \overline{\bigcup_{\iiv \in \bbL}\, \spec \calD^{\Kv_\iiv} (\mu + \delta^{-1} \gamma_\iiv)}.
    \end{equation*}
    As stated in Proposition \ref{prop:pties_Dirac_operator} and shown in Figure \ref{fig:spectrum_Dirac_operator}, the essential spectrum of each effective operator $\calD^{\Kv_\star} (\muhat)$ is purely absolutely continuous, and is given by the complement in $\R$ of some interval $(-\theta_\GAP (\muhat), \theta_\GAP (\muhat))$ containing finitely many simple and isolated eigenvalues. If $r$ is rational, then $\gamma_\iiv$ vanishes for $\delta$ small enough and $\iiv \in \bbL(\delta^{3/4})$ (see Proposition \ref{prop:veps_neighborhood_K_rational}), so that $\calD^{\Kv_\iiv} (\mu + \delta^{-1} \gamma_\iiv) = \calD^{\Kv_\iiv} (\mu)$. In that case, the spectrum of $\spec \calD^{\vts \delta} (\mu)$ contains \emph{finitely} many eigenvalues of \emph{infinite} multiplicity. However, if $r$ is irrational, then the family $(\gamma_\iiv)_{\iiv \in \bbL}$ is dense in $[-\pi, \pi]$. As a consequence, the spectrum of $\calD^{\vts \delta} (\mu)$ is $\R$, and contains a \emph{dense} set of eigenvalues of \emph{finite} multiplicity, a subset of which is dense in the bulk gap $(-\theta_\GAP, \theta_\GAP)$, and a subset of which is embedded in its absolutely continuous part $\R \setminus (-\theta_\GAP, \theta_\GAP)$. This is illustrated in Figure \ref{fig:spectrum_block_diagonal_operator}. In particular, for $r \in \R \setminus \Q$, the assumption on $z$ in Theorem \ref{thm:resolvent_expansion} reduces to $|\Imag z| > \eta$.

    \begin{figure}[ht!]
        \makebox[\textwidth][c]{
            \includegraphics[page=6]{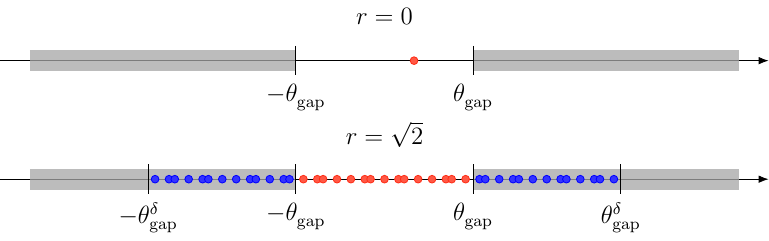}
        }
        \caption{The spectrum of the block-diagonal operator $\calD^{\vts \delta} (\mu)$, obtained by sampling the spectrum of $\calD^{\Kv_\iiv} (\muhat)$ at $\muhat = \mu + \delta^{-1} \gamma_\iiv$ for $\iiv \in \bbL$. The red dots are the eigenvalues embedded in the gap $(-\theta_\GAP, \theta_\GAP)$. The blue dots represent the eigenvalues embedded in the absolutely continuous spectrum of $\calD^{\vts \delta} (\mu)$, and extend to $\theta^\delta_\GAP := \theta_\GAP (\mu + \delta^{-1}\vts \pi)$. For rational edges, $\calD^{\vts \delta} (\mu)$ has a finite number of eigenvalues of \emph{infinite} multiplicity (top panel), whereas for irrational edges, $\calD^{\vts \delta} (\mu)$ has a \emph{dense} set of simple eigenvalues (bottom panel).\label{fig:spectrum_block_diagonal_operator}}
    \end{figure}

    \begin{remark}\label{rmk:technical_exponent}
        For technical reasons, the quantity $\delta^{3/4}$ appearing in the resolvent expansion cannot be replaced by $\delta$. More generally, one can replace $\delta^{3/4}$ by $\delta^\nu$, which changes the order of the error term in \eqref{eq:resolvent_expansion_Haug}. The exponent $\nu$ has to be taken in a subinterval of $(0, 1)$, which is determined in the proof of Theorem \ref{thm:resolvent_expansion}. 
        For rational edges, $\nu \in (1/2, 1)$, and $\nu = 2/3$, as chosen in \cite{drouot2020edge}, yields an $\calO(\delta^{1/3})$--error term. For irrational edges, $\nu \in (2/3, 1)$, and based on our calculations, the optimal (though likely non-sharp) error order is $\calO(\delta^{1/4})$, achieved when $\nu = 3/4$. The difference in outcomes is due to the existence of constants $C_0, C_1 > 0$ such that
        \begin{equation*}
            \spforall \iiv \in \bbL (\delta^\nu), \ \ 
            \| \calD^{\Kv_\iiv} (\mu + \delta^{-1} \gamma_\iiv) \|_{\scrL(H^1(\R), L^2(\R))} \underset{\delta \to 0}{\sim} C_0 + C_1\, |\delta^{-1} \gamma_\iiv|.
        \end{equation*}
        For rational edges, $\gamma_\iiv$ vanishes for a small enough $\delta$ and for $\iiv \in \bbL (\delta^\nu)$ (Proposition \ref{prop:veps_neighborhood_K_rational}), so that the $H^1(\R) \to L^2(\R)$ norm of the Dirac operator is bounded in $\delta$. However, if $r$ is irrational, then since $(\delta^{-1} \gamma_\iiv)_{\iiv \in \bbL(\delta^\nu)}$ is dense in $[-\delta^{\nu-1}, \delta^{\nu-1}]$, and $\nu < 1$, the above norm is allowed to tend to infinity as $\delta \to 0$. The treatment of this unboundedness property is what leads to an $\calO(\delta^{1/4})$--error, which is larger than the $\calO(\delta^{1/3})$--error of  the rational case.
    \end{remark}

    \subsection{Strategy for the proof of Theorem \ref{thm:resolvent_expansion}}\label{sec:strategy_proof}
    For convenience, we conjugate $\calH^\delta_{\AUG, \Kv \cdot \vvh_1 + \delta \mu}$ to an operator $\calH^\delta_{\AUG, \Kv \cdot \vvh_1} (\mu)$ which acts, for each  $\mu \in \R$, on the fixed space $L^2_{\Kv \cdot \vvh_1} (\cylaug)$:
    \begin{align}
        \calH^\delta_{\AUG, \Kv \cdot \vvh_1} (\mu) &:= \euler^{-\icplx\vts \delta \mu \kvh_1 \cdot \xv}\, \calH^\delta_{\AUG, \Kv \cdot \vvh_1 + \delta \mu}\, \euler^{\vts\icplx\vts \delta \mu \kvh_1 \cdot \xv} \nonumber
        \retss
        &= -\big(\vts \nabla_\xv + \icplx \delta \mu \kvh_1 \vts\big) \cdot A^\delta_\AUG (\xv, s)\, \big(\vts \nabla_\xv + \icplx \delta \mu \kvh_1 \vts\big) + V_\AUG (\xv, s). \label{eq:conjugated_augmented_operator}
    \end{align}
    The proof of Theorem \ref{thm:resolvent_expansion} uses the overall strategy in \cite{drouot2019characterization,drouot2020edge}:
    \begin{enumerate}[label={($\arabic*$). }]
        \bypierre{\item {\it Spectral properties of the unperturbed operator $\calH^0_{\AUG, \Kv \cdot \vvh_1}$, Section \ref{sec:spectrum_unperturbed_augmented_operator}:}}
        \\
        In Section \ref{sec:Fourier_analysis_L2aug}, we study the spectral properties of the unperturbed augmented operator $\calH^0_{\AUG, \Kv \cdot \vvh_1} (\mu) = \calH^0_{\AUG, \Kv \cdot \vvh_1} = -\Delta_\xv + V(\xv)$ acting on $L^2_{\Kv \cdot \vvh_1} (\cylaug)$. Using the partial Fourier transform in $s$ and Floquet-Bloch theory, we establish in Proposition \ref{prop:completeness_slices} a generalized eigenfunction expansion for elements of $L^2_{\Kv \cdot \vvh_1} (\cylaug)$. The generalized eigenfunctions appearing in this expansion are constructed from Floquet-Bloch modes whose quasi-momenta are located on the line $\Kv + \R \kvh_2$. This motivates the study in Section \ref{sec:edge_slices} of the line $\Kv + \R\, \kvh_2$ projected to the torus $\R^2 / \Lambda^*$. For irrational edges, this projection is dense. Section \ref{sec:veps_neighborhood} then characterizes the values of $\lvar \in \R$ for which $\Kv + \lvar \kvh_2$ lies in a neighborhood of the lattice of all high-symmetry quasi-momenta: $\bbK = (\Kv + \Lambda^*) \cup (\Kv' + \Lambda^*)$. This characterization is less trivial for irrational edges, and relies on the Wrapping Lemma \ref{lem:wrapping_lemma}.
        \item {\it Projections on the modes near and away from the Dirac point, Section \ref{sec:near_far_quasimomentum_energy_components}}:
        \\
        The multiple scale analysis of Section \ref{sec:ansatz_and_multi_scale_analysis} produces wave-packet type states, which are spectrally concentrated in a neighborhood of the Dirac point of $\calH^0=-\Delta+V$. This motivates, in Section \ref{sec:near_far_quasimomentum_energy_components}, our introducing projectors:  $\Pi_\NEAR$ and $\Pi_\FAR = \Id - \Pi_\NEAR$, to spectrally localize elements of $L^2_{\Kv \cdot \vvh_1} (\cylaug)$ near and  away from the Dirac point. The construction of $\Pi_\NEAR$ relies heavily on the step $(1)$ above. 
        \item {\it Schur complement / Lyapunov-Schmidt reduction scheme, Section \ref{sec:schur_complement}}:
        \\
        Writing $L^2_{\Kv \cdot \vvh_1} (\cylaug) = \scrX_\NEAR \oplus \scrX_\FAR$ where $\scrX_\NEAR := \Ran (\Pi_\NEAR)$ and $\scrX_\FAR := \Ran (\Pi_\FAR)$, the operator $\calH^\delta_{\AUG, \Kv \cdot \vvh_1} (\mu) - E_D - \delta z$ can be expressed in block form:
        \begin{equation*}
            \displaystyle
            \calH^\delta_{\AUG, \Kv \cdot \vvh_1} (\mu) - E_D - \delta z = %
            \begin{pmatrix}
            \bfA & \bfB
            \\
            \bfC & \bfD
            \end{pmatrix},
        \end{equation*}
        where $\bfA: \scrX_\NEAR \to \scrX_\NEAR$, $\bfB: \scrX_\FAR \to \scrX_\NEAR$, $\bfC: \scrX_\NEAR \to \scrX_\FAR$, and $\bfD: \scrX_\FAR \to \scrX_\FAR$. In Section \ref{sec:schur_complement}, we prove that $\bfD = \Pi_\FAR\, (\vts \calH^\delta_{\AUG, \Kv \cdot \vvh_1} (\mu) - E_D - \delta z \vts)\, \Pi_\FAR$ is invertible under the omnidirectional no-fold condition stated in Assumption \ref{hyp:omnidirectional_no_fold}. Further, we prove that the off-diagonal operators $\bfB$ and $\bfC$ have norms that tend to zero, as $\delta \to 0$, and hence the invertibility of  $\calH^\delta_{\AUG, \Kv \cdot \vvh_1} (\mu) - E_D - \delta z$ on $L^2_{\Kv \cdot \vvh_1} (\cylaug)$ reduces to that of $\bfA = \Pi_\NEAR\, (\vts \calH^\delta_{\AUG, \Kv \cdot \vvh_1} (\mu) - E_D - \delta z \vts)\, \Pi_\NEAR$ on $\scrX_\NEAR$.
        \item {\it Asymptotic expansion of the resolvent for $\delta \to 0$, Section \ref{sec:asymptotic_expansions}}:
        \\
        We expand the inverse  of $\bfA$. In Section \ref{sec:spectral_localization_tools}, we prove some useful properties of the operators $\calT_\iiv$, $\calJ_{\delta, \iiv}$, $\calT$, $\calJ_\delta$, and their adjoints. Section \ref{sec:unperturbed_operator_expansion} expands the near projector $\Pi_\NEAR$ in terms of the operators $\calJ_{\delta, \iiv}$ involved in Theorem \ref{thm:resolvent_expansion}; see Proposition \ref{prop:asymptotics_near_projectors}. Then, in Section \ref{sec:proof_asymptotics_HtildeDeltaNEAR}, we show that $\bfA$ is asymptotically governed by a family of band-limited Dirac operators; see Proposition \ref{prop:asymptotics_HdeltaNEAR}. This result is more technical than the rational case, due to the \emph{infinite} family of Dirac operators. At last, Section \ref{sec:proof_resolvent_expansion_HtildeDeltaNEAR} contains the expansion of $(\bfA/\delta)^{-1}$ (Proposition \ref{prop:resolvent_expansion_HtildeDeltaNEAR}) in terms of  $(\calD^\delta (\mu)-z)^{-1}$, which we use to conclude  the proof of Theorem \ref{thm:resolvent_expansion}.
    \end{enumerate}

    \section{Band dispersion slices}\label{sec:spectrum_unperturbed_augmented_operator}
    \noindent
    In what follows, the parallel quasi-momentum is set to $\kpar = \Kv \cdot \vvh_1$. 
    
    \subsection{Spectral representation of \texorpdfstring{$\calH^0_{\AUG, \Kv \cdot \vvh_1}$}{H0aug}}\label{sec:Fourier_analysis_L2aug}
    The goal of this section is to derive a spectral representation of $\calH^0_{\AUG, \Kv \cdot \vvh_1} = -\Delta_\xv + V (\xv)$. Since this operator commutes with translations in $s$, a natural tool for its study is the \emph{partial Fourier transform in $s$}, $\scrF_\Kv$, defined by \eqref{eq:directional_Fourier}. The properties of $\scrF_\Kv$ are presented in Lemma \ref{lem:directional_Fourier}. Combining these properties with Floquet-Bloch theory, we obtain in Proposition \ref{prop:completeness_slices} a generalized eigenfunction expansion for elements of $L^2_{\Kv \cdot \vvh_1} (\cylaug)$. This result is the analogue of \cite[Theorem 4.2]{fefferman2016edge}, which treats rational edges.
    
    It can be seen from \eqref{eq:def_L2aug} that $\smash{F \in L^2_{\Kv \cdot \vvh_1} (\cylaug)}$ if and only if $\smash{\euler^{- \icplx \vts \Kv \cdot (\xv + s \vvv_2)}\; F \in L^2(\cylaug)}$. Assuming that $\smash{\euler^{- \icplx \vts \Kv \cdot (\xv + s \vvv_2)}\; F \in \scrC^\infty_0 (\cylaug)}$, we can then introduce
    \begin{equation}\label{eq:directional_Fourier}
        (\scrF_\Kv F) (\xv, \lvar) = \widetilde{F} (\xv, \lvar) := \frac{1}{\sqrt{2\pi}} \int_{\R} F (\xv, s)\, \euler^{- \icplx\vts (\Kv + \lvar \kvh_2) \cdot s \vvv_2}\, d s, \quad (\xv, \lvar) \in \R^2 \times \R .
    \end{equation}
    From its definition, $\lvar \mapsto \widetilde{F} (\xv, \lvar)$ is the Fourier transform of $s \mapsto F(\xv, s)\, \euler^{- \icplx\vts (\Kv \cdot \vvv_2)\vts s}$, where the term $\euler^{- \icplx\vts (\Kv \cdot \vvv_2)\vts s}$ was added simply for convenience. To study the properties of $\scrF_\Kv$, introduce the space
    \begin{align}\label{eq:L2aug_alt}
        L^2(\R_\lvar; L^2_{\Kv + \lvar \kvh_2}) := \big\{\vts \widetilde{F} \in L^2_{\LOC} (\R^3) \ \ \big|\ \ \widetilde{F}(\cdot\,, \lvar) \in L^2_{\Kv + \lvar \kvh_2}, \ \widetilde{F}(\xv, \cdot) \in L^2(\R)\vts \big\},
        \intertext{which is equipped with the $L^2(\Omega \times \R)$--norm. Similarly, for any $q \in \N_0$, define}
        L^2(\R_\lvar; H^q_{\Kv + \lvar \kvh_2}) := \big\{\vts \widetilde{F} \in L^2_{\LOC} (\R^3) \ \ \big|\ \ \widetilde{F}(\cdot\,, \lvar) \in H^q_{\Kv + \lvar \kvh_2}, \ \widetilde{F}(\xv, \cdot) \in L^2(\R)\vts \big\},\label{eq:Hqaug_alt}
    \end{align}
    equipped with the $L^2(\R_\lvar; H^q (\Omega))$--norm. We emphasize that $L^2(\R_\lvar; H^q_{\Kv + \lvar \kvh_2})$ is anisotropic, as its definition involves derivatives with respect to the variable $\xv$ only. We next show that the transform $\scrF_\Kv$ defines an isomorphism from $H^q_{\kpar = \Kv \cdot \vvh_1, \xv} (\cylaug)$ to $L^2(\R_\lvar; H^q_{\Kv + \lvar \kvh_2})$ for any $q \in \N_0$, and we provide some of its properties.

    \begin{lemma}[Properties of $\scrF_\Kv$]\label{lem:directional_Fourier}~
        \begin{enumerate}[label={$(\alph*).$}, ref={\theproposition.$(\alph*)$}, wide=0pt]
            \item The transform $\scrF_\Kv$, defined in \eqref{eq:directional_Fourier} for smooth functions, extends by density to a continuous mapping from $L^2_{\kpar = \Kv \cdot \vvh_1} (\cylaug)$ to $L^2(\R_\lvar; L^2_{\Kv + \lvar \kvh_2})$.
            \item The adjoint of $\scrF_\Kv$ is the operator $\scrF_\Kv^*: L^2(\R_\lvar; L^2_{\Kv + \lvar \kvh_2}) \to L^2_{\kpar = \Kv \cdot \vvh_1} (\cylaug)$ defined for any $G \in L^2(\R_\lvar; L^2_{\Kv + \lvar \kvh_2})$ as
            \begin{equation}\label{eq:directional_Fourier-adjoint}
            (\scrF_\Kv^*\, G) (\xv, s) := \frac{1}{\sqrt{2\pi}} \int_{\R} G (\xv, \lvar)\, \euler^{\vts \icplx\vts (\Kv + \lvar \kvh_2) \cdot s \vvv_2}\, d \lvar, \quad (\xv, s) \in \R^2 \times \R.
            \end{equation}
            \item  We have the Plancherel-like identity:
            \begin{equation}\label{eq:directional_Plancherel}
            \spforall F, G \in L^2_{\Kv \cdot \vvh_1} (\cylaug), \quad \lla\vts G,\, F\vts \rra_{L^2_{\Kv \cdot \vvh_1} (\cylaug)} = \int_\R \big\langle\vts \scrF_\Kv G(\cdot\,, \lvar),\, \scrF_\Kv F (\cdot\,, \lvar) \vts\big\rangle_{L^2_{\Kv + \lvar \kvh_2}}\! d\lvar.
            \end{equation}
            \item We have the relations:
            \begin{equation}\label{eq:inverse_directional_Fourier}
            \scrF_\Kv^*\, \scrF^{}_\Kv = \Id_{L^2_{\Kv \cdot \vvh_1} (\cylaug)}, \quad \scrF^{}_\Kv\, \scrF_\Kv^* = \Id_{L^2(\R_\lvar; L^2_{\Kv + \lvar \kvh_2})}.
            \end{equation}
            \item The transform $\scrF_\Kv$ commutes with multiplication by functions constant in $s$: given a function $p \in \scrC^0(\R^2 / \Lambda)$ identified with its extension $(\xv, s) \mapsto p(\xv)$, we have 
            \begin{equation}\label{eq:Fourier_commutes_with_x_functions}
            \displaystyle
            \spforall F \in L^2_{\Kv \cdot \vvh_1} (\cylaug), \quad \scrF_\Kv\, (p\, F)\, (\xv, \lvar) = p(\xv)\, \scrF_\Kv F (\xv, \lvar), \quad (\xv, \lvar) \in \R^2 \times \R.
            \end{equation}
            \item The transform $\scrF_\Kv$ commutes with differential operators in $\xv$:
            \begin{equation}\label{eq:Fourier_commutes_with_differential_operators}
            \spforall F \in H^1_{\kpar = \Kv \cdot \vvh_1, \xv} (\cylaug), \quad \nabla_\xv\, \scrF_\Kv\vts F = \scrF_\Kv\vts \nabla_\xv\, F.
            \end{equation}
            In particular, $\scrF_\Kv$ is isomorphic from $H^q_{\kpar = \Kv \cdot \vvh_1, \xv} (\cylaug)$ to $L^2(\R_\lvar; H^q_{\Kv + \lvar \kvh_2})$ for any $q \in \N$.
        \end{enumerate}
    \end{lemma}

    \noindent
    \bypierre{The proof of Lemma \ref{lem:directional_Fourier} is delayed to Appendix \ref{sec:directional_Fourier}, since it relies on classical arguments. Here, however, let us just show that} $\scrF_\Kv F(\cdot\,, \lvar)$ is $\Kv + \lvar \kvh_2$--pseudo-periodic with respect to the triangular lattice $\Lambda = \Z \vvv_1 + \Z \vvv_2$ for any $\lvar \in \R$. Assume that $F \in L^2_{\Kv \cdot \vvh_1} (\cylaug)$ is such that $\euler^{- \icplx \vts \Kv \cdot (\xv + s \vvv_2)}\; F \in \scrC^\infty_0 (\cylaug)$. On one hand, using the relations
    \begin{equation*}
        F(\xv + \vvv_1, s + r) = \euler^{\vts \icplx\vts \Kv \cdot \vvh_1}\, F(\xv, s), \quad \Kv \cdot \vvh_1 = (\Kv + \lvar \kvh_2) \cdot \vvh_1, \textAND \vvh_1 = \vvv_1 + r\vts \vvv_2,
    \end{equation*}
    we obtain
    \begin{align*}
        \scrF_\Kv F (\xv + \vvv_1, \lvar) &%
        = \euler^{\vts \icplx\vts \Kv \cdot \vvh_1} \frac{1}{\sqrt{2\pi}} \int_\R F (\xv, s - r)\, \euler^{- \icplx\vts (\Kv + \lvar \kvh_2) \cdot s \vvv_2}\, d s
        \\
        &= \euler^{\vts \icplx\vts (\Kv + \lvar \kvh_2) \cdot \vvh_1}\, \euler^{- \icplx\vts (\Kv + \lvar \kvh_2) \cdot r \vvv_2}\, \frac{1}{\sqrt{2\pi}} \int_\R F (\xv, s)\, \euler^{- \icplx\vts (\Kv + \lvar \kvh_2) \cdot s \vvv_2}\, d s 
        \\
        &= \euler^{\vts \icplx\vts (\Kv + \lvar \kvh_2) \cdot \vvv_1}\, \scrF_\Kv F (\xv, \lvar).
        \\
        \intertext{On the other hand, the property $F(\xv + \vvv_2, s - 1) = F(\xv, s)$ leads to}
        \scrF_\Kv F (\xv + \vvv_2, \lvar) &%
        = \frac{1}{\sqrt{2\pi}} \int_\R F (\xv, s + 1)\, \euler^{- \icplx\vts (\Kv + \lvar \kvh_2) \cdot s \vvv_2}\, d s = \euler^{\vts \icplx\vts (\Kv + \lvar \kvh_2) \cdot \vvv_2}\, \scrF_\Kv F  (\xv, \lvar).
    \end{align*}
    Consequently, $\scrF_\Kv F(\cdot\,, \lvar) \in L^2_{\Kv + \lvar \kvh_2}$ for any $\lvar \in \R$.

    \bypierre{Recall that $\calH^0_\kv = -\Delta + V(\xv)$ acts on $L^2_\kv$ for any $\kv \in \scrB$. By Lemma \ref{lem:directional_Fourier},
    \begin{equation*}
        \calH^0_{\AUG, \Kv \cdot \vvh_1} = \scrF^*_\Kv\, \left( \int^\oplus_{\lvar \in \R} \calH^0_{\Kv + \lvar \kvh_2}\, d\lvar \right)\, \scrF^{}_\Kv.
    \end{equation*}
    Therefore, for every $b \geq 1$ and $\lvar \in \R$, $E_b (\Kv + \lvar \kvh_2)$ is a generalized eigenvalue of $\calH^0_{\AUG, \Kv \cdot \vvh_1}$ with associated generalized eigenfunction $\Phi_b (\xv; \Kv + \lvar \kvh_2)\, \euler^{\vts \icplx\vts (\Kv + \lvar \kvh_2) \cdot s \vvv_2}$.} The completeness of the Floquet-Bloch eigenfunctions $\{\Phi_b (\cdot\,; \kv)\}_{b \geq 1}$ in $L^2_{\kv}$ for any $\kv \in \scrB$ leads directly to the following result.
    \begin{proposition}[Generalized eigenpairs of $\calH^0_{\AUG, \Kv \cdot \vvh_1} (\cylaug)$]\label{prop:completeness_slices}~
        \begin{enumerate}[label={$(\alph*).$}, ref={\theproposition.$(\alph*)$}, wide=0pt]
            \item Any $F \in L^2_{\Kv \cdot \vvh_1} (\cylaug)$ can be expressed as a continuous superposition of Floquet-Bloch modes with quasi-momenta located on the line $\Kv + \R\, \kvh_2$:
            \begin{subequations}
                \begin{align}
                    F(\xv, s) &= \frac{1}{\sqrt{2\pi}} \int_\R \widetilde{F} (\xv, \lvar)\, \euler^{\vts \icplx\vts (\Kv + \lvar \kvh_2) \cdot s \vvv_2}\, d \lvar
                    \\
                    &= \int_\R \sum_{b \geq 1} \widetilde{F}_b(\lvar)\, \Phi_b (\xv; \Kv + \lvar \kvh_2)\, \euler^{\vts \icplx\vts (\Kv + \lvar \kvh_2) \cdot s \vvv_2}\, d\lvar, \textFOR (\xv, s) \in \R^2 \times \R.\label{eq:Fxs-expanded-2}
                    \\
                    \intertext{Here,}
                    \widetilde{F}_b (\lvar) &:= \frac{1}{\sqrt{2\pi}} \big\langle \Phi_b (\cdot\,; \Kv + \lvar \kvh_2),\, \widetilde{F} (\cdot\,, \lvar) \big\rangle_{L^2_{\Kv + \lvar \kvh_2}} \textFOR (b, \lvar) \in \N \times \R.\label{eq:def_Ftilde_b}
                \end{align}
            \end{subequations}
            \item The estimate \eqref{eq:link_Sobolev_regularity_decay_FB_coeffs} combined with the properties of $\scrF_\Kv$ in Lemma \ref{lem:directional_Fourier} provides the following characterization for any $q \in \N_0$ and $F \in H^q_{\Kv \cdot \vvh_1, \xv} (\cylaug)$:
            \begin{equation}\label{eq:link_augmented_Sobolev_regularity_decay_FB_coeffs}
                \|\vts F \vts\|^2_{H^q_{\Kv \cdot \vvh_1, \xv} (\cylaug)} = \|\vts \widetilde{F} \vts\|^2_{L^2(\R_\lvar; H^q_{\Kv + \lvar \kvh_2})} \hyperref[item:notation_lesssim_simeq]{\simeq} \int_\R \sum_{b \geq 1} (1 + b^2)^{q/2}\, |\widetilde{F}_b (\lvar)|^2\, d\lvar.
            \end{equation}
            \item For any $F \in H^2_{\Kv \cdot \vvh_1, \xv} (\cylaug)$ and for $(\xv, s) \in \R^2 \times \R$,
            \begin{align}
                \calH^0_{\AUG, \Kv \cdot \vvh_1}\, F (\xv, s) &= \frac{1}{\sqrt{2\pi}} \int_\R \big[\vts \calH^0_{\Kv + \lvar \kvh_2} \widetilde{F} (\cdot\,, \lvar) \vts\big] (\xv)\, \euler^{\vts \icplx\vts (\Kv + \lvar \kvh_2) \cdot s \vvv_2}\, d \lvar \label{eq:fibered_H0aug}
                \\
                &= \int_\R \sum_{b \geq 1} E_b (\Kv + \lvar \kvh_2)\, \widetilde{F}_b(\lvar)\, \Phi_b (\xv; \Kv + \lvar \kvh_2)\, \euler^{\vts \icplx\vts (\Kv + \lvar \kvh_2) \cdot s \vvv_2}\, d\lvar. \label{eq:spectral_representation_H0aug}
            \end{align}
            In particular,
            \begin{equation*}
                \displaystyle
                \spec(\calH^0_{\AUG, \Kv \cdot \vvh_1}) = \big\{\vts E_b (\Kv + \lvar \kvh_2),\ \lvar \in \R,\ b \geq 1 \vts\big\}.
            \end{equation*}
        \end{enumerate}
    \end{proposition}

    \noindent
    Following the completeness property stated in Proposition \ref{prop:completeness_slices}, we give particular attention to the \emph{$\kvh_2$--band dispersion functions}, defined as
    \begin{equation}
    \lvar \in \R \mapsto E_b (\Kv + \lvar \kvh_2), \quad b \geq 1.
    \end{equation}
    We refer to their graphs as \emph{$\kvh_2$--band dispersion slices}.

    \subsection{Behavior of band dispersion slices in the Brillouin zone}\label{sec:edge_slices}
    By periodicity, restricting the dispersion functions $\kv \mapsto E_b (\kv)$ to the line $\Kv + \R\, \kvh_2$ is equivalent to restricting them to the broken line
    \begin{equation*}
        \displaystyle
        \scrD := (\Kv + \R\, \kvh_2 + \Lambda^*) \cap \scrB.
    \end{equation*}
    In Figure \ref{fig:broken_quasi_momentum_line}, we plot the truncation of $\scrD$, obtained from $(\Kv + [-L, L]\, \kvh_2 + \Lambda^*) \cap \scrB$ for $L \gg 1$, for different choices of edges: the zigzag edge ($r = 0$, first panel), the armchair edge ($r = 1$, second panel), a rational edge ($r = 1/2$, third panel), and an irrational edge ($r = \sqrt{2}$, fourth panel). These panels show different behaviors for rational and irrational edges.
    \begin{figure}[ht!]
        \foreach \idI/\idJ in {1/{Zigzag $(r = 0)$},2/{Armchair $(r = 1)$},3/{$r = 0.5$},4/{$r = \sqrt{2}$}} {
            \begin{subfigure}[b]{0.21\textwidth}
            \caption{\idJ\label{fig:broken_quasi_momentum_line_\idI}}
            \makebox[\textwidth][c]{
                \includegraphics[page=\idI]{figures/slice_dispersion.pdf}
            }
            \end{subfigure}
            \hfill
        }
        \caption{The truncated broken line $(\Kv + [-L, L]\, \kvh_2 + \Lambda^*) \cap \scrB$ for $L = 20$.\label{fig:broken_quasi_momentum_line}}
    \end{figure}

    %
    %
    %
    If $r = b_1 / a_1$ is rational, then the broken line $\scrD$ has a finite number of components. Moreover, the function $\lvar \mapsto E_b (\Kv + \lvar \kvh_2)$ is $2\pi\vts a_1$--periodic, and therefore can be reconstructed from its values on the interval $[-\pi\vts a_1, \pi\vts a_1]$. If $r$ is irrational, then $\lvar \mapsto E_b (\Kv + \lvar \kvh_2)$ is no longer periodic, and Figure \ref{fig:broken_quasi_momentum_line} suggests that $\scrD$ is dense in $\scrB$. Indeed, we have

    \begin{theorem}[Kronecker's theorem, \textnormal{\cite[Theorem 442]{hardy1979introduction}}]\label{thm:kronecker_approximation_theorem}
    Let $r$ be irrational and recall that $\kvh_2 := -r \kv_1 + \kv_2$. Then for any $\kv \in \R^2$, the set $\kv + \R\, \kvh_2 + \Lambda^*$ is dense in $\R^2$. In particular, the broken line $\scrD$ comes arbitrarily close to all vertices of the Brillouin zone $\scrB$, the union of the two high-symmetry sublattices.
    \end{theorem}

    \subsection{Neighborhoods of high-symmetry quasi-momenta}\label{sec:veps_neighborhood}
    We take $r$ to be irrational. As in previous work devoted to rational edges \cite{fefferman2016edge,lee2019elliptic,drouot2019characterization,drouot2020edge}, our goal is to construct edge states comprised dominantly of Floquet-Bloch modes $\{\Phi_b(\cdot\,; \kv)\}_{b \geq 1}$ with quasi-momenta $\kv = \Kv + \lvar \kvh_2$ located near the high-symmetry quasi-momenta, $\bbK := (\Kv + \Lambda^*) \cup (\Kv' + \Lambda^*)$. Recall from Proposition \ref{prop:completeness_slices} that the family of Floquet-Bloch modes with quasi-momenta parameterized by the mapping $\lvar \in \R \mapsto \Kv + \lvar \kvh_2$ is complete in $L^2_{\Kv \cdot \vvh_1} (\cylaug)$. Hence, we next characterize\ those $\lvar\in\R$ for which $\Kv + \lvar \kvh_2$ is located in some specified small neighborhood of $\bbK$.

    Since $r$ is irrational, Theorem \ref{thm:kronecker_approximation_theorem} (see also Figure \ref{fig:broken_quasi_momentum_line}) implies that there are points on the line $\Kv + \R \kvh_2$ which are arbitrarily close to points of both the $\Kv$ and $\Kv'$--sublattices. Thus we treat $\Kv$ and $\Kv'$--type points equivalently, unless the distinction is specifically needed. We introduce the countable set 
    \begin{equation}
    \bbL := \{\Kv, \Kv'\} \times \Z.
    \label{eq:bbL-def}
    \end{equation}
    A generic element of $\bbL$ is denoted $\iiv = (\Kv_\iiv, m_\iiv)$.

    We find it natural to introduce the following parallelogram-shaped neighborhoods in $\R^2_\kv$: 
    \begin{align}\label{eq:cell_aligned_with_edge_vectors}
        \spforall 0 < \veps \leq \pi, \quad \widehat{\Omega}_\veps &:= \big\{\vts \yrm_1\vts \kvh_1 + \yrm_2\vts \kvh_2\ \ /\ \ |\yrm_j| < \veps \vts\big\},
        \rets
        \widehat{\Omega} &:= \big\{\vts \yrm_1\vts \kvh_1 + \yrm_2\vts \kvh_2\ \ /\ \ -\pi \leq \yrm_j < \pi \vts\big\}.\label{eq:unit_cell_aligned_with_edge_vectors}
    \end{align}
    Note that $\widehat{\Omega}$ is a fundamental parallelogram of the lattice $\Z\vts (2\pi \kvh_1) + \Z\vts (2\pi \kvh_2)$, with centering at $\bf{0}$.

    We introduce the open set, $B_\veps$, consisting of $\veps$--neighborhoods of all points in the union of high-symmetry quasi-momenta sublattices $\mathbb K=(\Kv+\Lambda^*)\cup (\Kv'+\Lambda^*)$; see Figure \ref{fig:near_quasimomenta}:
    \begin{equation}\label{eq:def_veps_neighborhood}
    \spforall \veps \in (0, 1), \quad B_\veps := B^{\Kv}_\veps \cup B^{\Kv'}_\veps \textWHERE B^{\Kv_\star}_\veps := \Kv_\star + \widehat{\Omega}_\veps + \Lambda^*, \quad \spforall \Kv_\star \in \{\Kv, \Kv'\}.
    \end{equation}
    The shape of these neighborhoods has the desirable property that $\{\lvar \in \R\ /\ \Kv + \lvar \kvh_2 \in B^{\Kv_\star}_\veps\}$ consists of intervals of \emph{equal} length $\veps$. To introduce the centers of these $\veps$--sized intervals, we begin with the following lemma, which shows that any quasi-momentum of the form $\Kv + \lvar \kvh_2$ coincides with a unique element of $\Kv_\star + \widehat{\Omega}$, up to translations by elements of the dual lattice $\Lambda^*$. Recall that $\modulo{\lvar} := \lvar - \lfloor \lvar \rfloor \in [0, 1)$ denotes the fractional part of $\lvar \in \R$.

    \begin{lemma}[Wrapping Lemma]\label{lem:wrapping_lemma}
        Let $r$ be irrational, and consider the infinite line in $\R^2: \Kv+\R \kvh_2$. 
        Here,  
        $\kvh_2= -r\kv_1+\kv_2$ and $\kvh_1=\kv_1$.
        \begin{enumerate}[label={$(\alph*).$}, ref={$(\alph*)$}]
            \item\label{lem:veps_neighborhood_K_a} Fix $\Kv_\star\in\{\Kv,\Kv'\}$. For any $\lvar \in \R$, there is a unique pair:
            \begin{equation}\label{eq:lvar-def}
            \modK (\lvar)= \yrm_1 (\lvar)\, \kvh_1+ \yrm_2 (\lvar)\, \kvh_2 \in \widehat{\Omega}\quad \textrm{(with $|\yrm_j|\le\pi$) and} \quad \floK(\lvar) \in \Lambda^*,
            \end{equation}
            depending on $\Kv_\star$, such that
            \begin{equation}\label{eq:representation_K_plus_lk2} 
            \Kv + \lvar \kvh_2 = \Kv_\star + \modK (\lvar) + \floK(\lvar). %
            \end{equation}
            The vector $\modK(\lvar)$ represents the relative position of $\Kv + \lvar \kvh_2$ within the local coordinate system centered at $\Kv_\star$ and aligned with the lattice $\Z\vts (2\pi \kvh_1) + \Z\vts (2\pi \kvh_2)$, up to translation by an element of  $ \Lambda^*$ (which is $\floK (\lvar)$); see Figure \ref{fig:relative_position_eps_neighborhoods}. 
            \item\label{lem:veps_neighborhood_K_b} For any $\Kv_\star \in \{\Kv,\Kv'\}$, the mapping $\Kv + \lvar \kvh_2 \mapsto \Kv_\star + \modK (\lvar)$, where $\modK (\lvar)$ is given by \eqref{eq:lvar-def}, is a  bijection between the line $\Kv + \R \kvh_2$ and a countable union of segments which is dense in $\Kv_\star + \widehat{\Omega}$ (see Figure \ref{fig:relative_position_eps_neighborhoods}). We label these segments with indices:  $\iiv = (\Kv_\star, m) \in \{\Kv, \Kv'\} \times \Z  = \bbL$, and they have parameterizations:
            \begin{equation}
            \lvar\in \lvar_\iiv + [-\pi  , \pi ) \mapsto  \Kv_\star + \modK_\iiv(\lvar),\nonumber 
            \end{equation}
            where 
            \begin{equation}
            \modK_\iiv(\lvar) := \displaystyle\gamma_\iiv\, \kvh_1 + (\lvar - \lvar_\iiv)\, \kvh_2\in \widehat\Omega, 
            \label{eq:def-l_I}
            \end{equation}
            and where the pair $(\gamma_\iiv, \lvar_\iiv)$ is defined by \eqref{eq:gamma_Iandlambda_I}. The mid-point of the $\iiv$--indexed segment is the point $\gamma_\iiv\, \kvh_1$, which is attained at $\lvar = \lvar_\iiv$.
            \item\label{lem:veps_neighborhood_K_c} As a subset of the infinite line $\Kv+\R\kvh_2$ 
            the $\iiv$--indexed segment in $\Kv_\star + \widehat\Omega$ is given by
            \begin{equation}\label{eq:representation_K_plus_lk2-iiv} 
            \lvar\in \lvar_\iiv + [-\pi  , \pi ] \mapsto  \Kv + \lvar \kvh_2 = \Kv_\star + \modK_\iiv (\lvar) + \floK_\iiv(\lvar), %
            \end{equation}
            where 
            \begin{equation}\label{eq:def_floK}
            \floK_\iiv (\lvar) = \floK_\iiv = \displaystyle 2\pi \left\lfloor \frac{(\Kv - \Kv_\star) \cdot \vvh_1}{2\pi} - m\vts r + \frac{1}{2} \right\rfloor\vts \kv_1 + 2\pi\vts m\vts \kv_2 \in \Lambda^*.
            \end{equation}
            \item\label{lem:veps_neighborhood_K_d} The relation \eqref{eq:representation_K_plus_lk2-iiv}, namely $\Kv + \lvar \kvh_2 = \Kv_\star + \modK_\iiv (\lvar) + \floK_\iiv$, extends to all $\lvar \in \R$. However, $\modK_\iiv (\lvar)$ might not belong to $\widehat{\Omega}$ if $\lvar \not\in \lvar_\iiv + [-\pi, \pi)$.
        \end{enumerate}
    \end{lemma}

    \begin{figure}[ht!]
        \centering
        \begin{subfigure}{\textwidth}
            \makebox[\textwidth][c]{
            \includegraphics[page=2]{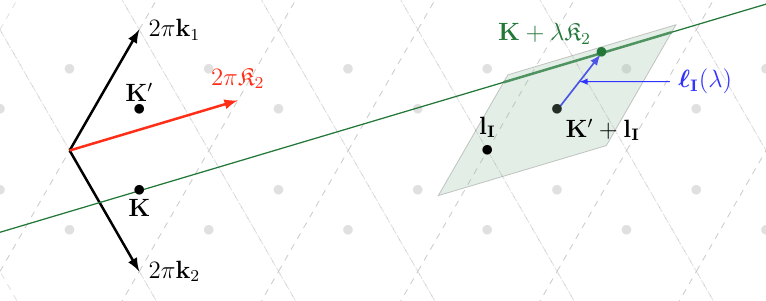}
            }
        \end{subfigure}
        \begin{subfigure}{\textwidth}
            \makebox[\textwidth][c]{
            \includegraphics[page=3]{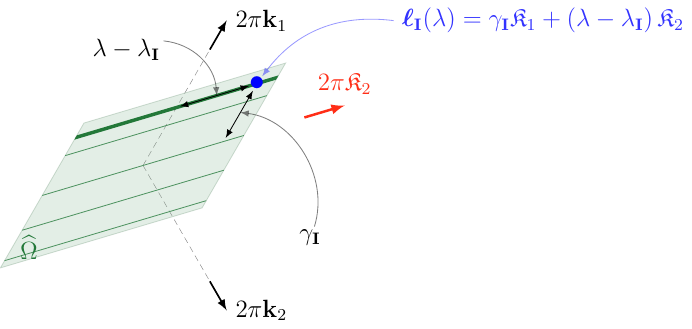}
            }
        \end{subfigure}
        \caption{Illustration of the Wrapping Lemma \ref{lem:wrapping_lemma} with $\Kv_\star = \Kv'$;  $\Kv + \lvar \kvh_2 =  \Kv_\star +\modK_\iiv (\lvar) +  \floK_\iiv$, where 
            $\Kv_\star + \modK_\iiv (\lvar)$ is in the fundamental cell $\Kv_\star + \widehat{\Omega}$. Here, $\floK_\iiv = 2\pi\, (3\vts \kv_1 + 3\vts \kv_2)$ and $\modK_\iiv (\lvar)$ is the relative position of $\Kv + \lvar \kvh_2$ in $\Kv_\star + \widehat{\Omega}$ up to translation by $\floK_\iiv$. Recall the relations \eqref{eq:def_edge_vectors}, e.g. that $\kvh_1=\kv_1$. \label{fig:relative_position_eps_neighborhoods}
        }
    \end{figure}

    \begin{dem}
        Our proof of the Part \ref{lem:veps_neighborhood_K_a} of Lemma  \ref{lem:wrapping_lemma} is constructive. Recall that by Kronecker's theorem, the set $\Kv + \R \kvh_2 + \Lambda^*$ restricted to $\Kv_\star + \widehat{\Omega}$ is a countable union of line segments which is dense in $\Kv_\star + \widehat{\Omega}$ (where $\widehat{\Omega}$ is defined by \eqref{eq:unit_cell_aligned_with_edge_vectors}). Let us now parameterize all such line segments in $\widehat\Omega$. %

        Fix $\lvar \in \R$. Let $\modK := \yrm_1\vts \kvh_1 + \yrm_2\vts \kvh_2\in\widehat\Omega$ and $\floK = 2\pi\vts n_1\vts \kv_1 + 2\pi\vts n_2\vts \kv_2\in \Lambda^*$. The constraint that $\modK$ lie in $\widehat\Omega$ implies that $-\pi \le \yrm_j < \pi$.  Part $(a)$ is then equivalent to \emph{finding $\yrm_j \in [-\pi, \pi)$ and $n_j \in \Z$, $j = 1, 2$, such that}
        \begin{equation}\label{eq:yjnj-fit}
            \Kv + \lvar \kvh_2 = \Kv_\star + \yrm_1\vts \kvh_1 + \yrm_2\vts \kvh_2 + 2\pi\vts n_1\vts \kv_1 + 2\pi\vts n_2\vts \kv_2.
        \end{equation}
        Taking the scalar product with the vectors $\vvh_1$ and $\vvh_2$, and using the relations $\kvh_n \cdot \vvh_l = \delta_{n l}$ as well as $\kv_1 \cdot \vvh_1 = \kv_2 \cdot \vvh_2 = 1$, $\kv_1 \cdot \vvh_2 = 0$, and $\kv_2 \cdot \vvh_1 = r$, the above equation reduces to
        \begin{equation}\label{eq:representation_system}
            \left\{
            \begin{array}{r@{\ =\ }l}
                \Kv \cdot \vvh_1 & \Kv_\star \cdot \vvh_1 + \yrm_1 + 2\pi\, (n_1 + n_2\vts r),
                \rets
                \Kv \cdot \vvh_2 + \lvar & \Kv_\star \cdot \vvh_2 + \yrm_2 + 2\pi\vts n_2.
            \end{array}
            \right.
        \end{equation}
        Let us note that each equation in \eqref{eq:representation_system} boils down to the toy problem:
        \begin{subequations}\label{eq:integer_fractional_decomposition}
            \begin{equation}\label{eq:integer_fractional_decomposition_a}
                \textit{Given $a \in \R$, \quad find $(\yrm, n) \in [-\pi, \pi) \times \Z$ such that} \quad \yrm + 2\pi n = a , 
            \end{equation}
            which may be rewritten as $\yrm/(2\pi) + 1/2 = a/(2\pi) + 1/2 - n$. Taking the integer and fractional parts, we find that the unique solution $(\yrm, n)\in [-\pi, \pi) \times \Z$ is given by:
            \begin{equation}
                \yrm = 2\pi\, \modulo{\frac{a}{2\pi} + \frac{1}{2}} - \pi \textAND n = \left\lfloor \frac{a}{2\pi} + \frac{1}{2} \right\rfloor.
            \end{equation}
        \end{subequations}
        Applying \eqref{eq:integer_fractional_decomposition} to the second equation of \eqref{eq:representation_system} with $a := (\Kv - \Kv_\star) \cdot \vvh_2 + \lvar$, we get
        \begin{equation}\label{eq:y2n2}
            \yrm_2(\lvar) = 2\pi\, \modulo{\frac{(\Kv - \Kv_\star) \cdot \vvh_2 + \lvar}{2\pi} + \frac{1}{2}} - \pi, \quad n_2(\lvar) = \left\lfloor \frac{(\Kv - \Kv_\star) \cdot \vvh_2 + \lvar}{2\pi} + \frac{1}{2} \right\rfloor.
        \end{equation}
        Next, substitute the expression \eqref{eq:y2n2} for $n_2 \equiv n_2(\lvar)$ into the first equation of \eqref{eq:representation_system} and apply \eqref{eq:integer_fractional_decomposition} with $a = (\Kv - \Kv_\star) \cdot \vvh_1 - 2\pi\vts n_2 (\lvar)\vts r$, to obtain:
        \begin{equation}\label{eq:y1n1}
            \begin{aligned}
            \yrm_1(\lvar) &= 2\pi\, \modulo{\frac{(\Kv - \Kv_\star) \cdot \vvh_1}{2\pi} - n_2(\lvar)\vts r + \frac{1}{2}} - \pi,
            \\
            n_1(\lvar) &= \left\lfloor \frac{(\Kv - \Kv_\star) \cdot \vvh_1}{2\pi} - n_2(\lvar)\vts r + \frac{1}{2} \right\rfloor.
            \end{aligned}
        \end{equation}
        Thus we have a representation of any point of $\Kv+\R\kvh_2$ as the sum of a point in $\widehat{\Omega}$ and a point in $\Lambda^*$, in terms of parameters $\yrm_j \in [-\pi,\pi)$ and $n_j \in \Z$. This concludes the proof of Part \ref{lem:veps_neighborhood_K_a}.

        Turning to Part \ref{lem:veps_neighborhood_K_b}, we note that the center of the line segment:   $\lvar\mapsto \Kv_\star+ \modK(\lvar)\in \Kv_\star + \widehat\Omega$  occurs for $\lvar=\widetilde{\lvar}$, such that $\yrm_2(\widetilde{\lvar})=0$. From \eqref{eq:y2n2}, this implies
        \begin{equation*}
            \widetilde{\lvar} \in (\Kv_\star - \Kv) \cdot \vvh_2 + 2\pi\vts \Z.
        \end{equation*}
        For each $\Kv_\star\in\{\Kv,\Kv'\}$ and $m\in\Z$, we define
        \begin{equation}
            \lvar_\iiv := (\Kv_\star - \Kv) \cdot \vvh_2 + 2\pi\vts m,\quad {\rm where}\quad \iiv =(\Kv_\star,m). \label{eq:def-lvar_iiv}
        \end{equation}
        Note that 
        \begin{equation*}
            \frac{(\Kv - \Kv_\star) \cdot \vvh_2 + \lvar}{2\pi} + \frac{1}{2} = \frac{\lvar - \lvar_\iiv}{2\pi} + \frac{1}{2} + m\ \ \in [0, 1) + m,
        \end{equation*}
        with the latter inclusion holding if $\lvar \in \lvar_\iiv + [-\pi, \pi)$.
        Hence, the expressions \eqref{eq:y2n2} for $(\yrm_2, n_2)$ reduce to:
        \begin{subequations}\label{eq:simplified_exp_yj_nj}
            \begin{equation}
            \yrm_2(\lvar) = \lvar - \lvar_\iiv \textAND n_2 = m,
            \end{equation}
            and further, the expressions for $(\yrm_1, n_1)$, become:
            \begin{align}
            \yrm_1(\lvar) &= \gamma_\iiv := 2\pi\vts \modulo{\frac{(\Kv - \Kv_\star) \cdot \vvh_1}{2\pi} - m\vts r + \frac{1}{2}} - \pi, \label{eq:def-gamma_iiv-NEW}
            \\
            n_1 &= \left\lfloor \frac{(\Kv - \Kv_\star) \cdot \vvh_1}{2\pi} - m\vts r + \frac{1}{2} \right\rfloor,
            \end{align}
        \end{subequations}
        where $\lvar \in \lvar_\iiv + [-\pi, \pi)$. We conclude that the $\iiv = (\Kv_\star, m)$--indexed segment in $\Kv_\star + \widehat\Omega$ is given by 
        \begin{equation*}
            \lvar\in \lvar_\iiv + [-\pi, \pi)\mapsto \Kv_\star+ \modK_\iiv (\lvar),\textWHERE \modK_\iiv(\lvar) = \gamma_\iiv\vts \kvh_1 + (\lvar-\lvar_\iiv)\vts \kvh_2\in\widehat\Omega.
        \end{equation*}
        This concludes the proof of Part \ref{lem:veps_neighborhood_K_b} of Lemma \ref{lem:wrapping_lemma}.

        Part \ref{lem:veps_neighborhood_K_c} of Lemma \ref{lem:wrapping_lemma} follows from the decomposition 
        \[ 
            \Kv + \lvar \kvh_2 = \Kv_\star + \modK_\iiv (\lvar) + \floK_\iiv, \quad {\rm for}\quad   \lvar \in \lvar_\iiv + [-\pi, \pi),
        \]
        where 
        \begin{equation}\label{eq:ell_I&l_I}
            \floK_\iiv := 2\pi\vts n_1 \vts \kv_1 + 2\pi\vts n_2\vts \kv_2\in \Lambda^*, \textFOR \lvar \in \lvar_\iiv + [-\pi, \pi),
        \end{equation}
        with $n_1$ and $n_2$ as in \eqref{eq:simplified_exp_yj_nj}. Finally, Part \ref{lem:veps_neighborhood_K_d} can be verified from direct computations. Thus the proof of the Wrapping Lemma \ref{lem:wrapping_lemma} is complete.
    \end{dem}

    \begin{figure}[ht!]  
    \makebox[\textwidth][c]{
        \includegraphics[page=4]{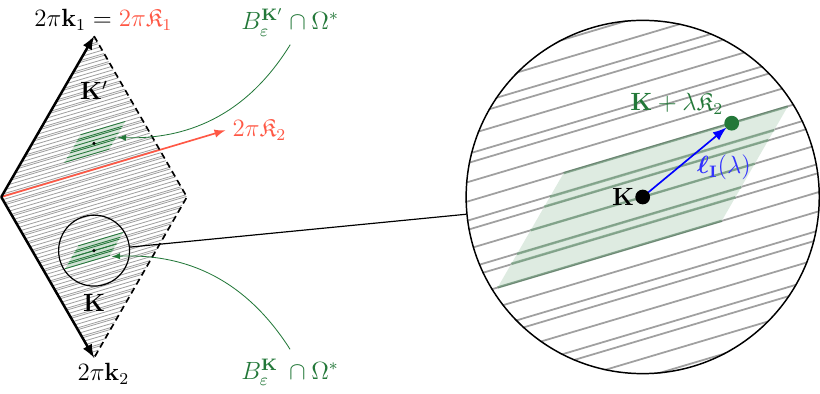}
    }
    \caption{The neighborhoods $B^{\Kv_\star}_\veps$ restricted to the dual periodicity cell $\Omega^* = \{\vts z_1\vts \kv_1 + z_2\vts \kv_2\ /\ z_1, z_2 \in (0, 2\pi) \vts\}$, with $r = - \sqrt{2}$. The set of gray lines represents the set $\Kv + \R\, \kvh_2 \mod \Omega^*$, which is dense in $\Omega^*$. The set of green segments represents the subset of $\Kv + \R\, \kvh_2 \mod \Omega^*$ that is within $B^{\Kv_\star}_\veps$.\label{fig:near_quasimomenta}}
    \end{figure}

    \vspace{1\baselineskip} \noindent
    By Part \ref{lem:veps_neighborhood_K_d} of Lemma \ref{lem:wrapping_lemma}, $\Kv + \lvar \kvh_2 = \Kv_\star+\gamma_\iiv\kvh_1 + (\lvar - \lvar_\iiv) \kvh_2 + \floK_\iiv$, and so the coordinate $\gamma_\iiv$ controls displacement along $\kvh_1 = \kv_1$  and $\lvar-\lvar_\iiv$, displacement along $\kvh_2$. A consequence is a characterization of segments in $\Kv + \lvar \kvh_2$ within the $\veps$--neighborhood $B_\veps$, defined in \eqref{eq:def_veps_neighborhood}. The simplification of this result for rational edges is presented in Section \ref{sec:app:neighborhoods_quasimomenta}.
    \begin{proposition}\label{prop:veps_neighborhood_K}
        Let $\veps > 0$ and recall from \eqref{eq:def_L_veps} that $\bbL (\veps) := \left\{\vts \iiv \in \bbL,\ \ |\gamma_\iiv| \leq \veps \vts\right\}$, where $\bbL (\pi) = \{\Kv, \Kv'\} \times \Z = \bbL$, and where $\gamma_\iiv \in [-\pi, \pi)$ is the coefficient defined by \eqref{eq:def-gamma_iiv-NEW}. Then, the set of reals $\lvar$ such that $\Kv + \lvar \kvh_2 \in B_\veps$ is a countable union of $2\vts \veps$--sized segments centered at the points $\lvar_\iiv$ defined in \eqref{eq:def-lvar_iiv}, for $\iiv \in \bbL (\veps)$: 
        \begin{equation}\label{eq:carac_eps_neighborhood_K}
            \textnormal{For} \   \lvar \in \R, \quad \Kv + \lvar \kvh_2 \in B_\veps \quad \Longleftrightarrow \quad \lvar \in \bigcup_{\iiv \in \bbL (\veps)} \lvar_\iiv + (-\veps, \veps)\ .
        \end{equation}
        Further, if $\veps < 1/3$, then  the union of intervals on the right-hand side of \eqref{eq:carac_eps_neighborhood_K} is a disjoint union.
    \end{proposition}

    \begin{dem}
    Fix $\Kv_\star \in \{\Kv, \Kv'\}$. For any $\lvar \in \R$, there exists a unique $m \in \Z$, with $\iiv = (\Kv_\star, m)$, such that $\lvar \in \lvar_\iiv + [-\pi, \pi)$. It follows from Part \ref{lem:veps_neighborhood_K_b} of the Wrapping Lemma \ref{lem:wrapping_lemma} that $\modK_\iiv (\lvar) := \gamma_\iiv\, \kvh_1 + (\lvar - \lvar_\iiv)\, \kvh_2$ and $\floK_\iiv$ are the unique vectors in $\widehat{\Omega}$ and $\Lambda^*$ respectively, such that $\Kv + \lvar \kvh_2 = \Kv_\star + \modK_\iiv (\lvar) + \floK_\iiv$. Therefore, using the definitions of $\modK_\iiv (\lvar)$ and $B^{\Kv_\star}_\veps$, we deduce that $\Kv + \lvar \kvh_2$ belongs to $B^{\Kv_\star}_\veps$ if and only if $|\gamma_\iiv| \leq \veps$ and $|\lvar - \lvar_\iiv| \leq \veps$. Using this conclusion for $\Kv_\star = \Kv$ and $\Kv_\star = \Kv'$, we deduce the characterization \eqref{eq:carac_eps_neighborhood_K}.
    \end{dem}

    \section{Schur complement reduction to near components}\label{sec:proof_Theorems_A_B}
    \subsection{Projectors near and far from Dirac points }\label{sec:near_far_quasimomentum_energy_components}
    Let $(\Kv_\star, E_D)$, with $\Kv_\star \in \{\Kv, \Kv'\}$, denote a Dirac point in the sense of Section \ref{sec:dirac_points}. We recall from Section \ref{sec:FB_states} that for any $\kv \in \scrB$, the $L^2_\kv$--spectrum of $\calH^0_\kv = -\Delta + V(\xv)$ consists of eigenvalues $E_b (\kv)$, with associated eigenfunctions $\Phi_b (\cdot\,; \kv)$, for $b \geq 1$. By Assumption \ref{hyp:Dirac_point}, there exists an integer $b_\star \geq 1$ such that, setting $\{-, +\} := \{b_\star, b_\star + 1\}$, the energy $E_D = E_+ (\Kv_\star) = E_- (\Kv_\star)$ is a double eigenvalue of $\calH^0_{\Kv_\star}$, for $\Kv_\star \in \{\Kv, \Kv'\}$. For arbitrary $\kv \in \scrB$, we  introduce the orthogonal projector $\Pi^0_\kv \in \scrL(L^2_\kv)$:
    \begin{equation}\label{eq:Pi_0_k}
        \spforall u \in L^2_\kv, \quad \Pi^0_\kv\, u (\xv) = \sum_{\pm} \big\langle\vts \Phi_\pm (\cdot\, ; \kv), \, u \vts\big\rangle_{L^2_\kv}\, \Phi_\pm (\xv; \kv), \quad \xv \in \R^2.
    \end{equation}
    For $\Kv_\star \in \{\Kv, \Kv'\}$, $\Pi^0_{\kv = \Kv_\star}$ can be rewritten in terms of the Dirac eigenbasis $\{\Phi^{\Kv_\star}_1, \Phi^{\Kv_\star}_2\}$:
    \begin{equation}\label{eq:Pi_0_K_star}
        \spforall u \in L^2_{\Kv_\star}, \quad \Pi^0_{\Kv_\star}\vts u (\xv) = \Phi^{\Kv_\star} (\xv)^\transp\, \big\langle\vts \Phi^{\Kv_\star}, \, u \vts\big\rangle_{L^2_{\Kv_\star}}, \quad \xv \in \R^2.
    \end{equation}
    Recall, from Proposition \ref{prop:completeness_slices}, that any $F \in L^2_{\Kv \cdot \vvh_1} (\cylaug)$ can be represented as a continuous superposition of Floquet-Bloch modes of $\calH^0 = -\Delta + V(\xv)$ with quasi-momenta located on the line $\Kv + \R\, \kvh_2$. %
    For $\delta > 0$, Proposition \ref{prop:veps_neighborhood_K} ensures that $\Kv + \lvar \kvh_2$ belongs to the $\delta^{3/4}$--neighborhood, $\hyperref[eq:def_veps_neighborhood]{B_{\delta^{3/4}}}$ of $\mathbb K$, the lattice of high-symmetry quasi-momenta, if and only if $\lvar \in \lvar_\iiv + (-{\delta^{3/4}}, {\delta^{3/4}})$ for some $\iiv \in \bbL (\delta^{3/4})$. Here, $\lvar_\iiv \in \R$ is defined by \eqref{eq:gamma_Iandlambda_I}, and the set $\bbL (\delta^{3/4}) \subset \bbL$ is introduced in \eqref{eq:def_L_veps}. Using these observations, we introduce, for any $\iiv \in \bbL (\delta^{3/4})$, the orthogonal projector on $L^2_{\Kv \cdot \vvh_1} (\cylaug)$, associated with the two dispersion surfaces touching at the Dirac point. %
    In \eqref{eq:Fxs-expanded-2}, replacing the infinite sum over $b \geq 1$ by a finite sum over $b \in \{b_\star, b_\star + 1\}$ and inserting a Fourier-truncation to the interval $|\lvar -\lvar_\iiv| \leq \delta^{3/4}$, we define
    for $F \in L^2_{\Kv \cdot \vvh_1} (\cylaug)$:
    \begin{subequations}
    \label{eq:def_Pi_I}
    \begin{align}
        \Pi_\iiv\, F (\xv, s) &:= \int_\R \chi \big(\delta^{-3/4} (\lvar - \lvar_\iiv)\big)\, \sum_{\pm} \widetilde{F}_\pm(\lvar)\, \Phi_\pm (\xv; \Kv + \lvar \kvh_2)\, \euler^{\vts \icplx\vts (\Kv + \lvar \kvh_2) \cdot s \vvv_2}\, d\lvar, \label{eq:def_Pi_I_3}
        \retss
        &:= \frac{1}{\sqrt{2\pi}} \int_\R \chi \big(\delta^{-3/4} (\lvar - \lvar_\iiv)\big)\,  \big[\vts\Pi^0_{\Kv + \lvar \kvh_2}\, \widetilde{F} (\cdot\,, \lvar)\vts\big] (\xv)\, \euler^{\vts \icplx\vts (\Kv + \lvar \kvh_2) \cdot s \vvv_2}\, d \lvar, \label{eq:def_Pi_I_2}
    \end{align}
    where $\widetilde{F}_- (\lvar) := \widetilde{F}_{b_\star} (\lvar)$, $\widetilde{F}_+ (\lvar) := \widetilde{F}_{b_\star + 1} (\lvar)$. Here, $\chi$ denotes the characteristic function of the interval $[-1, 1]$. The second equality follows from the definition \eqref{eq:def_Ftilde_b} of $\widetilde{F}_b (\lvar)$ and the definition \eqref{eq:Pi_0_k} of $\Pi^0_{\Kv + \lvar \kvh_2}$. Note that the projection $\Pi_\iiv$ only involves the Floquet-Bloch eigenfunctions $\Phi_\pm (\xv; \Kv + \lvar\vts \kvh_2)$. The corresponding components are then integrated over the quasi-momentum segment $\{\Kv + \lvar\vts \kvh_2,\ |\lvar - \lvar_\iiv| \leq \delta^{3/4}\}$, which is included in $B_{\delta^{3/4}}$, provided that $\iiv \in \bbL (\delta^{3/4})$. Using the definition \eqref{eq:inverse_directional_Fourier} of $\scrF^*_\Kv$, we find the equivalent definition:
    \begin{equation}
        \Pi_\iiv\, F := \scrF^*_\Kv \Bigl\{\vts (\xv, \lvar) \mapsto \chi \big(\delta^{-3/4} (\lvar - \lvar_\iiv)\big)\,  \big[\vts\Pi^0_{\Kv + \lvar \kvh_2}\, \widetilde{F} (\cdot\,, \lvar)\vts\big] (\xv) \vts\Bigr\}. \label{eq:def_Pi_I_1}
    \end{equation}
    \end{subequations}
    We also introduce the following operators:
    \begin{subequations}
    \label{eq:def_near_far_projections}
    \begin{align}
    \Pi_{\bbL'} &:= \sum_{\iiv \in \bbL'} \Pi_\iiv, \quad  \textrm{where}\quad  \bbL' \subseteq \bbL(\delta^{3/4}), \quad 
    \\
    \Pi_\NEAR &:= \Pi_{\bbL(\delta^{3/4})}= \sum_{\iiv \in \bbL(\delta^{3/4})} \Pi_\iiv, \textAND\\
    \Pi_\FAR &:= \Id_{L^2_{\Kv \cdot \vvh_1} (\cylaug)} - \Pi_\NEAR. 
    \end{align} 
    \end{subequations}
    Following Proposition \ref{prop:veps_neighborhood_K}, we choose $\delta$ so that 
    \begin{equation}\label{eq:pairwise_disjoint_near_intervals}
    0 < \delta^{3/4} < 1/3; \quad \textnormal{hence $\lvar_\iiv + ( -\delta^{3/4},\, \delta^{3/4} )$ are pairwise disjoint for $\iiv \in \bbL (\delta^{3/4})$.}
    \end{equation}
    Under Condition \eqref{eq:pairwise_disjoint_near_intervals}, 
    \begin{subequations}\label{eq:PiIPiJeq0}
    \begin{align}
    &\Pi_\iiv\, \Pi_\jjv = 0,\quad \textrm{for}\quad \iiv \neq \jjv,\quad  
    \retss
    &\Pi_\NEAR\, \Pi_\FAR=0, \textAND
    \retss
    &\textnormal{$\Pi_\NEAR$, $\Pi_\FAR$ and $\Pi_{\bbL'}$ for $\bbL' \subseteq \bbL(\delta^{3/4})$,  define orthogonal projectors on $L^2_{\Kv \cdot \vvh_1} (\cylaug)$.}
    \end{align}
    \end{subequations}
    As a consequence, their ranges
    \begin{equation}
    \scrX_\iiv := \Ran (\Pi_\iiv), \quad \scrX_{\bbL'} := \Ran(\Pi_{\bbL'}), \quad \scrX_{\NEAR} := \Ran(\Pi_\NEAR), \quad \scrX_{\FAR} := \Ran(\Pi_\FAR),
    \end{equation}
    are all closed subspaces of $L^2_{\Kv \cdot \vvh_1} (\cylaug)$.

    \begin{proposition}\label{prop:smoothing_near_projectors}
    For $q \in \N_0$, and for any $\bbL' \subseteq \bbL (\delta^{3/4})$, the orthogonal projector $\Pi_{\bbL'}$ is bounded from $L^2_{\Kv \cdot \vvh_1} (\cylaug)$ to $H^q_{\Kv \cdot \vvh_1, \xv} (\cylaug)$.
    \end{proposition}

    \begin{dem}
    Let $F \in L^2_{\Kv \cdot \vvh_1} (\cylaug)$ and $q \in \N_0$. By definition \eqref{eq:def_Pi_I_3} of $\Pi_\iiv$,
    \begin{equation*}
        \Pi_{\bbL'}\, F (\xv, s) = \sum_{\pm} \int_\R \Big[ \sum_{\iiv \in \bbL'} \chi \big(\delta^{-3/4} (\lvar - \lvar_\iiv)\big) \Big]\, \widetilde{F}_\pm(\lvar)\, \Phi_\pm (\xv; \Kv + \lvar \kvh_2)\, \euler^{\vts \icplx\vts (\Kv + \lvar \kvh_2) \cdot s \vvv_2}\, d\lvar.
    \end{equation*}
    Applying the characterization \eqref{eq:link_augmented_Sobolev_regularity_decay_FB_coeffs}, we then get
    \begin{multline*}
        \|\vts \Pi_{\bbL'}\, F \vts\|^2_{H^q_{\Kv \cdot \vvh_1, \xv} (\cylaug)} \underset{\eqref{eq:link_augmented_Sobolev_regularity_decay_FB_coeffs}}{\hyperref[item:notation_lesssim_simeq]{\simeq}} \int_\R \sum_{b \in \{b_\star, b_\star + 1\}} (1 + b^2)^{q/2}\, \Big[ \sum_{\iiv \in \bbL'} \chi \big(\delta^{-3/4} (\lvar - \lvar_\iiv)\big) \Big]^2 |\widetilde{F}_b (\lvar)|^2\, d\lvar 
        \\
        \lesssim C_q\, \int_\R \Big[\sum_{\iiv \in \bbL'} \chi \big(\delta^{-3/4} (\lvar - \lvar_\iiv)\big) \Big]^2\, \sum_{\pm} |\widetilde{F}_\pm (\lvar)|^2\, d\lvar, \textWHERE C_q := \bigl[\vts 1 + (b_\star + 1)^2 \vts\bigr]^{q/2}.
    \end{multline*}
    Since $\chi$ is the characteristic function of $[-1, 1]$ and the intervals $\lvar_\iiv + (-\delta^{3/4}, \delta^{3/4})$, $\iiv \in \bbL (\delta^{3/4})$, are pairwise disjoint by \eqref{eq:pairwise_disjoint_near_intervals}, we have
    \begin{equation}\label{eq:sum_indicators_bounded}
        \spforall \bbL' \subset \bbL(\delta^{3/4}),\ \ \bbL' \neq \emptyset, \quad \Big\|\vts \sum_{\iiv \in \bbL'} \chi \big(\vts\delta^{-3/4} (\cdot - \lvar_\iiv) \vts\big) \vts\Big\|_{L^\infty (\R)} = 1.
    \end{equation}
    As a consequence,
    \begin{equation*}
        \|\vts \Pi_{\bbL'}\, F \vts\|^2_{H^q_{\Kv \cdot \vvh_1, \xv} (\cylaug)} \lesssim C_q\, \int_\R \sum_{\pm} |\widetilde{F}_\pm (\lvar)|^2\, d\lvar \lesssim C_q\, \int_\R \|\widetilde{F} (\cdot\,, \lvar)\|^2_{L^2_{\Kv + \lvar \kvh_2}} = C_q\, \|\vts F \vts\|^2_{L^2_{\Kv \cdot \vvh_1} (\cylaug)},
    \end{equation*}
    where the last equality follows from the Plancherel-like formula \eqref{eq:directional_Plancherel}. The above estimate shows that $\Pi_{\bbL'}$ is indeed bounded from $L^2_{\Kv \cdot \vvh_1} (\cylaug)$ to $H^q_{\Kv \cdot \vvh_1, \xv} (\cylaug)$.
    \end{dem}

    \vspace{1\baselineskip} \noindent
    Even though the individual eigenmode maps $\kv \mapsto \Phi_\pm (\xv; \kv)$ are discontinuous in a neighborhood of $\Kv_\star$ (see \cite[Theorem 3.2]{fefferman2014wave} or \cite[Corollary 3.3]{fefferman2016edge}), the projection-valued map $\kv \mapsto \euler^{-\icplx\vts \kv \cdot \xv}\; \Pi^0_\kv\; \euler^{\vts \icplx\vts \kv \cdot \xv}$ is smooth, and can be expanded using Taylor series \cite{kato1995perturbation}. Using this observation, in Section \ref{sec:unperturbed_operator_expansion} (see Part $(b)$ of Proposition \ref{prop:asymptotics_H0near}), we prove the following result, on the asymptotic behavior of the near quasi-momentum projectors.

    \begin{proposition}\label{prop:asymptotics_near_projectors}
    As $\delta \to 0$, uniformly in subsets $\bbL' \subseteq \bbL(\delta^{3/4})$,
    \begin{equation}
        \Pi_{\bbL'} = \calJ^*_\delta\, \mathds{1}_{\bbL'}\, \chi (\delta^{1/4}\, D_\zeta)\, \mathds{1}_{\bbL'}\, \calJ^{}_\delta + \calO (\delta^{3/4}) \quad \textnormal{in}\ \ \scrL(L^2_{\Kv \cdot \vvh_1} (\cylaug)). \label{eq:asymptotics_Pi_near} 
    \end{equation}
    Here, $\calJ^{}_\delta$ and $\calJ^*_\delta$ are defined by \eqref{eq:mode_for_averaging_operator}--\eqref{eq:concatenated_averaging_operator}, $\mathds{1}_{\bbL'}$ denotes the block-diagonal projection operator on $\ell^2(\bbL;L^2(\R; \C^2))$ defined by \eqref{eq:block_diagonal_restriction_operator}, and $\chi (\delta^{1/4}\, D_\zeta) := \calF^*\, \chi (\delta^{1/4}\, \cdot)\, \calF$ (see \eqref{eq:def_Fourier_multiplier}), where $\calF$ denotes the Fourier transform on $L^2(\R)$.%
    \end{proposition}

    \subsection{A Schur complement reduction to near spectral components}\label{sec:schur_complement}
    The projectors $\Pi_\NEAR$ and $\Pi_\FAR$ allow us to decompose elements of $L^2_{\Kv \cdot \vvh_1} (\cylaug)$ in terms of ``near'' components and ``far'' components: 
    \begin{equation*}
    \displaystyle
    L^2_{\Kv \cdot \vvh_1} (\cylaug) = \scrX_{\NEAR} \oplus \scrX_{\FAR}, \textWHERE \scrX_{\NEAR} := \Ran(\Pi_\NEAR) \textAND \scrX_{\FAR} := \Ran(\Pi_\FAR). 
    \end{equation*}
    (Recall that the projections $\Pi_\NEAR$ and $\Pi_\FAR$, and thus the corresponding subspaces $\scrX_{\NEAR}$ and $\scrX_{\FAR}$ are $\delta$-dependent.) %
    Accordingly, $\calH^\delta_{\AUG, \Kv \cdot \vvh_1} (\mu) - E_D - \delta z$, for $z \in \C$, admits the following block representation, as an operator acting on $\scrX_{\NEAR} \oplus \scrX_{\FAR}$: 
    \begin{equation*}
    \displaystyle
    \begin{pmatrix}
        \Pi_\NEAR\, \big(\vts \calH^\delta_{\AUG, \Kv \cdot \vvh_1} (\mu) - E_D - \delta z \vts\big)\, \Pi_\NEAR %
        \ & %
        \Pi_\NEAR\, \calH^\delta_{\AUG, \Kv \cdot \vvh_1} (\mu)\, \Pi_\FAR %
        \rets
        \Pi_\FAR\, \calH^\delta_{\AUG, \Kv \cdot \vvh_1} (\mu)\, \Pi_\NEAR %
        \ & %
        \Pi_\FAR\, \big(\vts \calH^\delta_{\AUG, \Kv \cdot \vvh_1} (\mu) - E_D - \delta z \vts\big)\, \Pi_\FAR
    \end{pmatrix}.
    \end{equation*}
    For $\delta$ small, we shall implement a reduction to  $\scrX_\NEAR$ based on the Schur complement.

    \begin{lemma}[Schur complement]\label{lem:schur_complement}
    Let $\scrX, \scrY$ denote Banach spaces. Consider the block operator matrix
    \begin{equation*}
        \calM = \begin{pmatrix} \bfA & \bfB \\ \bfC & \bfD \end{pmatrix}: \DOM(\calM) = \scrX \oplus \DOM(\bfD) \subset \scrX \oplus \scrY \to \scrX \oplus \scrY,
    \end{equation*}
    where $\bfA \in \scrL(\scrX)$ and $\bfC \in \scrL(\scrX, \scrY)$ are bounded, and $\bfB$ and $\bfD$ are possibly unbounded with dense domains $\DOM(\bfD) \subseteq \DOM(\bfB) \subseteq \scrY$. Suppose that $\bfD$ is invertible on $\scrY$ with bounded inverse $\bfD^{-1} \in \scrL(\scrY, \DOM (\bfD))$. Introduce the Schur complement $\bfS := \bfA - \bfB \bfD^{-1} \bfC \in \scrL(\scrX)$. 
    Finally, suppose  $\bfS$ is invertible on $\scrX$. 

    Then, $\calM$ is invertible on $\scrX \oplus \scrY$, and its inverse, $\calM^{-1}$, is bounded from $\scrX \oplus \scrY$ to $\DOM(\calM) = \scrX \oplus \DOM(\bfD)$ and is explicitly given by:
    \begin{equation*}
        \displaystyle
        \calM^{-1} = %
        \begin{pmatrix}
        \bfS^{-1} & %
        -\bfS^{-1} \bfB \bfD^{-1} \rets %
        -\bfD^{-1} \bfC \bfS^{-1} & %
        \bfD^{-1} + \bfD^{-1} \bfC \bfS^{-1} \bfB \bfD^{-1}  
        \end{pmatrix}\ .
    \end{equation*}
    \end{lemma}

    \vspace{1\baselineskip} \noindent
    To study the resolvent of $\calH^\delta_{\AUG, \Kv \cdot \vvh_1} (\mu)$ near the Dirac energy, $E_D$, we shall apply Lemma \ref{lem:schur_complement} with Banach spaces $\scrX = \scrX_\NEAR$, $\scrY = \scrX_\FAR$, and block operators:
    \begin{equation*}
    \begin{array}{l@{\qquad}l}
        \bfA := \Pi_\NEAR\, \big(\vts \calH^\delta_{\AUG, \Kv \cdot \vvh_1} (\mu) - E_D - \delta z \vts\big)\, \Pi_\NEAR, & %
        \bfB := \Pi_\NEAR\, \calH^\delta_{\AUG, \Kv \cdot \vvh_1} (\mu)\, \Pi_\FAR,
        \ret
        \bfC := \Pi_\FAR\, \calH^\delta_{\AUG, \Kv \cdot \vvh_1} (\mu)\, \Pi_\NEAR, & %
        \bfD := \Pi_\FAR\, \big(\vts \calH^\delta_{\AUG, \Kv \cdot \vvh_1} (\mu) - E_D - \delta z \vts\big)\, \Pi_\FAR.
    \end{array}
    \end{equation*}
    By Proposition \ref{prop:smoothing_near_projectors}, we have that $\Pi_\NEAR$ is bounded from $L^2_{\Kv \cdot \vvh_1} (\cylaug)$ to $H^2_{\Kv \cdot \vvh_1, \xv} (\cylaug)$, and it follows that $\bfA \in \scrL(\scrX_\NEAR)$. On the other hand, $\bfD$ is unbounded, with domain
    \begin{equation}
    \DOM(\bfD) := H^2_{\Kv \cdot \vvh_1, \xv} (\cylaug) \cap \scrX_\FAR.\label{eq:H2far}
    \end{equation}
    Moreover, note from its definition \eqref{eq:conjugated_augmented_operator} that
    \begin{multline}
    \calH^\delta_{\AUG, \Kv \cdot \vvh_1} (\mu) = \calH^0_{\AUG, \Kv \cdot \vvh_1} - 2 \icplx\vts \mu\vts \delta\, \kvh_1 \cdot \nabla_\xv + (\mu\vts \delta)^2\, |\kvh_1|^2
    \retss
    + \delta\, \big(\vts \nabla_\xv + \icplx\vts \mu\vts \delta\vts \kvh_1 \vts\big) \cdot \Big[\vts a(\xv)\, \sigma_2\, \kappa\big(\delta\vts \kvh_2 \cdot (\xv + s \vvv_2)\big) \vts\Big]\, \big(\vts \nabla_\xv + \icplx\vts \mu\vts \delta\vts \kvh_1 \vts\big). \label{eq:conjugated_augmented_operator_alt}
    \end{multline}
    From \eqref{eq:fibered_H0aug} and the definition \eqref{eq:def_Pi_I_2} of $\Pi_\iiv$, it follows that $\calH^0_{\AUG, \Kv \cdot \vvh_1}$ commutes with each projector $\Pi_\iiv$, and therefore with $\Pi_\NEAR$. Since $\Pi_\NEAR\, \Pi_\FAR = 0$, we deduce that
    \begin{equation*}
    \Pi_\NEAR\, \calH^0_{\AUG, \Kv \cdot \vvh_1}\, \Pi_\FAR = \Pi_\FAR\, \calH^0_{\AUG, \Kv \cdot \vvh_1}\, \Pi_\NEAR = 0.
    \end{equation*}
    Therefore, \eqref{eq:conjugated_augmented_operator_alt} combined with the boundedness of $\Pi_\NEAR$ from $L^2_{\Kv \cdot \vvh_1} (\cylaug)$ to $H^2_{\Kv \cdot \vvh_1, \xv} (\cylaug)$ implies that
    \begin{subequations}\label{eq:C1_B1_O_delta}
    \begin{equation}
    \bfC \in \scrL(\scrX_\NEAR, \scrX_\FAR) \textAND \bfC = \calO (\delta) \quad \textnormal{in}\ \ \scrL(\scrX_\NEAR, \scrX_\FAR).
    \end{equation}
    Taking the adjoint bound yields
    \begin{equation}
    \bfB \in \scrL(\scrX_\FAR, \scrX_\NEAR) \textAND \bfB = \calO (\delta) \quad \textnormal{in}\ \ \scrL(\scrX_\FAR, \scrX_\NEAR).
    \end{equation}
    \end{subequations}
    The implicit constants in the bounds \eqref{eq:C1_B1_O_delta} are uniform in $(\mu, z)$ varying in compact subsets of $\R \times \C$.

    \vspace{1\baselineskip} \noindent
    To apply Lemma \ref{lem:schur_complement}, we first need to verify that $\bfD$ is boundedly invertible from $\scrX_{\FAR}$ to its domain $\DOM (\bfD)$ defined by \eqref{eq:H2far} and equipped with the norm on $H^2_{\Kv \cdot \vvh_1, \xv} (\cylaug)$. This is precisely where the omnidirectional no-fold condition, \eqref{eq:omnidirectional_no_fold}, becomes relevant. Under this assumption, the $\kvh_2$--band dispersion slices $\lvar \mapsto E_b (\Kv + \lvar \kvh_2)$ are bounded away from the Dirac energy $E_D$ whenever $\Kv + \lvar \kvh_2$ lies outside a neighborhood of $\bbK := (\Kv + \Lambda^*) \cup (\Kv' + \Lambda^*)$, the lattice of high-symmetry quasi-momenta:

    \begin{lemma}\label{lem:dispersion_curves_bounded_away_from_Dirac}
        Suppose that Assumption \ref{hyp:omnidirectional_no_fold}, namely the omnidirectional no-fold condition \eqref{eq:omnidirectional_no_fold} holds. Then, there exist constants $\delta_0, C_0, C_1, C_2 > 0$ such that 
        \begin{subequations}
            \begin{align}
            \displaystyle
            \delta_0 < \dist (\kv, \bbK)\ \ &\Longrightarrow\ \ |E_\pm (\kv) - E_D| \geq C_0\label{eq:dispersion_curves_bounded_away_from_Dirac_1}
            \rets
            \delta^{3/4} < \dist (\kv, \bbK) \leq \delta_0 \ \ &\Longrightarrow \ \ |E_\pm (\kv) - E_D| \geq C_1\vts \delta^{3/4},\label{eq:dispersion_curves_bounded_away_from_Dirac_2}
            \rets
            b \in \N \setminus \{b_\star, b_\star+1\}  \ \ &\Longrightarrow \ \ |E_b (\kv) - E_D| \geq C_2\, (1 + b), \quad \kv \in \R^2.\label{eq:dispersion_curves_bounded_away_from_Dirac_3}
            \end{align} 
        \end{subequations}
    \end{lemma}

    \begin{dem}
    {\textit{Estimate \eqref{eq:dispersion_curves_bounded_away_from_Dirac_1}.}} We proceed by contradiction: assume that for $\delta_0 > 0$ and $n \in \N$, there exists $\kv_n \in \scrB$ such that $\delta_0 < \dist(\kv_n, \bbK)$ and $|E_\pm (\kv_n) - E_D| \leq 1/n$. Since the sequence $(\kv_n)_n$ is bounded, one can extract a subsequence converging to some $\widehat{\kv} \in \scrB$ as $n \to +\infty$. It follows that $0 < \delta_0 \leq \dist(\widehat{\kv}, \bbK)$, whereas, by continuity of $\kv\mapsto E_\pm(\kv)$, we have $E_\pm (\widehat{\kv}) = E_D$. But Condition \eqref{eq:omnidirectional_no_fold} requires that $\hat\kv\in\mathbb K$, which contradicts that $\hat\kv$ is bounded away from $\mathbb K$.

    \vspace{1\baselineskip}\noindent
    {\textit{Estimate \eqref{eq:dispersion_curves_bounded_away_from_Dirac_2}.}} The conical behavior \eqref{eq:conical_behavior_Dirac} of the dispersion surfaces $E_\pm$ around the Dirac point implies the existence of constants $\delta_1, C > 0$, as well as functions $e_\pm (\kv)$ defined for $\dist (\kv, \bbK) < \delta_1$ such that 
    \begin{equation}
        |e_\pm (\kv)| \leq C \dist (\kv, \bbK), \quad E_\pm (\kv) - E_D = \pm \upsilon_D\, \dist (\kv, \bbK)\,(1 + e_\pm(\kv)).\label{eq:dispersion_curves_bounded_away_from_Dirac_2_dem}
    \end{equation}
    Since $\dist (\kv, \bbK) < \delta_1 \Longrightarrow 1 + e_\pm(\kv) \geq 1 - C \dist (\kv, \bbK)$, we choose $\delta_0 := \min(\delta_1, 1/C)/2$ so that $\dist (\kv, \bbK) \leq \delta_0 < \delta_1 \Longrightarrow 1 + e_\pm(\kv) > 1 - C \delta_0 =: \widetilde{C} > 0$. For such values of $\kv$, we deduce directly from \eqref{eq:dispersion_curves_bounded_away_from_Dirac_2_dem} that $\dist (\kv, \bbK) \in (\delta^{3/4}, \delta_0] \Longrightarrow |E_\pm (\kv) - E_D| \geq \upsilon_D\, \dist (\kv, \bbK)\, \widetilde{C} \geq C_1 \delta^{3/4}$, where $C_1 := \upsilon_D\vts \widetilde{C}  > 0$.

    \vspace{1\baselineskip}\noindent
    {\textit{Estimate \eqref{eq:dispersion_curves_bounded_away_from_Dirac_3}.}} The first step of the proof of \cite[Lemma 4.1]{drouot2019characterization} implies that if the no-fold condition \eqref{eq:omnidirectional_no_fold} holds, then $\pm (E_\pm (\kv) - E_D) \geq 0$ for $\kv \in \R^2$. Hence, \eqref{eq:dispersion_curves_bounded_away_from_Dirac_3} follows from the fact that the dispersion functions are labeled in increasing order, and from Weyl's law; see also \cite[Equation (7.21)]{fefferman2016edge}.
    \end{dem}

    \vspace{1\baselineskip} \noindent
    The invertibility of $\bfD$ then follows from the behavior of the dispersion surfaces far from the Dirac point:

    \begin{proposition}\label{prop:Hfar_is_invertible}
        Suppose the omnidirectional no-fold condition from Assumption \ref{hyp:omnidirectional_no_fold} satisfied. Consider $\mu_0 > 0$ and a compact set $S \subset \C$. There exists a constant $\delta_0 > 0$ depending on $\mu_0, S$ such that for $\delta \in (0, \delta_0)$, $|\mu| \leq \mu_0$, $z \in S$, the operator 
        \[\bfD = \Pi_\FAR\, \big(\vts \calH^\delta_{\AUG, \Kv \cdot \vvh_1} (\mu) - E_D - \delta z \vts\big)\, \Pi_\FAR\] is invertible from $\scrX_{\FAR}$ to its domain $\DOM (\bfD)$ defined by \eqref{eq:H2far}. In addition, equipping $\DOM (\bfD)$ with the $H^2_{\Kv \cdot \vvh_1, \xv} (\cylaug)$--norm, we have $\bfD^{-1} = \calO (\delta^{-3/4})$ in $\scrL(\scrX_{\FAR}, \DOM (\bfD))$, uniformly in $z \in S$ and $\mu \in [-\mu_0, \mu_0]$. 
    \end{proposition}

    \begin{dem}
        Let $(\mu, z) \in [-\mu_0, \mu_0] \times S$. By \eqref{eq:conjugated_augmented_operator_alt}, uniformly in $(\mu, z)$ varying in compact subsets,
        \begin{equation*}
            \displaystyle
            \bfD = \Pi_\FAR\, \big(\vts \calH^0_{\AUG, \Kv \cdot \vvh_1} - E_D \vts\big)\, \Pi_\FAR + \calO (\delta) \quad \textnormal{in}\ \ \scrL(\DOM (\bfD), \scrX_\FAR).
        \end{equation*}
        Therefore, it suffices to first prove that $\Pi_\FAR\, (\vts \calH^0_{\AUG, \Kv \cdot \vvh_1} - E_D \vts)\, \Pi_\FAR$ is invertible from $\scrX_\FAR$ to $\DOM (\bfD)$ with an $\calO(\delta^{-3/4})$--bounded inverse, and then conclude to the invertibility of $\bfD$ using a Neumann series argument.

        By completeness of the Floquet-Bloch modes $\{\Phi_b (\cdot\,; \Kv + \lvar \kvh_2)\}_{b \geq 1, \lvar \in \R}$ in $L^2_{\Kv \cdot \vvh_1} (\cylaug)$ (see Proposition \ref{prop:completeness_slices}), and by definition (\ref{eq:def_Pi_I_3}, \ref{eq:def_near_far_projections}) of $\Pi_\FAR$, any $F \in \scrX_\FAR$ is given by
        \begin{multline*}
            F (\xv, s) = \int_\R \chi^\delta_\FAR (\lvar)\, \sum_{\pm} \widetilde{F}_\pm (\lvar)\, \Phi_\pm (\xv; \Kv + \lvar \kvh_2)\, \euler^{\vts \icplx\vts (\Kv + \lvar \kvh_2) \cdot s \vvv_2}\, d\lvar
            \\[-4pt]
            + \int_\R \sum_{b \in \N \setminus \{b_\star, b_\star + 1\}} \widetilde{F}_b (\lvar)\, \Phi_b (\xv; \Kv + \lvar \kvh_2)\, \euler^{\vts \icplx\vts (\Kv + \lvar \kvh_2) \cdot s \vvv_2}\, d\lvar, \quad (\xv, s) \in \R^2 \times \R.
        \end{multline*}
        Here, $\widetilde{F}_- (\lvar) := \widetilde{F}_{b_\star} (\lvar)$, $\widetilde{F}_+ (\lvar) := \widetilde{F}_{b_\star + 1} (\lvar)$, where $\widetilde{F}_b (\lvar)$ is defined by \eqref{eq:def_Ftilde_b}, and
        \begin{equation*}
            \chi^\delta_\FAR (\lvar) := 1 - \sum_{\iiv \in \bbL(\delta^{3/4})} \chi \big(\delta^{-3/4} (\lvar - \lvar_\iiv)\big).
        \end{equation*}
        We introduce the function in $\scrX_\FAR$, defined by
        \begin{equation*}
            U (\xv, s) = \sum_{b \geq 1} \int_\R \widetilde{U}_b(\lvar)\, \Phi_b (\xv; \Kv + \lvar \kvh_2)\, \euler^{\vts \icplx\vts (\Kv + \lvar \kvh_2) \cdot s \vvv_2}\, d\lvar, \textFOR (\xv, s) \in \R^2 \times \R,
        \end{equation*}
        with $\widetilde{U}_- (\lvar) = \widetilde{U}_{b_\star} (\lvar)$, $\widetilde{U}_+ (\lvar) = \widetilde{U}_{b_\star + 1} (\lvar)$, and where
        \begin{equation*}
            \widetilde{U}_b (\lvar) := \left\{ 
            \begin{array}{r@{\textFOR}l}
                \displaystyle \chi^\delta_\FAR (\lvar)\, \big(\vts E_- (\Kv + \lvar \kvh_2) - E_D \vts\big)^{-1}\, \widetilde{F}_- (\lvar) & b = b_\star,
                \rets
                \displaystyle \chi^\delta_\FAR (\lvar)\, \big(\vts E_+ (\Kv + \lvar \kvh_2) - E_D \vts\big)^{-1}\, \widetilde{F}_+ (\lvar) & b = b_\star + 1,
                \rets
                \displaystyle \big(\vts E_b (\Kv + \lvar \kvh_2) - E_D \vts\big)^{-1}\, \widetilde{F}_b (\lvar) & b \not\in \{b_\star, b_\star + 1\}.
            \end{array}
            \right.
        \end{equation*}
        It follows from Lemma \ref{lem:dispersion_curves_bounded_away_from_Dirac} that for $\delta$ small enough,
        \begin{equation*}
            |\widetilde{U}_\pm (\lvar)| \underset{(\ref{eq:dispersion_curves_bounded_away_from_Dirac_1}, \ref{eq:dispersion_curves_bounded_away_from_Dirac_2})}{\lesssim} \delta^{-3/4}\, |\widetilde{F}_\pm (\lvar)| \textAND (1 + b)\, |\widetilde{U}_b (\lvar)| \underset{\eqref{eq:dispersion_curves_bounded_away_from_Dirac_3}}{\lesssim} |\widetilde{F}_b (\lvar)| \quad \spforall b \neq \pm,
        \end{equation*}
        where the implied constants are independent of $\lvar$. These estimates, squared and summed over $b \geq 1$, $\lvar \in \R$, lead to
        \begin{equation*}
            \int_\R \sum_{b \geq 1} (1 + b)^2\, |\widetilde{U}_b (\lvar)|^2 \lesssim \delta^{-3/2}\, \int_\R \sum_{b \geq 1} |\widetilde{F}_b (\lvar) |^2,
        \end{equation*}
        or equivalently, using the characterization \eqref{eq:link_augmented_Sobolev_regularity_decay_FB_coeffs}:
        \begin{equation}\label{eq:Hfar_is_invertible_dem_1}
        \|\vts U \vts\|_{H^2_{\Kv \cdot \vvh_1, \xv} (\cylaug)} \lesssim \delta^{-3/4}\, \|\vts F \vts\|_{L^2_{\Kv \cdot \vvh_1} (\cylaug)}.
        \end{equation}
        In particular, $U \in H^2_{\Kv \cdot \vvh_1, \xv} (\cylaug)$, and from the spectral representation \eqref{eq:spectral_representation_H0aug} of $\calH^0_{\AUG, \Kv \cdot \vvh_1}$, we see that
        \begin{equation}\label{eq:Hfar_is_invertible_dem_2}
            \Pi_\FAR\, \big(\vts \calH^0_{\AUG, \Kv \cdot \vvh_1} - E_D \vts\big)\, \Pi_\FAR\, U = F.
        \end{equation}
        Equations (\ref{eq:Hfar_is_invertible_dem_1}, \ref{eq:Hfar_is_invertible_dem_2}) imply that $\Pi_\FAR\, (\vts \calH^0_{\AUG, \Kv \cdot \vvh_1} - E_D \vts)\, \Pi_\FAR$ is boundedly invertible from $\scrX_\FAR$ to $H^2_{\Kv \cdot \vvh_1, \xv} (\cylaug) \cap \scrX_\FAR = \DOM (\bfD)$ with an $\calO(\delta^{-3/4})$ inverse.
    \end{dem}

    \vspace{1\baselineskip} \noindent
    Now that the invertibility of $\bfD$ has been established, it remains to study the invertibility of the Schur complement $\bfS := \bfA - \bfB \bfD^{-1} \bfC$ on $\scrX_\NEAR$. Using Proposition \ref{prop:Hfar_is_invertible} and \eqref{eq:C1_B1_O_delta}, one has
    \begin{align}
    \bfS &= \Pi_\NEAR\, \big(\vts \calH^\delta_{\AUG, \Kv \cdot \vvh_1} (\mu) - E_D - \delta z \vts\big)\, \Pi_\NEAR + \calO (\delta^{1 - 3/4 + 1}) &\textnormal{in}\ \ \scrL(\scrX_\NEAR) \nonumber
    \\
    &=\Pi_\NEAR\, \big(\vts \calH^\delta_{\AUG, \Kv \cdot \vvh_1} (\mu) - E_D - \delta z \vts\big)\, \Pi_\NEAR + \calO (\delta^{5/4}) &\textnormal{in}\ \ \scrL(\scrX_\NEAR),\label{eq:link_schur_complement_H_near}
    \end{align}
    uniformly in $(\mu, z)$ varying in compact subsets of $\R \times \C$. In particular, it suffices to study the invertibility of $\Pi_\NEAR\, (\vts \calH^\delta_{\AUG, \Kv \cdot \vvh_1} (\mu) - E_D - \delta z \vts)\, \Pi_\NEAR$, and then pass to $\bfS$ using a Neumann series argument. This is the object of the next results.

    Recall the block-diagonal effective Dirac operator, $\calD^{\vts \delta} (\mu) $ (see \eqref{eq:block_diagonal_dirac_operator}), emerging from the multiple scale asymptotic analysis of Section \ref{sec:ansatz_and_multi_scale_analysis}.
    \begin{proposition}\label{prop:asymptotics_HdeltaNEAR}
        Consider $\mu_0 > 0$ and a compact set $S \subset \C$. There exists $\delta_0 > 0$, depending on $\mu_0$ and $S$, such that for all $0<\delta<\delta_0$, the following operator expansions hold, with error terms which are  uniform in $\mu \in [-\mu_0, \mu_0]$, $z \in S$ and $\bbL', \bbL'' \subseteq \bbL (\delta^{3/4})$:
        \begin{enumerate}[label={$(\alph*).$}, ref={\theproposition.$(\alph*)$}]
            \item\label{prop:asymptotics_HdeltaNEAR_diagonal} The asymptotic behavior of $\Pi_{\bbL'}\, (\vts \calH^\delta_{\AUG, \Kv \cdot \vvh_1} (\mu) - E_D - \delta z \vts)\, \Pi_{\bbL'}$ is governed by a band-limited block-diagonal effective Dirac operator. More precisely, 
            \begin{multline*}
            \Pi_{\bbL'}\, \big(\vts \calH^\delta_{\AUG, \Kv \cdot \vvh_1} (\mu) - E_D - \delta z \vts\big)\, \Pi_{\bbL'} 
            \retss
            = \delta\, \calJ^*_\delta\, \chi (\delta^{1/4}\, D_\zeta)\, \mathds{1}_{\bbL'}\, \big( \calD^{\vts \delta} (\mu) - z \big)\, \mathds{1}_{\bbL'}\, \chi (\delta^{1/4}\, D_\zeta)\, \calJ^{}_\delta + \calO (\delta^{3/2}) \quad \textnormal{in}\ \ \scrL(\scrX_{\bbL'}).
            \end{multline*}
            We recall that $\scrX_{\bbL'} = \Ran (\Pi_{\bbL'})$ and $\chi (\delta^{1/4}\, D_\zeta) := \calF^*\, \chi (\delta^{1/4}\, \cdot)\, \calF$; see also \eqref{eq:def_Fourier_multiplier}.
            \item\label{prop:asymptotics_HdeltaNEAR_off_diagonal} If $\bbL'$ and $\bbL''$ are disjoint, then
            \begin{equation*}
            \Pi_{\bbL'}\, \big(\vts \calH^\delta_{\AUG, \Kv \cdot \vvh_1} (\mu) - E_D - \delta z \vts\big)\, \Pi_{\bbL''} = \calO (\delta^{7/4}) \quad \textnormal{in}\ \ \scrL(\vts\scrX_{\bbL''},\, \scrX_{\bbL'}\vts).
            \end{equation*}
        \end{enumerate}
    \end{proposition}

    \noindent
    Proposition \ref{prop:asymptotics_HdeltaNEAR} describes the asymptotic behavior of $\Pi_{\bbL'}\, (\vts \calH^\delta_{\AUG, \Kv \cdot \vvh_1} (\mu) - E_D - \delta z \vts)\, \Pi_{\bbL'}$. Its inverse is expanded in the next result.
    \begin{proposition}\label{prop:resolvent_expansion_HtildeDeltaNEAR}
    Fix $\mu_0, \eta > 0$ and a compact set $S \subset \C$. Suppose that
    \begin{multline}\label{eq:z-constraint}
        \spexists \delta_0, C > 0, \quad \spforall \delta \in (0, \delta_0), \quad \bbL' \subseteq \bbL(\delta^{3/4}), \quad \mu \in [-\mu_0, \mu_0], \quad z \in S,
        \retss
        \big\|\vts \mathds{1}_{\bbL'}\, \big(\vts \calD^{\vts \delta} (\mu) - z \vts\big)^{-1}\, \mathds{1}_{\bbL'}\; \calJ^{}_\delta\; \Pi_{\bbL'}\vts\big\|_{\scrL(\vts L^2_{\Kv \cdot \vvh_1} (\cylaug), \ell^2 (\bbL; H^1 (\R; \C^2))\vts)} \leq C/\eta.
    \end{multline}
    Then $\Pi_{\bbL'}\, (\vts \calH^\delta_{\AUG, \Kv \cdot \vvh_1} (\mu) - E_D - \delta z \vts)\, \Pi_{\bbL'}$ is invertible on $\scrX_{\bbL'}$, and its inverse admits the following expansion:
    \begin{multline}\label{eq:resolvent_on_near_components}
        \Pi_{\bbL'}\, \bigg(\vts \frac{\calH^\delta_{\AUG, \Kv \cdot \vvh_1} (\mu) - E_D}{\delta} - z \vts\bigg)^{-1}\! \Pi_{\bbL'} %
        \\
        = \Pi_{\bbL'}\, \calJ_\delta^*\, \mathds{1}_{\bbL'}\, \big(\vts \calD^{\vts \delta} (\mu) - z \vts\big)^{-1}\, \mathds{1}_{\bbL'}\, \calJ^{}_\delta\; \Pi_{\bbL'} + \frac{1}{\eta}\, \calO (\delta^{1/4}) \quad \textnormal{in}\ \ \scrL(\vts\scrX_{\bbL'}\vts).
    \end{multline}
    \end{proposition}
    \begin{remark}\label{rmk:dist_to_spectrum_not_equivalent_to_bound_on_D}
        \begin{enumerate}[label={$(\alph*).$}, ref={\theremark.$(\alph*)$}, wide=0pt]
            \item\label{rmk:dist_to_spectrum_implies_bound_on_D} The condition $z \in \C$ with $\dist (z, \spec \calD^{\vts \delta} (\mu)) > \eta$ implies \eqref{eq:z-constraint}. In fact, if $\dist (z, \spec \calD^{\vts \delta} (\mu)) > \eta$, then the resolvent $(\calD^\delta (\mu) - z)^{-1}$ is bounded in $\ell^2 (\bbL; L^2 (\R; \C^2))$ with a norm less than $1/\eta$. Moreover, the boundedness of this inverse from $\ell^2 (\bbL; L^2 (\R; \C^2))$ to $\ell^2 (\bbL; H^1 (\R; \C^2))$ follows from classical graph estimates. Combining this with the boundedness \eqref{eq:T_J_delta_are_bounded} of $\calJ_\delta$, we deduce \eqref{eq:z-constraint} with $C := \|\calJ_\delta\| = \|\calT\|$ which is independent of $\delta$.
            
            \item The condition $\dist (z, \spec \calD^{\vts \delta} (\mu)) > \eta$ implies \eqref{eq:z-constraint} (as explained in the first part of this remark) but the converse is \underline{not} true. In fact, in subsequent work \cite{amenoagbadji2026dense}, under a Diophantine condition on the edge parameter $r$, we shall prove that \eqref{eq:z-constraint} holds even when $z$ is an eigenvalue of the effective Dirac operator $\calD^{\Kv_\iiv} (\mu + \delta^{-1}\, \gamma_\iiv)$ such that $\iiv \not\in \bbL'$. This result will be used in the construction of augmented edge states.
        \end{enumerate}
    \end{remark}

    \vspace{1\baselineskip} \noindent
    Finally, we use the next result to reformulate the dominant term in \eqref{eq:resolvent_on_near_components} for $\bbL' = \bbL (\delta^{3/4})$. 
    \begin{proposition}\label{prop:dominant_term_without_near_projectors}
    Consider $\eta, \mu_0 > 0$, a compact set $S \subset \C$, and assume that there exists $\delta_0 > 0$ such that
    \begin{equation*}
        \delta \in (0, \delta_0), \quad \mu \in [-\mu_0, \mu_0], \quad \kpar = \Kv \cdot \vvh_1 + \delta \mu, \quad z \in S, \quad \dist \big(z,\; \spec \calD^{\vts \delta} (\mu)\big) > \eta,
    \end{equation*}
    as in Theorem \ref{thm:resolvent_expansion}. Then, as $\delta \to 0$,
    \begin{multline*}
        \Pi_\NEAR\, \calJ^*_\delta\; \mathds{1}_{\bbL (\delta^{3/4})}\; \big(\vts \calD^{\vts \delta} (\mu) - z \vts\big)^{-1}\; \mathds{1}_{\bbL (\delta^{3/4})}\; \calJ^{}_\delta\, \Pi_\NEAR%
        \\
        = \calJ^*_\delta\; \mathds{1}_{\bbL (\delta^{3/4})}\; \big(\vts \calD^{\vts \delta} (\mu) - z \vts\big)^{-1}\; \mathds{1}_{\bbL (\delta^{3/4})}\; \calJ^{}_\delta + \frac{1}{\eta}\, \calO (\delta^{1/4}) \quad \textnormal{in}\ \ \scrL(L^2_{\Kv \cdot \vvh_1} (\cylaug)).
    \end{multline*}
    In other words, the near projectors can be removed from the left-hand side, at the cost of introducing an $\calO(\delta^{1/4})$--error. 
    \end{proposition}

    \vspace{1\baselineskip} \noindent
    Proposition \ref{prop:asymptotics_HdeltaNEAR} is the analog of \cite[Proposition 6.4]{drouot2020edge}, while Propositions \ref{prop:resolvent_expansion_HtildeDeltaNEAR} and \ref{prop:dominant_term_without_near_projectors} together form an analog of \cite[Proposition 6.5]{drouot2020edge}. Their proofs are presented in Section \ref{sec:asymptotic_expansions}. We now show that Theorem \ref{thm:resolvent_expansion} follows from these results.

    \vspace{0\baselineskip} \noindent
    \begin{dem}[of Theorem \ref{thm:resolvent_expansion} using Propositions \ref{prop:asymptotics_HdeltaNEAR}, \ref{prop:resolvent_expansion_HtildeDeltaNEAR}, and \ref{prop:dominant_term_without_near_projectors}]
    The assumption on $z$ in Theorem \ref{thm:resolvent_expansion} implies \eqref{eq:z-constraint}; see Remark \ref{rmk:dist_to_spectrum_implies_bound_on_D}. Thus, by applying Proposition \ref{prop:resolvent_expansion_HtildeDeltaNEAR} with $\bbL' = \bbL(\delta^{3/4})$, we obtain that $\Pi_{\NEAR}\, (\vts \calH^\delta_{\AUG, \Kv \cdot \vvh_1} (\mu) - E_D - \delta z \vts)\, \Pi_{\NEAR}$ is invertible on $\scrX_\NEAR$. Using a Neumann series argument, we deduce from \eqref{eq:link_schur_complement_H_near}, for $\delta$ sufficiently small,  that $\bfS$ is invertible on $\scrX_\NEAR$ and the following holds in $\scrL(\scrX_\NEAR)$:
    \begin{align}
        (\bfS/\delta)^{-1} &= \bigg[\vts \Pi_\NEAR\, \bigg(\vts \frac{\calH^\delta_{\AUG, \Kv \cdot \vvh_1} (\mu) - E_D}{\delta} - z \vts\bigg)\, \Pi_\NEAR + \calO_{\scrL(\scrX_\NEAR)} (\delta^{1/4}) \vts\bigg]^{-1} \nonumber
        \\
        &= \Pi_\NEAR\, \bigg(\vts \frac{\calH^\delta_{\AUG, \Kv \cdot \vvh_1} (\mu) - E_D}{\delta} - z \vts\bigg)^{-1}\, \Pi_\NEAR + \calO (\delta^{1/4}) \nonumber 
        \\
        &= \Pi_\NEAR\, \calJ^*_\delta\, \mathds{1}_{\bbL (\delta^{3/4})}\, \big(\vts \calD^{\vts \delta} (\mu) - z \vts\big)^{-1}\, \mathds{1}_{\bbL (\delta^{3/4})}\, \calJ^{}_\delta\, \Pi_\NEAR + \frac{1}{\eta}\, \calO (\delta^{1/4}). \label{eq:inverse_schur_complement_inverse_H_near}
    \end{align}
    The last equality follows from the expansion \eqref{eq:resolvent_on_near_components}.
    In particular, $\bfS^{-1} = \calO (\delta^{-1})$ in $\scrL(\scrX_\NEAR)$. Moreover, since by \eqref{eq:C1_B1_O_delta} and Proposition \ref{prop:Hfar_is_invertible},
    \begin{equation*}
        \bfC = \calO_{\scrL(\scrX_\NEAR, \scrX_{\FAR})} (\delta), \quad \bfB = \calO_{\scrL(\mathscr{H}^2_{\FAR}, \scrX_\NEAR)} (\delta), \textAND \bfD_1^{-1} = \calO_{\scrL(\scrX_{\FAR}, \mathscr{H}^2_\FAR)}(\delta^{-3/4}),
    \end{equation*}
    Lemma \ref{lem:schur_complement} implies the following in $\scrL(L^2_{\Kv \cdot \vvh_1} (\cylaug))$:
    \begin{equation}
        (\calM/\delta)^{-1} = \bigg( \frac{\calH^\delta_{\AUG, \Kv \cdot \vvh_1} (\mu) - E_D}{\delta} - z \bigg)^{-1}  = \begin{pmatrix}  (\bfS/\delta)^{-1} & 0 \\ 0 & 0 \end{pmatrix} + \calO (\delta^{1/4}).\label{eq:inverse_M_inverse_schur_complement}
    \end{equation}
    Substituting \eqref{eq:inverse_schur_complement_inverse_H_near} in \eqref{eq:inverse_M_inverse_schur_complement}, and using Proposition \ref{prop:dominant_term_without_near_projectors} to remove the projectors $\Pi_\NEAR$, we obtain in $\scrL(L^2_{\Kv \cdot \vvh_1} (\cylaug))$ the following:
    \begin{equation*}
        \bigg( \frac{\calH^\delta_{\AUG, \Kv \cdot \vvh_1} (\mu) - E_D}{\delta} - z \bigg)^{-1}\!\! = \calJ^*_\delta\, \mathds{1}_{\bbL (\delta^{3/4})}\, \big(\vts \calD^{\vts \delta} (\mu) - z \vts\big)^{-1}\, \mathds{1}_{\bbL (\delta^{3/4})}\, \calJ^{}_\delta + \frac{1}{\eta}\, \calO (\delta^{1/4}).
    \end{equation*}
    By combining this expansion with the definition \eqref{eq:conjugated_augmented_operator_alt} of $\calH^\delta_{\AUG, \Kv \cdot \vvh_1} (\mu)$ in terms of $\calH^\delta_{\AUG, \Kv \cdot \vvh_1 + \delta \mu}$, we deduce the resolvent expansion for general edges, Theorem \ref{thm:resolvent_expansion}.
    \end{dem}

    \section{Asymptotic expansions and estimates near the Dirac point}\label{sec:asymptotic_expansions}
    \noindent
    In this final but technical section, we carefully derive the asymptotic expansions and estimates stated in Propositions \ref{prop:asymptotics_near_projectors}, \ref{prop:asymptotics_HdeltaNEAR}, \ref{prop:resolvent_expansion_HtildeDeltaNEAR}, and \ref{prop:dominant_term_without_near_projectors}. We follow the strategy of \cite[Section 6]{drouot2020edge}. To simplify the presentation, we set
    \begin{equation*}
    \mu = 0, \quad \textnormal{so that} \quad \calH^\delta_{\AUG, \Kv \cdot \vvh_1} (\mu) = \calH^\delta_{\AUG, \Kv \cdot \vvh_1}.
    \end{equation*}
    The forthcoming calculations extend to any $\mu$ varying in a compact subset of $\R$. 

    In Section \ref{sec:spectral_localization_tools}, we present some properties of the operators $\calT_\iiv$ and $\calJ_{\delta, \iiv}$. Section \ref{sec:unperturbed_operator_expansion} contains the proof of Proposition \ref{prop:asymptotics_near_projectors} on the asymptotic behavior of near projectors, as well as part of the proof of Proposition \ref{prop:asymptotics_HdeltaNEAR} on the asymptotic behavior of the operator $\calH^\delta_{\AUG, \Kv \cdot \vvh_1}$ projected on near components. The proof of Proposition \ref{prop:asymptotics_HdeltaNEAR} is concluded in Section \ref{sec:proof_asymptotics_HtildeDeltaNEAR}. Finally, in Section \ref{sec:proof_resolvent_expansion_HtildeDeltaNEAR}, we prove the resolvent estimates in Propositions \ref{prop:resolvent_expansion_HtildeDeltaNEAR} and \ref{prop:dominant_term_without_near_projectors}, therefore completing the proof of Theorem \ref{thm:resolvent_expansion}.

    \subsection{Spectral localization tools: properties of \texorpdfstring{$\calT_\iiv$}{Ti} and \texorpdfstring{$\calJ_{\delta, \iiv}$}{Jiδ}}\label{sec:spectral_localization_tools}
    We refer to Section \ref{sec:main_results} and Equations \eqref{eq:mode_for_averaging_operator}--\eqref{eq:def_calJ} for the definitions of $\calT_\iiv$, $\calJ_{\delta, \iiv}$, $\calT$, $\calJ_\delta$, and their adjoints. In this section, we summarize some useful properties of these operators, via the following proposition. Recall that the Fourier transform $\calF: g(\zeta) \in L^2(\R) \mapsto \calF g (\lvar) \in L^2(\R)$ is given by \eqref{eq:def_unitary_Fourier_transform}, and the partial Fourier transform $\scrF_\Kv: F (\xv, s) \in L^2_{\Kv \cdot \vvh_1} (\cylaug) \mapsto \widetilde{F} (\xv, \lvar) \in L^2 (\R_\lvar; L^2_{\Kv + \lvar \kvh_2})$ is defined by \eqref{eq:directional_Fourier}. 

    \begin{proposition}[Properties of \texorpdfstring{$\calT_\iiv$}{T_i} and \texorpdfstring{$\calJ_{\delta, \iiv}$}{J_iδ}]\, \label{prop:pties_averaging_operator}
    Let $\iiv = (\Kv_\iiv, m_\iiv) \in \bbL = \{\Kv,\Kv'\} \times \Z$ and further let $\lvar_\iiv$ and $\modK_\iiv (\lvar) := \gamma_\iiv \kvh_1 + (\lvar - \lvar_\iiv)\, \kvh_2$ be as in the Wrapping Lemma \ref{lem:wrapping_lemma}.  
    \begin{enumerate}[label={$(\alph*).$}, ref={\theproposition.$(\alph*)$}, wide=0pt]
        \item Expression for $\calF\, \calT_\iiv$:  Let   $\lvar-\lvar_\iiv\in[-\pi,\pi)$. Then, for all $ F \in L^2_{\Kv \cdot \vvh_1} (\cylaug),$
        \begin{equation}
        (\calF\, \calT_\iiv)\, F\, (\lvar - \lvar_\iiv) = \left\langle \euler^{\icplx\, \modK_\iiv (\lvar) \cdot \xv}\; \Phi^{\Kv_\iiv}(\xv),\; \scrF_\Kv{F}(\xv,\lvar)  \right\rangle_{L_\xv^2(\Omega)} 
        \label{eq:averaging_operator_alt}
        \end{equation}
        \item Expression for $\calT_\iiv^*\calF^*$: For  $\lvar\in\R$ and   all $ g \in L^2(\R; \C^2)$:
        \begin{equation}
        \scrF_\Kv\, (\calT^*_\iiv\, \calF^*)\, g(\xv, \lvar) = \euler^{\vts \icplx\vts \modK_\iiv (\lvar) \cdot \xv}\; \Phi^{\Kv_\iiv} (\xv)^\transp\, g(\lvar - \lvar_\iiv), \quad \xv\in \R^2\ .\label{eq:averaging_operator_alt_adjoint}
        \end{equation}
        \item\label{prop:pties_averaging_operator:orthogonality_relations}(Orthogonality relations) $\calT^*_\iiv$ and $\calJ^*_{\delta, \iiv}$ are isometries:
        \begin{equation}\label{eq:J_left_inverse_is_J_star}
        \calT^{}_{\iiv}\, \calT^*_{\iiv} = \calJ^{}_{\delta, \iiv}\, \calJ^*_{\delta, \iiv} = \Id_{L^2(\R; \C^2)}\ .
        \end{equation}

        \item\label{prop:pties_averaging_operator:square_summability}(Boundedness of $\calT$ and $\calJ_\delta$) For any $q \in \N_0$, the column operators $\calT$ and $\calJ_\delta$ are bounded from the space $H^q_{\Kv \cdot \vvh_1, \xv} (\cylaug)$ defined by \eqref{eq:def_Hqaug_edge} to the space $\ell^2_q (\bbL; L^2(\R; \C^2))$ defined by \eqref{eq:def_ell2_L_X}, uniformly in $\delta$: in fact, there exists a constant $C_q > 0$ depending on the $\scrC^q$--norm of $(\Phi^{\Kv}, \Phi^{\Kv'})^\transp$ such that
        \begin{equation}\label{eq:square_summability_estimate}
        \spforall F \in H^q_{\Kv \cdot \vvh_1, \xv} (\cylaug), \quad \sum_{\iiv \in \bbL} (1 + m^2_\iiv)^q\; \|\vts \calT_\iiv\, F\vts \|^2_{L^2(\R)} \leq C_q\, \|F\|^2_{H^q_{\Kv \cdot \vvh_1, \xv} (\cylaug)}.
        \end{equation}
    \end{enumerate}
    \end{proposition}

    \begin{dem}
    \textit{$(a)$ and $(b)$. Relations satisfied by $\calF\, \calT_\iiv$ and its adjoint}. It suffices to establish the expression \eqref{eq:averaging_operator_alt_adjoint}; the relation \eqref{eq:averaging_operator_alt} follows by a duality argument.
    To prove \eqref{eq:averaging_operator_alt_adjoint}, we replace $g(\cdot)$ by $\mathcal F^*g(\cdot)$ in the definition of $\calT^*_\iiv$, and use \eqref{eq:mode_for_averaging_operator}, which defines $\varphi_\iiv (\xv, s)$ in terms of $\Phi^{\Kv_\iiv}=(\Phi^{\Kv_\iiv}_1, \Phi^{\Kv_\iiv}_2)$:
    \begin{equation*}
        (\calT^*_\iiv\; \calF^*\, g) (\xv, s) = %
        \euler^{\vts\icplx\, \scalebox{0.7}{$\big($} \gamma_\iiv\, \kvh_1 \cdot \xv + (\lvar_\iiv\, \kvh_2 + \Kv) \cdot s\vts \vvv_2 \scalebox{0.7}{$\big)$} }\ \Phi^{\Kv_\iiv}(\xv)^\top\ (\mathcal F^*g)\big(\kvh_2 \cdot (\xv + s \vvv_2)\big).
    \end{equation*}
    Applying the partial Fourier transform $\scrF_\Kv$, \eqref{eq:directional_Fourier}, yields 
    \begin{align*}
        &\scrF_\Kv\, (\calT^*_\iiv\, \calF^*\, g) (\xv, \lvar) = \frac{1}{\sqrt{2\pi}} \int_{\R} (\calT^*_\iiv\; \calF^*\, g) (\xv, s)\, \euler^{- \icplx\vts (\Kv + \lvar \kvh_2) \cdot s \vvv_2}\, d s \\
        &= \euler^{\vts \icplx\vts \gamma_\iiv\, \kvh_1 \cdot \xv}\, \Phi^{\Kv_\iiv} (\xv)^\transp\, 
        \frac{1}{\sqrt{2\pi}} \int_{\R} (\calF^*\, g)\vts \big(\kvh_2 \cdot (\xv + s \vvv_2)\big)\
        \euler^{\vts\icplx\, (\lvar_\iiv\kvh_2+\Kv)\cdot s\vvv_2 }\ \euler^{- \icplx\vts (\Kv + \lvar \kvh_2) \cdot s \vvv_2}\, ds\\
        &= \euler^{\vts \icplx\vts \gamma_\iiv\, \kvh_1 \cdot \xv}\, \Phi^{\Kv_\iiv} (\xv)^\transp\, \ 
        \frac{1}{\sqrt{2\pi}} \int_{\R} (\calF^*\, g)\vts ( s + \kvh_2 \cdot \xv )\
        \euler^{-\icplx\vts (\lvar - \lvar_\iiv) \vts s }\, ds\\
        &=  \euler^{\vts \icplx\vts (\gamma_\iiv\, \kvh_1 
        + (\lvar-\lvar_\iiv)\kvh_2)\cdot \xv}\,\ \Phi^{\Kv_\iiv} (\xv)^\transp\, \frac{1}{\sqrt{2\pi}}\, \int_{\R} (\calF^*\, g)\vts (t)\ \euler^{-\icplx\vts (\lvar - \lvar_\iiv)\vts t}\, dt\ =\  \euler^{\icplx\modK_\iiv (\lvar)\cdot \xv}\ 
        \Phi^{\Kv_\iiv} (\xv)^\transp\ g(\lvar-\lvar_\iiv)\ .
    \end{align*}
    In the last equality we used the relation \eqref{eq:def-l_I}: %
    $\modK_\iiv(\lvar) =  \displaystyle\gamma_\iiv\, \kvh_1 + (\lvar - \lvar_\iiv)\, \kvh_2\in \widehat\Omega$ %
    of the Wrapping Lemma \ref{lem:wrapping_lemma}. This completes the proof of \eqref{eq:averaging_operator_alt_adjoint}.

    \vspace{1\baselineskip} \noindent
    \textit{$(c)$. Orthogonality relations}. \eqref{eq:J_left_inverse_is_J_star} follows from the definition of $\calT_\iiv$ and its adjoint, and from the fact that $\{\Phi^{\Kv_\iiv}_1, \Phi^{\Kv_\iiv}_2\}$ is orthonormal; see Condition \ref{item:dirac_point_item_2} in Proposition \ref{prop:sufficient_conditions_dirac_point}.

    \vspace{1\baselineskip} \noindent
    \textit{$(d)$. Bound \eqref{eq:square_summability_estimate}}. Let $F \in H^q_{\Kv \cdot \vvh_1, \xv} (\cylaug)$. Fix $\Kv_\star \in \{\Kv, \Kv'\}$, and let $\iiv \in \{\Kv_\star\} \times \Z$. By substituting \eqref{eq:representation_K_plus_lk2}, \emph{i.e.} the relation $\Kv + \lvar \kvh_2 = \Kv_\star + \modK_\iiv (\lvar) + \floK_\iiv$, in the relation \eqref{eq:averaging_operator_alt}, we find
    \begin{multline}
        \calF\, \calT_\iiv\, F\, (\lvar - \lvar_\iiv) = \widehat{F}^{\Kv_\star}_{\floK_\iiv} (\lvar), \textWHERE \floK_\iiv \in \Lambda^* \textAND
        \rets
        \widehat{F}^{\Kv_\star}_{\kv} (\lvar) := \int_{\Omega} \euler^{\vts \icplx\, \kv \cdot \xv}\; \left[\vts \euler^{\vts\icplx\, \Kv_\star \cdot \xv}\; \overline{\Phi^{\Kv_\star} (\xv)} \vts\right] \left[\vts\euler^{-\icplx\, (\Kv + \lvar \kvh_2) \cdot \xv}\; \scrF_\Kv F(\xv, \lvar) \vts\right]\, d\xv, \quad \spforall \kv \in \Lambda^*.\label{eq:Fourier_coefficient_Fhat}
    \end{multline}
    Since $\Phi^{\Kv_\star}$ is $\Kv_\star$-pseudo-periodic and $\xv \mapsto \scrF_\Kv F (\xv, \lvar)$ is $\Kv + \lvar \kvh_2$-pseudo-periodic with respect to $\Lambda$, the bracketed terms in the integrand are $\Lambda$-- periodic functions of $\xv$.  \emph{Hence, for $\lvar \in \R$, $\widehat{F}^{\Kv_\star}_{\kv} (\lvar)$ is the Fourier coefficient with index $\kv$ of a $\Lambda$--periodic function.} Moreover, from its definition \eqref{eq:def_floK}, the family $(\floK_\iiv)_{\iiv \in \{\Kv_\star\} \times \Z}$ is pairwise disjoint in $\Lambda^*$. Therefore, there exists a constant $C > 0$ independent of $F$, such that
    {
    \allowdisplaybreaks
    \begin{align*}
        \sum_{\iiv \in \{\Kv_\star\} \times \Z} (1 + m^2_\iiv)^q\; \| \calT_\iiv\, F\, \|^2_{L^2(\R)} &= \sum_{\iiv \in \{\Kv_\star\} \times \Z} (1 + m^2_\iiv)^q\; \| \calF\, \calT_\iiv\, F\, \|^2_{L^2(\R)} &&\!\!\textnormal{by Plancherel's theorem}
        \\
        &=\sum_{\iiv \in \{\Kv_\star\} \times \Z} (1 + m^2_\iiv)^q\; \| \widehat{F}^{\Kv_\star}_{\floK_\iiv} \|^2_{L^2(\R)} &&\!\!\textnormal{by \eqref{eq:Fourier_coefficient_Fhat}}
        \\
        &\leq C\, \sum_{\iiv \in \{\Kv_\star\} \times \Z} (1 + |\floK_\iiv|^2)^q\; \| \widehat{F}^{\Kv_\star}_{\floK_\iiv} \|^2_{L^2(\R)} &&\!\!\textnormal{because $\floK_\iiv \cdot \vvv_2 = 2 \pi\vts m_\iiv$}
        \\
        &\leq C \sum_{\kv \in \Lambda^*} (1 + |\kv|^2)^q\; \| \widehat{F}^{\Kv_\star}_\kv \|^2_{L^2(\R)} 
        \\
        &= C\, \int_\R \sum_{\kv \in \Lambda^*} (1 + |\kv|^2)^q\; |\widehat{F}^{\Kv_\star}_{\kv} (\lvar)|^2\, d\lvar,
    \end{align*}
    where the last equality is obtained by interchanging the sum over $\kv \in \Lambda^*$ and the integral over $\lvar \in \R$, via Tonelli's theorem. Now, note that $\scrF_\Kv F(\cdot\,, \lvar) \in H^q_{\Kv + \lvar \kvh_2}$, since $\scrF_\Kv$ is isomorphic from $H^q_{\Kv \cdot \vvh_1, \xv} (\cylaug)$ to $L^2(\R_\lvar; H^q_{\Kv + \lvar \kvh_2})$, according to Lemma \ref{lem:directional_Fourier}. Therefore, using the characterization \eqref{eq:link_Sobolev_regularity_decay_Fourier_coeffs} of elements of $H^q (\R^2 / \Lambda)$ in terms of their Fourier series, and the fact that $\Phi^{\Kv_\star} \in \scrC^\infty(\R^2)$, we deduce the estimate}
    \begin{align*}
        \sum_{\iiv \in \{\Kv_\star\} \times \Z} (1 + m^2_\iiv)^q\; &\| \calT_\iiv\, F\, \|^2_{L^2(\R)} 
        \\
        &\underset{\eqref{eq:link_Sobolev_regularity_decay_Fourier_coeffs}}{\lesssim} \int_\R \left\| \xv \mapsto \left[\vts \euler^{\vts\icplx\, \Kv_\star \cdot \xv}\, \overline{\Phi^{\Kv_\star} (\xv)} \vts\right] \left[\vts\euler^{-\icplx\, (\Kv + \lvar \kvh_2) \cdot \xv}\, \scrF_\Kv F(\xv, \lvar) \vts\right] \right\|^2_{H^q (\R^2 / \Lambda)}\, d\lvar
        \retss
        &\lesssim \|\Phi^{\Kv_\star}\|^2_{\scrC^q}\; \|\scrF_\Kv F\|^2_{L^2(\R_\lvar; H^q_{\Kv + \lvar \kvh_2})}\ \lesssim\ \|\Phi^{\Kv_\star}\|^2_{\scrC^q}\; \|F\|^2_{H^q_{\Kv \cdot \vvh_1, \xv} (\cylaug)}.
    \end{align*}
    For the last bound, we reused the boundedness of $\scrF_\Kv$ from $H^q_{\Kv \cdot \vvh_1, \xv} (\cylaug)$ to $L^2(\R_\lvar; H^q_{\Kv + \lvar \kvh_2})$. Summing this bound over $\Kv_\star \in \{\Kv, \Kv'\}$ leads to \eqref{eq:square_summability_estimate}, which proves the boundedness of $\calT$. Moreover since $\calJ_\delta = \calU_\delta\, \calT$ and $\calU_\delta$ is bounded uniformly in $\delta$, it follows that $\calJ_\delta$ is bounded uniformly in $\delta$.
    \end{dem}

    \subsection{Proof of Proposition \ref{prop:asymptotics_near_projectors} and expansion of \texorpdfstring{$\calH^0_{\AUG, \Kv \cdot \vvh_1}$}{H0AUG} on near components}\label{sec:unperturbed_operator_expansion}
    In this section, we study the asymptotic behavior of the unperturbed operator $\calH^0_{\AUG, \Kv \cdot \vvh_1}$, projected on near components. More precisely, in Proposition \ref{prop:asymptotics_H0near}, we derive an expansion for $\Pi_{\bbL'}\; (\vts \calH^0_{\AUG, \Kv \cdot \vvh_1} - E_D \vts)^d\; \Pi_{\bbL''}$, where $d = 0, 1$, and $\bbL', \bbL''$ are subsets of $\bbL (\delta^{3/4})$. The motivation behind this expansion is two-fold: first, for $d = 0$, it provides an expansion for the near projectors, hence proving Proposition \ref{prop:asymptotics_near_projectors}. Second, for $d = 1$, it captures the effective Dirac Hamiltonian character of the unperturbed operator $\calH^0_{\AUG, \Kv \cdot \vvh_1}$, when spectrally localized near the Dirac (conical) point. The latter result will be used in the subsequent proof of Proposition \ref{prop:asymptotics_HdeltaNEAR} in Section \ref{sec:proof_asymptotics_HtildeDeltaNEAR}.
    \begin{proposition}\label{prop:asymptotics_H0near}
    Let $d \in \{0, 1\}$.
    \begin{enumerate}[label={$(\alph*).$}, ref={\theproposition.$(\alph*)$}]
        \item Assume $\bbL', \bbL'' \subseteq \bbL (\delta^{3/4})$ and $\bbL'\cap \bbL'' =\emptyset$. Then, $\Pi_{\bbL'}\; (\vts \calH^0_{\AUG, \Kv \cdot \vvh_1} - E_D \vts)^d\; \Pi_{\bbL''} = 0.$
        \item As $\delta \to 0$, the following holds in $\scrL(L^2_{\Kv \cdot \vvh_1} (\cylaug))$, uniformly in $\bbL' \subset \bbL (\delta^{3/4})$:
        \begin{multline}\label{eq:asymptotics_H0near}
        \Pi_{\bbL'}\; \big(\vts \calH^0_{\AUG, \Kv \cdot \vvh_1} - E_D \vts\big)^d\; \Pi_{\bbL'} 
        \retss
        = \delta^d \sum_{\iiv \in \bbL'} \calJ^*_{\delta, \iiv}\, \chi (\delta^{1/4}\, D_\zeta)\, \big[\vts \sigma^{\Kv_\iiv} (\kvh_2)\, D_\zeta + \delta^{-1}\, \gamma_\iiv\, \sigma^{\Kv_\iiv} (\kvh_1) \vts\big]^d\, \chi (\delta^{1/4}\, D_\zeta)\, \calJ^{}_{\delta, \iiv} + \calO \big(\delta^{\frac{3}{4} (d+1)} \big),
        \end{multline}
        where $\calJ_{\delta, \iiv}$ is given by \eqref{eq:def_calJ}, and $\sigma^{\Kv_\iiv} (\kvh_j)$ is the $2\times2$ Hermitian matrix defined in  \eqref{eq:def_twisted_Pauli_matrix_1}--\eqref{eq:def_twisted_Pauli_matrix_2}. In particular, for $d = 0$,  \eqref{eq:asymptotics_H0near} reduces to the expansion of  $\Pi_{\bbL'}$, displayed in \eqref{eq:asymptotics_Pi_near} of Proposition \ref{prop:asymptotics_near_projectors}.
    \end{enumerate}
    \end{proposition}

    \begin{dem}
    Part $(a)$ of Proposition \ref{prop:asymptotics_H0near} derives from the fact that $\calH^0_{\AUG, \Kv \cdot \vvh_1}$ commutes with $\Pi_\iiv$, and $\Pi_\iiv\, \Pi_\jjv = 0$ for $\iiv \neq \jjv$; see \eqref{eq:PiIPiJeq0}. This also implies that
    \begin{equation}\label{eq:asymptotics_H0near_dem1}
        \Pi_{\bbL'}\; \big(\vts \calH^0_{\AUG, \Kv \cdot \vvh_1} - E_D \vts\big)^d\; \Pi_{\bbL'} = \sum_{\iiv \in \bbL'} \Pi_\iiv\; \big(\vts \calH^0_{\AUG, \Kv \cdot \vvh_1} - E_D \vts\big)^d\; \Pi_\iiv.
    \end{equation}
    We now turn to the proof of  Part $(b)$ of Proposition \ref{prop:asymptotics_H0near}. We first show that the operator $\Pi_{\bbL'}\; (\calH^0_{\AUG, \Kv \cdot \vvh_1} - E_D)^d\; \Pi_{\bbL'}$ is approximately equal to a sum of operators: $\sum_{\iiv \in \bbL'} \calN_\iiv$, where each $\calN_\iiv$ is given explicitly in terms of a Dirac eigenspace projection. The latter expression is then shown, for $\delta$ small, to be equal to the dominant term in \eqref{eq:asymptotics_H0near}.

    \vspace{1\baselineskip} \noindent
    \textbf{Step 1.} \textit{Expansion near the Dirac point}. Fix $\iiv \in \bbL$. From the integral representation \eqref{eq:fibered_H0aug} of $\calH^0_{\AUG, \Kv \cdot \vvh_1}$ in terms of $\calH^0_{\Kv + \lvar \kvh_2}$, and from the definition \eqref{eq:def_Pi_I_1} of $\Pi_\iiv$, it follows that
    \begin{align*}
        &\Pi_\iiv\; \big(\vts \calH^0_{\AUG, \Kv \cdot \vvh_1} - E_D \vts\big)^d\; \Pi_\iiv\, F
        \retss
        &= \scrF^*_\Kv\,  \biggl\{\vts (\xv, \lvar) \mapsto \chi \big(\delta^{-3/4} (\lvar - \lvar_\iiv)\big)\, \Big[\vts \Pi^0_{\Kv + \lvar \kvh_2}\; \big(\vts \calH^0_{\Kv + \lvar \kvh_2} - E_D \vts\big)^d\; \Pi^0_{\Kv + \lvar \kvh_2}\, \widetilde{F} (\cdot\,, \lvar)\vts \Big] (\xv) \vts\biggr\}.
    \end{align*}
    To expand this operator, we use the following asymptotic expansion for $d = 0, 1$, shown for instance in \cite[Equation ($6.9$) ($d = 0$) and Step $2$ of Proposition $6.6$ ($d = 1$)]{drouot2020edge}. For any $\Kv_\star \in \{\Kv, \Kv'\}$, as $|\modK| \to 0$, there holds in $\scrL(L^2_{\Kv_\star + \modK})$:
    \begin{equation}\label{eq:asymptotics_Pi0k_H_E_Pi0k}
        \Pi^0_{\Kv_\star + \modK}\; (\calH^0_{\Kv_\star + \modK} - E_D)^d\; \Pi^0_{\Kv_\star + \modK} = \euler^{\vts \icplx\, \modK \cdot \xv}\, \Pi^0_{\Kv_\star}\; (\vts -2\vts \icplx\vts \modK \cdot \nabla \vts)^d \; \Pi^0_{\Kv_\star}\, \euler^{- \icplx\, \modK \cdot \xv} + \calO (|\modK|^{d+1}).
    \end{equation}
    Let $\lvar \in \R$. By Part \ref{lem:veps_neighborhood_K_d} of the Wrapping Lemma \ref{lem:wrapping_lemma}, we have $\Kv + \lvar \kvh_2 = \Kv_\iiv + \modK_\iiv (\lvar) + \floK_\iiv$, where $\floK_\iiv \in \Lambda^*$. Therefore, using the $\Lambda^*$--periodicity of $\kv \mapsto (\Pi^0_\kv, \calH^0_\kv)$, we have
    \begin{equation}
        \Pi^0_{\smash{\Kv + \lvar \kvh_2\vphantom{\Kv_\iiv + \modK_\iiv (\lvar)}}}\; %
        \big(\vts \calH^0_{\smash{\Kv + \lvar \kvh_2\vphantom{\Kv_\iiv + \modK_\iiv (\lvar)}}} - E_D \vts\big)^d\; %
        \Pi^0_{\smash{\Kv + \lvar \kvh_2\vphantom{\Kv_\iiv + \modK_\iiv (\lvar)}}} = %
        \Pi^0_{\smash{\Kv_\iiv + \modK_\iiv (\lvar)}}\; \big(\vts \calH^0_{\smash{\Kv_\iiv + \modK_\iiv (\lvar)}} - E_D \vts\big)^d\; \Pi^0_{\smash{\Kv_\iiv + \modK_\iiv (\lvar)}}.
    \end{equation}
    Now, assume $\iiv \in \bbL (\delta^{3/4})$ and $|\lvar - \lvar_\iiv| \leq \delta^{3/4}$, so that $\modK_\iiv (\lvar) = \gamma_\iiv\, \kvh_1 + (\lvar - \lvar_\iiv)\, \kvh_2 = \calO(\delta^{3/4})$. Then, applying the expansion \eqref{eq:asymptotics_Pi0k_H_E_Pi0k} with $\modK := \modK_\iiv (\lvar)$ leads to the following in $\scrL(L^2_{\Kv + \lvar \kvh_2})$:
    \begin{gather}
        \Pi^0_{\Kv + \lvar \kvh_2}\; \big(\vts \calH^0_{\Kv + \lvar \kvh_2} - E_D \vts\big)^d\; \Pi^0_{\Kv + \lvar \kvh_2} = \widetilde{\calN}_\iiv (\lvar) + \calO \big(\delta^{\frac{3}{4} (d+1)} \big), \quad |\lvar-\lvar_\iiv| \leq \delta^{3/4},\ \iiv\in \bbL(\delta^{3/4}), \nonumber
        \rets
        \textWITH \widetilde{\calN}_\iiv(\lvar) := \euler^{\vts \icplx\, \modK_\iiv (\lvar) \cdot \xv}\, \Pi^0_{\Kv_\iiv}\; \big(\vts -2\vts \icplx\vts \modK_\iiv (\lvar) \cdot \nabla_\xv \vts\big)^d \; \Pi^0_{\Kv_\iiv}\, \euler^{- \icplx\, \modK_\iiv (\lvar) \cdot \xv}, \label{eq:asymptotics_Pi0k_p_lvark2_H_E_Pi0k_p_lvark2}
    \end{gather}
    and where the implicit constant in the error term can be chosen to be independent of $\iiv$. Introducing the operator $\calN_\iiv$ defined for $F \in L^2_{\Kv \cdot \vvh_1} (\cylaug)$ by 
    \begin{equation}\label{eq:def_N_I}
        \calN_\iiv\, F := \scrF^*_\Kv\,  \Bigl\{\vts (\xv, \lvar) \mapsto \chi \big(\delta^{-3/4} (\lvar - \lvar_\iiv)\big)\, \big[ \widetilde{\calN}_\iiv(\lvar)\, \widetilde{F} (\cdot\,, \lvar) \vts \big] (\xv) \vts\Bigr\},
    \end{equation}
    it follows from the Plancherel-like formula \eqref{eq:directional_Plancherel} and the estimates (\ref{eq:sum_indicators_bounded}, \ref{eq:asymptotics_Pi0k_p_lvark2_H_E_Pi0k_p_lvark2}) that
    \begin{align*}
        &\Big\|\vts \sum_{\iiv \in \bbL'} \Big[\vts \Pi_\iiv\; \big(\vts \calH^0_{\AUG, \Kv \cdot \vvh_1} - E_D \vts\big)^d\; \Pi_\iiv - \calN_\iiv \vts\Big]\vts F\, \Big\|^2_{L^2_{\Kv \cdot \vvh_1} (\cylaug)}
        \\
        &\underset{\eqref{eq:directional_Plancherel}}{=} \int_\R \Big|\sum_{\iiv \in \bbL'} \!\chi \big(\delta^{-3/4} (\lvar - \lvar_\iiv)\big) \Big|^2\, \Big\|\vts \Big[\vts \Pi^0_{\Kv + \lvar \kvh_2}\, \big(\vts \calH^0_{\Kv + \lvar \kvh_2} - E_D \vts\big)^d\, \Pi^0_{\Kv + \lvar \kvh_2} - \widetilde{\calN}_\iiv(\lvar) \vts\Big] \widetilde{F} (\cdot\,, \lvar) \vts\Big\|^2_{L^2_{\Kv + \lvar \kvh_2}}\!\!\!\!\!\! d\lvar
        \\
        &\underset{(\ref{eq:sum_indicators_bounded}, \ref{eq:asymptotics_Pi0k_p_lvark2_H_E_Pi0k_p_lvark2})}{\lesssim} \delta^{\frac{3}{2}(d+1)} \int_\R \|\vts \widetilde{F} (\cdot\,, \lvar) \vts\|^2_{L^2_{\Kv + \lvar \kvh_2}}\, d\lvar \underset{\eqref{eq:directional_Plancherel}}{=} \delta^{\frac{3}{2}(d+1)}\; \|\vts F \vts\|^2_{L^2_{\Kv \cdot \vvh_1} (\cylaug)}.
    \end{align*}
    Combining this bound with \eqref{eq:asymptotics_H0near_dem1}, we deduce the following expansion in $\scrL(L^2_{\Kv \cdot \vvh_1} (\cylaug))$, uniformly in $\bbL' \subseteq \bbL (\delta^{3/4})$:
    \begin{equation}\label{eq:asymptotics_H0near_dem3}
        \Pi_{\bbL'}\, \big(\vts \calH^0_{\AUG, \Kv \cdot \vvh_1} - E_D \vts\big)^d\, \Pi_{\bbL'} = \sum_{\iiv \in \bbL'} \calN_\iiv + \calO \big( \delta^{\frac{3}{4}(d+1)} \big).
    \end{equation}
    It remains to prove that $\sum_{\iiv \in \bbL'} \calN_\iiv$ coincides with the first expression on the right-hand side of \eqref{eq:asymptotics_H0near}. This is the object of the next step.

    \vspace{1\baselineskip} \noindent
    \textbf{Step 2.} \textit{Reformulation of the dominant operator $\calN_\iiv$}. Let $\iiv \in \bbL' \subseteq \bbL (\delta^{3/4})$ be arbitrary. In this step, we prove the identity:
    \begin{equation}\label{eq:asymptotics_H0near_step2_goal}
        \calN_{\iiv}\, F = \delta^d\, \calJ^*_{\delta, \iiv}\, \chi(\delta^{1/4}\, D_\zeta)\; \big[\vts \sigma^{\Kv_\iiv} (\kvh_2)\, D_\zeta + \delta^{-1}\, \gamma_\iiv\, \sigma^{\Kv_\iiv} (\kvh_1) \vts\big]^d\; \chi(\delta^{1/4}\, D_\zeta)\, \calJ^{}_{\delta, \iiv}\, F.
    \end{equation}
    Then Part $(b)$ of Proposition \ref{prop:asymptotics_H0near} will follow from summing \eqref{eq:asymptotics_H0near_step2_goal} over $\iiv \in \bbL'$, and using the expansion \eqref{eq:asymptotics_H0near_dem3} of Step 1.

    Let $F \in L^2_{\Kv \cdot \vvh_1} (\cylaug)$ and recall that $\widetilde{F} := \scrF_\Kv F$ (see \eqref{eq:directional_Fourier}). We begin by reformulating the function $\widetilde{\calN}_\iiv (\lvar)\, \widetilde{F} (\cdot\,, \lvar)$ defined by \eqref{eq:asymptotics_Pi0k_p_lvark2_H_E_Pi0k_p_lvark2}, using the alternative expression \eqref{eq:Pi_0_K_star} for $\Pi^0_{\Kv_\iiv}$:
    \begin{align}
        &\big[ \widetilde{\calN}_\iiv (\lvar)\, \widetilde{F} (\cdot\,, \lvar) \big] (\xv) \nonumber
        \retss
        &= \euler^{\vts \icplx\, \modK_\iiv (\lvar) \cdot \xv}\, \Phi^{\Kv_\iiv} (\xv)^\transp\, \big\langle \Phi^{\Kv_\iiv}, (-2\vts \icplx\vts \modK_\iiv (\lvar) \cdot \nabla)^d\; \Phi^{\Kv_\iiv} \big\rangle_{L^2_{\Kv_\iiv}}\, \big\langle\vts  \euler^{\vts \icplx\, \modK_\iiv (\lvar) \cdot \yv}\, \Phi^{\Kv_\iiv} (\yv), \, \widetilde{F}(\yv, \lvar) \vts\big\rangle_{L^2_\yv (\Omega)}. \nonumber
        \intertext{Recall that $\modK_\iiv (\lvar) = \gamma_\iiv\, \kvh_1 + (\lvar - \lvar_\iiv)\, \kvh_2$. The product $\langle\vts \Phi^{\Kv_\iiv}, (-2\vts \icplx\vts \modK_\iiv (\lvar) \cdot \nabla)^d\; \Phi^{\Kv_\iiv} \vts\rangle \in \C^{2 \times 2}$ is the identity matrix when $d = 0$, by orthonormality of $\{\Phi^{\Kv_\iiv}_1, \Phi^{\Kv_\iiv}_2\}$. For $d = 1$, it equals the matrix $\sigma^{\Kv_\iiv} (\modK_\iiv (\lvar)) = (\lvar - \lvar_\iiv)\, \sigma^{\Kv_\iiv} (\kvh_2) + \gamma_\iiv\, \sigma^{\Kv_\iiv} (\kvh_1)$ defined by \eqref{eq:def_twisted_Pauli_matrix_1}--\eqref{eq:def_twisted_Pauli_matrix_2}, by Proposition \ref{prop:expression_Phi_l_nabla_Phi_j}. Therefore,}
        &\big[ \widetilde{\calN}_\iiv (\lvar)\, \widetilde{F} (\cdot\,, \lvar) \big] (\xv) \label{eq:asymptotics_H0near_dem4}
        \retss
        &= \euler^{\vts \icplx\, \modK_\iiv (\lvar) \cdot \xv}\, \Phi^{\Kv_\iiv} (\xv)^\transp\, \big[\vts (\lvar - \lvar_\iiv)\, \sigma^{\Kv_\iiv} (\kvh_2) + \gamma_\iiv\, \sigma^{\Kv_\iiv} (\kvh_1) \vts\big]^d\, \big\langle\vts  \euler^{\vts \icplx\, \modK_\iiv (\lvar) \cdot \yv}\, \Phi^{\Kv_\iiv} (\yv), \, \widetilde{F}(\yv, \lvar) \vts\big\rangle_{L^2_\yv (\Omega)}.\nonumber
    \end{align}
    Multiplying each side of \eqref{eq:asymptotics_H0near_dem4} by $\chi (\delta^{-3/4} (\lvar - \lvar_\iiv))$ yields an expression for $\scrF_\Kv\, \calN_\iiv\, F (\xv, \lvar)$, where $\calN_\iiv$ is defined by \eqref{eq:def_N_I}. Further, using the formulas (\ref{eq:averaging_operator_alt}, \ref{eq:averaging_operator_alt_adjoint}) satisfied by $\calT_\iiv$, \eqref{eq:averaging_operator}, and its adjoint, \eqref{eq:averaging_operator_adjoint}, we obtain
    \begin{align*}
        &\scrF_\Kv\, \calN_\iiv\, F (\xv, \lvar) 
        \retss
        &\underset{\eqref{eq:averaging_operator_alt}}{=} \euler^{\vts \icplx\, \modK_\iiv (\lvar) \cdot \xv}\, \Phi^{\Kv_\iiv} (\xv)^\transp\, \big[\vts (\lvar - \lvar_\iiv)\, \sigma^{\Kv_\iiv} (\kvh_2) + \gamma_\iiv\, \sigma^{\Kv_\iiv} (\kvh_1) \vts\big]^d\, \chi \big(\delta^{-3/4} (\lvar - \lvar_\iiv)\big)\, \calF\, \calT^{}_\iiv\, F\, (\lvar - \lvar_\iiv)
        \retss
        &\underset{\eqref{eq:averaging_operator_alt_adjoint}}{=} \scrF^{}_\Kv\, \calT^*_\iiv\, \calF^* \Big\{\vts \widehat{\lvar} \mapsto \big[\vts \widehat{\lvar}\, \sigma^{\Kv_\iiv} (\kvh_2) + \gamma_\iiv\, \sigma^{\Kv_\iiv} (\kvh_1) \vts\big]^d\,\chi (\delta^{-3/4}\, \widehat{\lvar}) \vts\Big\}\, \calF\, \calT^{}_\iiv\, F\, (\xv, \lvar)
        \retss
        &= \scrF^{}_\Kv\, \calT^*_\iiv\, \big[\vts \sigma^{\Kv_\iiv} (\kvh_2)\, D_\zeta + \gamma_\iiv\, \sigma^{\Kv_\iiv} (\kvh_1) \vts\big]^d\, \chi(\delta^{-3/4}\, D_\zeta)\, \calT^{}_\iiv\, F\, (\xv, \lvar).
    \end{align*}
    Using the relation $f(D_\zeta) = \calU_\delta^*\, f(\delta\, D_\zeta)\, \calU^{}_\delta$ with $f(\lvar) := [\vts \sigma^{\Kv_\iiv} (\kvh_2)\, \lvar + \gamma_\iiv\, \sigma^{\Kv_\iiv} (\kvh_1) \vts]^d\, \chi(\delta^{-3/4}\, \lvar)$, we deduce that
    \begin{gather*}
        \scrF_\Kv\, \calN_\iiv\, F (\xv, \lvar) = \scrF_\Kv\, \calT^*_\iiv\, \calU^*_\delta\, \chi(\delta^{1/4}\, D_\zeta)\, \big[\vts \delta\, \sigma^{\Kv_\iiv} (\kvh_2)\, D_\zeta + \gamma_\iiv\, \sigma^{\Kv_\iiv} (\kvh_1) \vts\big]^d\, \chi(\delta^{1/4}\, D_\zeta)\, \calU_\delta\, \calT_\iiv\, F\, (\xv, \lvar)
        \retss
        = \delta^d\, \scrF^{}_\Kv\, \calJ^*_{\delta, \iiv}\, \chi(\delta^{1/4}\, D_\zeta)\, \big[\vts \sigma^{\Kv_\iiv} (\kvh_2)\, D_\zeta + \delta^{-1}\, \gamma_\iiv\, \sigma^{\Kv_\iiv} (\kvh_1) \vts\big]^d\, \chi(\delta^{1/4}\, D_\zeta)\, \calJ^{}_{\delta, \iiv}\, F\, (\xv, \lvar),
    \end{gather*}
    where the last equality follows from the definition of $\calJ_{\delta, \iiv} := \calU_\delta\, \calT_\iiv$. Finally, applying $\scrF^*_\Kv$ leads to \eqref{eq:asymptotics_H0near_step2_goal}.
    \end{dem}

    \subsection{Proof of Proposition \ref{prop:asymptotics_HdeltaNEAR}: expansion of the domain wall contribution}\label{sec:proof_asymptotics_HtildeDeltaNEAR}
    In Section \ref{sec:contribution_domain_wall}, we study the contributions from the domain wall perturbation, when projected on near components; see Proposition \ref{prop:DW_contributions}. In Section \ref{sec:completion_expansion_HdeltaNEAR}, we will combine this with results from Section \ref{sec:unperturbed_operator_expansion}, and complete the proof of Proposition \ref{prop:asymptotics_HdeltaNEAR}.

    \subsubsection{Expansion of $\ \Pi_{\bbL'}\, (\vts \nabla_\xv \cdot a(\xv)\, \sigma_2\, \kappa_\delta (\xv, s)\, \nabla_\xv \vts)\, \Pi_{\bbL''}$\ ; contribution from the domain wall}\label{sec:contribution_domain_wall}
    Define 
    \begin{equation*}
    \kappa_\delta (\xv, s) := \kappa \big(\delta\vts \kvh_2 \cdot (\xv + s \vvv_2) \big).
    \end{equation*}
    The main result of this section is the following.
    \begin{proposition}\label{prop:DW_contributions}
    As $\delta \to 0$, the following expansion holds in $\scrL(\scrX_{\bbL''}, \scrX_{\bbL'})$, uniformly in subsets $\bbL', \bbL'' \subseteq \bbL (\delta^{3/4})$:
    \begin{multline}\label{eq:K_to_K_dw_coupling}
        \Pi_{\bbL'}\, \big(\vts \nabla_\xv \cdot a(\xv)\, \sigma_2\, \kappa_\delta (\xv, s)\, \nabla_\xv \vts\big)\, \Pi_{\bbL''}
        \retss
        = \sum_{\iiv \in \bbL' \cap \bbL''} \calJ^*_{\delta, \iiv}\, \chi (\delta^{1/4}\, D_\zeta)\, \vartheta^{\Kv_\iiv}\, \sigma_3\, \kappa (\zeta)\, \chi (\delta^{1/4}\, D_\zeta)\, \calJ^{}_{\delta, \iiv} + \calO (\delta^{3/4}).
    \end{multline}
    \end{proposition}

    \noindent
    The main difficulty in proving Proposition \ref{prop:DW_contributions} is that, unlike $\calH^0_{\AUG, \Kv \cdot \vvh_1} - E_D$, the domain wall term $\nabla_\xv \cdot a(\xv)\, \sigma_2\, \kappa_\delta (\xv, s)\, \nabla_\xv$ does \emph{not} commute with $\Pi_\iiv$. As a result, it couples elements lying in the ranges of distinct projectors $\Pi_\iiv$ and $\Pi_\jjv$ for $\iiv \neq \jjv$. We address this issue in two steps, starting with the following lemma.

    \begin{lemma}\label{lem:K_to_K_dw_coupling_in_between}
    As $\delta \to 0$, the following holds in $\scrL(\scrX_{\bbL''}, \scrX_{\bbL'})$, uniformly $\bbL', \bbL'' \subseteq \bbL (\delta^{3/4})$:
    \begin{multline}\label{eq:K_to_K_dw_coupling_in_between}
        \Pi_{\bbL'}\, \big(\vts \nabla_\xv \cdot a(\xv)\, \sigma_2\, \kappa_\delta (\xv, s)\, \nabla_\xv \vts\big)\, \Pi_{\bbL''}
        \retss
        = \sum_{\iiv \in \bbL', \jjv \in \bbL''} \calJ^*_{\delta, \iiv}\, \chi (\delta^{1/4}\, D_\zeta)\, M_{\iiv, \jjv}\, \euler^{\vts \icplx\, \delta^{-1}\, (\lvar_\jjv - \lvar_\iiv)\, \zeta}\, \kappa (\zeta)\, \chi (\delta^{1/4}\, D_\zeta)\, \calJ^{}_{\delta, \jjv} + \calO (\delta^{3/4}).
    \end{multline}
    Here $M_{\iiv, \jjv} \in \C^{2 \times 2}$ is defined for any $\iiv, \jjv \in \bbL$ as
    \begin{equation}\label{eq:def_M_I_J}
        M_{\iiv, \jjv} := \int_\Omega \euler^{\vts \icplx\, (\Kv_\iiv + \floK_\iiv - \Kv_\jjv - \floK_\jjv) \cdot \xv}\; \overline{\Phi^{\Kv_\iiv} (\xv)}\, \nabla \cdot a (\xv)\, \sigma_2\, \nabla \Phi^{\Kv_\jjv} (\xv)^\transp\, d\xv,
    \end{equation}
    where $\floK_\iiv \in \Lambda^*$ is introduced in the Wrapping Lemma \ref{lem:wrapping_lemma}; see \eqref{eq:def_floK}.
    \end{lemma}

    \begin{dem}
    By Proposition \ref{prop:asymptotics_near_projectors}, $\Pi_{\bbL'} = \calJ^*_\delta\, \mathds{1}_{\bbL'}\, \chi (\delta^{1/4}\, D_\zeta)\, \mathds{1}_{\bbL'}\, \calJ^{}_\delta + \calO (\delta^{3/4})$ in $\scrL(L^2_{\Kv \cdot \vvh_1} (\cylaug))$, uniformly in $\bbL' \subseteq \bbL (\delta^{3/4})$. By writing the same expansion for $\Pi_{\bbL''}$, and by using the boundedness properties of $\Pi_{\bbL'}, \Pi_{\bbL''}$  (Proposition \ref{prop:smoothing_near_projectors}), we deduce the following asymptotic expansion in $\scrL(\scrX_{\bbL''}, \scrX_{\bbL'})$, uniformly in $\bbL', \bbL'' \subseteq \bbL (\delta^{3/4})$:
    \begin{equation}\label{eq:dem_DW_0}
        \Pi_{\bbL'}\, \big(\vts \nabla_\xv \cdot a(\xv)\, \sigma_2\, \kappa_\delta (\xv, s)\, \nabla_\xv \vts\big)\, \Pi_{\bbL''} = \calJ^*_\delta\, \mathds{1}_{\bbL'}\, \calA^\delta\, \mathds{1}_{\bbL''}\, \calJ^{}_\delta + \calO (\delta^{3/4}).
    \end{equation}
    Here, $\calA^\delta \in \scrL(\ell^2 (\bbL; L^2(\R; \C^2)))$ is the block operator whose entries are given by
    \begin{equation}\label{eq:def_A_I_J_dem_DW}
        \calA^\delta_{\iiv, \jjv} := \left\{%
        \begin{array}{cl}
        \displaystyle \chi (\delta^{1/4}\, D_\zeta)\, \calJ^{}_{\delta, \iiv}\, \nabla_\xv\, a (\xv)\, \sigma_2\, \kappa_\delta (\xv, s)\, \nabla_\xv\, \calJ^*_{\delta, \jjv}\, \chi (\delta^{1/4}\, D_\zeta), & (\iiv, \jjv) \in \bbL' \times \bbL'',
        \ret
        0 & (\iiv, \jjv) \not\in \bbL' \times \bbL''.
        \end{array}
        \right.
    \end{equation}
    In what follows, we decompose $\calA^\delta_{\iiv, \jjv}$, with the goal of extracting the dominant term in \eqref{eq:K_to_K_dw_coupling_in_between}. 

    \vspace{1\baselineskip}\noindent
    \textbf{Step 1.} \textit{Preliminary decomposition of $\calA^\delta$}. We introduce the notations: for $g \in L^2 (\R; \C^2)$ and $\jjv \in \bbL$,
    \begin{align}
        a(\xv)\, \sigma_2 &:= (\widetilde{a}_{j, l} (\xv))_{j, l = 1, 2}, \nonumber
        \retss
        \xi_\jjv (\xv, s) &:= \gamma_\jjv\, \kvh_1 \cdot \xv + (\lvar_\jjv + \Kv \cdot \vvv_2)\vts s \ \ \textnormal{so that}\ \ \varphi_\jjv (\xv, s) \underset{\eqref{eq:mode_for_averaging_operator}}{=} \euler^{\vts \icplx\, \xi_\jjv (\xv, s)}\; \Phi^{\Kv_\jjv} (\xv), \nonumber
        \retss
        \Upsilon^\delta_{n, \jjv} (D_\zeta) &:= \delta\, \Upsilon^\delta_\jjv (0; D_\zeta) \cdot \bm{e}_n := \icplx \left(\gamma_\jjv \kvh^{(n)}_1 + \delta\, \kvh^{(n)}_2 D_\zeta \right), \quad n = 1, 2, \label{eq:def_Upsilon_l_J}
        \rets
        g_\delta (\zeta) &:= \chi (\delta^{1/4}\, D_\zeta)\, g (\zeta). \nonumber
    \end{align}
    Here, $\Upsilon^\delta_\jjv (\mu; D_\zeta)$ is the operator introduced in \eqref{eq:diff_op_slow_variable}, and $\bm{e}_1 := (1, 0)^\transp$, $\bm{e}_2 := (0, 1)^\transp$. Note that $g_\delta \in \scrC^\infty (\R)$, since $\chi (\delta^{1/4}\, \cdot)$ is compactly supported. We also require the following identity for future reference: for $u (\xv, \zeta) \in \scrC^1 (\R^3; \C^2)$, $f (\zeta) \in \scrC^1 (\R; \C^2)$, and $l' = 1, 2$, by the product rule,
    \begin{multline}
        U(\xv, s) = \euler^{\vts \icplx\, \xi_\jjv (\xv, s)}\; u (\xv, s)^\transp\, f (\zeta) \big|_{\zeta = \delta\, \kvh_2 \cdot (\xv + s \vvv_2)}
        \rets
        \begin{aligned}[b]
        \Longrightarrow \ \ \partial_{x_{l'}} U (\xv, s) &= \euler^{\vts \icplx\, \xi_\jjv (\xv, s)} \; \partial_{x_{l'}} u (\xv, s)^\transp f (\zeta) \big|_{\zeta = \delta\, \kvh_2 \cdot (\xv + s \vvv_2)}
        \\
        &+ \euler^{\vts \icplx\, \xi_\jjv (\xv, s)}\; u (\xv, s)^\transp \left(\vts \Upsilon^\delta_{l', \jjv} (D_\zeta)\, f \vts\right)\vts (\zeta) \big|_{\zeta = \delta\, \kvh_2 \cdot (\xv + s \vvv_2)}.
        \end{aligned}\label{eq:multi_scale_product_rule}
    \end{multline}
    Let $(\iiv, \jjv) \in \bbL' \times \bbL''$, $j, l \in \{1,2\}$ and $g \in L^2 (\R; \C^2)$. With the above notations, it follows from the definition $\calJ^*_{\delta, \jjv} := \calT^*_\jjv\, \calU^*_\delta$ and from the definition (\ref{eq:mode_for_averaging_operator}, \ref{eq:averaging_operator_adjoint}) of $\calT^*_\jjv$ that 
    \begin{equation*}
        \calJ^*_{\delta, \jjv}\, g_\delta\vts (\xv, s) = \delta^{1/2}\, \euler^{\vts \icplx\, \xi_\jjv (\xv, s)}\; \Phi^{\Kv_\jjv} (\xv)^\transp\, g_\delta (\zeta) \big|_{\zeta = \delta\, \kvh_2 \cdot (\xv + s \vvv_2)}.
    \end{equation*}
    Differentiating the above expression with respect to $x_l$, and applying the product rule \eqref{eq:multi_scale_product_rule} with $u(\xv, s) = \Phi^{\Kv_\jjv} (\xv)$, $f = g_\delta$, $l' = l$, we obtain:
    \begin{multline}
        \partial_{x_l}\, \calJ^*_{\delta, \jjv}\, g_\delta \vts (\xv, s) = \delta^{1/2}\, \euler^{\vts \icplx\, \xi_\jjv (\xv, s)} \; \partial_{x_l} \Phi^{\Kv_\jjv} (\xv)^\transp g_\delta (\zeta) \big|_{\zeta = \delta\, \kvh_2 \cdot (\xv + s \vvv_2)}
        \retss
        + \delta^{1/2}\, \euler^{\vts \icplx\, \xi_\jjv (\xv, s)}\; \Phi^{\Kv_\jjv} (\xv)^\transp \left(\vts \Upsilon^\delta_{l, \jjv} (D_\zeta)\, g_\delta \vts\right) (\zeta) \big|_{\zeta = \delta\, \kvh_2 \cdot (\xv + s \vvv_2)}. \label{eq:dem_DW_1}
    \end{multline}
    Multiplying \eqref{eq:dem_DW_1} by $\widetilde{a}_{j, l} (\xv)\, \kappa_\delta (\xv, s)$ yields an expression with two terms on the right-hand side. The first term is of the form $U (\xv, s)$ in \eqref{eq:multi_scale_product_rule} with $u(\xv, s) = \widetilde{a}_{j, l} (\xv)\, \kappa_\delta (\xv, s)\, \partial_{x_l} \Phi^{\Kv_\jjv} (\xv)$ and $f = g_\delta$. The second term is also of the form $U (\xv, s)$ with $u (\xv, s) = \widetilde{a}_{j, l} (\xv)\, \kappa_\delta (\xv, s)\, \Phi^{\Kv_\jjv} (\xv)$ and $f = \Upsilon^\delta_{l, \jjv} (D_\zeta)\, g_\delta$. Therefore, we can differentiate these two terms with respect to $x_j$ and use the product rule \eqref{eq:multi_scale_product_rule} with $l' = j$ to obtain
    \begin{align}
        &\displaystyle \partial_{x_j} \left(\vts \widetilde{a}_{j, l} (\xv)\, \kappa_\delta (\xv, s)\, \partial_{x_l}\, \calJ^*_{\delta, \jjv}\, g_\delta \vts\right) (\xv, s) \nonumber
        \rets
        &\begin{aligned}\label{eq:dem_DW_2}
        &= \left. \delta^{1/2}\, \euler^{\vts \icplx\, \xi_\jjv (\xv, s)}\; %
        \partial_{x_j} \left(\vts \widetilde{a}_{j, l} (\xv)\, \kappa_\delta (\xv, s)\, \partial_{x_l} \Phi^{\Kv_\jjv} (\xv)^\transp \vts\right) \times g_\delta (\zeta) \right|_{\zeta = \delta\, \kvh_2 \cdot (\xv + s \vvv_2)}
        \retss
        &+\; \delta^{1/2}\, \euler^{\vts \icplx\, \xi_\jjv (\xv, s)} \left(\vts \widetilde{a}_{j, l} (\xv)\, \kappa_\delta (\xv, s)\, \partial_{x_l}\, \Phi^{\Kv_\jjv} (\xv)^\transp \vts\right) \times \left. \left(\vts \Upsilon^\delta_{j, \jjv} (D_\zeta)\, g_\delta \vts\right) (\zeta) \right|_{\zeta = \delta\, \kvh_2 \cdot (\xv + s \vvv_2)}
        \retss
        &+\; \delta^{1/2}\, \euler^{\vts \icplx\, \xi_\jjv (\xv, s)} \partial_{x_j} \left(\vts \widetilde{a}_{j, l} (\xv)\, \kappa_\delta (\xv, s)\, \Phi^{\Kv_\jjv} (\xv)^\transp \vts\right) \times \left. \left(\vts \Upsilon^\delta_{l, \jjv} (D_\zeta)\, g_\delta \vts\right) (\zeta) \right|_{\zeta = \delta\, \kvh_2 \cdot (\xv + s \vvv_2)}
        \retss
        &+\; \delta^{1/2}\, \euler^{\vts \icplx\, \xi_\jjv (\xv, s)} \left(\vts \widetilde{a}_{j, l} (\xv)\, \kappa_\delta (\xv, s)\, \Phi^{\Kv_\jjv} (\xv)^\transp \vts\right) \times \left. \left(\vts \Upsilon^\delta_{j, \jjv} (D_\zeta)\, \Upsilon^\delta_{l, \jjv} (D_\zeta)\, g_\delta \vts\right) (\zeta) \right|_{\zeta = \delta\, \kvh_2 \cdot (\xv + s \vvv_2)}
        \end{aligned}
        \rets
        &:= G^{(0)}_{j, l, \delta, \jjv} (\xv, s) + G^{(1)}_{j, l, \delta, \jjv} (\xv, s) + G^{(2)}_{j, l, \delta, \jjv} (\xv, s) + G^{(3)}_{j, l, \delta, \jjv} (\xv, s), \label{eq:dem_DW_2bis}
    \end{align}
    where $G^{(n)}_{j, l, \delta, \jjv} (\xv, s) \in L^2_{\Kv \cdot \vvh_1} (\cylaug)$, $n = 0, \dots, 3$, are the four terms on the right-hand side of \eqref{eq:dem_DW_2}. Now, for $n \in \{0, \dots, 3\}$, let $\calA^{\delta, (n)}$ be the block operator acting on $\ell^2 (\bbL; L^2(\R; \C^2))$, whose entries are given for any $g \in L^2(\R; \C^2)$ by
    \begin{equation}\label{eq:def_A_n_I_J_delta_dem_DW}
        \calA^{\delta, (n)}_{\iiv, \jjv}\, g := \left\{
        \begin{array}{cl}
            \displaystyle\sum_{j, l = 1}^2 \chi (\delta^{1/4}\, D_\zeta)\, \calJ^{}_{\delta, \iiv}\, G^{(n)}_{j, l, \delta, \jjv} \in L^2(\R; \C^2), & (\iiv, \jjv) \in \bbL' \times \bbL'',
            \rets
            0, & (\iiv, \jjv) \not\in \bbL' \times \bbL''.
        \end{array}
        \right.
    \end{equation}
    Then summing \eqref{eq:dem_DW_2bis} over $j, l = 1, 2$, applying $\chi (\delta^{1/4}\, D_\zeta)\, \calJ^{}_{\delta, \iiv}$, and using $g_\delta = \chi (\delta^{1/4}\, D_\zeta)\, g$ yields an expression for the entries of $\calA^\delta$, \eqref{eq:def_A_I_J_dem_DW}: for $(\iiv, \jjv) \in \bbL' \times \bbL''$,
    \begin{equation}
        \calA^\delta_{\iiv, \jjv} \underset{\eqref{eq:def_A_I_J_dem_DW}}{=} \sum_{j, l = 1}^2 \chi (\delta^{1/4}\, D_\zeta)\, \calJ^{}_{\delta, \iiv}\, \partial_{x_j} \widetilde{a}_{j, l}(\xv)\, \kappa_\delta (\xv, s)\, \partial_{x_l}\, \calJ^*_{\delta, \jjv}\, \chi (\delta^{1/4}\, D_\zeta) \underset{(\ref{eq:dem_DW_2bis}, \ref{eq:def_A_n_I_J_delta_dem_DW})}{=} \sum_{n = 0}^3 \calA^{\delta, (n)}_{\iiv, \jjv}. \label{eq:dem_DW_2_summed}
    \end{equation}
    We claim (and prove in the next steps below) that
    \begin{align}\label{eq:calJ_G_0_I_J}
        \spforall (\iiv, \jjv) \in \bbL' \times \bbL'', \quad \calA^{\delta, (0)}_{\iiv, \jjv} &= \chi (\delta^{1/4}\, D_\zeta)\, M_{\iiv, \jjv}\, \euler^{\vts \icplx\, \delta^{-1}\, (\lvar_\jjv - \lvar_\iiv)\, \zeta}\, \kappa (\zeta)\, \chi (\delta^{1/4}\, D_\zeta),
        \rets
        \spforall n \in \{1, \dots, 3\}, \quad \calA^{\delta, (n)} &= \calO (\delta^{3/4}), \quad \textnormal{in}\ \ \scrL(\ell^2 (\bbL; L^2(\R; \C^2))), \label{eq:G_n_123_negligible}
    \end{align}
    uniformly in $\bbL', \bbL'' \subseteq \bbL (\delta^{3/4})$, where $M_{\iiv, \jjv} \in \C^{2 \times 2}$ is defined by \eqref{eq:def_M_I_J}. Note that if these two statements hold, then together with \eqref{eq:dem_DW_2_summed} and \eqref{eq:dem_DW_0} they directly imply the desired result, Lemma \ref{lem:K_to_K_dw_coupling_in_between}. It therefore remains to establish (\ref{eq:calJ_G_0_I_J}, \ref{eq:G_n_123_negligible}). We do so in the next two steps.

    \vspace{1\baselineskip} \noindent
    \textbf{Step 2.} \textit{Proof of Identity \eqref{eq:calJ_G_0_I_J}}. Let $(\iiv, \jjv) \in \bbL' \times \bbL''$. From the definitions \eqref{eq:mode_for_averaging_operator}--\eqref{eq:def_calJ}, it follows that for $\zeta \in \R$,
    \begin{equation*}
        \calJ^{}_{\delta, \iiv}\, G^{(0)}_{j, l, \delta, \jjv} (\zeta) = \delta^{-1/2}\,\int_\Omega \euler^{-\icplx\, \xi_\iiv \scalebox{0.82}{$($} \xv,\, \delta^{-1}\, \zeta - \kvh_2 \cdot \xv \scalebox{0.82}{$)$}}\; \overline{\Phi^{\Kv_\iiv} (\xv)}\, G^{(0)}_{j, l, \delta, \jjv} (\xv, \delta^{-1}\, \zeta - \kvh_2 \cdot \xv)\, d\xv.
    \end{equation*}
    Then substituting $G^{(0)}_{j, l, \delta, \jjv}$ with its definition \eqref{eq:dem_DW_2}, and using $\kappa_\delta (\xv, s) = \kappa (\delta\, \kvh_2 \cdot (\xv + s\vvv_2))$ as well as the identity $\xi_\jjv (\xv, s) - \xi_\iiv (\xv, s) = (\gamma_\jjv - \gamma_\iiv)\, \kvh_1 \cdot \xv + (\lvar_\jjv - \lvar_\iiv)\, s$, we obtain 
    \begin{gather*}
        \sum_{j, l = 1}^2 \calJ^{}_{\delta, \iiv}\, G^{(0)}_{j, l, \delta, \jjv} (\zeta) = M_{\iiv, \jjv}\, \euler^{\vts \icplx\, \delta^{-1}\, (\lvar_\iiv - \lvar_\jjv)\, \zeta}\, \kappa (\zeta)\, \left(\vts \chi (\delta^{1/4}\, D_\zeta)\, g \vts\right) (\zeta), \textWHERE
        \retss
        M_{\iiv, \jjv} := \sum_{j, l = 1}^2 \int_\Omega \euler^{\vts \icplx\, \scalebox{0.82}{$($} (\gamma_\jjv - \gamma_\iiv)\, \kvh_1 - (\lvar_\iiv - \lvar_\jjv) \kvh_2 \scalebox{0.82}{$)$} \cdot \xv}\; \overline{\Phi^{\Kv_\iiv} (\xv)}\, \partial_{x_j}\, \widetilde{a}_{j, l}(\xv)\, \partial_{x_l} \Phi^{\Kv_\jjv} (\xv)^\transp\, d\xv.
    \end{gather*}
    Finally, using the relation $(\gamma_\jjv - \gamma_\iiv)\, \kvh_1 - (\lvar_\jjv - \lvar_\iiv)\, \kvh_2 = \Kv_\iiv + \floK_\iiv - \Kv_\jjv - \floK_\jjv$, which follows from Part \ref{lem:veps_neighborhood_K_d} of the Wrapping Lemma \ref{lem:wrapping_lemma}, as well as the definition $a(\xv)\, \sigma_2 = (\widetilde{a}_{j, l} (\xv))_{j, l = 1, 2}$, we deduce that the above definition of $M_{\iiv, \jjv}$ coincides with \eqref{eq:def_M_I_J}. Substituting this in the definition \eqref{eq:def_A_n_I_J_delta_dem_DW}, we then obtain \eqref{eq:calJ_G_0_I_J}.

    \vspace{1\baselineskip} \noindent
    \textbf{Step 3.} \textit{Proof of Bound \eqref{eq:G_n_123_negligible}}. We only present the arguments for $n = 1$, since they extend straightforwardly to $n = 2, 3$. Let $g \in L^2 (\R; \C^2)$ and $\jjv \in \bbL$. By \eqref{eq:dem_DW_2},
    \begin{equation}\label{eq:def_G_1_delta_J_0}
        G^{(1)}_{j, l, \delta, \jjv} (\xv, s) := \delta^{1/2}\vts \euler^{\vts \icplx\, \xi_\jjv (\xv, s)}\vts \widetilde{a}_{j, l}(\xv)\vts \kappa_\delta (\xv, s)\, \partial_{x_l} \Phi^{\Kv_\jjv} (\xv)^\transp \left(\vts \Upsilon^\delta_{j, \jjv} (D_\zeta)\, g_\delta \vts\right) (\zeta) \big|_{\zeta = \delta\, \kvh_2 \cdot (\xv + s \vvv_2)}.
    \end{equation}
    To reformulate this definition, introduce the operator $\calS^{}_{l, \jjv} \in \scrL (L^2_{\Kv \cdot \vvh_1} (\cylaug), L^2 (\R; \C^2))$ whose adjoint is defined by
    \begin{equation*}
        \calS^*_{l, \jjv} \in \scrL(L^2 (\R; \C^2), L^2_{\Kv \cdot \vvh_1} (\cylaug)), \quad \calS^*_{l, \jjv}\, g (\xv, s) := \euler^{\vts \icplx\, \xi_\jjv (\xv, s)}\, \partial_{x_l}\, \Phi^{\Kv_\jjv} (\xv)^\transp g (\kvh_2 \cdot (\xv + s \vvv_2)).
    \end{equation*} 
    Note that $\calS^*_{l, \jjv}$ is defined similarly to $\calT^*_\jjv$, \eqref{eq:averaging_operator_adjoint}, where $\Phi^{\Kv_\jjv}$ is replaced by \smash{$\partial_{x_l}\, \Phi^{\Kv_\jjv}$}. The same holds for $\calS^{}_{l, \jjv}$ in terms of $\calT_\jjv$. The definitions of $\calS^*_{l, \jjv}$, $\calU^*_\delta$, and $g_\delta = \chi (\delta^{1/4}\, D_\zeta)\, g$, substituted in Equation \eqref{eq:def_G_1_delta_J_0}, lead to
    \begin{equation}
        G^{(1)}_{j, l, \delta, \jjv} (\xv, s) := \widetilde{a}_{j, l} (\xv)\, \kappa_\delta (\xv, s)\, \calS^*_{l, \jjv}\; \calU^*_\delta\, \Upsilon^\delta_{j, \jjv} (D_\zeta)\, \chi (\delta^{1/4}\, D_\zeta)\, g\vts (\xv, s). \label{eq:def_G_1_delta_J}
    \end{equation}
    Summing \eqref{eq:def_G_1_delta_J} over $j, l = 1, 2$ and applying $\chi (\delta^{1/4}\, D_\zeta)\, \calJ_{\delta, \iiv}$ yields an expression for $\calA^{\delta, (1)}_{\iiv, \jjv}$, the operator defined by \eqref{eq:def_A_n_I_J_delta_dem_DW}. This expression can be written in the block operator form:
    \begin{align}
        &\calA^{\delta, (1)} = \sum_{j, l = 1}^2 \chi (\delta^{1/4}\, D_\zeta)\, \mathds{1}_{\bbL'}\; \calJ^{}_\delta\, \widetilde{a}_{j, l} (\xv)\, \kappa_\delta (\xv, s)\, \calS^*_l\; \calU^*_\delta\, \Upsilon^\delta_j (D_\zeta)\, \chi (\delta^{1/4}\, D_\zeta)\; \mathds{1}_{\bbL''}, \label{eq:dem_DW_5}
        \retss
        &\textnormal{where} \quad \calS_l := \begin{pmatrix} \vdots \\ \calS_{l, \jjv} \\ \vdots \end{pmatrix} \textAND \Upsilon^\delta_j (D_\zeta) := %
        \begin{pmatrix}
        \ddots & & (0)
        \\
        & \Upsilon^\delta_{j, \jjv} (D_\zeta) & 
        \\
        (0) & & \ddots
        \end{pmatrix}_{\jjv\in\bbL}. \nonumber
    \end{align}
    Recall that $\calT$ is bounded from $L^2_{\Kv \cdot \vvh_1} (\cylaug)$ to $\ell^2 (\bbL; L^2(\R; \C^2))$, uniformly in $\delta$, by Proposition \ref{prop:pties_averaging_operator:square_summability}. Moreover, note that the proof of this result only relies on the smoothness of $\Phi^{\Kv_\jjv}$, and therefore can be adapted to show a similar result for $\calS_l$. As a consequence, $\calS^*_l$ is bounded from $\ell^2 (\bbL; L^2(\R; \C^2))$ to $L^2_{\Kv \cdot \vvh_1} (\cylaug)$. Finally, since $\Upsilon^\delta_j (D_\zeta)$ is block-diagonal with entries given by $\smash{\icplx\, ( \gamma_\jjv \kvh^{(l)}_1 + \delta \kvh^{(l)}_2 D_\zeta )}$, and $\bbL (\delta^{3/4})$ is, by definition \eqref{eq:def_L_veps}, the set of indices $\jjv \in \bbL$ such that $|\gamma_\jjv| \leq \delta^{3/4}$, there exists a constant $C > 0$ depending on $\kvh_1, \kvh_2$ only, such that
    \begin{align*}
        \big\|\vts \Upsilon^\delta_j (D_\zeta)\, \chi (\delta^{1/4}\, D_\zeta)\, \mathds{1}_{\bbL''} \vts\big\|_{\scrL(\ell^2(\bbL; L^2 (\R)))} &= \sup_{\jjv \in \bbL''} \big\|\vts \big(\vts \gamma_\jjv\, \kvh_1^{(j)} + \delta\, \kvh_2^{(j)}\, D_\zeta \vts\big)\, \chi (\delta^{1/4}\, D_\zeta) \vts\big\|_{\scrL(L^2 (\R))} 
        \retss
        &= \sup_{\jjv \in \bbL''} \big\|\vts \lvar \mapsto \big(\vts \gamma_\jjv\, \kvh_1^{(j)} + \delta\, \kvh_2^{(j)}\, \lvar \vts\big)\, \chi (\delta^{1/4}\, \lvar) \vts\big\|_\infty \leq C\, \delta^{3/4}.
    \end{align*}
    Therefore, since $a(\xv)$ and $\kappa (\zeta)$ are bounded functions, we deduce from \eqref{eq:dem_DW_5} that $\calA^{\delta, (1)}$ is bounded in $\ell^2(\bbL; L^2 (\R; \C^2))$ with an $\calO (\delta^{3/4})$ norm.
    \end{dem}

    \vspace{1\baselineskip} \noindent
    To prove Proposition \ref{prop:DW_contributions}, we also require the following useful lemma.

    \begin{lemma}\label{lem:dominated_boundedness}
    For $\iiv, \jjv \in \bbL$, let $\xv\mapsto u_{\Kv_\iiv, \Kv_\jjv}(\xv) \in \scrC^\infty (\R^2 / \Lambda; \C^{2 \times 2})$ be a smooth, $2 \times 2$ matrix-valued, $\Lambda$--periodic function.  For $\iiv, \jjv \in \bbL$, consider an operator $\calA_{\iiv, \jjv} \in \scrL(\scrX, \scrY)$, where $(\scrX, \|\cdot\|_\scrX)$ and $(\scrY, \|\cdot\|_\scrY)$ are Hilbert spaces. Suppose that
    \begin{equation}\label{eq:assumption_dominated_boundedness}
        \spforall \iiv, \jjv \in \bbL, \quad \| \calA_{\iiv, \jjv} \|_{\scrL(\scrX, \scrY)} \leq  \left| \int_\Omega \euler^{\vts \icplx\, (\floK_\iiv - \floK_\jjv) \cdot \xv}\; u_{\Kv_\iiv, \Kv_\jjv} (\xv)\, d\xv \, \right|,
    \end{equation}
    where $|\cdot|$ denotes the  matrix $2$--norm. Then, the block operator $\calA := (\calA_{\iiv, \jjv})_{\iiv, \jjv \in \bbL}$ is bounded from $\ell^2(\bbL; \scrX)$ to $\ell^2(\bbL; \scrY)$.
    \end{lemma}

    \begin{dem}
    By the following Schur-Holmgren-Young type inequality, the operator norm of $\calA$ can be controlled in terms of its supremum column sum norm and supremum row sum norm:
    \begin{multline*}
        \|\calA\|^2_{\scrL(\ell^2(\bbL; \scrX), \ell^2(\bbL; \scrY))} \leq \vertiii{\calA}_1\, \vertiii{\calA}_\infty,
        \retss
        \textWHERE \vertiii{\calA}_1 = \sup_{\jjv \in \bbL} \sum_{\iiv \in \bbL} \| \calA_{\iiv, \jjv} \|_{\scrL(\scrX, \scrY)} \textAND \vertiii{\calA}_\infty = \sup_{\iiv \in \bbL} \sum_{\jjv \in \bbL} \| \calA_{\iiv, \jjv} \|_{\scrL(\scrX, \scrY)}.
    \end{multline*}
    Therefore, it suffices to bound $\vertiii{\calA}_1$ and $\vertiii{\calA}_\infty$. 
    Fix $\Kv_\star \in \{\Kv, \Kv'\}$ and let $\jjv \in \bbL$. From its definition \eqref{eq:def_floK}, the family $(\floK_\iiv)_{\iiv \in \{\Kv_\star\} \times \Z}$ is pairwise disjoint in $\Lambda^*$. Therefore, summing \eqref{eq:assumption_dominated_boundedness} over $\iiv \in \{\Kv_\star\} \times \Z$ and using the smoothness of $u_{\Kv_\star, \Kv_\jjv}$, we get
    \begin{equation*}
        \sum_{\iiv \in \{\Kv_\star\} \times \Z} \| \calA_{\iiv, \jjv} \|_{\scrL(\scrX, \scrY)} \leq \sum_{\kv \in \Lambda^*} \left| \int_\Omega \euler^{\vts \icplx\, \kv \cdot \xv}\; u_{\Kv_\star, \Kv_\jjv} (\xv)\, d\xv \, \right| < \infty.
    \end{equation*}
    Summing this inequality over $\Kv_\star \in \{\Kv, \Kv'\}$, and taking the supremum over $\jjv \in \bbL$, we deduce that
    \begin{equation*}
        \vertiii{\calA}_1 \leq \sup_{\Kv_\star, \Kv_\circ \in \{\Kv, \Kv'\}} \bigg[ \sum_{\kv \in \Lambda^*} \left| \int_\Omega \euler^{\vts \icplx\, \kv \cdot \xv}\; u_{\Kv_\star, \Kv_\circ} (\xv)\, d\xv \, \right| \bigg].
    \end{equation*}
    By interchanging the roles of $\iiv$ and $\jjv$, we deduce a similar inequality for $\vertiii{\calA}_\infty$. Combining both inequalities on $\vertiii{\calA}_1$ and $\vertiii{\calA}_\infty$, we then obtain the boundedness of $\calA$.
    \end{dem}

    \vspace{1\baselineskip} \noindent
    We now turn to the proof of Proposition \ref{prop:DW_contributions}.
    \begin{dem}[of Proposition \ref{prop:DW_contributions}] %
    We start from the expansion \eqref{eq:K_to_K_dw_coupling_in_between} established in Lemma \ref{lem:K_to_K_dw_coupling_in_between}. The dominant contribution in \eqref{eq:K_to_K_dw_coupling_in_between} is a sum over $\iiv \in \bbL', \jjv \in \bbL''$, that can be decomposed into diagonal summands ($\iiv = \jjv \in \bbL' \cap \bbL''$) and off-diagonal summands ($\iiv \neq \jjv$). Then, comparing \eqref{eq:K_to_K_dw_coupling_in_between} with the statement of Proposition \ref{prop:DW_contributions}, we see that it suffices to prove the identity
    \begin{equation}
        \spforall \iiv \in \bbL, \quad M_{\iiv, \iiv} = \vartheta^{\Kv_\iiv}\, \sigma_3,\label{eq:K_to_K_dw_coupling_dem_diagonal}
    \end{equation}
    as well as the following bound in $\scrL(L^2_{\Kv \cdot \vvh_1} (\cylaug))$, uniformly in $\bbL', \bbL'' \subset \bbL (\delta^{3/4})$:
    \begin{equation}
        \sum_{\substack{\iiv \in \bbL', \jjv \in \bbL'' \\ \iiv \neq \jjv}}\!\!\! \calJ^*_{\delta, \iiv}\, \chi (\delta^{1/4}\, D_\zeta)\, M_{\iiv, \jjv}\, \euler^{\vts \icplx\, \delta^{-1}\, (\lvar_\jjv - \lvar_\iiv)\, \zeta}\, \kappa (\zeta)\, \chi (\delta^{1/4}\, D_\zeta)\, \calJ^{}_{\delta, \jjv} = \calO (\delta^{3/4}).
        \label{eq:K_to_K_dw_coupling_dem_off_diagonal}
    \end{equation}
    In other words, the asymptotic behavior of $\Pi_{\bbL'}\, (\vts \nabla_\xv \cdot a(\xv)\, \sigma_2\, \kappa_\delta (\xv, s)\, \nabla_\xv \vts)\, \Pi_{\bbL''}$ is mainly governed by the diagonal summands, given by \eqref{eq:K_to_K_dw_coupling_dem_diagonal}, while the off-diagonal summands are negligible due to \eqref{eq:K_to_K_dw_coupling_dem_off_diagonal}. The next steps are devoted to proving these statements.

    \vspace{1\baselineskip} \noindent
    \textbf{Step 1.} \textit{Diagonal summands; identity \eqref{eq:K_to_K_dw_coupling_dem_diagonal}}. The identity \eqref{eq:K_to_K_dw_coupling_dem_diagonal} follows directly from the definition \eqref{eq:def_M_I_J} of $M_{\iiv, \jjv}$ and Proposition \ref{prop:expression_Phi_l_W_Phi_j}: in fact,
    \begin{equation*}
        \spforall \iiv \in \bbL, \quad M_{\iiv, \iiv} = \int_\Omega \overline{\Phi^{\Kv_\iiv} (\xv)}\; \nabla \cdot a(\xv)\, \sigma_2\, \nabla \Phi^{\Kv_\iiv} (\xv)^\transp\, d\xv = \vartheta^{\Kv_\iiv}\, \sigma_3.
    \end{equation*}

    \noindent
    \textbf{Step 2.} \textit{Off-diagonal summands; bound \eqref{eq:K_to_K_dw_coupling_dem_off_diagonal}}. The left-hand side of \eqref{eq:K_to_K_dw_coupling_dem_off_diagonal}, which we wish to bound, can be rewritten as
    \begin{equation*}
        \calJ^*_\delta\, \mathds{1}_{\bbL'}\, \OFF (\calA^{\delta, (0)})\, \mathds{1}_{\bbL''}\, \calJ^{}_\delta,
    \end{equation*}
    where $\OFF (\calA^{\delta, (0)})$ is the block operator whose entries are given for $\iiv, \jjv \in \bbL$ by
    \begin{equation*}
        \left\{
        \begin{array}{cl}
        \chi (\delta^{1/4}\, D_\zeta)\, M_{\iiv, \jjv}\, \euler^{\vts \icplx\, \delta^{-1}\, (\lvar_\jjv - \lvar_\iiv)\, \zeta}\, \kappa (\zeta)\, \chi (\delta^{1/4}\, D_\zeta) &\textnormal{if}\quad \iiv \neq \jjv
        \rets
        0 &\textnormal{if}\quad \iiv = \jjv.
        \end{array}
        \right.
    \end{equation*}
    Our strategy is to first bound each entry, and then extend the bound to the block operator using Lemma \ref{lem:dominated_boundedness}. By definition \eqref{eq:def_M_I_J} of $M_{\iiv, \jjv} (\zeta)$, the definition of Fourier multiplier \eqref{eq:def_Fourier_multiplier}, and the properties of the Fourier transform:
    \begin{multline}
        \chi (\delta^{1/4}\, D_\zeta)\, M_{\iiv, \jjv}\, \euler^{\vts \icplx\, \delta^{-1}\, (\lvar_\jjv - \lvar_\iiv)\, \zeta}\, \kappa (\zeta)\, \chi (\delta^{1/4}\, D_\zeta)
        \\
        = M_{\iiv, \jjv}\, \euler^{\vts \icplx\, \delta^{-1}\, (\lvar_\jjv - \lvar_\iiv)\, \zeta}\; \calU^{}_\delta\; \Big[\vts \chi \big(\vts\delta^{-3/4}\, (D_\zeta + \lvar_\jjv - \lvar_\iiv) \vts\big)\, \kappa (\delta\, \zeta)\, \chi \big(\vts \delta^{-3/4}\, D_\zeta \vts\big) \vts\Big]\; \calU^*_\delta. \label{eq:K_to_K_dw_coupling_dem_off_diagonal_1}
    \end{multline}
    The operator $\chi (\vts\delta^{-3/4}\, (D_\zeta + \lvar_\jjv - \lvar_\iiv) \vts)\, \kappa (\delta\, \zeta)\, \chi (\vts \delta^{-3/4}\, D_\zeta \vts) \in \scrL(L^2(\R))$ arises in \cite[Step 3 of proof of Proposition 6.7]{drouot2020edge}. To estimate its norm, first note, by \eqref{eq:pairwise_disjoint_near_intervals}  that if  $\iiv \neq \jjv$, then the intervals $ (-\delta^{3/4}, \delta^{3/4})$ and $\lvar_\jjv - \lvar_\iiv + (-\delta^{3/4}, \delta^{3/4})$ are disjoint if $\delta^{3/4} < 1/3$. %
    Therefore, we may write $\chi(\vts \delta^{-3/4}\, (D^{}_\zeta + \lvar_\jjv - \lvar_\iiv) \vts) = \chi(\vts \delta^{-3/4}\, (D^{}_\zeta + \lvar_\jjv - \lvar_\iiv) \vts)\, D_\zeta^{-1} \cdot D^{}_\zeta$, where the first factor is a bounded operator on $L^2(\R)$. Then, using that $\kappa$ is smooth, $\kappa' \in L^\infty (\R)$ and the properties of the Fourier transform, we obtain that there is a constant $C>0$, such that for all $\delta$ sufficiently small: if $\iiv, \jjv \in \bbL$ and  $\iiv\ne \jjv$,  then 
    \begin{equation}\label{eq:dw-bound}
        \big\|\vts \chi \big(\vts\delta^{-3/4}\, (D_\zeta + \lvar_\jjv - \lvar_\iiv) \vts\big)\, \kappa (\delta\, \zeta)\, \chi \big(\vts \delta^{-3/4}\, D_\zeta \vts\big) \vts\big\|_{\scrL(L^2(\R))} \leq C\, \delta^{3/4}\ .
    \end{equation}
    The constant $C > 0$ depends on $\|\kappa\|_{L^\infty}$ and $\|\kappa'\|_{L^\infty}$. Using \eqref{eq:dw-bound} and the definition \eqref{eq:def_M_I_J} of $M_{\iiv, \jjv}$ to bound \eqref{eq:K_to_K_dw_coupling_dem_off_diagonal_1}, we deduce that for all $\iiv, \jjv \in \bbL$ such that $\iiv \neq \jjv$,
    \begin{multline}\label{eq:K_to_K_dw_coupling_dem_off_diagonal_2}
        \big\|\vts \chi (\delta^{1/4}\, D_\zeta)\, M_{\iiv, \jjv}\, \euler^{\vts \icplx\, \delta^{-1}\, (\lvar_\jjv - \lvar_\iiv)\, \zeta}\, \kappa (\zeta)\, \chi (\delta^{1/4}\, D_\zeta) \vts\big\|_{\scrL(L^2(\R))} 
        \\
        \leq C\, \delta^{3/4}\, \left| \int_\Omega \euler^{\vts \icplx\, (\Kv_\iiv + \floK_\iiv - \Kv_\jjv - \floK_\jjv) \cdot \xv}\; \overline{\nabla \Phi^{\Kv_\iiv} (\xv)}\, a (\xv)\, \sigma_2\, \nabla \Phi^{\Kv_\jjv} (\xv)^\transp\, d\xv \vts\right|\ .
    \end{multline}
    Since this corresponds to \eqref{eq:assumption_dominated_boundedness} with $u_{\Kv_\iiv, \Kv_\jjv} := C\, \delta^{3/4}\, \euler^{\vts \icplx\vts (\Kv_\iiv - \Kv_\jjv) \cdot \xv}\vts \overline{\nabla \Phi^{\Kv_\iiv} (\xv)}\vts a (\xv)\vts \sigma_2\vts \nabla \Phi^{\Kv_\jjv} (\xv)^\transp$, we can use Lemma \ref{lem:dominated_boundedness}, and deduce that \eqref{eq:K_to_K_dw_coupling_dem_off_diagonal_2} implies the bound 
    \begin{equation*}
        \|\vts \OFF (\calA^{\delta, (0)}) \vts\|_{\scrL(\ell^2 (\bbL; L^2(\R)))} \lesssim \delta^{3/4}.
    \end{equation*}
    Finally, using the uniform boundedness of $\calJ_\delta$ in $\delta$, we deduce \eqref{eq:K_to_K_dw_coupling_dem_off_diagonal}.
    \end{dem}

    \subsubsection{Completion of the proof of Proposition \ref{prop:asymptotics_HdeltaNEAR}}\label{sec:completion_expansion_HdeltaNEAR}
    Since we set $\mu = 0$, we have from \eqref{eq:conjugated_augmented_operator_alt} that $\calH^\delta_{\AUG, \Kv \cdot \vvh_1} (\mu) = \calH^\delta_{\AUG, \Kv \cdot \vvh_1}$, and
    \begin{multline}
    \Pi_{\bbL'}\, \big(\vts \calH^\delta_{\AUG, \Kv \cdot \vvh_1} - E_D - \delta z \vts\big)\, \Pi_{\bbL''} 
    \retss
    = \Pi_{\bbL'}\,\big(\vts \calH^0_{\AUG, \Kv \cdot \vvh_1} - E_D \vts\big)\, \Pi_{\bbL''} - \Pi_{\bbL'}\,\big(\vts \delta z \vts\big)\, \Pi_{\bbL''} + \Pi_{\bbL'}\,\big(\vts \delta\, \nabla_\xv \cdot a(\xv)\, \sigma_2\, \kappa_\delta (\xv, s)\, \nabla_\xv \vts\big)\, \Pi_{\bbL''}, \label{eq:decomp_Haug}
    \end{multline}
    where $\kappa_\delta (\xv, s) := \kappa(\delta\vts \kvh_2 \cdot (\xv + s \vvv_2))$. For $\bbL' = \bbL''$, we expand the first term using \eqref{eq:asymptotics_H0near} with $d = 1$, the second term using \eqref{eq:asymptotics_H0near} with $d = 0$, and the third term using \eqref{eq:K_to_K_dw_coupling}. We then get that uniformly in $\bbL' \subseteq \bbL (\delta^{3/4})$ and in $z$ varying in compact subsets:
    \begin{multline*}
    \Pi_{\bbL'}\, \big(\vts \calH^\delta_{\AUG, \Kv \cdot \vvh_1} - E_D - \delta z \vts\big)\, \Pi_{\bbL'} 
    \retss
    = \delta \sum_{\iiv \in \bbL'} \calJ^*_{\delta, \iiv}\, \chi (\delta^{1/4}\, D_\zeta)\, \Big[\vts \sigma(\kvh_2)\, D_\zeta + \delta^{-1}\, \gamma_\iiv\, \sigma(\kvh_1) + \vartheta^{\Kv_\iiv}\, \sigma_3\, \kappa (\zeta) - z \vts\Big]\, \chi (\delta^{1/4}\, D_\zeta)\, \calJ^{}_{\delta, \iiv}
    \\
    + \calO (\delta^{3/2}), \quad \textnormal{in}\ \ \scrL(\scrX_{\bbL'}, \scrX_{\bbL'}).
    \end{multline*}
    By definition \eqref{eq:effective_Dirac_operator} of $\calD^{\Kv_\iiv} (\delta^{-1}\, \gamma_\iiv)$, this expansion is precisely Proposition \ref{prop:asymptotics_HdeltaNEAR_diagonal}. On the other hand, if $\bbL'$ and $\bbL'' \subseteq \bbL (\delta^{3/4})$ are disjoint, then the first part of Proposition \ref{prop:asymptotics_H0near} shows that the first two terms in \eqref{eq:decomp_Haug} are zero, while the third term, according to \eqref{eq:K_to_K_dw_coupling}, is of order $\delta \cdot \delta^{3/4}$ (the dominant term in \eqref{eq:K_to_K_dw_coupling} is zero, since $\bbL' \cap \bbL'' = \emptyset$). We then obtain:
    \begin{equation*}
    \Pi_{\bbL'}\, \big(\vts \calH^\delta_{\AUG, \Kv \cdot \vvh_1} - E_D  - \delta z \vts\big)\, \Pi_{\bbL''} = \calO (\delta^{3/4 + 1}) = \calO (\delta^{7/4}), \quad \textnormal{in}\ \ \scrL(\scrX_{\bbL''}, \scrX_{\bbL'}),
    \end{equation*}
    which is precisely Proposition \ref{prop:asymptotics_HdeltaNEAR_off_diagonal}. This completes the proof of Proposition \ref{prop:asymptotics_HdeltaNEAR}.

    \subsection{Resolvent expansion, proof of Propositions \ref{prop:resolvent_expansion_HtildeDeltaNEAR} and \ref{prop:dominant_term_without_near_projectors}}\label{sec:proof_resolvent_expansion_HtildeDeltaNEAR} 
    We now turn to the proof of Proposition \ref{prop:resolvent_expansion_HtildeDeltaNEAR}. As shown at the end of Section \ref{sec:schur_complement}, Theorem \ref{thm:resolvent_expansion} then follows directly. We recall that $\mu$ is set to zero for simplicity, so that $\calH^\delta_{\AUG, \Kv \cdot \vvh_1} (\mu) = \calH^\delta_{\AUG, \Kv \cdot \vvh_1}$.

    As a starting point, the following technical result provides useful bounds for the proof of the resolvent expansion.
    \begin{lemma}\label{lem:technical_res}~
    \begin{enumerate}[label={$(\alph*).$}, ref={\theproposition.$(\alph*)$}]
        \item As $\delta$ tends to 0,
        \begin{equation}
        \calJ^{}_{\delta}\, \calJ_{\delta}^* = \Id + \calO (\delta) \quad \textnormal{in}\ \ \scrL(\ell^2(\bbL; H^1(\R; \C^2)), \vts\ell^2(\bbL; H^{-1}(\R; \C^2))).\label{eq:J_times_Jstar_off_diagonal}
        \end{equation}
        \item $\calJ^*_\delta$ is asymptotically invariant under $\Pi_{\bbL'}$: more precisely, uniformly in $\bbL' \subseteq \bbL (\delta^{3/4})$,
        \begin{align}
            \Pi_{\bbL'}\, \calJ^*_\delta &= \calJ^*_\delta\, \mathds{1}_{\bbL'} + \calO (\delta^{1/4}) \quad \textnormal{in}\ \ \scrL(\ell^2(\bbL; H^1(\R; \C^2)),\, L^2_{\Kv \cdot \vvh_1} (\cylaug)).\label{eq:Jstar_almost_invariant_under_Pi}
            \retss
            \calJ^{}_\delta\, \Pi_{\bbL'} &= \mathds{1}_{\bbL'}\, \calJ^{}_\delta + \calO (\delta^{1/4}) \quad \textnormal{in}\ \ \scrL(L^2_{\Kv \cdot \vvh_1} (\cylaug),\, \ell^2(\bbL; H^{-1}(\R; \C^2))).\label{eq:dual_Jstar_almost_invariant_under_Pi}
        \end{align}
    \end{enumerate}
    \end{lemma}

    \begin{dem}
    ~
    \\
    $(a)$. \textit{The product $\calJ^{}_\delta\, \calJ_\delta^*$}. We decompose $\calJ^{}_\delta\, \calJ_\delta^*$ into its diagonal and off-diagonal parts. According to Proposition \ref{prop:pties_averaging_operator:orthogonality_relations}, the diagonal part of $\calJ^{}_\delta\, \calJ_\delta^*$ is the identity operator. Therefore, it suffices to prove that the off-diagonal part of $\calJ^{}_\delta\, \calJ_\delta^*$ is negligible. Let $\iiv, \jjv \in \bbL$. Using the definition \eqref{eq:mode_for_averaging_operator}--\eqref{eq:def_calJ} of $\calJ^{}_{\delta, \iiv} = \calU^{}_\delta\, \calT^{}_\iiv$, we see that $\smash{\calJ^{}_{\delta, \iiv}\, \calJ^*_{\delta, \jjv}}$ is a multiplication operator: more precisely, for any $g \in L^2(\R; \C^2)$,
    \begin{align*}
        \calJ^{}_{\delta, \iiv}\, \calJ^*_{\delta, \jjv}\; g (\zeta) 
        &= \left[\vts \int_\Omega \euler^{\vts \icplx\, \scalebox{0.675}{$\big($}\vts (\gamma_\jjv - \gamma_\iiv)\, \kvh_1 - (\lvar_\jjv - \lvar_\iiv)\, \kvh_2 \vts\scalebox{0.675}{$\big)$} \cdot \xv}\; \overline{\Phi^{\Kv_\iiv} (\xv)}\; \Phi^{\Kv_\jjv} (\xv)^\transp\, d\xv \vts\right]\, \euler^{\vts \icplx\vts \delta^{-1}\, (\lvar_\jjv - \lvar_\iiv)\, \zeta\vphantom{\scalebox{0.675}{$\big($}}}\, g(\zeta)
        \retss
        &= \left[\vts \int_\Omega \euler^{\vts \icplx\, (\Kv_\iiv + \floK_\iiv - \Kv_\jjv - \floK_\jjv) \cdot \xv}\; \overline{\Phi^{\Kv_\iiv} (\xv)}\; \Phi^{\Kv_\jjv} (\xv)^\transp\, d\xv \vts\right]\, \euler^{\vts \icplx\vts \delta^{-1}\, (\lvar_\jjv - \lvar_\iiv)\, \zeta}\, g(\zeta),
    \end{align*}
    where we used the identity $(\gamma_\jjv - \gamma_\iiv)\, \kvh_1 - (\lvar_\jjv - \lvar_\iiv)\, \kvh_2 = \Kv_\iiv + \floK_\iiv - \Kv_\jjv - \floK_\jjv$, which follows from Part \ref{lem:veps_neighborhood_K_d} of the Wrapping Lemma. Next, we prove that if $\iiv \neq \jjv$, then the multiplication by $\zeta \mapsto \euler^{- \icplx\vts \delta^{-1}\vts ( \lvar_\iiv - \lvar_{\smash{\jjv}})\vts \zeta}$ is a bounded operator from $H^1(\R; \C^2)$ to $H^{-1}(\R; \C^2)$, with norm $\calO(\delta)$. For any $g, h \in H^1(\R; \C^2)$, by integration by parts,
    \begin{gather*}
        \left| \int_\R \overline{h(\zeta)} \cdot g(\zeta)\; \euler^{- \icplx\vts \delta^{-1}\, ( \lvar_\iiv - \lvar_{\smash{\jjv}})\vts \zeta}\, d\zeta \right| = \frac{\delta}{|\lvar_\iiv - \lvar_{\smash{\jjv}}|} \left| \int_\R \big[\, \overline{h'(\zeta)} \cdot g(\zeta) + \overline{h(\zeta)} \cdot g'(\zeta) \,\big]\; \euler^{- \icplx\vts \delta^{-1}\, ( \lvar_\iiv - \lvar_{\smash{\jjv}})\vts \zeta}\, d\zeta \right|
        \\
        \leq \frac{\delta}{|\lvar_\iiv - \lvar_{\smash{\jjv}}|}\, \|h\|_{H^1(\R; \C^2)}\; \|g\|_{H^1(\R; \C^2)} \ \leq\ C_0\, \delta\, \|h\|_{H^1(\R; \C^2)}\; \|g\|_{H^1(\R; \C^2)},
    \end{gather*}
    where the last inequality follows from the fact that the set $\{|\lvar_\iiv - \lvar_{\smash{\jjv}}|\ /\ \iiv, \jjv \in \bbL,\ \iiv \neq \jjv\}$ is bounded from below by a constant $C^{-1}_0 > 0$, by definition \eqref{eq:gamma_Iandlambda_I} of $\lvar_\iiv$. As a consequence, 
    \begin{multline}\label{eq:J_times_Jstar_dem2}
        \spforall \iiv, \jjv \in \bbL, \quad \iiv \neq \jjv, \quad \|\vts \calJ^{}_{\delta, \iiv}\, \calJ^*_{\delta, \jjv} \vts\|_{\scrL({H^1(\R; \C^2),\, H^{-1}(\R; \C^2)})} 
        \\
        \leq C_0\, \delta\, \left|\, \int_\Omega \euler^{\vts \icplx\, (\floK_\iiv - \floK_\jjv) \cdot \xv}\, \bigl[\vts \euler^{\vts \icplx\, \Kv_\iiv \cdot \xv}\; \overline{\Phi^{\Kv_\iiv} (\xv)}\vts\bigr]\; \bigl[\vts \euler^{- \icplx\, \Kv_\jjv \cdot \xv}\; \Phi^{\Kv_\jjv} (\xv)^\transp \vts\bigr]\, d\xv \,\right|.
    \end{multline}
    This inequality corresponds to \eqref{eq:assumption_dominated_boundedness} with $u_{\Kv_\iiv, \Kv_\jjv} (\xv) := [\vts \euler^{\vts \icplx\, \Kv_\iiv \cdot \xv}\; \overline{\Phi^{\Kv_\iiv} (\xv)}\vts]\, [\vts \euler^{- \icplx\, \Kv_\jjv \cdot \xv}\; \Phi^{\Kv_\jjv} (\xv)^\transp \vts]$. Therefore, we can apply Lemma \ref{lem:dominated_boundedness} to deduce that the off-diagonal part of $\calJ^{}_\delta\, \calJ_\delta^*$ is bounded from $\ell^2 (\bbL; H^1 (\R; \C^2))$ to $\ell^2 (\bbL; H^{-1} (\R; \C^2))$ with norm $\calO (\delta)$. This proves \eqref{eq:J_times_Jstar_off_diagonal}.

    \vspace{1\baselineskip}\noindent
    $(b)$. \textit{Near-invariance of $\calJ^*_\delta$ under $\Pi_{\bbL'}$}. By the expansion of $\Pi_{\bbL'}$ in Proposition \ref{prop:asymptotics_near_projectors}, and the boundedness of $\calJ_\delta^*$ uniformly in $\delta$ (Proposition \ref{prop:pties_averaging_operator:square_summability}), we have
    \begin{equation*}
        \Pi_{\bbL'}\, \calJ_\delta^* = \calJ^*_\delta\, \mathds{1}_{\bbL'}\, \chi (\delta^{1/4}\, D_\zeta)\, \mathds{1}_{\bbL'}\, \calJ^{}_\delta\, \calJ_\delta^* + \calO (\delta^{3/4}) \quad \textnormal{in}\ \ \scrL(\ell^2(\bbL; L^2(\R; \C^2)), L^2_{\Kv \cdot \vvh_1} (\cylaug)).
    \end{equation*}
    Using the decomposition $\calJ^{}_\delta\, \calJ_\delta^* = \Id + (\calJ^{}_\delta\, \calJ_\delta^* - \Id)$ leads to
    \begin{multline}\label{eq:Jstar_almost_invariant_under_Pi_dem_1}
        \Pi_{\bbL'}\, \calJ_\delta^* = \calJ^*_\delta\, \mathds{1}_{\bbL'}\, \chi (\delta^{1/4}\, D_\zeta)\, \mathds{1}_{\bbL'} + \calJ^*_\delta\, \mathds{1}_{\bbL'}\, \chi (\delta^{1/4}\, D_\zeta)\, \mathds{1}_{\bbL'}\, \big(\vts \calJ^{}_\delta\, \calJ_\delta^* - \Id \vts\big)
        \retss
        + \calO (\delta^{3/4}) \ \ \textnormal{in}\ \ \scrL(\ell^2(\bbL; L^2(\R; \C^2)), L^2_{\Kv \cdot \vvh_1} (\cylaug)).
    \end{multline}
    By \eqref{eq:J_times_Jstar_off_diagonal}, $\calJ^{}_\delta\, \calJ_\delta^* - \Id$ is bounded from $\ell^2(\bbL; H^1(\R; \C^2))$ to $\ell^2(\bbL; H^{-1}(\R; \C^2))$ with an $\calO(\delta)$ norm. Moreover, 
    \begin{equation*}
        \big\|\vts \chi (\delta^{1/4}\, D_\zeta) \vts\big\|_{\scrL(H^{-1}(\R),\, L^2 (\R))} \leq \sup_{|\lvar| \leq \delta^{-1/4}} \big( 1 + |\lvar|^2 \big)^{1/2} = \calO (\delta^{-1/4}).
    \end{equation*}
    Therefore, $\chi (\delta^{1/4}\, D_\zeta)\, \mathds{1}_{\bbL'}\, (\vts \calJ^{}_\delta\, \calJ_\delta^* - \Id \vts)$ is bounded from $\ell^2(\bbL; H^1(\R; \C^2))$ to $\ell^2 (\bbL; L^2(\R; \C^2))$ with an $\calO (\delta^{3/4})$ norm. Because $\calJ^*_\delta$ is bounded uniformly in $\delta$, \eqref{eq:Jstar_almost_invariant_under_Pi_dem_1} becomes
    \begin{equation*}
        \Pi_{\bbL'}\, \calJ_\delta^* = \calJ^*_\delta\, \mathds{1}_{\bbL'}\, \chi (\delta^{1/4}\, D_\zeta)\, \mathds{1}_{\bbL'} + \calO (\delta^{3/4}), \quad \textnormal{in}\ \ \scrL(\ell^2(\bbL; H^1(\R; \C^2)), L^2_{\Kv \cdot \vvh_1} (\cylaug)).
    \end{equation*}
    Next, the decomposition $\chi (\delta^{1/4}\, D_\zeta) = \Id + (\vts \chi (\delta^{1/4}\, D_\zeta) - \Id \vts)$ leads to
    \begin{multline}\label{eq:Jstar_almost_invariant_under_Pi_dem_2}
        \Pi_{\bbL'}\, \calJ_\delta^* = \calJ^*_\delta\, \mathds{1}_{\bbL'} \retss
        + \calJ^*_\delta\, \mathds{1}_{\bbL'}\, \big(\vts \chi (\delta^{1/4}\, D_\zeta) - \Id \vts\big)\, \mathds{1}_{\bbL'} + \calO (\delta^{3/4}) \quad \textnormal{in}\ \ \scrL(\ell^2(\bbL; H^1(\R; \C^2)), L^2_{\Kv \cdot \vvh_1} (\cylaug)).
    \end{multline}
    Using the bound
    \begin{equation}\label{eq:bound_band_width}
        \big\|\vts \chi (\delta^{1/4}\, D_\zeta) - \Id \vts\big\|_{\scrL(H^1(\R),\, L^2(\R))} \leq \sup_{|\lvar| \geq \delta^{-1/4}} (1 + |\lvar|^2)^{-1/2} = \calO(\delta^{1/4}),
    \end{equation}
    and the boundedness of $\calJ^*_\delta$, we then deduce that $\calJ^*_\delta\, \mathds{1}_{\bbL'}\, (\vts \chi (\delta^{1/4}\, D_\zeta) - \Id \vts)\, \mathds{1}_{\bbL'}$ is bounded from $\ell^2(\bbL; H^1(\R; \C^2))$ to $L^2_{\Kv \cdot \vvh_1} (\cylaug)$ with an $\calO (\delta^{1/4})$ norm. By combining this with \eqref{eq:Jstar_almost_invariant_under_Pi_dem_2}, we then get \eqref{eq:Jstar_almost_invariant_under_Pi}. Moreover, \eqref{eq:dual_Jstar_almost_invariant_under_Pi} is simply the dual bound to \eqref{eq:Jstar_almost_invariant_under_Pi}.
    \end{dem}

    \vspace{1\baselineskip} \noindent
    \begin{dem}[of Proposition \ref{prop:resolvent_expansion_HtildeDeltaNEAR}]
    Our starting point is Proposition \ref{prop:asymptotics_HdeltaNEAR_diagonal}, which states that uniformly in $\bbL' \subset \bbL (\delta^{3/4})$,
    \begin{equation}\label{eq:asymptotics_HdeltaNEAR}
        \Pi_{\bbL'}\, \big(\vts \calH^\delta_{\AUG, \Kv \cdot \vvh_1} - E_D - \delta\, z \vts\big)\, \Pi_{\bbL'} = \delta\, \mathfrak{L}^\delta (z) + \calO (\delta^{3/2}) \quad \textnormal{in}\ \ \scrL(\scrX_{\bbL'}),
    \end{equation}
    where
    \begin{equation*}
        \mathfrak{L}^\delta (z) := \calJ^*_\delta\, \chi (\delta^{1/4}\, D_\zeta)\, \mathds{1}_{\bbL'}\, \big( \calD^{\vts \delta} (0) - z \big)\, \mathds{1}_{\bbL'}\, \chi (\delta^{1/4}\, D_\zeta)\, \calJ^{}_\delta.
    \end{equation*}
    Define
    \begin{equation*}
        \mathfrak{R}^{\vts \delta} (z) := \calJ_\delta^*\, \mathds{1}_{\bbL'}\, \big(\vts \calD^{\vts \delta} (0) - z \vts\big)^{-1}\, \mathds{1}_{\bbL'}\, \calJ^{}_\delta. 
    \end{equation*}
    Step $1$ is devoted to showing that $\Pi_{\bbL'}\, \mathfrak{L}^\delta (z)\, \mathfrak{R}^{\vts \delta} (z)\, \Pi_{\bbL'} = \Pi_{\bbL'} + \calO (\delta^{1/4})$. In Step $2$, we then prove the invertibility of $\Pi_{\bbL'}\, (\vts \calH^\delta_{\AUG, \Kv \cdot \vvh_1} - E_D - \delta\, z \vts)\, \Pi_{\bbL'}$ and the expansion of its inverse in terms of $\Pi_{\bbL'}\, \mathfrak{R}^{\vts \delta} (z)\, \Pi_{\bbL'}$.

    \vspace{1\baselineskip} \noindent
    \textbf{Step 1.} \textit{The expression $\Pi_{\bbL'}\, \mathfrak{L}^\delta (z)\, \mathfrak{R}^{\vts \delta} (z)\, \Pi_{\bbL'}$ is close to $\Pi_{\bbL'}$}. One has
    \begin{multline*}
        \Pi_{\bbL'}\, \mathfrak{L}^\delta (z)\, \mathfrak{R}^{\vts \delta} (z)\, \Pi_{\bbL'} 
        \retss
        = \Pi_{\bbL'}\, \calJ^*_\delta\, \chi (\delta^{1/4}\, D_\zeta)\, \mathds{1}_{\bbL'}\, \big( \calD^{\vts \delta} (0) - z \big)\, \mathds{1}_{\bbL'}\, \chi (\delta^{1/4}\, D_\zeta)\, \calJ^{}_\delta\, \calJ_\delta^*\, \mathds{1}_{\bbL'}\, \big(\vts \calD^{\vts \delta} (0) - z \vts\big)^{-1}\, \mathds{1}_{\bbL'}\, \calJ^{}_\delta\, \Pi_{\bbL'}.
    \end{multline*}
    This product features $\calJ^{}_\delta\, \calJ_\delta^*$, which we express as $\calJ^{}_\delta\, \calJ_\delta^* = \Id + (\calJ^{}_\delta\, \calJ_\delta^* - \Id)$. Then
    \begin{align}\label{eq:dem_RE_1}
        &\Pi_{\bbL'}\, \mathfrak{L}^\delta (z)\, \mathfrak{R}^{\vts \delta} (z)\, \Pi_{\bbL'} \nonumber
        \retss 
        &= \Pi_{\bbL'}\, \calJ^*_\delta\, \chi (\delta^{1/4} D_\zeta)\, \mathds{1}_{\bbL'}\, \big( \calD^{\vts \delta} (0) - z \big)\, \mathds{1}_{\bbL'} \chi (\delta^{1/4} D_\zeta)\, \mathds{1}_{\bbL'}\, \big(\vts \calD^{\vts \delta} (0) - z \vts\big)^{-1} \mathds{1}_{\bbL'}\, \calJ^{}_\delta\, \Pi_{\bbL'} + \calE^\delta_1,
    \end{align}
    where
    \begin{equation*}
        \calE^\delta_1 := \Pi_{\bbL'}\vts \calJ^*_\delta\, \chi (\delta^{1/4}\, D_\zeta)\vts \mathds{1}_{\bbL'}\vts \big( \calD^{\vts \delta} (0) - z \big)\vts \mathds{1}_{\bbL'}\vts \chi (\delta^{1/4}\, D_\zeta)\vts \big(\vts \calJ^{}_\delta\vts \calJ_\delta^* - \Id \vts\Big)\vts \mathds{1}_{\bbL'}\vts \big(\vts \calD^{\vts \delta} (0) - z \vts\big)^{-1}\, \mathds{1}_{\bbL'}\vts \calJ^{}_\delta\vts \Pi_{\bbL'}.
    \end{equation*}
    We bound the $L^2_{\Kv \cdot \vvh_1} (\cylaug)$--norm of $\calE^\delta_1$ using the triangle inequality:
    \begin{align}
        \|\vts \calE^\delta_1 \vts\|_{\scrL(L^2_{\Kv \cdot \vvh_1} (\cylaug))} &\leq %
        \big\|\, \Pi_{\bbL'}\, \calJ^*_\delta\, \chi (\delta^{1/4}\, D_\zeta) \,\|_{\scrL(\ell^2(\bbL; L^2(\R)),\, L^2_{\Kv \cdot \vvh_1} (\cylaug))} \nonumber
        \\
        &\times \left\|\, \mathds{1}_{\bbL'}\, \big( \calD^{\vts \delta} (0) - z \big)\, \mathds{1}_{\bbL'} \,\right\|_{\scrL(\ell^2(\bbL; H^1(\R)),\, \ell^2(\bbL; L^2(\R)))}\; \nonumber
        \\
        & \times \|\, \chi (\delta^{1/4}\, D_\zeta) \,\|_{\scrL(\ell^2(\bbL; H^{-1}(\R)),\, \ell^2(\bbL; H^1(\R)))}\,  \nonumber
        \\
        & \times \big\|\, \calJ^{}_\delta\, \calJ_\delta^* - \Id \,\big\|_{\scrL(\ell^2(\bbL; H^{1}(\R)),\, \ell^2(\bbL; H^{-1}(\R)))} \nonumber
        \\
        & \times \big\|\, \mathds{1}_{\bbL'}\, \big(\vts \calD^{\vts \delta} (0) - z \vts\big)^{-1}\, \mathds{1}_{\bbL'}\, \calJ^{}_\delta\, \Pi_{\bbL'} \,\big\|_{\scrL(L^2_{\Kv \cdot \vvh_1} (\cylaug),\, \ell^2(\bbL; H^1(\R)))}.\label{eq:technical_res_eq_1_dem1}
    \end{align}
    We bound each term on the right-hand side, starting from the last one. First, by assumption, the operator $\mathds{1}_{\bbL'}\, (\vts \calD^{\vts \delta} (0) - z \vts)^{-1} \mathds{1}_{\bbL'}\, \calJ^{}_\delta\, \Pi_{\bbL'}$ is bounded from $L^2_{\Kv \cdot \vvh_1} (\cylaug)$ to $\ell^2(\bbL; H^1(\R; \C^2))$, with norm less than $C/\eta$. Second, by \eqref{eq:J_times_Jstar_off_diagonal}, $\calJ^{}_\delta\, \calJ_\delta^* - \Id$ is bounded from $\ell^2(\bbL; H^1(\R; \C^2))$ to $\ell^2(\bbL; H^{-1}(\R; \C^2))$ with an $\calO(\delta)$ norm. Third, we have the bound
    \begin{equation*}
        \big\|\vts \chi (\delta^{1/4}\, D_\zeta) \vts\big\|_{\scrL(H^{-1}(\R),\, H^{1}(\R))} \leq \sup_{|\lvar| \leq \delta^{-1/4}} \big( 1 + |\lvar|^2 \big) = \calO (\delta^{-1/2}).
    \end{equation*}
    Fourth, as already discussed in Remark \ref{rmk:technical_exponent}, there exist constants $C_0, C_1 > 0$ such that
    \begin{equation*}
        \left\|\, \mathds{1}_{\bbL'}\, \big(\vts \calD^{\vts \delta} (0) - z \vts\big)\, \mathds{1}_{\bbL'} \,\right\|_{\scrL(\ell^2(\bbL; H^1(\R)),\, \ell^2(\bbL; L^2(\R)))} \underset{\delta \to 0}{\sim} C_0 + C_1\, \sup_{\iiv \in \bbL (\delta^{3/4})} |\delta^{-1} \gamma_\iiv| = \calO(\delta^{-1/4}).
    \end{equation*}
    Finally, $\calJ^*_\delta\, \chi (\delta^{1/4}\, D_\zeta)$ is bounded from $\ell^2 (\bbL; L^2(\R; \C^2))$ to $L^2_{\Kv \cdot \vvh_1} (\cylaug)$ uniformly in $\delta$, by Proposition \ref{prop:pties_averaging_operator:square_summability}. Substituting these observations in \eqref{eq:technical_res_eq_1_dem1}, we get that $\calE^\delta_1 = \eta^{-1}\, \calO (\delta^{1/4})$ in $\scrL(L^2_{\Kv \cdot \vvh_1} (\cylaug))$, and \eqref{eq:dem_RE_1} becomes:
    \begin{multline*}
        \Pi_{\bbL'}\, \mathfrak{L}^\delta (z)\, \mathfrak{R}^{\vts \delta} (z)\, \Pi_{\bbL'} 
        \retss
        = \Pi_{\bbL'}\, \calJ^*_\delta\, \chi (\delta^{1/4}\, D_\zeta)\, \mathds{1}_{\bbL'} \big( \calD^{\vts \delta} (0) - z \big)\, \mathds{1}_{\bbL'} \chi (\delta^{1/4}\, D_\zeta)\, \mathds{1}_{\bbL'}\, \big( \calD^{\vts \delta} (0) - z \big)^{-1}\, \mathds{1}_{\bbL'}\, \calJ^{}_\delta\, \Pi_{\bbL'} + %
        \eta^{-1}\, \calO (\delta^{1/4}).
    \end{multline*}
    Moreover, by substituting the identity $\chi (\delta^{1/4}\, D_\zeta) = \Id + (\vts \chi (\delta^{1/4}\, D_\zeta) - \Id \vts)$, we get
    \begin{multline}\label{eq:dem_RE_2}
        \Pi_{\bbL'}\, \mathfrak{L}^\delta (z)\, \mathfrak{R}^{\vts \delta} (z)\, \Pi_{\bbL'} = \Pi_{\bbL'}\, \calJ^*_\delta\, \chi (\delta^{1/4}\, D_\zeta)\, \mathds{1}_{\bbL'}\, \calJ^{}_\delta\, \Pi_{\bbL'}
        \retss
        + \Pi_{\bbL'}\, \calJ^*_\delta\, \chi (\delta^{1/4}\, D_\zeta)\, \mathds{1}_{\bbL'}\, \big( \calD^{\vts \delta} (0) - z \big)\, \mathds{1}_{\bbL'}\, \big(\vts \chi (\delta^{1/4}\, D_\zeta) - \Id \vts\big)\, \mathds{1}_{\bbL'}\, \big(\vts \calD^{\vts \delta} (0) - z \vts\big)^{-1}\, \mathds{1}_{\bbL'}\, \calJ^{}_\delta\, \Pi_{\bbL'}%
        \retss
        + \eta^{-1}\, \calO (\delta^{1/4}) \quad \textnormal{in}\ \ \scrL(\scrX_{\bbL'}).
    \end{multline}
    Call $\calE^\delta_2$ the second term on the right-hand side of \eqref{eq:dem_RE_2}. We bound $\calE^\delta_2$ by following closely \cite[Step 3, proof of Proposition 6.5]{drouot2020edge}. We write
    \begin{multline*}
        \calE^\delta_2 = \calJ^*_\delta\, \chi (\delta^{1/4}\, D_\zeta)\, \big(\vts \chi (\delta^{1/4}\, D_\zeta) - \Id \vts\big)\, \mathds{1}_{\bbL'}\, \big( \calD^{\vts \delta} (0) - z \big)\, \big(\vts \calD^{\vts \delta} (0) - z \vts\big)^{-1}\, \mathds{1}_{\bbL'}\, \calJ^{}_\delta\, \Pi_{\bbL'}
        \retss
        + \calJ^*_\delta\, \chi (\delta^{1/4}\, D_\zeta)\, \mathds{1}_{\bbL'}\, \Big[\vts \big( \calD^{\vts \delta} (0) - z \big),\; \big(\vts \chi (\delta^{1/4}\, D_\zeta) - \Id \vts\big) \vts\Big]\, \mathds{1}_{\bbL'}\, \big(\vts \calD^{\vts \delta} (0) - z \vts\big)^{-1}\, \mathds{1}_{\bbL'}\, \calJ^{}_\delta\, \Pi_{\bbL'}.
    \end{multline*}
    The first term vanishes because $\chi (\delta^{1/4}\, D_\zeta)\, (\vts \chi (\delta^{1/4}\, D_\zeta) - \Id \vts) = 0$. In the second term, the commutator can be simplified using that the massless part of the Dirac operator $\calD^{\vts \delta} (0) - z$ commutes with $\chi (\delta^{1/4}\, D_\zeta)$: more precisely, calling $\DIAG ((\vartheta^{\Kv_\iiv}\, \sigma_3)_{\iiv}) \in \scrL (\ell^2 (\bbL; L^2 (\R; \C^2)))$ the block-diagonal operator whose entries are $\vartheta^{\Kv_\iiv}\, \sigma_3$, the second term can be rewritten as
    \begin{equation*}
        \calJ^*_\delta\, \chi (\delta^{1/4}\, D_\zeta)\, \mathds{1}_{\bbL'}\, \DIAG ((\vartheta^{\Kv_\iiv}\, \sigma_3)_{\iiv})\, \big[\vts \kappa(\cdot),\; \chi (\delta^{1/4}\, D_\zeta) - \Id \vts\big]\, \mathds{1}_{\bbL'}\, \big(\vts \calD^{\vts \delta} (0) - z \vts\big)^{-1}\, \mathds{1}_{\bbL'}\, \calJ^{}_\delta\, \Pi_{\bbL'}.
    \end{equation*}
    By assumption, $\mathds{1}_{\bbL'}\, (\vts \calD^{\vts \delta} (0) - z \vts)^{-1} \mathds{1}_{\bbL'}\, \calJ^{}_\delta\, \Pi_{\bbL'}$ is bounded from $L^2_{\Kv \cdot \vvh_1} (\cylaug)$ to $\ell^2(\bbL; H^1(\R; \C^2))$, with norm less than $1/\eta$. Moreover, the bound \eqref{eq:bound_band_width}, combined with $\kappa \in L^\infty (\R)$, implies that $[\vts \kappa, \chi (\delta^{1/4}\, D_\zeta) - \Id \vts] = \calO (\delta^{1/4})$ as an operator from $\ell^2(\bbL; H^1(\R; \C^2))$ to $\ell^2 (\bbL; L^2(\R; \C^2))$. Finally, as $\calJ^*_\delta\, \chi (\delta^{1/4}\, D_\zeta)\, \mathds{1}_{\bbL'}\, \DIAG ((\vartheta^{\Kv_\iiv}\, \sigma_3)_{\iiv})$ is bounded from $\ell^2 (\bbL; L^2(\R; \C^2))$ to $L^2_{\Kv \cdot \vvh_1} (\cylaug)$ uniformly in $\delta$ (Proposition \ref{prop:pties_averaging_operator}), we deduce that $\calE^\delta_2 = \eta^{-1}\, \calO (\delta^{1/4})$ in $\scrL(L^2_{\Kv \cdot \vvh_1} (\cylaug))$. The expansion \eqref{eq:dem_RE_2} becomes
    \begin{equation*}
        \Pi_{\bbL'}\, \mathfrak{L}^\delta (z)\, \mathfrak{R}^{\vts \delta} (z)\, \Pi_{\bbL'} = \Pi_{\bbL'}\, \calJ^*_\delta\, \chi (\delta^{1/4}\, D_\zeta)\, \mathds{1}_{\bbL'}\, \calJ^{}_\delta\, \Pi_{\bbL'} + \eta^{-1}\, \calO (\delta^{1/4}) \quad \textnormal{in}\ \ \scrL(\scrX_{\bbL'}).
    \end{equation*}
    Finally, Proposition \ref{prop:asymptotics_near_projectors} shows that the dominant term in this expansion is equal to $\Pi_{\bbL'}$, plus an $\calO (\delta^{3/4})$--error in $\scrL(\scrX_{\bbL'})$. Therefore,
    \begin{equation}\label{eq:res_exp_dem_1}
        \Pi_{\bbL'}\, \mathfrak{L}^\delta (z)\, \mathfrak{R}^{\vts \delta} (z)\, \Pi_{\bbL'} = \Pi_{\bbL'} + \eta^{-1}\, \calO (\delta^{1/4}) \quad \textnormal{in}\ \ \scrL(\scrX_{\bbL'}).
    \end{equation}
    This concludes the first step.

    \vspace{1\baselineskip} \noindent
    \textbf{Step 2.} \textit{Inverting $\Pi_{\bbL'}\, (\vts \calH^\delta_{\AUG, \Kv \cdot \vvh_1} - E_D - \delta\, z \vts)\, \Pi_{\bbL'}$ in terms of $\Pi_{\bbL'}\, \mathfrak{R}^{\vts \delta} (z)\, \Pi_{\bbL'}$}. We multiply \eqref{eq:asymptotics_HdeltaNEAR} on the right with $\delta^{-1}\, \mathfrak{R}^{\vts \delta} (z)\, \Pi_{\bbL'}$. Because $\mathfrak{R}^{\vts \delta} (z)\, \Pi_{\bbL'}$ is, by assumption \eqref{eq:z-constraint}, bounded in $L^2_{\Kv \cdot \vvh_1} (\cylaug)$ with norm less than $C / \eta$, uniformly in $z \in S$ and $\bbL' \subset \bbL (\delta^{3/4})$, we get that
    \begin{equation*}
        \Pi_{\bbL'}\, \big(\vts \calH^\delta_{\AUG, \Kv \cdot \vvh_1} - E_D - \delta\, z \vts\big)\, \Pi_{\bbL'}\; \frac{1}{\delta}\, \mathfrak{R}^{\vts \delta} (z)\, \Pi_{\bbL'} = \Pi_{\bbL'}\, \mathfrak{L}^\delta (z)\, \mathfrak{R}^{\vts \delta} (z)\, \Pi_{\bbL'} + \calO ( \delta^{1/2} ) \quad \textnormal{in}\ \ \scrL(\scrX_{\bbL'}).
    \end{equation*}
    We then use the expansion \eqref{eq:res_exp_dem_1} shown in the first step, to obtain
    \begin{equation}\label{eq:res_exp_dem_2}
        \Pi_{\bbL'}\, \big(\vts \calH^\delta_{\AUG, \Kv \cdot \vvh_1} - E_D - \delta\, z \vts\big)\, \Pi_{\bbL'}\; \frac{1}{\delta}\, \mathfrak{R}^{\vts \delta} (z)\, \Pi_{\bbL'} = \Pi_{\bbL'} + \eta^{-1}\, \calO ( \delta^{1/4} ) \quad \textnormal{in}\ \ \scrL(\scrX_{\bbL'}).
    \end{equation}
    Denote by $\calE^\delta_3 (z)$ the $\eta^{-1}\, \calO(\delta^{1/4})$--error term in \eqref{eq:res_exp_dem_2}. Using a Neumann series argument, we deduce the expansion
    \begin{align}
        \Pi_{\bbL'}\, \big(\vts \calH^\delta_{\AUG, \Kv \cdot \vvh_1} - E_D - \delta\, z \vts\big)^{-1}\, \Pi_{\bbL'} &= \frac{1}{\delta}\; \Pi_{\bbL'}\; \mathfrak{R}^{\vts \delta} (z) \; \big[\vts \Id_{\scrX_{\bbL'}} + \calE^\delta_3 (z) \vts\big]^{-1}\; \Pi_{\bbL'} \nonumber
        \\
        &= \frac{1}{\delta}\; \Pi_{\bbL'}\; \Bigl[\vts \mathfrak{R}^{\vts \delta} (z) + \frac{1}{\eta}\, \calO ( \delta^{1/4} ) \vts\Bigr]\; \Pi_{\bbL'}\quad \textnormal{in}\ \ \scrL(\scrX_{\bbL'}),\label{eq:res_exp_dem_3}
    \end{align} 
    which is precisely \eqref{eq:resolvent_on_near_components}. Therefore, the proof of Proposition \ref{prop:resolvent_expansion_HtildeDeltaNEAR} is complete.
    \end{dem}

    \vspace{1\baselineskip} \noindent
    We finish with the proof of Proposition \ref{prop:dominant_term_without_near_projectors}.
    \begin{dem}[of Proposition \ref{prop:dominant_term_without_near_projectors}]%
    Define
    \begin{equation*}
        \mathfrak{R}^{\vts \delta} (z) := \calJ_\delta^*\, \mathds{1}_{\bbL (\delta^{3/4})}\, \big(\vts \calD^{\vts \delta} (0) - z \vts\big)^{-1}\, \mathds{1}_{\bbL (\delta^{3/4})}\, \calJ^{}_\delta. 
    \end{equation*}
    Then Proposition \ref{prop:dominant_term_without_near_projectors} is equivalent to the expansion $\Pi_\NEAR\, \mathfrak{R}^{\vts \delta} (z)\, \Pi_\NEAR = \mathfrak{R}^{\vts \delta} (z) + \eta^{-1} \calO(\delta^{1/4})$ in $\scrL(L^2_{\Kv \cdot \vvh_1} (\cylaug))$. We write
    \begin{equation*}
        \Pi_\NEAR\; \mathfrak{R}^{\vts \delta} (z) \; \Pi_\NEAR = \big(\vts \Pi_\NEAR\; \calJ_\delta^* \vts\big)\, \mathds{1}_{\bbL (\delta^ {3/4})}\, \big(\vts \calD^{\vts \delta} (0) - z \vts\big)^{-1}\, \mathds{1}_{\bbL (\delta^ {3/4})}\, \big(\vts \calJ^{}_\delta \; \Pi_\NEAR \vts\big).
    \end{equation*}
    By applying \eqref{eq:Jstar_almost_invariant_under_Pi} with $\bbL' = \bbL(\delta^{3/4})$, we find that $\Pi_\NEAR\, \calJ^*_\delta - \calJ^*_\delta\, \mathds{1}_{\bbL (\delta^ {3/4})}$ is bounded from $\ell^2(\bbL; H^1(\R; \C^2))$ to $L^2_{\Kv \cdot \vvh_1} (\cylaug)$ with an $\calO (\delta^{1/4})$ norm. Since $(\vts \calD^{\vts \delta} (0) - z \vts)^{-1}$ is by assumption bounded from $\ell^2 (\bbL; L^2(\R; \C^2))$ to $\ell^2(\bbL; H^1(\R; \C^2))$, uniformly in $z \in S$, with norm less than $1/\eta$, we deduce the following expansion in $\scrL(L^2_{\Kv \cdot \vvh_1} (\cylaug))$:
    \begin{equation*}
        \Pi_\NEAR\; \mathfrak{R}^{\vts \delta} (z) \; \Pi_\NEAR = \calJ^*_\delta\, \mathds{1}_{\bbL (\delta^ {3/4})}\, \big(\vts \calD^{\vts \delta} (0) - z \vts\big)^{-1}\, \mathds{1}_{\bbL (\delta^ {3/4})}\, \big(\vts \calJ^{}_\delta \; \Pi_\NEAR \vts\big) + \eta^{-1}\, \calO (\delta^{1/4}).
    \end{equation*}
    Moreover, by applying \eqref{eq:dual_Jstar_almost_invariant_under_Pi} with $\bbL' = \bbL(\delta^{3/4})$, we get that $\calJ^{}_\delta\, \Pi_\NEAR - \mathds{1}_{\bbL (\delta^ {3/4})}\, \calJ^{}_\delta$ is bounded from $L^2_{\Kv \cdot \vvh_1} (\cylaug)$ to $\ell^2(\bbL; H^{-1}(\R; \C^2))$ with an $\calO (\delta^{1/4})$ norm. Since $(\vts \calD^{\vts \delta} (0) - z \vts)^{-1}$ is bounded from $\ell^2(\bbL; H^{-1}(\R; \C^2))$ to $\ell^2 (\bbL; L^2(\R; \C^2))$ uniformly in $z \in S$, with norm less than $1/\eta$, we obtain the expansion
    \begin{align}
        \Pi_\NEAR\; \mathfrak{R}^{\vts \delta} (z) \; \Pi_\NEAR &= \calJ^*_\delta\, \mathds{1}_{\bbL (\delta^ {3/4})}\, \big(\vts \calD^{\vts \delta} (0) - z \vts\big)^{-1}\, \mathds{1}_{\bbL (\delta^ {3/4})}\, \calJ^{}_\delta + \eta^{-1}\, \calO (\delta^{1/4}) \nonumber
        \retss
        &= \mathfrak{R}^{\vts \delta} (z) + \eta^{-1}\, \calO (\delta^{1/4}) \quad \textnormal{in}\ \ \scrL(L^2_{\Kv \cdot \vvh_1} (\cylaug)),\label{eq:res_exp_dem_4}
    \end{align}
    which proves Proposition \ref{prop:dominant_term_without_near_projectors}.
    \end{dem}

    \appendix
    \section{Rational matters} \label{sec:rational_matters}
    \noindent
    The goal of this section is to show how the resolvent expansion in Theorem \ref{thm:resolvent_expansion} simplifies when $r$ is a rational number. Note that for rational edges, the lifting strategy adopted in this paper is not necessary, since Theorem \ref{thm:resolvent_expansion} can be directly formulated at the level of $\calH^\delta_\EDGE$ in $\R^2$; see \cite{fefferman2016edge,lee2019elliptic,drouot2019characterization,drouot2020edge}. Thus the present section's sole purpose is to complete the presentation.

    In Section \ref{sec:app:neighborhoods_quasimomenta}, we recall the classification of rational edges as zigzag-type and armchair-type, and we characterize the set $\bbL (\veps)$ defined by \eqref{eq:def_L_veps} for $\veps$ small enough and for such edges. Section \ref{sec:resolvent_expansion_rational_edges} shows how Theorem \ref{thm:resolvent_expansion} simplifies for rational edges.

    \subsection{Neighborhoods of high-symmetry quasi-momenta}\label{sec:app:neighborhoods_quasimomenta}
    For rational edges, it turns out that the set $\bbL (\veps) \subseteq \bbL$, defined by \eqref{eq:def_L_veps}, can be characterized explicitly for $\veps$ small enough. In fact, assume that $r = b_1 / a_1$ where $b_1 \in \Z$ and $a_1 \in \N$ are coprime, with the convention that $(a_1, b_1) = (1, 0)$ if $r = 0$. In \cite{fefferman2024discrete,drouot2020edge}, rational edges are divided into two categories: 
    \begin{itemize}
    \item \emph{Zigzag-type edges}, defined by the condition $a_1 \neq b_1 \mod 3$, are such that the line $\Kv + \R\, \kvh_2$ intersects the lattice $\Kv + \Lambda^*$ without being arbitrarily close to $\Kv' + \Lambda^*$.
    \item \emph{Armchair-type edges}, defined by the condition $a_1 = b_1 \mod 3$, are such that $\Kv + \R\, \kvh_2$ intersects both $\Kv + \Lambda^*$ and $\Kv' + \Lambda^*$.
    \end{itemize}
    We prove that Proposition \ref{prop:veps_neighborhood_K} is consistent with this classification.

    \begin{proposition}\label{prop:veps_neighborhood_K_rational}
    Let $\vvh_1 = \vvv_1 + r \vvv_2$ and $\kvh_2 = -r \kv_1 + \kv_2$, and suppose that $r = b_1 / a_1$, where $b_1 \in \Z$ and $a_1 \in \N$ are coprime, with the convention that $(a_1, b_1) = (1, 0)$ if $r = 0$.
    \begin{enumerate}[label={$(\alph*).$}, ref={$\alph*$}]
        \item \textit{Zigzag-type edges}: If $a_1 \neq b_1 \mod 3$, then for $\veps > 0$ small enough, 
        \begin{equation*}
        \bbL (\veps) = \{\Kv\} \times a_1\, \Z \textAND \gamma_\iiv = 0 \quad \spforall \iiv \in \bbL (\veps).
        \end{equation*}
        In other words, as $\veps \to 0$, the line $\Kv + \R\, \kvh_2$ intersects $\Kv + \Lambda^*$, but is bounded away from $\Kv' + \Lambda^*$.
        \item \textit{Armchair-type edges}: If $a_1 = b_1 \mod 3$, then there exists an integer $m_0$ depending on $a_1$ and $b_1$, such that for $\veps > 0$ small enough,
        \begin{equation*}
        \bbL (\veps) = \Big( \{\Kv\} \times a_1\, \Z \Big) \cup \Big( \{\Kv'\} \times (m_0 + a_1\, \Z) \Big),  \textAND \gamma_\iiv = 0 \quad \spforall \iiv \in \bbL (\veps).
        \end{equation*}
        In other words, for $\veps$ small enough, $\Kv + \R\, \kvh_2$ intersects both $\Kv + \Lambda^*$ and $\Kv' + \Lambda^*$.
    \end{enumerate}
    \end{proposition}

    \begin{dem}
    \textbf{Step 1.} We first consider the lattice $\Kv + \Lambda^*$, by letting $\iiv = (\Kv, m)$, with $m \in \Z$. From its definition \eqref{eq:gamma_Iandlambda_I}, we have that $\gamma_\iiv = 2\pi \modulo{-mr + 1/2} - \pi$. Therefore, the condition $\iiv \in \bbL (\veps)$, which implies that $\gamma_\iiv$ is arbitrarily small, holds if and only if $mr \in \Z$, that is, for integers $m \in a_1\, \Z$. Thus, for $\veps > 0$ small enough and for both zigzag-type and armchair-type edges, $\{\Kv\} \times a_1\, \Z \subset \bbL (\veps)$ and $\gamma_\iiv = 0$ for $m \in a_1\, \Z$.

    \vspace{1\baselineskip} \noindent
    \textbf{Step 2.} We now focus on $\Kv' + \Lambda^*$ by letting $\iiv = (\Kv', m)$, $m \in \Z$. From \eqref{eq:gamma_Iandlambda_I}, we obtain
    \begin{equation*}
        \gamma_\iiv = 2\pi\vts \modulowoparen{\big( 2/3 - (m + 2/3)\, r  + 1/2 \big)} - \pi.
    \end{equation*}
    Therefore, the condition $\iiv \in \bbL (\veps)$, which implies that $\gamma_\iiv$ is arbitrarily small, holds if and only if $2/3 - (m + 2/3)\, r$ is an integer, say $n \in \Z$. Hence, using the expression $r = b_1 / a_1$, we are reduced to finding all the solutions to the equation
    \begin{equation}\label{eq:veps_neighborhood_K_rational_dem1}
        \textit{Find $(m, n) \in \Z^2$ such that} \qquad m\vts b_1 + n\vts a_1 = 2\, (a_1 - b_1) / 3.
    \end{equation}
    First, we note that this equation has a solution \emph{only if} $(a_1 - b_1) / 3$ is an integer, that is, only if $a_1 = b_1 \mod 3$. The contrapositive is that if $a_1 \neq b_1 \mod 3$, then as $\veps \to 0$, $\gamma_\iiv$ lies outside a neighborhood of $0$, so $\iiv \not\in \bbL (\veps)$ for $\veps$ small enough. This, combined with the first part (on $\Kv + \Lambda^*$), concludes the proof for zigzag-type edges.

    Now we consider armchair-type edges, by assuming that $a_1 = b_1 \mod 3$. Then $c := 2\, (a_1 - b_1) / 3 \in \Z$. Since $a_1$ and $b_1$ are coprime, we have from Bézout's theorem the existence of $(a_2, b_2) \in \Z^2$ such that $a_1\vts b_2 - a_2\vts b_1 = 1$. Moreover, any solution of \eqref{eq:veps_neighborhood_K_rational_dem1} has the form $(m, n) = (c\vts b_2 + d\vts a_1, -c\vts a_2 - d\vts b_1)$ for $d \in \Z$. Thus, for $\veps$ small enough, we deduce that $\iiv = (\Kv', m) \in \bbL (\veps)$ if and only if $m \in m_0 + a_1\, \Z$ with $m_0 := c\vts b_2$, and $\gamma_\iiv = 0$. This, combined with the first part (on $\Kv + \Lambda^*$), concludes the proof for armchair-type edges.
    \end{dem}

    \subsection{The resolvent expansion for rational edges}\label{sec:resolvent_expansion_rational_edges}
    We highlight in this section the fibered structure of $\calH^\delta_{\AUG, \kpar}$ to illustrate how Theorem \ref{thm:resolvent_expansion} relates to earlier works \cite{drouot2020defect} on rational edges.

    Assume that $r = b_1 / a_1 \in \Q$, where $b_1 \in \Z$ and $a_1 \in \N$ are coprime, with the convention that $(a_1, b_1) = (1, 0)$ if $r = 0$. Recall that $A^\delta_\AUG (\xv, s)$ is the augmented coefficient defined by \eqref{eq:def_augmented_potential}. Fixing $s \in \R$ as a parameter, one sees from the periodicity properties \eqref{eq:augmented_potential_property} of $A^\delta_\AUG (\xv, s)$ that $\xv \mapsto A^\delta_\AUG (\xv, s)$ is $\Z (a_1\, \vvh_1)$--periodic. Define the cylinder $\Sigma := \R^2 / \Z (a_1 \vvh_1)$, and let $L^2_{\kpar} (\Sigma)$ denote the space of functions $u \in L^2_{\LOC} (\R^2)$ such that $\smash{\euler^{- \icplx\vts \kpar\vts (\kvh_1 \cdot \xv)}\, u \in L^2 (\Sigma)}$. Such functions are $\kpar$-pseudo-periodic in the $(a_1 \vvh_1)$--direction and decay in the transverse direction:
    \begin{equation*}
    \spforall u \in L^2_{\kpar} (\Sigma), \quad u(\xv + a_1\, \vvh_1) = \euler^{\vts\icplx\vts \kpar}\, u(\xv), \quad \zeta \mapsto u (\xv + \zeta\, \vvh_2) \in L^2 (\R), \quad \xv \in \R^2,
    \end{equation*}
    Then Theorem \ref{thm:resolvent_expansion} reduces to:
    \begin{theorem}\label{thm:resolvent_expansion_rational}
    Let $\vvh_1 = \vvv_1 + r \vvv_2$, and assume that $r = b_1 / a_1$ where $b_1 \in \Z$ and $a_1 \in \N$ are coprime, with the convention that $(a_1, b_1) = (1, 0)$ if $r = 0$. Let $\Kv_\star \in \{\Kv, \Kv'\}$ and consider the operator $\calT_{\Kv_\star} \in \scrL(L^2_{\kpar} (\Sigma), L^2(\R; \C^2))$ and its adjoint
    \begin{align*}
        \spforall u \in L^2_{\kpar} (\Sigma), \qquad \calT_{\Kv_\star}\, u(\zeta) &:= \int_0^{1} (\overline{\Phi^{\Kv}}\, u)\vts (\vts \yrm_1\vts (a_1 \vvh_1) + \zeta\vts \vvh_2 \vts)\, d\yrm_1,
        \rets
        \spforall \alpha \in L^2(\R; \C^2), \qquad \calT^*_{\Kv_\star}\, \alpha (\xv) &\phantom{:}= \Phi^{\Kv_\star} (\xv)^\transp\, \alpha (\kvh_2 \cdot \xv).
    \end{align*}
    From these operators, define $\calJ_{\delta, \Kv_\star} (\mu) := \calU_\delta\, \calT_{\Kv_\star}\, \euler^{- \icplx\vts \delta \mu\vts \kvh_1 \cdot \xv}$ where $\euler^{- \icplx\vts \delta \mu\vts \kvh_1 \cdot \xv}$ is identified with the multiplication operator. Moreover, let $T_a: \alpha \in L^2(\R) \mapsto \alpha(\cdot - a)$ and $\eta  > 0$. Under the assumptions of Theorem \ref{thm:resolvent_expansion}, there exist constants $\delta_0, \mu_0 > 0$ such that if
    \begin{equation*}
        \delta \in (0, \delta_0), \quad \mu \in (-\mu_0, \mu_0), \quad \min_{\Kv_\star \in \{\Kv, \Kv'\}} \Big[ \dist \big(z,\; \spec \calD^{\Kv_\star} (\mu) \big) \Big] > \eta,
    \end{equation*}
    then setting $\kpar = \Kv \cdot \vvh_1 + \delta \mu$, the operator $\calH^\delta_{\AUG, \kpar} - E_D - \delta\, z$ is invertible on $L^2_{\kpar} (\cylaug)$, and the following holds: 
    \begin{enumerate}[label={$(\alph*).$}, ref={$\alph*$}]
        \item \textit{Zigzag-type edges}: if $a_1 \neq b_1 \mod 3$, then
        \begin{multline} \label{eq:resolvent_expansion_zigzag}
        \bigg( \frac{\calH^\delta_{\AUG, \kpar} - E_D}{\delta} - z \bigg)^{-1} = \int^\oplus_\R \calJ_{\delta, \Kv} (\mu)^*\, T^*_{\delta\vts s}\, \big( \calD^{\Kv} (\mu) - z \big)^{-1}\, T_{\delta\vts s}\, \calJ_{\delta, \Kv} (\mu)\; ds 
        + \calO( \delta^{1/4} )
        \end{multline}
        \item \textit{Armchair-type edges}: if $a_1 = b_1 \mod 3$, then
        \begin{multline} \label{eq:resolvent_expansion_armchair}
        \bigg( \frac{\calH^\delta_{\AUG, \kpar} - E_D}{\delta} - z \bigg)^{-1}\\
        = \sum_{\Kv_\star \in \{\Kv, \Kv'\}} \int^\oplus_\R \calJ_{\delta, \Kv_\star} (\mu)^*\, T^*_{\delta\vts s}\, \big( \calD^{\Kv_\star} (\mu) - z \big)^{-1}\, T_{\delta\vts s}\, \calJ_{\delta, \Kv_\star} (\mu)\; ds 
        + \calO( \delta^{1/4} ),\ 
        \end{multline}
    \end{enumerate}
    where the error terms in \eqref{eq:resolvent_expansion_zigzag} and \eqref{eq:resolvent_expansion_armchair} are with respect to the $\scrL(L^2_{\kpar} (\cylaug), H^2_{\kpar, \xv} (\cylaug))$ norm. 
    \end{theorem}

    \begin{dem}
    For simplicity, we present the calculations for zigzag-type edges and for $\mu = 0$, the extension to armchair-type edges or to $\mu \neq 0$ being straightforward. By Proposition \ref{prop:veps_neighborhood_K_rational}, for zigzag-type edges, one has $\bbL (\delta^{3/4}) = \{\Kv \} \times a_1 \Z$ and $\gamma_\iiv = 0$ for $\iiv \in \bbL (\delta^{3/4})$ and for $\delta$ small enough. In particular, from \eqref{eq:mode_for_averaging_operator}, $\varphi_{(\Kv, m)} (\xv, s) := \euler^{\vts \icplx\vts (\Kv \cdot \vvv_2 + 2\pi m)\vts s}\, \Phi^{\Kv} (\xv)$ for $m \in a_1 \Z$. Our goal is to show that in this case, the dominant term in \eqref{eq:resolvent_expansion_Haug} coincides with the dominant term in \eqref{eq:resolvent_expansion_zigzag}, that is,
    \begin{equation}\label{eq:dominant_terms_equality_rational_edges}
        \sum_{\iiv \in \{\Kv \} \times a_1 \Z} \calJ^*_{\delta, \iiv}\; \big(\vts \calD^\Kv (0) - z \vts\big)^{-1}\; \calJ^{}_{\delta, \iiv} = \int^\oplus_\R \calT^*_\Kv\, \calU^*_\delta\, T^*_{\delta\vts s}\, \big( \calD^{\Kv} (0) - z \big)^{-1}\, T_{\delta\vts s}\, \calU^{}_\delta\, \calT^{}_\Kv\; ds.
    \end{equation}
    Fix $m \in a_1 \Z$. By a density argument, since the operators in \eqref{eq:dominant_terms_equality_rational_edges} are bounded, it is sufficient to show \eqref{eq:dominant_terms_equality_rational_edges} for functions $F \in L^2_{\Kv \cdot \vvh_1} (\cylaug)$ such that $\euler^{- \icplx\, \Kv \cdot \xv}\, F \in \scrC^\infty_0 (\cylaug)$. For any $\zeta \in \R$, using the change of variables $\yv \mapsto \yv + \zeta \vvv_2$ in the definition \eqref{eq:averaging_operator} of $\calT_\iiv$ leads to
    \begin{equation*}
        \calT_{(\Kv, m)}\, F (\zeta) = \int_{\Omega - \zeta \vvv_2} (\vts\overline{\varphi_{(\Kv, m)}}\, F\vts)\vts (\yv + \zeta \vvv_2, -\kvh_2 \cdot \yv)\, d\yv = \int_{\Omega} (\vts\overline{\varphi_{(\Kv, m)}}\, F\vts)\vts (\yv + \zeta \vvv_2, -\kvh_2 \cdot \yv)\, d\yv,
    \end{equation*}
    where the last equality follows from the $\Lambda$--periodicity of the integrand with respect to $\yv$. We apply the operator $\calU^*_\delta\, (\calD^{\Kv} (0) - z)^{-1}\, \calU^{}_\delta$, interchange it with the integral over $\Omega$, and use the simplified expression for $\varphi_{(\Kv, m)}$ to obtain
    \begin{align*}
        \calU^*_\delta\, (\calD^{\Kv} (0) &- z)^{-1}\, \calU^{}_\delta\, \calT_{(\Kv, m)}\, F (\zeta) = \int_\Omega \euler^{\icplx\vts 2\pi m\vts \kvh_2 \cdot \yv}\, g(\yv, \zeta)\, d\yv, \textWHERE 
        \rets
        g(\yv, \cdot) &:= \euler^{\icplx\vts (\Kv \cdot \vvv_2)\, \kvh_2 \cdot \yv}\; \big[\vts \calU^*_\delta\, (\calD^{\Kv} (0) - z)^{-1}\, \calU^{}_\delta \vts\big]\, f (\yv, \cdot), \textWITH 
        \retss
        f(\yv, \zeta) &:= \overline{\Phi^{\Kv} (\yv + \zeta \vvv_2)}\, F (\yv + \zeta \vvv_2, -\kvh_2 \cdot \yv).
    \end{align*}
    Applying $\calT^*_{(\Kv, m)}$ and using $\calJ^{}_{\delta, (\Kv, m)} = \calU^{}_\delta\, \calT^{}_{(\Kv, m)}$, we find that the left-hand side of \eqref{eq:dominant_terms_equality_rational_edges} is given by
    \begin{multline}
        \!\!\!\sum_{\iiv \in \{\Kv \} \times a_1 \Z} \!\!\!\!\calT^*_\iiv\, \calU^*_\delta\, (\calD^{\Kv} (0) - z)^{-1}\, \calU^{}_\delta\, \calT^{}_\iiv\, F\vts (\xv, s) = \sum_{m \in a_1 \Z} \calT^*_{(\Kv, m)}\, \calU^*_\delta\, (\calD^{\Kv} (0) - z)^{-1}\, \calU^{}_\delta\, \calT^{}_{(\Kv, m)}\, F\vts (\xv, s)
        \\
        = \euler^{\vts \icplx\vts (\Kv \cdot \vvv_2)\vts s}\, \Phi^{\Kv} (\xv)^\transp \sum_{m \in a_1 \Z} \euler^{\vts \icplx\vts 2\pi m\vts s}\, \int_\Omega \euler^{\vts \icplx\vts 2\pi m\vts \kvh_2 \cdot \yv}\, g(\yv, \kvh_2 \cdot (\xv + s \vvv_2))\, d\yv.\label{eq:dominant_term_rational_1}
    \end{multline}
    To further reformulate \eqref{eq:dominant_term_rational_1}, we will rely on the following Fourier series identity: since the vector $2\pi\vts a_1 \kvh_2 = 2\pi (-b_1 \kv_1 + a_1 \kv_2)$ belongs to $\Lambda^*$ and $|\Omega| = |\vvv_1 \wedge \vvv_2| = 1$, for $u \in L^2(\R^2 / \Lambda)$,
    \begin{equation}\label{eq:dominant_term_rational_2}
        \sum_{m \in a_1\, \Z} \euler^{\vts \icplx\vts 2\pi m \vts s}\, \int_\Omega \euler^{\vts \icplx\vts 2\pi m\vts \kvh_2 \cdot \yv}\, u (\yv)\, d\yv = \int_0^{1} u(\yrm_1\vts (a_1 \vvh_1) - s\vts \vvh_2)\, d\yrm_1, \quad s \in \R.
    \end{equation}
    In fact, this identity is easily verified for $u(\yv) = \euler^{\vts \icplx\vts 2\pi\vts \kv \cdot \yv}$, $\kv \in \Lambda^*$; it extends to trigonometric polynomials by linearity, and then to elements of $L^2(\R^2 / \Lambda)$ by density. 

    We now note that $\zeta \mapsto g(\yv, \zeta) \in L^2 (\R)$ and $\yv \mapsto g(\yv, \zeta) \in L^2 (\R^2 / \Lambda)$, the latter following from the definition \eqref{eq:def_L2aug} of $L^2_{\Kv \cdot \vvh_1} (\cylaug)$ and from the fact that $\Phi^{\Kv} \in L^2_{\Kv}$. This allows us to apply the identity \eqref{eq:dominant_term_rational_2} to $u: \yv \mapsto g(\yv, \kvh_2 \cdot (\xv + s \vvv_2))$ so that \eqref{eq:dominant_term_rational_1} becomes
    \begin{align}
        \sum_{\iiv \in \{\Kv \} \times a_1 \Z} &\calT^*_\iiv\, \calU^*_\delta\, (\calD^{\Kv} (0) - z)^{-1}\, \calU^{}_\delta\, \calT^{}_\iiv\, F\vts (\xv, s) \nonumber
        \\
        &= \euler^{\vts \icplx\vts (\Kv \cdot \vvv_2)\vts s}\, \Phi^{\Kv} (\xv)^\transp\, \Big[\vts \int_0^{1} g\vts \big(\yrm_1\vts (a_1 \vvh_1) - s\vts \vvh_2,\, \kvh_2 \cdot (\xv + s \vvv_2)\big)\, d\yrm_1 \vts\Big] \nonumber
        \\
        &= \Phi^{\Kv} (\xv)^\transp\, \bigg\{\vts \calU^*_\delta\, (\calD^{\Kv} (0) - z)^{-1}\, \calU^{}_\delta \, \Big[\vts \int_0^1 f\vts (\yrm_1\vts (a_1 \vvh_1) - s\vts \vvh_2,\, \cdot\vts)\, d\yrm_1 \vts\Big]\, \big(\kvh_2 \cdot (\xv + s \vvv_2) \big) \vts\bigg\} \nonumber
        \\
        &= \bigg\{\vts \calT^*_{\Kv}\, T^*_s\, \calU^*_\delta\, (\calD^{\Kv} (0) - z)^{-1}\, \calU^{}_\delta\, \Big[\vts \int_0^1 f\vts (\yrm_1\vts (a_1 \vvh_1) - s\vts \vvh_2,\, \cdot\vts)\, d\yrm_1 \vts\Big] \vts\bigg\}\, (\xv).\label{eq:dominant_term_rational_3}
    \end{align}
    Finally, using the definition of $f$,
    \begin{align}
        \int_0^1 f\vts (\yrm_1\vts (a_1 \vvh_1) - s\vts \vvh_2,  \zeta)\, d\yrm_1 &= \int_0^1 \overline{\Phi^{\Kv} (\yrm_1\vts (a_1 \vvh_1) + (\zeta - s)\vts \vvv_2)}\, F (\yrm_1\vts (a_1 \vvh_1) + (\zeta - s)\vts \vvv_2, s)\, d\yrm_1 \nonumber
        \retss
        &= \big[\vts T^{}_s\, \calT^{}_\Kv F (\cdot\,, s) \vts\big] (\zeta). \label{eq:dominant_term_rational_4}
    \end{align}
    By combining \eqref{eq:dominant_term_rational_3} and \eqref{eq:dominant_term_rational_4}, we find that the left-hand side of \eqref{eq:dominant_terms_equality_rational_edges} is given by 
    \begin{equation*}
        \sum_{\iiv \in \{\Kv \} \times a_1 \Z} \calT^*_\iiv\, \calU^*_\delta\, (\calD^{\Kv} (0) - z)^{-1}\, \calU^{}_\delta\, \calT^{}_\iiv\, F\vts (\cdot\,, s) = \calT^*_{\Kv}\, \calU^*_\delta\, T^*_{\delta\vts s}\, (\calD^{\Kv} (0) - z)^{-1}\, T^{}_{\delta\vts s}\, \calU^{}_\delta\, \calT^{}_\Kv F (\cdot\,, s),
    \end{equation*}
    where the right-hand side in the above equality is exactly the right-hand side in \eqref{eq:dominant_terms_equality_rational_edges}. Therefore, the proof of \eqref{eq:dominant_terms_equality_rational_edges}, or equivalently, Theorem \ref{thm:resolvent_expansion_rational}, is complete.
    \end{dem}

    \vspace{1\baselineskip} \noindent
    Theorem \ref{thm:resolvent_expansion_rational} can be established more directly by exploiting the fibered nature of $\calH^\delta_{\AUG, \kpar}$, together with results on rational edges. In fact, for $s \in \R$ and $\kpar \in [-\pi, \pi)$, consider the two-dimensional Schrödinger operator
    \begin{equation*}
    \calH^\delta_{\kpar} (s) := -\nabla \cdot A^\delta_\AUG (\xv, s)\, \nabla + V (\xv)\quad \textnormal{acting on $L^2_{\kpar} (\Sigma)$}.
    \end{equation*}
    The operator $\smash{\calH^\delta_{\kpar} (s)}$, $s \in \R$, is a dislocation of $\smash{\calH^\delta_{\kpar} (0) = -\nabla \cdot A^\delta (\xv)\, \nabla + V(\xv)}$, whose eigenpairs define edge states; see also Remark \ref{rmk:quasiperiodic_nature_domain_wall_potential_a}. In particular, given the dependence of $A^\delta_\AUG$ in $s$, the results in \cite{fefferman2016edge,lee2019elliptic,drouot2019characterization,drouot2020edge} can be applied to $\smash{\calH^\delta_{\kpar} (s)}$ by simply replacing the domain wall $\kappa (\zeta)$ with $\kappa(\zeta + \delta\vts s)$. Furthermore, by definition \eqref{eq:def_L2aug} of $\smash{L^2_{\kpar} (\cylaug)}$,
    \begin{multline*}
    L^2_{\kpar} (\cylaug) = \int^\oplus_\R L^2_{\kpar} (\Sigma)\, ds \textAND \calH^\delta_{\AUG, \kpar} = \int^\oplus_\R \calH^\delta_{\kpar} (s)\, ds, \quad \textnormal{that is,}
    \rets
    \spforall F \in L^2_{\kpar} (\cylaug), \quad \xv \mapsto F(\xv, s) \in L^2_{\kpar} (\Sigma) \textAND (\vts \calH^\delta_{\AUG, \kpar}\, F\vts)\vts (\cdot\,, s) = \calH^\delta_{\kpar} (s)\, F(\cdot\,, s).
    \end{multline*}
    In particular, applying the resolvent expansion from \cite{drouot2020edge} to each $\calH^\delta_{\kpar} (s)$, $s \in \R$, leads directly to Theorem \ref{thm:resolvent_expansion_rational}, with \cite{drouot2020edge} even yielding an $\calO(\delta^{1/3})$--error term instead of $\calO(\delta^{1/4})$. We refer to Remark \ref{rmk:technical_exponent} for an explanation of this difference in error orders.

    \section{Application of Theorem \ref{thm:resolvent_expansion} to the functional calculus}\label{sec:functional_calculus}
    \noindent
    Theorem \ref{thm:resolvent_expansion} allows us to expand functions of $\delta^{-1}\, (\calH^\delta_{\AUG, \kpar} - E_D)$ using the Helffer-Sjöstrand formula \cite[Proposition 7.2]{helffer1989equation}, \cite{davies1995functional}: for $w \in \scrC^\infty_0 (\R)$, and for a self-adjoint operator $\calA$ on a Hilbert space $\scrX$,
    \begin{equation}\label{eq:helffer_sjostrand_formula}
    w(\calA) = \frac{1}{\icplx\vts 2\pi}\, \int_\C \frac{\partial \widetilde{w}}{\partial \overline{z}} (z)\, (\calA - z)^{-1}\, dz \wedge d\overline{z},
    \end{equation}
    with $\partial_{\overline{z}} := (\vts \partial_x + \icplx\, \partial_y \vts) / 2$ for $z = x + \icplx\vts y$, and where $\widetilde{w}$ is an \emph{almost analytic extension} of $w$, namely $\widetilde{w} \in \scrC^\infty_0 (\C)$, $\widetilde{w}|_\R = w$, and for any $N \geq 0$, there exists a constant $C_N > 0$ such that
    \begin{equation}\label{eq:almost_analytic_extension_bound}
    \Big|\vts \frac{\partial \widetilde{w}}{\partial \overline{z}} (z) \vts\Big| \leq C_N\, \| w^{(N + 1)} \|_\infty\, |\Imag z|^N, \quad \spforall z \in \C.
    \end{equation}
    Using the Helffer-Sjöstrand formula, Theorem \ref{thm:resolvent_expansion} leads to the following result.

    \begin{corollary}\label{cor:functional_calculus_estimates}
    Let $w \in \scrC^\infty_0 (\R)$ with $\operatorname{supp} w \subset (-\theta_\GAP, \theta_\GAP)$. For $\eta_0, \mu_0 > 0$ and $N \in \N$, there exists a constant $C > 0$ such that for $\eta \in (0, \eta_0)$, one can find $\delta_0 \equiv \delta_0 (\eta, \mu_0, N) > 0$ with the following estimate: For any $\delta \in (0, \delta_0)$, $\mu \in [-\mu_0, \mu_0]$ and $\kpar = \Kv \cdot \vvh_1 + \delta \mu$,
    \begin{multline}\label{eq:HS_estimate}
        \bigg\|\vts w \vts \Big(\vts \frac{\calH^\delta_{\AUG, \kpar} - E_D}{\delta} \vts\Big) - \sum_{\iiv \in \bbL (\delta^{3/4})} \sum_{j = -N}^N w\vts \big( z_j (\mu + \delta^{-1} \gamma_\iiv) \big)\; \Psi^{\delta, (0)}_{\AUG, \iiv, j} \hyperref[item:notation_otimes]{\otimes} \Psi^{\delta, (0)}_{\AUG, \iiv, j} \vts\bigg\|_{\scrL(L^2_{\kparsubscript} (\cylaug))}
        \\
        \leq C\, \eta^N\, \| w^{(N+1)}\|_\infty.
    \end{multline}
    Here, $z_j(\muhat)$ is an eigenvalue of $\calD^{\Kv_\iiv} (\muhat)$, and $\Psi^{\delta, (0)}_{\AUG, \iiv, j}$ is introduced in (\ref{eq:approximate_augmented_eigenfunction_1}, \ref{eq:approximate_augmented_eigenfunction_2}) as an $\calO(\delta)$--approximate eigenfunction of $\smash{\calH^\delta_{\AUG, \kpar}}$ with the approximate eigenvalue $E_D + \delta\vts z_j(\mu + \delta^{-1} \gamma_\iiv)$; see also Section \ref{sec:recap_multiscale}.
    \end{corollary}

    \begin{dem}
        We begin with the Helffer-Sjöstrand formula \eqref{eq:helffer_sjostrand_formula} with $\calA := \delta^{-1}\, (\calH^\delta_{\AUG, \kpar} - E_D)$:
        \begin{equation*}
            w \vts \Big(\vts \frac{\calH^\delta_{\AUG, \kpar} - E_D}{\delta} \vts\Big) = \frac{1}{\icplx\vts 2\pi}\, \int_\C \frac{\partial \widetilde{w}}{\partial \overline{z}} (z)\, \Big(\vts \frac{\calH^\delta_{\AUG, \kpar} - E_D}{\delta} - z \vts\Big)^{-1}\, dz \wedge d\overline{z}. 
        \end{equation*}
        We would now like to apply Theorem \ref{thm:resolvent_expansion}; however, because the resolvent expansion holds only when $z$ is bounded away from the spectrum of $\calD^\delta (\mu)$, which can be the full real axis, we use an approximation argument. More precisely, following closely \cite[Step $3$ in the proof of Theorem 3.7]{shapiro2022tight}, for a self-adjoint operator $\calA$ and for $\eta > 0$, we define
        \begin{equation*}
            w_\eta \vts (\calA) := \frac{1}{\icplx\vts 2\pi}\, \int_{\{z' \in \C\, /\, |\Imag z'| > \eta\}} \frac{\partial \widetilde{w}}{\partial \overline{z}} (z)\, (\vts \calA - z \vts)^{-1}\, dz \wedge d\overline{z}.
        \end{equation*}
        (We emphasize that $w_\eta (\calA)$ is simply a notation, rather than the result of applying a function to $\calA$.) Next, we write
        \begin{equation*}
            w \vts \Big(\vts \frac{\calH^\delta_{\AUG, \kpar} - E_D}{\delta} \vts\Big) - \euler^{\vts\icplx\vts \delta \mu \kvh_1 \cdot \xv}\, \calJ^*_\delta\; \mathds{1}_{\bbL (\delta^{3/4})}\; w (\calD^{\vts \delta} (\mu))\; \mathds{1}_{\bbL (\delta^{3/4})}\; \calJ^{}_\delta\, \euler^{-\icplx\vts \delta \mu \kvh_1 \cdot \xv} = \calN_1 + \calN_2 + \calN_3,
        \end{equation*}
        where
        \begin{align*}
            \calN_1 &:= w \vts \Big(\vts \frac{\calH^\delta_{\AUG, \kpar} - E_D}{\delta} \vts\Big) - w_\eta \vts \Big(\vts \frac{\calH^\delta_{\AUG, \kpar} - E_D}{\delta} \vts\Big),
            \retss
            \calN_2 &:= w_\eta \vts \Big(\vts \frac{\calH^\delta_{\AUG, \kpar} - E_D}{\delta} \vts\Big) - \euler^{\vts\icplx\vts \delta \mu \kvh_1 \cdot \xv}\, \calJ^*_\delta\; \mathds{1}_{\bbL (\delta^{3/4})}\; w_\eta (\calD^{\vts \delta} (\mu))\; \mathds{1}_{\bbL (\delta^{3/4})}\; \calJ^{}_\delta\, \euler^{-\icplx\vts \delta \mu \kvh_1 \cdot \xv},
            \retss
            \calN_3 &:= \euler^{\vts\icplx\vts \delta \mu \kvh_1 \cdot \xv}\, \calJ^*_\delta\; \mathds{1}_{\bbL (\delta^{3/4})}\; \big[\vts w_\eta (\calD^{\vts \delta} (\mu)) - w (\calD^{\vts \delta} (\mu))\vts\big]\; \mathds{1}_{\bbL (\delta^{3/4})}\; \calJ^{}_\delta\, \euler^{-\icplx\vts \delta \mu \kvh_1 \cdot \xv}.
        \end{align*} 
        We first bound $\calN_1$ and $\calN_3$. In what follows, $z = x + iy$. Using the Helffer-Sjöstrand formula and the definition of $w_\eta (\calA)$, one obtains
        {
        \allowdisplaybreaks
        \begin{align}
            &\|\vts \calN_1 \vts\|_{\scrL(L^2_{\kparsubscript} (\cylaug))} = \bigg\|\vts \frac{1}{\icplx\vts 2\pi} \int_{\{z' \in \C\, /\, |\Imag z'| \leq \eta\}} \frac{\partial \widetilde{w}}{\partial \overline{z}} (z)\, \Big(\vts \frac{\calH^\delta_{\AUG, \kpar} - E_D}{\delta} - z \vts\Big)^{-1}\, dz \wedge d\overline{z} \vts\bigg\|_{\scrL(L^2_{\kparsubscript} (\cylaug))} \nonumber
            \retss
            &\leq \frac{1}{2\pi} \int_{\{z' \in \C\, /\, |\Imag z'| \leq \eta\}\, \cap\, \operatorname{supp} \widetilde{w}} \bigg| \frac{\partial \widetilde{w}}{\partial \overline{z}} (z)\bigg|\, \bigg\| \Big(\vts \frac{\calH^\delta_{\AUG, \kpar} - E_D}{\delta} - z \vts\Big)^{-1} \bigg\|_{\scrL(L^2_{\kparsubscript} (\cylaug))}\, dz \wedge d\overline{z} \nonumber
            \retss
            &\lesssim \frac{C\, C_N}{2\pi}\, \| w^{(N + 1)} \|_\infty\, \int_{-\eta}^\eta |y|^N \frac{1}{|y|}\, dy \lesssim \| w^{(N + 1)} \|_\infty\, \eta^N, \label{eq:cor_H_S_est_1}
        \end{align}
        uniformly in $\delta$, where the last inequalities follow from the the bound \eqref{eq:almost_analytic_extension_bound} on $\partial_{\overline{z}}\, \widetilde{w}$ and the resolvent estimate $\|(\delta^{-1}\, (\calH^\delta_{\AUG, \kpar} - E_D) - z)^{-1}\| \leq |\Imag z|^{-1} = |y|^{-1}$. Similarly, since $\calJ_\delta$ is bounded uniformly in $\delta$ (Proposition \ref{eq:square_summability_estimate}) we find that}
        \begin{equation}
            \|\vts \calN_3 \vts\|_{\scrL(L^2_{\kparsubscript} (\cylaug))} \lesssim \| w^{(N + 1)} \|_\infty\, \eta^N. \label{eq:cor_H_S_est_2}
        \end{equation}
        We now bound $\calN_2$. By the triangle inequality,
        \begin{equation*}
            \|\vts \calN_2 \vts\|_{\scrL(L^2_{\kparsubscript} (\cylaug))} \leq \frac{1}{2\pi}\, \sup_{|\Imag z| > \eta} \|\scrP^\delta (z)\|_{\scrL(L^2_{\kparsubscript} (\cylaug))} \int_{\{z' \in \C\, /\, |\Imag z'| > \eta\}} \bigg| \frac{\partial \widetilde{w}}{\partial \overline{z}} (z)\bigg|\, dz \wedge d\overline{z},
        \end{equation*}
        where 
        \begin{equation*}
            \scrP^\delta (z) := \bigg(\vts \frac{\calH^\delta_{\AUG, \kpar} - E_D}{\delta} - z \vts\bigg)^{-1}\!\! - \euler^{\vts\icplx\vts \delta \mu \kvh_1 \cdot \xv}\, \calJ^*_\delta\; \mathds{1}_{\bbL (\delta^{3/4})}\; \big(\vts \calD^{\vts \delta} (\mu) - z \vts\big)^{-1}\; \mathds{1}_{\bbL (\delta^{3/4})}\; \calJ^{}_\delta\, \euler^{-\icplx\vts \delta \mu \kvh_1 \cdot \xv}.
        \end{equation*}
        If $z$ satisfies $\Imag z \geq \eta$, then its distance to the spectrum of $\calD^{\vts \delta} (\mu)$ is bounded from below by $1/\eta$. Thus, by Theorem \ref{thm:resolvent_expansion}, there exists $\delta_\star (\eta, \mu_0) > 0$ such that for any $\delta \in (0, \delta_\star (\eta, \mu_0))$,
        \begin{equation*}
            \sup_{|\Imag z| > \eta} \|\scrP^\delta (z)\|_{\scrL(L^2_{\kparsubscript} (\cylaug))} \leq C\, \eta^{-1}\, \delta^{1/4}.
        \end{equation*}
        It follows that
        \begin{align}
            \|\vts \calN_2 \vts\|_{\scrL(L^2_{\kparsubscript} (\cylaug))} &\leq \frac{C}{2\pi}\, \eta^{-1}\, \delta^{1/4}\, \int_{\{z' \in \C\, /\, |\Imag z'| > \eta\}\, \cap\, \operatorname{supp} \widetilde{w}} \bigg| \frac{\partial \widetilde{w}}{\partial \overline{z}} (z)\bigg|\, dz \wedge d\overline{z} \nonumber
            \retss
            &\hspace{-1.5cm}\leq \frac{C C_N}{2\pi}\, \| w^{(N + 1)} \|_\infty\, \eta^{-1}\, \delta^{1/4}\, \int_{\{z' \in \C\, /\, |\Imag z'| > \eta\}\, \cap\, \operatorname{supp} \widetilde{w}} |\Imag z|^N\, dz \wedge d\overline{z} \quad \textnormal{by \eqref{eq:almost_analytic_extension_bound}}\nonumber
            \retss
            &\hspace{-1.5cm}\lesssim \| w^{(N + 1)} \|_\infty\, \eta^{-1}\, \delta^{1/4}.\label{eq:cor_H_S_est_3}
        \end{align}
        Combining the estimates \eqref{eq:cor_H_S_est_1}--\eqref{eq:cor_H_S_est_3}, we find that for $\delta \in (0, \delta_\star (\eta, \mu_0))$,
        \begin{multline}
            \bigg\|\vts w \vts \Big(\vts \frac{\calH^\delta_{\AUG, \kpar} - E_D}{\delta} \vts\Big) - \euler^{\vts\icplx\vts \delta \mu \kvh_1 \cdot \xv}\, \calJ^*_\delta\; \mathds{1}_{\bbL (\delta^{3/4})}\; w (\calD^{\vts \delta} (\mu))\; \mathds{1}_{\bbL (\delta^{3/4})}\; \calJ^{}_\delta\, \euler^{-\icplx\vts \delta \mu \kvh_1 \cdot \xv} \vts\bigg\|_{\scrL(L^2_{\kparsubscript} (\cylaug))}
            \\
            \lesssim \| w^{(N+1)}\|_\infty\, (\eta^N + \eta^{-1}\, \delta^{1/4}).\label{eq:cor_H_S_exp}
        \end{multline}
        Finally, setting $\delta_0 := \min (\eta^{4(N+1)}, \delta_\star (\eta, \mu_0))$ so that $\eta^{-1}\, \delta^{1/4} \leq \eta^N$ for any $\delta \in (0, \delta_0)$, it follows that the right-hand side of \eqref{eq:cor_H_S_exp} is bounded by $\| w^{(N+1)}\|_\infty\, \eta^N$ whenever $\delta \in (0, \delta_0)$.

        To finish, we rewrite the second term on the left-hand side of \eqref{eq:cor_H_S_exp}. Because $w$ is compactly supported in $(-\theta_\GAP, \theta_\GAP)$ and the spectrum of $\calD^\delta (\mu)$ in that interval is pure point,
        \begin{equation*}
            w\vts (\calD^{\vts \delta} (\mu)) = \sum_{j = -N}^N %
            \begin{pmatrix}
            \ddots & & (0)
            \\
            & w(z_j (\mu + \delta^{-1} \gamma_\iiv))\, \alpha^{\Kv_\iiv} (\cdot; \mu + \delta^{-1} \gamma_\iiv) \otimes \alpha^{\Kv_\iiv} (\cdot; \mu + \delta^{-1} \gamma_\iiv) &  
            \\
            (0) & & \ddots
            \end{pmatrix}.
        \end{equation*}
        Combining this expression with the definition \eqref{eq:approximate_augmented_eigenfunction_2} of $\Psi^{\delta, (0)}_{\AUG, \iiv, j}$ shows that the second term on the left-hand side of \eqref{eq:cor_H_S_exp} coincides with the second term on the left-hand side of \eqref{eq:HS_estimate}. This completes the proof.
    \end{dem}

    \section{The partial Fourier transform, proof of Lemma \ref{lem:directional_Fourier}}
    \label{sec:directional_Fourier}
    \noindent
    The goal of this section is to prove Lemma \ref{lem:directional_Fourier}.

    \begin{dem}
        \textit{Point $(a)$. Continuity of $\scrF_\Kv$}. Let $F \in L^2_{\Kv \cdot \vvh_1} (\cylaug)$ with $\euler^{- \icplx \vts \Kv \cdot (\xv + s \vvv_2)}\; F \in \scrC^\infty_0 (\cylaug)$. The calculations following the statement of the lemma show that $\scrF_\Kv F(\cdot\,, \lvar)$ is $\Kv + \lvar \kvh_2$--pseudo-periodic with respect to $\Lambda$ for any $\lvar \in \R$. Furthermore, by the Plancherel formula for the one-dimensional Fourier transform in $s$,
        \begin{equation*}
            \int_\R |F(\xv, s)|^2\, ds = \int_\R |\scrF_\Kv F(\xv, \lvar)|^2\, d\lvar, \quad \xv \in \Omega.
        \end{equation*}
        Integrating this relation over $\xv \in \Omega$ yields $\|F\|_{L^2(\Omega \times \R)} = \|\scrF_\Kv F\|_{L^2(\Omega \times \R)}$, which is a continuity estimate for smooth functions. Since $\scrC^\infty_0 (\cylaug)$ is dense in $L^2(\cylaug)$, it follows that $\scrF_\Kv$ extends uniquely as a continuous map from $L^2_{\Kv \cdot \vvh_1} (\cylaug)$ to $L^2(\R_\lvar; L^2_{\Kv + \lvar \kvh_2})$.

        \vspace{1\baselineskip} \noindent
        \textit{Point $(b)$ -- $(e)$} The expression \eqref{eq:directional_Fourier-adjoint} for the adjoint $\scrF^*_{\Kv}$ follows directly from the definition of the adjoint of the one-dimensional Fourier transform in $s$. \eqref{eq:directional_Plancherel} follows by applying the Plancherel formula for the one-dimensional Fourier transform in $s$ to $F (\xv, \cdot), G(\xv, \cdot) \in L^2(\R)$, and by integrating the result over $\xv \in \Omega$. The relations \eqref{eq:inverse_directional_Fourier} follow from the unitarity nature of the one-dimensional Fourier transform in $s$. The commutation property \eqref{eq:Fourier_commutes_with_x_functions} is direct for smooth functions; it extends to $L^2_{\Kv \cdot \vvh_1} (\cylaug)$ by density of $\scrC^\infty_0 (\cylaug)$ in $L^2 (\cylaug)$. 

        \vspace{1\baselineskip} \noindent
        \textit{Point $(d)$}. Although \eqref{eq:Fourier_commutes_with_differential_operators} can be easily verified for smooth functions, its proof in the general case requires more care. For $F \in H^1_{\Kv \cdot \vvh_1, \xv} (\cylaug)$, we first show that \eqref{eq:Fourier_commutes_with_differential_operators} holds in a weak sense. For any test function $\widetilde{G} (\xv, \lvar) \in \scrC^\infty_0 (\Omega \times \R)$, derivation in the sense of distributions combined with $\scrF_\Kv\vts F \in L^2(\Omega \times \R)$ leads to
        \begin{align*}
            \big\langle\vts \widetilde{G},\, \nabla_\xv\, \scrF_\Kv\vts F \vts\big\rangle_{[\scrC^\infty_0 (\Omega \times \R_\lvar)]',\vts \scrC^\infty_0 (\Omega \times \R_\lvar)} &:= -\big\langle\vts \nabla_\xv\, \widetilde{G},\, \scrF_\Kv\vts F \vts\big\rangle_{[\scrC^\infty_0 (\Omega \times \R_\lvar)]',\vts \scrC^\infty_0 (\Omega \times \R_\lvar)}
            \\
            &= -\int_\R \int_\Omega \overline{\nabla_\xv\, \widetilde{G} (\xv, \lvar)}\, \scrF_\Kv\vts F (\xv, \lvar) \, d\xv d\lvar
            \\
            &= -\int_\R \int_\Omega \overline{\scrF_\Kv^*\, \nabla_\xv\, \widetilde{G} (\xv, s)}\, F (\xv, s) \, d\xv ds,
            \intertext{by the Plancherel-like formula \eqref{eq:directional_Plancherel} and the relations \eqref{eq:inverse_directional_Fourier}. Since $\widetilde{G}$ is a smooth function, the integral involved in the definition \eqref{eq:directional_Fourier-adjoint} of $\scrF_\Kv^*$ can be interchanged with $\nabla_\xv$, to obtain $\scrF_\Kv^*\, \nabla_\xv\, \widetilde{G} = \nabla_\xv\, \scrF_\Kv^*\, \widetilde{G}$. As a consequence,}
            \big\langle\vts \widetilde{G},\, \nabla_\xv\, \scrF_\Kv\vts F \vts\big\rangle_{[\scrC^\infty_0 (\Omega \times \R_\lvar)]',\vts \scrC^\infty_0 (\Omega \times \R_\lvar)} &= -\int_\R \int_\Omega \overline{\nabla_\xv\, \scrF_\Kv^* \widetilde{G} (\xv, s)}\, F (\xv, s) \, d\xv ds
            \\
            &= -\big\langle\vts \nabla_\xv\, \scrF_\Kv^* \widetilde{G},\, F \vts\big\rangle_{[\scrC^\infty_0 (\Omega \times \R_s)]',\vts \scrC^\infty_0 (\Omega \times \R_s)}
            \\
            &=: \big\langle\vts \scrF_\Kv^*\, \widetilde{G},\, \nabla_\xv\, F \vts\big\rangle_{[\scrC^\infty_0 (\Omega \times \R_s)]',\vts \scrC^\infty_0 (\Omega \times \R_s)},
            \intertext{where the last equality is obtained using derivation in the sense of distributions. Given that $\nabla_\xv\, F (\xv, s) \in L^2(\Omega \times \R)$, the above duality product can be written as an integral over $\Omega \times \R$. Then, reusing the Plancherel-like formula \eqref{eq:directional_Plancherel} and the relations \eqref{eq:inverse_directional_Fourier}, we obtain} 
            \big\langle\vts \widetilde{G},\, \nabla_\xv\, \scrF_\Kv\vts F \vts\big\rangle_{[\scrC^\infty_0 (\Omega \times \R_\lvar)]',\vts \scrC^\infty_0 (\Omega \times \R_\lvar)} &= \int_\R \int_\Omega \overline{\widetilde{G} (\xv, \lvar)}\, \scrF_\Kv\vts \nabla_\xv\, F (\xv, \lvar)\, d\xv d\lvar
            \\
            &=\big\langle\vts \widetilde{G},\, \scrF_\Kv\vts \nabla_\xv\, F \vts\big\rangle_{[\scrC^\infty_0 (\Omega \times \R_\lvar)]',\vts \scrC^\infty_0 (\Omega \times \R_\lvar)}.
        \end{align*}
        In other words, \eqref{eq:Fourier_commutes_with_differential_operators} holds in the sense of distributions. But since $\scrF_\Kv \nabla_\xv F \in L^2(\Omega \times \R)$, it follows that $\nabla_\xv\, \scrF_\Kv\vts F \in L^2(\Omega \times \R)$, and \eqref{eq:Fourier_commutes_with_differential_operators} holds in $L^2(\Omega \times \R)$.
    \end{dem}

    \printbibliography
\end{document}